\documentclass[11pt]{article}
\usepackage{leo}
\usepackage[margin=1in]{geometry}
\usepackage{array}

\title{Understanding Simulated Annealing: from Weak Poincaré Inequalities to a Discrete Optimal Transport Framework}
\author{%
  Yuchen He\thanks{Shanghai Jiao Tong University. \texttt{yuchen\_he@sjtu.edu.cn}.}
  \quad Tianhui Jiang\thanks{University of California San Diego. \texttt{tianhui-jiang@ucsd.edu}.}
  \quad Zhehan Lei\thanks{Shanghai Jiao Tong University. \texttt{hsfzlzh1@sjtu.edu.cn}.}\\[0.5em]
  Sihan Wang\thanks{Shanghai Jiao Tong University. \texttt{wangsihan\_leo@sjtu.edu.cn}.}
  \quad Chihao Zhang\thanks{Shanghai Jiao Tong University. \texttt{chihao@sjtu.edu.cn}.}
}
\date{}

\begin{document}

\pagenumbering{roman}
\maketitle

\thispagestyle{empty}

\begin{abstract}
Simulated annealing follows a prescribed path of distributions from an easy
initial law to a target distribution.  Standard analyses often require rapid
mixing of the Markov chain at every fixed temperature, a condition that can
fail for multimodal targets even when annealing succeeds.  Motivated by the
continuous optimal-transport analysis of annealed Langevin dynamics by Guo,
Tao, and Chen~\cite{GTC25}, we develop a discrete optimal-transport framework
for finite-state simulated annealing.

Our central quantity is the \emph{discrete action} of the prescribed
distribution path.  It integrates over the path the minimum electrical energy
required, at each time $s$, to route the demand $-\partial_s\pi_s$ through the
transition capacities of the Markov chain.  Thus action can be viewed as the
effective resistance of the moving demand, and as a path-dependent analogue of
inverse conductance.  A bottleneck makes action large only when the path must
move probability mass across it.  We also
introduce \emph{weak action}, which allows an $\varepsilon$ amount of the
prescribed motion to remain unrouted.  Up to initialization error, action
$\mathcal A$ gives reverse-KL error at most $\mathcal A/(4T)$ after $T$
expected updates, while weak action $\mathcal A^{\mathrm w}_\varepsilon$ gives
total-variation error at most
$\varepsilon+\sqrt{\mathcal A^{\mathrm w}_\varepsilon/(8T)}$.

We develop flow and spectral tools for bounding these quantities, paralleling
canonical paths and Poincar\'e inequalities in
classical MCMC theory.  In particular, pathwise Poincar\'e inequalities imply an
action bound, while pathwise weak Poincar\'e inequalities imply a weak-action bound.
We also show that, after removing mass proportional to its error term, a
finite reversible chain satisfying a weak Poincar\'e inequality has a
restriction satisfying a genuine Poincar\'e inequality.  As a pathwise
analogue of Markov chain decomposition, we obtain an initialization-aware
componentwise algorithm that can sample from mixture targets whose phase
weights are supplied before annealing.

We apply the framework to mean-field Ising and Potts models and to spin
systems on bipartite graphs.  Most notably, for every fixed $\delta>0$, we
give a polynomial-time unrestricted single-site annealing sampler for the
uniform distribution over independent sets of any $(d,D)$-biregular
bipartite graph with $D/d\geq1+\delta$ and $D$ sufficiently large.
To the best of our knowledge, this is the first polynomial-time sampler under only a fixed multiplicative degree imbalance.  It strictly extends the
$D\geq7d\log d$ sampling regime of Cannon and Perkins~\cite{CP20}.
\end{abstract}

\newpage
\setcounter{tocdepth}{1}
\tableofcontents
\newpage

\pagenumbering{arabic}

\section{Introduction}

Given a probability distribution $\pi$ on a large finite state space $\Omega$,
the sampling problem asks for a random state whose law is close to $\pi$.  A
standard approach is Markov chain Monte Carlo (MCMC): construct a local Markov
chain with stationary distribution $\pi$ and run it until its law is close to
equilibrium; see, for example,~\cite{LP17}.  Classical tools such as coupling, conductance, spectral gaps, canonical paths, and chain decompositions form a
powerful theory for deciding when this approach is efficient.

Simulated annealing takes a different approach.  Annealing-type methods have
become a common tool for sampling from multimodal and non-log-concave
distributions~\cite{Neal01,DDJ06}, and have recently attracted renewed theoretical
attention~\cite{GTC25,HMRW25,HLY26}.  Instead of running one fixed chain, we choose
a path of distributions $\tp{\pi_s}_{s\in[0,1]}$ from an easy initial law
$\pi_0$ to the target $\pi_1=\pi$, and run a time-inhomogeneous chain whose
instantaneous stationary law is $\pi_s$.  If the chain can keep up with the
changes in $\pi_s$, the output will be close to the target.

A simple sufficient condition for success is that, for every $s\in[0,1]$, the
time-homogeneous chain obtained by freezing the annealing dynamics at $s$
mixes rapidly to $\pi_s$ from every initial state.  This condition can be far
stronger than necessary.  The annealing chain is not restarted at each
temperature: it arrives with the law produced by the preceding part of the
schedule.  It therefore needs only to track the prescribed path from its
current law, rather than converge to each $\pi_s$ from every possible
initial law.

This distinction is especially important for multimodal distributions.  A
fixed-temperature chain may have an exponentially small spectral gap because
two modes communicate through a narrow bottleneck, and yet annealing can
still sample the target efficiently.  Rapid mixing at every fixed
temperature therefore does not capture what annealing actually requires.
This raises the following question:
\begin{quote}
    \emph{What quantity measures the cost for a
    local Markov chain to follow the prescribed change $\partial_s\pi_s$?}
\end{quote}

The work of Guo, Tao, and Chen~\cite{GTC25} gives a continuous-space answer to
this question.  Informally, their result has the form
\begin{equation*}
    D_{\mathrm{KL}}(\pi_1\|\nu_T)
    \leq D_{\mathrm{KL}}(\pi_0\|\nu_0)
       +\frac{1}{4T}\int_0^1\abs{\dot\pi}_s^2\,\dd s,
\end{equation*}
where $\nu_0$ and $\nu_T$ are respectively the initial and output laws of
annealed Langevin dynamics, and
$\abs{\dot\pi}_s$ is the speed of the prescribed path in Wasserstein space.
Thus the cost of annealing is governed by the motion of the target path, not
by the worst-case mixing time of every frozen diffusion.

The quantity $\abs{\dot\pi}_s^2$ has a useful physical meaning.  Through the
Benamou--Brenier formulation in optimal transport theory~\cite{BB00}, it is the
minimum kinetic energy of a velocity field that realizes the infinitesimal
change $\partial_s\pi_s$.  On a finite
state space, the available local moves form the edges of a transition graph, and a
reversible Markov chain assigns a capacity to each edge.  The corresponding
minimum kinetic energy becomes the effective resistance associated with the
infinitesimal mass redistribution $-\partial_s\pi_s$: the minimum electrical
energy of a flow that routes mass from states where $\pi_s$ decreases to states where it increases through the available moves. Equivalently,
this quantity is a path-dependent analogue of inverse conductance: a bottleneck
is charged only to the extent that the prescribed path moves probability mass
across it.  This explains the multimodal phenomenon above.  If two modes keep
their relative weights along the path, the prescribed change requires no net
transfer of mass across the bottleneck between them.  It can be routed inside
each mode, and the cost of such a routing does not depend on the bottleneck,
however small the spectral gap is.

Building on this observation, we define the discrete action of a finite-state
distribution path and develop a collection of tools for studying it and the
inhomogeneous MCMC algorithms it governs.  These tools form a pathwise
counterpart of classical Markov-chain theory: effective resistance and
conductance, canonical paths and multicommodity flows, functional inequalities,
and Markov-chain decompositions all reappear in forms adapted to a prescribed
distribution path.  This correspondence is the organizing principle of our
theory of discrete annealing.

\subsection{A structural theory of simulated annealing}
\label{subsec:intro-functional-inequalities}

We begin with the basic notions of our theory.  They are pathwise counterparts
of familiar notions from classical Markov-chain theory.  The central
distinction is directional.  Homogeneous MCMC must control every test
function, whereas annealing needs to control only the direction of the path
velocity.

\paragraph{Action and weak action.}

Fix a path of distributions $(\pi_s)_{s\in[0,1]}$ and, at each $s$, a
reversible Markov chain whose transitions specify the available local moves.
The transition capacities of these chains define a sampler-dependent discrete
$W_2$ geometry on the probability simplex.  In this geometry, the squared
speed $\abs{\dot\pi}_s^2$ is the minimum electrical energy required to route
the demand $-\partial_s\pi_s$.  Thus it is precisely the effective resistance
of this demand.  It may also be viewed as a path-dependent analogue of inverse
conductance.  Instead of charging the worst bottleneck of the chain, it charges
only the flow required to realize the prescribed change $\partial_s\pi_s$.  In
direct analogy with the continuous definition, we define the action of the
prescribed path by
\begin{equation*}
    \+A\tp{\tp{\pi_s}_{s\in[0,1]}}
    =\int_0^1\abs{\dot\pi}_s^2\,\dd s.
\end{equation*}

Exact action requires all of this prescribed motion to be routed.  For
$\varepsilon\geq0$, we write
$\+A^{\mathrm w}_\varepsilon((\pi_s)_{s\in[0,1]})$ for the minimum integrated
energy needed to route the prescribed motion when a total amount at most
$\varepsilon$ may be left unrouted.  In this sense, weak action is a pathwise
analogue of $s$-conductance~\cite{LS93}, since a small exceptional part of the
motion is not allowed to determine the annealing cost.

In \zcref{sec:discrete-action}, we develop the discrete optimal-transport
theory needed to make these definitions rigorous, including a discrete
Benamou--Brenier characterization of action.

\paragraph{Flow and routing viewpoint.}
In classical Markov-chain theory, canonical paths and multicommodity flows give
a dual flow viewpoint on conductance~\cite{SJ89,LP17}.  One routes stationary-weighted
demand between every pair of states.  Every cut must carry all demand between
its two sides, so a bound on the congestion of this all-pairs flow directly
gives a lower bound on conductance.  The requirement is global: the flow must
connect the entire state space.

We show that a similar characterization exists for action. Write $p_s=P_s-I_{\abs\Omega}$ for the rate
kernel and
$c_s(x,y)=\pi_s(x)p_s(x,y)=\pi_s(y)p_s(y,x)$ for its edge capacities.  To
compute action, we need to route only the signed demand
$-\partial_s\pi_s$.  We call an edge flow $J_s$ that routes this demand a \emph{flux}.  Its balance
condition and \emph{quadratic congestion} are
\begin{equation*}
    \partial_s\pi_s(x)+\sum_y J_s(x,y)=0,
    \qquad
    \sum_{\{x,y\}\in E}\frac{J_s(x,y)^2}{c_s(x,y)}
\end{equation*}
respectively. The first identity is the discrete continuity equation.  Minimizing the second
quantity over fluxes satisfying the first gives the squared speed
$\abs{\dot\pi}_s^2$, and hence the action after integration over $s$.

This demand involves only the mass that the path actually moves, which is
typically much less than the all-pairs demand used in a classical mixing
argument.  A low-congestion flux may therefore exist even
when every all-pairs flow has high congestion.  We develop this viewpoint in
\zcref{subsec:path-routing-action}, where canonical paths, multicommodity flows,
and rooted paths are converted directly into action and weak-action bounds.

\paragraph{Spectral viewpoint.}
The dual, or Dirichlet, principle for the squared speed is
\begin{equation*}
    \abs{\dot\pi}_s^2
    =\sup_{\+E_{p_s}(f,f)>0}
      \frac{\inner{\partial_s\log\pi_s}{f}_{\pi_s}^2}
      {\+E_{p_s}(f,f)}.
\end{equation*}
By comparison, the optimal Poincar\'e constant is
\begin{equation*}
    C_{\-{PI},s}
    =\sup_{\substack{\inner{f}{\mathbf 1}_{\pi_s}=0\\
                       \+E_{p_s}(f,f)>0}}
      \frac{\inner{f}{f}_{\pi_s}}{\+E_{p_s}(f,f)}.
\end{equation*}
Here $C_{\-{PI},s}$ is the inverse spectral gap, which is determined by the
worst direction among all mean-zero test functions.  Action instead depends on
a test function only through its projection onto the prescribed direction
$\partial_s\log\pi_s$.  A slow mode orthogonal to this direction may
make the spectral gap exponentially small while contributing nothing to
action.

For example, suppose two regions $L$ and $R$ are separated by a narrow
bottleneck and let $f=\mathbf1_L-\mathbf1_R$.  If their relative weights do
not change along the path, then
\begin{equation}
    \inner{\partial_s\log\pi_s}{f}_{\pi_s}
    =\frac{\dd}{\dd s}\tp{\pi_s(L)-\pi_s(R)}=0.
    \label{eq:intro-symmetry-orthogonality}
\end{equation}
The same $f$ can witness an exponentially small spectral gap, but it is
invisible to the path velocity.

This comparison allows a Poincar\'e inequality (PI) to control exact action using
the actual change $\partial_s\log\pi_s$ of the prescribed path.  The precise statement is given in
\zcref{prop:PI-implies-action}.

\paragraph{Weak action and weak Poincar\'e inequalities.}
Weak Poincar\'e inequalities (WPI) were recently used to prove efficient
convergence of simulated annealing from successive warm
starts~\cite{HMRW25,HLY26}.  A WPI relaxes a PI by an additive error
term, much as weak action relaxes exact action by leaving a small part of the
motion unrouted.  Accordingly, a WPI along the path controls weak action using
the actual change $\partial_s\log\pi_s$ of the prescribed path, with the WPI
error term determining the unrouted residual.  Weak action is therefore a more
general pathwise condition for the efficient convergence of simulated
annealing; see \zcref{prop:WPI-implies-weak-action}.

Comparing a WPI with a genuine PI clarifies this connection.  Huang, Mohanty,
Rajaraman, and Wu asked whether a distribution satisfying a WPI must be close
in total variation to one satisfying a genuine
PI~\cite[Problem~1.16]{HMRW25}.  For finite reversible chains, we partially
answer this question.  After removing stationary mass proportional to the WPI
error and restricting the original dynamics to the remaining states, the
restricted distribution satisfies a genuine PI; see \zcref{thm:poincare-core}.
Thus a WPI describes a well-connected core together with a small exceptional
part.  This exceptional mass is the static counterpart of the motion that weak
action leaves unrouted, which explains why a WPI leads to weak action.

Combined with the annealing theorem below, these structural results translate
directly into quantitative convergence guarantees for the inhomogeneous
Markov chain; see \zcref{thm:intro-annealing-bound}.

\subsection{Algorithmic results: one path and components}
\label{subsec:intro-componentwise-main-result}
In classical MCMC, geometric and algebraic quantities such as conductance and
the spectral gap yield fast-mixing guarantees.  Action plays the corresponding
role for inhomogeneous MCMC, since a bound on the action of the prescribed path
gives a convergence guarantee for annealing.

\paragraph{One-path annealing.}

Fix a prescribed annealing path $(\pi_s)_{s\in[0,1]}$ and, for every $s$, a
transition matrix $P_s$ reversible with respect to $\pi_s$.  Given a running
time $T$, the following algorithm traverses this path at speed $1/T$.

\begin{algorithm}[H]
\caption{Finite-state annealing}
\label{algo:poisson-annealing}
\begin{algorithmic}[1]
    \Require Annealing path $\tp{\pi_s}_{s\in[0,1]}$, initial law
    $\pi_0^{\-{\!{ALG}}}$, transition matrices $\tp{P_s}_{s\in[0,1]}$ with
    $P_s$ reversible with respect to $\pi_s$, and running time $T$.
    \State Sample $X\sim\pi_0^{\-{\!{ALG}}}$
    \State Run a unit-rate Poisson clock on $[0,T]$
    \For{each ring time $t$}
        \State Sample $X\sim P_{t/T}(X,\cdot)$
    \EndFor
    \State \Return $X$
\end{algorithmic}
\end{algorithm}

The expected number of updates is $T$.  If $\nu_t$ denotes the law of the
algorithm at time $t$, the goal is to keep $\nu_t$ close to $\pi_{t/T}$ and in
particular to make $\nu_T$ close to $\pi_1$.

\begin{theorem}[Action and weak action control annealing, informal]
    \label{thm:intro-annealing-bound}
    If the algorithm starts from $\pi_0$ and $\nu_T$ is its output law, then
    exact action gives
    \begin{equation*}
        D_{\mathrm{KL}}(\pi_1\|\nu_T)
        \leq \frac{\+A((\pi_s)_{s\in[0,1]})}{4T}.
    \end{equation*}
    For every $\varepsilon\geq0$, weak action gives
    \begin{equation*}
        D_{\mathrm{TV}}(\pi_1,\nu_T)
        \leq
        \varepsilon+
        \sqrt{\frac{\+A^{\mathrm w}_\varepsilon
        ((\pi_s)_{s\in[0,1]})}{8T}}.
    \end{equation*}
\end{theorem}

The formal theorem allows an arbitrary initial law and adds its reverse-KL
error to the action term; see \zcref{thm:annealing-ub}.

\paragraph{Componentwise annealing.}
Weak action may still be large when the target is a mixture of well-separated
components whose weights change along the path.  To handle such targets, our
second algorithmic contribution is a decomposition theorem for annealing,
analogous to the theory of decomposable Markov chains~\cite{JSTV04}.  In the
classical setting, suppose the state space is partitioned into components.  A
\emph{restriction chain} runs the original dynamics inside one component and
rejects every transition that would leave it.  A \emph{projection chain}
collapses each component to a single state and records the stationary flow
between components.  Classical decomposition theorems bound the mixing of the
full chain by combining mixing within the restriction chains with mixing of
the projection chain.

Our annealing theorem gives a pathwise version of this decomposition.  Suppose
the target admits a representation
\begin{equation*}
    \pi=\sum_{i=1}^k w_i\pi_i.
\end{equation*}
We choose phase domains $\Omega_i\subseteq\Omega$ and a reference path
$(\rho_s)_{s\in[0,1]}$ such that $\rho_1(\,\cdot\mid\Omega_i)=\pi_i$ for every
$i$, and let the transition matrices $P_s$ be reversible with respect to
$\rho_s$.  The component paths are
\begin{equation*}
    \mu_{i,s}=\rho_s(\,\cdot\mid\Omega_i).
\end{equation*}
In the simplest case, the domains partition $\Omega$ and $\rho_1=\pi$, so that
$w_i=\pi(\Omega_i)$ and $\pi_i=\pi(\,\cdot\mid\Omega_i)$.

Consider any path of the form $\bar\mu_s=\sum_iw_i(s)\mu_{i,s}$ that ends at
$\pi$, that is, with time-dependent weights satisfying $w(1)=w$.  In the
simplest case, the reference path itself has this form with
$w_i(s)=\rho_s(\Omega_i)$.  The derivative of such a path is
\begin{equation}
    \partial_s\bar\mu_s
    =\underbrace{\sum_{i=1}^k w_i(s)\partial_s\mu_{i,s}}
       _{\text{motion within components}}
     +\underbrace{\sum_{i=1}^k\dot w_i(s)\mu_{i,s}}
       _{\text{motion of component weights}}.
    \label{eq:intro-component-source-decomposition}
\end{equation}
The first term changes the conditional laws within components and is the
analogue of motion in the restriction chains.  The second term changes the
component weights.  It is the analogue of projection motion and can be very
expensive when the components are separated by low-capacity cuts.

Our algorithmic response differs from the classical one.  Instead of proving
that a projection chain mixes rapidly, we aim to follow the fixed-weight path
$\sum_iw_i\mu_{i,s}$, which ends at $\pi$ and has no projection motion.
Ideally, one would draw a phase label $I\sim w$, start from $\mu_{I,0}$, and
run the dynamics restricted to $\Omega_I$, which rejects every transition
leaving the phase.  This ideal procedure needs the label $I$ and a membership
test for $\Omega_I$.  Our algorithm uses neither.  Given approximate weights
$\widehat w_1,\ldots,\widehat w_k$ and component initializers
$\kappa_1,\ldots,\kappa_k$, it runs \zcref{algo:poisson-annealing} with the
ordinary dynamics $P_s$, started from
\begin{equation*}
    \kappa^{\-{ALG}}=\sum_{i=1}^k\widehat w_i\kappa_i.
\end{equation*}
In the analysis, we view this initial law as drawing $I\sim\widehat w$ and
then $X_0\sim\kappa_I$, and compare each phase with the ideal restricted
procedure.  The algorithm never observes $I$, so the initializer may be a
black box, for example a sampler built from data.

For each component, let $A_i$ bound the weak action, with residual budget
$\varepsilon_i$, of the component path $(\mu_{i,s})_{s\in[0,1]}$ under the
dynamics restricted to $\Omega_i$.  Let $d_i=D_{\mathrm{TV}}(\mu_{i,0},\kappa_i)$
be the initialization error, and let
\begin{equation*}
    \Lambda_i=\int_0^1\sum_{x\in\Omega_i}\sum_{y\notin\Omega_i}
      \mu_{i,s}(x)p_s(x,y)\,\dd s
\end{equation*}
be the integrated rate at which the component path flows out of $\Omega_i$
under the unrestricted dynamics.  Set
\begin{equation*}
    \delta_w=D_{\mathrm{TV}}(w,\widehat w),
    \qquad m_i=\min\{w_i,\widehat w_i\}.
\end{equation*}

\begin{theorem}[Initialization-aware componentwise annealing, informal]
    \label{thm:intro-componentwise-annealing}
    The output law of the unrestricted annealing algorithm satisfies
    \begin{equation}
        D_{\mathrm{TV}}(\pi,\nu_T)
        \leq
        \delta_w+
        \sum_{i=1}^k m_i\min\left\{1,
          d_i+\varepsilon_i+T\Lambda_i
          +\sqrt{\frac{A_i}{8T}}
        \right\}.
        \label{eq:intro-componentwise-annealing}
    \end{equation}
\end{theorem}

The terms have direct meanings.  The quantity $\delta_w$ is the error in the installed
component weights, $d_i$ is the within-component initialization error,
$\varepsilon_i$ and $A_i$ control tracking inside the component, and
$T\Lambda_i$ accounts for the unrestricted dynamics leaving its phase domain.
The crucial absence is a term for dynamically transporting the target weights
between components.  Initialization has already supplied that slow
coordinate.  The formal theorem is
\zcref{thm:componentwise-weak-action-annealing}.

This viewpoint complements data-initialized sampling.  Existing fixed-target
results use data or advice to initialize a mixture whose components satisfy
Poincar\'e inequalities~\cite{KLV25,HMRW25}.  Taking the initializer as given,
our theorem controls the subsequent annealing dynamics through
component actions along time-dependent paths, together with the leakage term
$T\Lambda_i$.

The structural and algorithmic parallels developed above are summarized in
the following table.

\begin{table}[H]
    \centering
    \small
    \renewcommand{\arraystretch}{1.13}
    \begin{tabular}{@{}
        >{\raggedright\arraybackslash}p{0.20\linewidth}
        >{\raggedright\arraybackslash}p{0.34\linewidth}
        >{\raggedright\arraybackslash}p{0.38\linewidth}@{}}
        \toprule
        \textbf{Feature} & \textbf{Homogeneous MCMC} &
        \textbf{Inhomogeneous MCMC / annealing} \\
        \midrule
        Basic object & A fixed reversible pair $(\pi,P)$ &
        A path of reversible pairs $(\pi_s,P_s)$ \\
        Geometric quantity & Inverse conductance and effective resistance &
        Action of the prescribed demand $-\partial_s\pi_s$ \\
        Small exceptional part & $s$-conductance and WPI &
        Weak action \\
        Flow viewpoint & Canonical paths and multicommodity flows routing
        all-pairs demand & A flux routing only $-\partial_s\pi_s$ \\
        Spectral viewpoint & The inverse gap controls every mean-zero test
        function & Action sees only the projection onto
        $\partial_s\log\pi_s$ \\
        Algorithmic consequence & Fast mixing of one fixed chain &
        Convergence while tracking one prescribed path \\
        Decomposition & Restriction chains and a projection chain &
        Component actions, with component weights supplied at initialization \\
        \bottomrule
    \end{tabular}
    \caption{A dictionary between classical homogeneous MCMC and the
    inhomogeneous MCMC theory developed here.}
    \label{tab:intro-mcmc-dictionary}
\end{table}

\subsection{Applications: three path geometries}
\label{subsec:intro-applications}

We apply our theory to two representative classes of multimodal
distributions, namely mean-field spin systems and spin systems on bipartite
graphs.  In both classes,
we obtain convergence guarantees for annealing in regimes where the
fixed-temperature dynamics can have an exponentially small spectral gap.
Rather than organize these results by model or parameter regime, we classify
them by the geometry of the distribution path.  Three geometries recur.
Either one phase dominates and the remaining mass is negligible, or symmetry
fixes or approximately fixes the phase weights, or nontrivial phase weights
are computed and installed before annealing.

\paragraph{One dominant phase and a negligible remainder.}
\begin{figure}[H]
    \centering
    \begin{tikzpicture}[
        dominant/.style={circle,draw=blue!70!black,very thick,fill=blue!13,
            minimum size=1.65cm,align=center},
        minor/.style={circle,draw=gray!70,thick,fill=gray!12,
            minimum size=3.6mm,inner sep=0pt}
    ]
        \node[dominant] at (-2.00,0) {dominant\\mode};
        \foreach \x/\y in {
            0.20/0.42,0.68/0.58,1.14/0.34,1.58/0.55,
            0.38/-0.18,0.88/-0.38,1.36/-0.12,1.78/-0.34}
            \node[minor] at (\x,\y) {};
        \draw[gray!70,densely dashed,rounded corners=5pt]
            (-0.08,-0.66) rectangle (2.08,0.82);
        \node[gray!65!black,align=center] at (1.00,-1.23)
            {negligible modes\\total mass $\leq\varepsilon$};
    \end{tikzpicture}
    \caption{One dominant mode carries all but an $\varepsilon$ fraction of
    the mass.  Weak action transports the dominant part and leaves the small
    remainder unrouted.}
    \label{fig:intro-one-dominant-geometry}
\end{figure}

Consider the uniform hard-core model on a $(d,D)$-biregular bipartite graph
on $N$ vertices.
When $D/d\geq1+\delta$, the larger side of the graph determines one dominant
phase.  Along a suitable annealing path, all but a negligible amount of mass
lies in a core around this phase.  The restricted dynamics satisfies a PI on
the core, while the small complement is left unrouted by weak action.

\begin{theorem}[Annealing on unbalanced biregular graphs, informal]
    \label{thm:intro-unbalanced-biregular}
    For every fixed $\delta>0$, there is $D_0=D_0(\delta)$ such that the
    following holds.  Let $G$ be any finite simple $(d,D)$-biregular bipartite
    graph with $D/d\geq1+\delta$ and $D\geq D_0$.  For every accuracy
    $\tau\in(0,1)$, unrestricted single-site annealing samples the uniform
    distribution over independent sets of $G$ to total-variation error at
    most $\tau$ using $\operatorname{poly}_\delta(N,1/\tau)$ expected updates.
\end{theorem}

The formal statement is \zcref{thm:fixed-gap-sampling}.  In the same uniform
biregular setting, the sampling result of Cannon and Perkins~\cite{CP20}
requires $D\geq7d\log d$.  Our condition requires only a fixed multiplicative
imbalance, $D/d\geq1+\delta$, once $D$ is sufficiently large.  It therefore
strictly extends the parameter range covered by their sampling theorem.  

The same geometry appears in mean-field Potts, but only after using its color
symmetry.  The Hamiltonian is invariant under permutations of the $q$ colors,
so the $q$ ordered modes are symmetry copies.  We fold these copies together
by sorting the empirical color counts.  Away from the coexistence inverse
temperature $\beta_{\mathsf c}$, the folded landscape has one dominant well.
In the two metastability windows on either side of $\beta_{\mathsf c}$, it
also has an exponentially smaller losing well.  There we choose the direction
of annealing for which the losing well remains small and leave it unrouted.
Outside these windows, a PI or an exact-action argument applies directly.  We use the $(q-1)$-block heat-bath chain, which resamples
$q-1$ uniformly random sites at a time from their conditional Gibbs law.

\begin{theorem}[Mean-field Potts annealing away from coexistence, informal]
    \label{thm:intro-potts-annealing}
    Fix $q\geq3$ and an inverse temperature
    $\beta\neq\beta_{\mathsf c}$.  For every inverse-polynomial accuracy
    $\tau$, an appropriate forward or reverse annealing schedule using the
    unrestricted $(q-1)$-block heat-bath chain produces a sample with
    total-variation error at most $\tau$ in a polynomial expected number of
    block updates.
\end{theorem}

The formal statement, including the accuracy-dependent window around
coexistence, is \zcref{thm:potts-noncritical-annealing}.

\paragraph{Symmetry and approximate balance.}
\begin{figure}[H]
    \centering
    \begin{tikzpicture}[
        leftmode/.style={circle,draw=blue!70!black,thick,fill=blue!13,
            minimum size=7.5mm,inner sep=0pt},
        rightmode/.style={circle,draw=orange!80!black,thick,fill=orange!18,
            minimum size=7.5mm,inner sep=0pt}
    ]
        \foreach \x/\t in {-3/{s=0},0/{s=1/2},3/{s=1}} {
            \node at (\x,0.72) {$\t$};
            \node[leftmode] at (\x-0.55,0) {$L$};
            \node[rightmode] at (\x+0.55,0) {$R$};
            \node at (\x-0.55,-0.62) {$1/2$};
            \node at (\x+0.55,-0.62) {$1/2$};
        }
        \draw[-{Stealth[length=2mm]},thick] (-2.00,0) -- (-1.00,0);
        \draw[-{Stealth[length=2mm]},thick] (1.00,0) -- (2.00,0);
    \end{tikzpicture}
    \caption{The two phase laws evolve along the annealing path, while
    symmetry keeps their weights equal to $1/2$.  Approximate balance allows a
    small error in these weights.}
    \label{fig:intro-balanced-geometry}
\end{figure}

Mean-field Ising has symmetry at the level of the entire state space.  Vertex
permutations reduce a configuration to its magnetization, and global spin flip
identifies magnetizations $M$ and $-M$.  We may therefore analyze the chain on
the one-dimensional variable $|M|$.  Canonical paths on this reduced space
give polynomial exact action at every inverse temperature, even when the
original fixed-temperature chain has exponentially slow worst-case
mixing~\cite{LLP10}.

\begin{theorem}[All-temperature mean-field Ising annealing, informal]
    \label{thm:intro-mean-field-ising-annealing}
    For the Ising model on the complete graph with $n$ vertices at every
    inverse temperature
    $\beta\geq0$, annealed single-site heat bath started from the uniform
    distribution achieves reverse-KL error at most $\zeta$ using
    $O(1+n^5\beta^2/\zeta)$ expected updates.
\end{theorem}

The formal statement is \zcref{thm:mean-field-ising-annealing}.

For spin systems on bipartite graphs, we generally do not have such a complete
reduction of the state space.  What may be known instead is only that the two
principal phase weights are equal, or approximately equal, to $1/2$.  This is
exactly the information needed by componentwise annealing: initialize the two
phases with probabilities $(1/2,1/2)$ and then control the action inside each
phase separately.  Approximate balance contributes only the error in these
initial weights.

\begin{theorem}[Annealing from balanced phase weights, informal]
    \label{thm:intro-balanced-bipartite-annealing}
    For the high-fugacity bipartite hard-core models considered in
    \zcref{sec:weak-action-two-spin}, suppose the two phase weights are equal,
    as guaranteed by an automorphism exchanging the two sides, or are
    sufficiently close to $(1/2,1/2)$, as guaranteed for regular graphs of
    sufficiently large girth.  Initializing the two phases with
    equal probability and then running unrestricted single-site annealing gives
    total-variation error at most $\tau$ in a polynomial expected number of
    updates, for every inverse-polynomial accuracy $\tau$.
\end{theorem}

The formal statements are
\zcref{thm:symmetric-bis-weak-action,thm:approximate-balance-annealing}.

\paragraph{Asymmetric phases with computed weights.}
\begin{figure}[H]
    \centering
    \begin{tikzpicture}[
        mode/.style={circle,thick,inner sep=0pt,font=\scriptsize}
    ]
        \node at (-3,0.95) {$s=0$};
        \node[mode,draw=blue!70!black,fill=blue!13,minimum size=12mm]
            at (-3.65,0) {$L$};
        \node[mode,draw=orange!80!black,fill=orange!18,minimum size=6mm]
            at (-2.35,0) {$R$};

        \node at (0,0.95) {$s=1/2$};
        \node[mode,draw=blue!70!black,fill=blue!13,minimum size=9mm]
            at (-0.65,0) {$L$};
        \node[mode,draw=orange!80!black,fill=orange!18,minimum size=9mm]
            at (0.65,0) {$R$};

        \node at (3,0.95) {$s=1$};
        \node[mode,draw=blue!70!black,fill=blue!13,minimum size=7mm]
            at (2.35,0) {$L$};
        \node[mode,draw=orange!80!black,fill=orange!18,minimum size=11mm]
            at (3.65,0) {$R$};

        \draw[-{Stealth[length=2mm]},thick] (-1.92,0) -- (-1.23,0);
        \draw[-{Stealth[length=2mm]},thick] (1.23,0) -- (1.87,0);
        \node at (0,-1.05) {the phase weights vary along the Gibbs path};
    \end{tikzpicture}
    \caption{Without symmetry, the relative phase weights can change along
    the ambient Gibbs path.  The target weights are computed and
    installed at initialization.}
    \label{fig:intro-varying-weights-geometry}
\end{figure}

In the third geometry, both phases have nonnegligible weight, but no symmetry
determines their relative sizes.  The annealing chain should not be asked to
discover these weights by moving mass through the low-capacity cut between the
phases.  Instead, the target phase weights must be supplied before annealing.

For the high-fugacity hard-core model on a bipartite expander, the two phase
partition functions from the polymer framework of Jenssen, Keevash, and
Perkins~\cite{JKP20} can be approximated separately.  Normalizing these
estimates approximates the coefficients of a phase mixture that is close to
the Gibbs target.  We use these approximate weights to choose the initial phase, discard the phase
label, and then run the ordinary unrestricted single-site dynamics.  The
componentwise theorem controls the error relative to this mixture.  Its small
distance from the Gibbs target contributes an additional TV error.

\begin{theorem}[Annealing with computed phase weights, informal]
    \label{thm:intro-computed-phase-weights}
    Fix the maximum degree and expansion parameters, and consider the
    hard-core model on a bipartite expander at a sufficiently high target
    fugacity.  Suppose the phase weights are estimated using the
    approximations of~\cite{JKP20}, and the two phases are initialized with
    these probabilities.  Then unrestricted single-site annealing produces a
    sample with total-variation error at most any inverse-polynomial $\tau$
    in a polynomial expected number of updates, and the whole algorithm runs
    in polynomial time.
\end{theorem}

The formal statement is \zcref{thm:computed-phase-weight-annealing}.
Approximate counting and sampling for the hard-core model on bipartite
expanders have been developed through polymer and container
methods~\cite{JKP20,CGGPSV21,JMP26}, and recent algorithms cover all
fugacities on regular bipartite graphs with suitable spectral expansion and
on random regular bipartite graphs~\cite{NP26,KOGT26}.  The point of this
application is that the ordinary unrestricted single-site
dynamics suffices once the phase weights are installed, so the chain never
needs to be restricted to one phase.
More generally, the phase weights may also be learned from data or supplied
as advice, as discussed after \zcref{thm:intro-componentwise-annealing}.

\subsection{Related work and organization}

\paragraph{Related work.}
The use of action to analyze annealing is inspired by the continuous-space
result of Guo, Tao, and Chen~\cite{GTC25}.  Action has also been used to bound
the complexity of normalizing-constant estimation by annealed importance
sampling~\cite{GTC26}.  Discrete transport metrics adapted
to reversible Markov chains were developed by Maas~\cite{Maas11},
Mielke~\cite{Mielke11}, and Chow et al.~\cite{CHLZ12}.  Those works emphasize
the gradient-flow geometry of Markov evolution.  Here the capacities are
chosen to match the time-inhomogeneous sampler, and the central object is the
cost of following one prescribed target path.

Electrical networks, canonical paths, and multicommodity flows are classical
tools for fixed-chain relaxation and mixing~\cite{SJ89,LP17}, while classical
chain decomposition separates restriction and projection chains~\cite{JSTV04}.
Weak Poincar\'e inequalities have long been used to study subgeometric
relaxation~\cite{RW01}, and have recently yielded sampling guarantees for
spin systems~\cite{AEGP26,HMRW25,DLSS26,HLY26}.  Our theory gives pathwise versions of these
ideas.  Flows route only the prescribed derivative, PI and WPI imply action and
weak-action bounds, and componentwise annealing uses initialization instead of
dynamically mixing the phase weights.  Related uses of initialization to
handle multimodality appear in data-initialized sampling~\cite{KLV25,HMRW25}
and low-temperature Ising dynamics~\cite{GS23}.  Chen and Liu~\cite{CL26}
also treat the component weights separately from the transport along
componentwise annealing paths, which they learn as transport maps.

\paragraph{Organization.}
\zcref{sec:preliminaries} provides preliminaries, and \zcref{sec:discrete-action}
develops the discrete transport geometry.  \zcref{section:annealing,
sec:componentwise-mixture-annealing} prove the one-path and componentwise
annealing theorems, respectively.  \zcref{sec:functional-inequalities-and-action}
develops the spectral and flow tools.  Finally,
\zcref{sec:mean-field-phase-transitions,sec:weak-action-two-spin} give the
mean-field and bipartite-spin-system applications.

\section{Preliminaries}\label{sec:preliminaries}
\subsection{Notations}
\paragraph{Distributions and functions.}
Throughout this paper, we consider a finite discrete state space $\Omega$. We define the probability simplex over $\Omega$ as
$$
    \+P(\Omega) = \left\{ \mu \in \mathbb{R}^{|\Omega|} \cmid \mu(x) \ge 0 \text{ for all } x \in \Omega, \sum_{x \in \Omega} \mu(x) = 1 \right\}.
$$
Its relative interior is
$$    \+P_{>0}(\Omega) = \left\{ \mu \in \+P\tp{\Omega}: \mu(x) > 0 \text{ for all } x \in \Omega\right\}.
$$

For $\pi\in\+P_{>0}\tp{\Omega}$, let $L^2\tp{\pi}$ denote the space of
functions $f:\Omega\to\bb R$ satisfying
$\sum_{x\in\Omega}\pi\tp{x}f\tp{x}^2<\infty$.  Since $\Omega$ is finite,
this space is simply $\bb R^\Omega$.  For $f,g\in L^2\tp{\pi}$, define
$$
    \E[\pi]{f}\defeq\sum_{x\in\Omega}\pi\tp{x}f\tp{x},
    \qquad
    \inner{f}{g}_{\pi}
    \defeq\E[\pi]{fg},
    \qquad
    \norm{f}_{L^2\tp{\pi}}
    \defeq\sqrt{\inner{f}{f}_{\pi}},
    \qquad
    \Var[\pi]{f}\defeq\E[\pi]{f^2}-\tp{\E[\pi]{f}}^2.
$$

For a subset $A\subseteq\Omega$, we write $\bar A\defeq\Omega\setminus A$ for its complement.

\paragraph{Divergences.}
For $\mu,\nu\in\+P\tp{\Omega}$, we say that $\mu$ is absolutely continuous with respect to $\nu$, denoted as $\mu \ll \nu$, if $\nu\tp{x} = 0$ implies $\mu\tp{x} = 0$ for every $x\in\Omega$. For $\mu \ll \nu$, the Kullback--Leibler (KL) divergence from $\mu$ to $\nu$ is defined as
\[
    D_{\mathrm{KL}}\tp{\mu\|\nu}
    \defeq\sum_{x:\,\mu\tp{x}>0}
      \mu\tp{x}\log\frac{\mu\tp{x}}{\nu\tp{x}},
\]
and we set $D_{\mathrm{KL}}\tp{\mu\|\nu}=+\infty$ when $\mu\not\ll\nu$.
We will use the Donsker--Varadhan variational formula
\cite{DZ98}: for distributions $\mu,\nu$ on a finite set $\Omega$,
\begin{equation}
    D_{\mathrm{KL}}\tp{\mu\|\nu}
    =\sup_{h:\Omega\to\bb R}\left(\sum_{x}\mu\tp{x}h\tp{x}
    -\log\sum_{x}\nu\tp{x}e^{h\tp{x}}\right).
    \label{eq:donsker-varadhan}
\end{equation}
The total variation distance on $\Omega$ is
\[
    D_{\mathrm{TV}}\tp{\mu,\nu}
    \defeq\frac12\sum_x\abs{\mu\tp{x}-\nu\tp{x}}
    =\sup_A\abs{\mu\tp{A}-\nu\tp{A}}.
\]

\paragraph{Poisson clocks.}
Our annealing algorithms perform their updates at the rings of a Poisson clock, which we now define. We denote the Poisson distribution with mean $\lambda$ by $\!{Pois}\tp{\lambda}$. An \emph{inhomogeneous Poisson process} with a non-negative, locally integrable intensity $\tp{\lambda_t}_{t\geq 0}$ is a continuous-time counting process with independent increments such that the number of events in any interval $(s,r]$ follows $\!{Pois}\tp{\int_s^r\lambda_t\,\dd t}$. We also refer to the event times of this process as the rings of a time-varying \emph{Poisson clock} with intensity $\tp{\lambda_t}_{t\geq 0}$. A \emph{unit-rate Poisson clock} is a Poisson clock with constant intensity $\lambda_t\equiv1$; in particular, the number of its rings in an interval of length $T$ follows $\!{Pois}\tp{T}$.

\paragraph{Absolutely continuous curves.}
Let $L^1\tp{[0,T]}$ denote the space of measurable functions
$m:[0,T]\to\bb R$ such that $\int_0^T\abs{m\tp{t}}\,\dd t<\infty$.
Let $\tp{M, d}$ be a metric space. A curve $\tp{\gamma_t}_{t \in [0, T]} \subseteq M$ is called \emph{absolutely continuous} if there exists a nonnegative function $m \in L^1\tp{[0, T]}$ such that
\[
    d\tp{\gamma_s, \gamma_t} \le \int_s^t m\tp{r} \, \dd r, \quad \text{for all } 0 \le s \le t \le T.
\]
For an absolutely continuous curve, the \emph{metric derivative} $\abs{\dot{\gamma}}_t \defeq \lim_{h \to 0,\, t+h\in[0,T]} d\tp{\gamma_{t+h}, \gamma_t} / \abs{h}$ exists for almost every $t \in [0, T]$ and satisfies $\abs{\dot{\gamma}}_t \le m\tp{t}$ a.e.; moreover, $\abs{\dot{\gamma}}$ is the smallest such $m$ in $L^1$ (see, e.g., \cite[Theorem 1.1.2]{AGS05}).

\paragraph{Graphs, norms, and matrices.}
For an undirected, connected graph $G = \tp{\Omega, E}$, we denote an edge between $x$ and $y$ as $e = \set{x, y} \in E$.

Without specific mention, the norm \(\norm{\cdot}\) is the standard Euclidean norm. For a function $f:\Omega\to\bb R$, its $L^1$ and $L^\infty$ norms are defined by $\norm{f}_1\defeq\sum_{x\in\Omega}\abs{f\tp{x}}$ and $\norm{f}_\infty\defeq\max_{x\in\Omega}\abs{f\tp{x}}$.

For a finite set $S$, we write $\!I_S$ for the identity matrix indexed by
$S$.  When only its dimension is relevant, we write $\!I_{\abs{S}}$.

\subsection{Markov Processes, Reversibility, and Capacities}
\label{subsec:prelim-markov-processes}

In this paper, we study continuous-time inhomogeneous Markov processes on a discrete state space $\Omega$. Such a process $\tp{X_t}_{t \ge 0}$ is driven by a time-varying family of transition rate kernels, denoted by $\tp{p_t}_{t \ge 0}$. For any time $t$, the kernel $p_t \in \bb R^{\abs{\Omega} \times \abs{\Omega}}$ satisfies $p_t\tp{x, y} \ge 0$ for $x \neq y$, and the zero-row-sum condition $\sum_{y \in \Omega} p_t\tp{x, y} = 0$ for all $x \in \Omega$. 
We define the underlying graph of the transition rate kernel $p_t$ as $\tp{\Omega, E_t}$, where $E_t = \set{\set{x, y} \subseteq \Omega: p_t\tp{x, y} > 0 \text{ or } p_t\tp{y, x} > 0}$.

A time-inhomogeneous Markov process can also be described by transition rates $\tp{\lambda_t}_{t \ge 0}$ and transition matrices $\tp{P_t}_{t \ge 0}$, where each $P_t$ is a stochastic matrix, i.e., its entries are nonnegative and each of its rows sums to one. Suppose the process is at state $x$ at time $t$. Transition events are triggered by an inhomogeneous Poisson process with rate $\lambda_t$. Whenever the Poisson clock rings at time $t$, the process transitions to a target state $y$ sampled from the probability distribution $P_t\tp{x, \cdot}$.

Therefore, the transition rate kernel can be written as
\[
    p_t \defeq \lambda_t \tp{P_t - \!I_{\abs{\Omega}}}.
\]
Note that the matrix $P_t$ may contain strictly positive diagonal entries $P_t\tp{x, x} > 0$. Consequently, while the Poisson clock triggers transition attempts at a total rate of $\lambda_t$, the true exit rate from state $x$ at time $t$ is $\sum_{y \neq x} p_t\tp{x, y} = \lambda_t\tp{1 - P_t\tp{x, x}} \le \lambda_t$.

The time evolution of the marginal distributions $\tp{\mu_t}_{t \ge 0}$ of the process is governed by the \emph{Kolmogorov forward equation} (also called the discrete Fokker--Planck equation):
\begin{equation*}
    \partial_t \mu_t\tp{x} = \sum_{y \in \Omega} \mu_t\tp{y} p_t\tp{y, x} =\sum_{y \in \Omega} \tp{\mu_t\tp{y} p_t\tp{y, x} - \mu_t\tp{x} p_t\tp{x, y}} , \quad \forall x \in \Omega,
\end{equation*}
where the second equality follows from the zero-row-sum condition. In matrix notation, this reads $\partial_t \mu_t = \mu_t p_t$, where $\mu_t$ is viewed as a row vector.
Dually, for $T>0$ and $f:\Omega\to\bb R$, the function $u_t\tp{x}\defeq\E{f\tp{X_T}\mid X_t=x}$ satisfies, for almost every $t\in[0,T]$ (when $t\mapsto p_t$ is measurable and locally bounded), the \emph{Kolmogorov backward equation}
\[
    \partial_t u_t\tp{x}+\sum_{y\in\Omega}p_t\tp{x,y}\tp{u_t\tp{y}-u_t\tp{x}}=0, \quad \forall x\in\Omega, \qquad u_T=f.
\]
For a time-homogeneous Markov process, the transition rate kernel $p_t$ is independent of time, and we simply denote it by $p$. A probability distribution $\pi \in \+P\tp{\Omega}$ is called a stationary distribution for $p$ if $\sum_{x \in \Omega} \pi\tp{x} p\tp{x, y} = 0$ for all $y \in \Omega$. The transition rate kernel $p$ is said to be reversible with respect to $\pi$ if it satisfies the detailed balance condition:
\[
    \pi\tp{x} p\tp{x, y} = \pi\tp{y} p\tp{y, x}, \quad \text{for all } x, y \in \Omega.
\]
Likewise, a transition matrix $P$ is said to be reversible with respect to $\pi$ if $\pi\tp{x}P\tp{x,y}=\pi\tp{y}P\tp{y,x}$ for all $x,y\in\Omega$. Since $\lambda\tp{P-\!I_{\abs{\Omega}}}$ coincides with $\lambda P$ off the diagonal, for $\lambda>0$ the matrix $P$ is reversible with respect to $\pi$ if and only if the kernel $\lambda\tp{P-\!I_{\abs{\Omega}}}$ is.
When $p$ is reversible with respect to $\pi$, we define the symmetric \emph{edge capacity} $c\tp{x,y}\defeq\pi\tp{x}p\tp{x,y}$ for $x\neq y$, and set $c\tp{x,x}=0$.
We say that $p$ is \emph{irreducible} if for all $x,y\in\Omega$ there exist $x=x_0,x_1,\dots,x_k=y$ such that $p\tp{x_{i-1},x_i}>0$ for every $1\le i\le k$. If $p$ is reversible with respect to some $\pi\in\+P_{>0}\tp{\Omega}$, then $p\tp{x,y}>0$ if and only if $p\tp{y,x}>0$, so $p$ is irreducible if and only if its underlying graph is connected.

For a real-valued function $f: \Omega \to \bb R$, the \emph{infinitesimal generator} associated with the transition rate kernel $p$ is the operator $\+L^p$ defined by
\[
    \+L^p f\tp{x} \defeq \sum_{y \in \Omega} p\tp{x, y} \tp{f\tp{y} - f\tp{x}}, \quad \forall x \in \Omega.
\]
For fixed $\pi\in\+P_{>0}\tp{\Omega}$, the associated \emph{Dirichlet form} $\+E_p\tp{f, g}$ for two functions $f, g: \Omega \to \bb R$ is given by $\+E_p\tp{f, g} \defeq \inner{f}{-\+L^p g}_{\pi}$.

\begin{proposition}
    \label{prop:generator-and-dirichlet-form-of-discrete-chains}
    Suppose that the transition rate kernel $p$ is reversible with respect to $\pi$. Then the Dirichlet form satisfies
    \[
        \+E_p\tp{f, g}
        = \frac{1}{2} \sum_{x, y \in \Omega}
        \tp{f\tp{y} - f\tp{x}}
        \tp{g\tp{y} - g\tp{x}}
        \pi\tp{x} p\tp{x, y}.
    \]
\end{proposition}
This identity follows from a summation by parts, see, e.g., \cite{LP17}.

In particular,
\[
    \+E_p\tp{f,f}
    =\frac12\sum_{x,y\in\Omega}
      c\tp{x,y}\tp{f\tp{y}-f\tp{x}}^2.
\]
Thus the Dirichlet form $\+E_p\tp{f,f}$ measures how much $f$ varies across the available
transitions, with each transition weighted by its edge capacity.

\subsection{Functional Inequalities}
\label{subsec:prelim-functional-inequalities}

The \Poincare inequality is a fundamental tool for bounding the mixing time of reversible Markov chains. Suppose that $p$ is reversible with respect to $\pi$. The \Poincare constant $C_{\-{PI}}\tp{p,\pi}$ is defined as the smallest constant such that for all $f: \Omega \to \bb R$,
\[
    \Var[\pi]{f} \le C_{\-{PI}}\tp{p,\pi}\cdot \+E_p\tp{f, f}.
\]
We write $C_{\-{PI}}$ when $\tp{p,\pi}$ is clear from context.

Suppose that $p$ is
irreducible and reversible with respect to $\pi$.  Then $-\+L^p$ is a
positive semidefinite, self-adjoint operator on $L^2\tp{\pi}$.  Its smallest
eigenvalue is zero, corresponding to the constant functions.  The
\emph{spectral gap of the continuous-time Markov chain with rate kernel $p$}, denoted as $\operatorname{gap}\tp{p}$,
is the smallest nonzero eigenvalue of $-\+L^p$.  We have
\[
    \operatorname{gap}\tp{p}
    =\inf_{f:\,\Var[\pi]{f}>0}
      \frac{\+E_p\tp{f,f}}{\Var[\pi]{f}}.
\]
This is the Rayleigh-quotient characterization of the spectral gap; see,
e.g., \cite{LP17}.  Comparing with the definition above,
$C_{\-{PI}}\tp{p,\pi}=1/\operatorname{gap}\tp{p}$.
In the common case $p=P-\!I_{\abs{\Omega}}$ for a reversible transition
matrix $P$, this spectral gap equals $1-\lambda_2\tp{P}$, where
$\lambda_2\tp{P}$ is the second-largest eigenvalue of $P$.
Accordingly, we sometimes also call this quantity the spectral gap of $P$
and write
$\operatorname{gap}\tp{P}\defeq\operatorname{gap}\tp{P-\!I_{\abs{\Omega}}}=1-\lambda_2\tp{P}$.
Similarly, we write
$C_{\-{PI}}\tp{P,\pi}\defeq C_{\-{PI}}\tp{P-\!I_{\abs{\Omega}},\pi}$.

For a function $f:\Omega\to\bb R$, define its oscillation by
\[
    \operatorname{Osc}\tp{f}\defeq\max_{x\in\Omega}f\tp{x}-\min_{x\in\Omega}f\tp{x}.
\]
For parameters $C_{\-{WPI}},\delta_{\-{WPI}}\geq0$, we say that $(p,\pi)$ satisfies a \emph{weak Poincar\'e inequality} (WPI) with parameters $\tp{C_{\-{WPI}},\delta_{\-{WPI}}}$ if
\begin{equation}
    \Var[\pi]{f}
    \leq C_{\-{WPI}}\cdot\+E_p\tp{f,f}
    +\delta_{\-{WPI}}\,\operatorname{Osc}\tp{f}^2,
    \qquad f:\Omega\to\bb R.
    \label{eq:weak-poincare-inequality}
\end{equation}
The case $\delta_{\-{WPI}}=0$ recovers the usual Poincar\'e inequality.  At the other extreme, $\Var[\pi]{f}\leq\operatorname{Osc}\tp{f}^2/4$ for every $f$, so the inequality becomes vacuous when $\delta_{\-{WPI}}\geq1/4$.  The useful regime allows a small oscillation error in exchange for a substantially smaller Dirichlet-form constant. 

\section{Wasserstein Distance and Action in Discrete Space}
\label{sec:discrete-action}
In this section, we introduce the discrete Wasserstein metric, the action, and
the weak action for distributions on finite state spaces. We first recall the
continuous-space definitions and then present our discrete framework. We then
give an equivalent flux formulation of the action and relax it to the weak
action.

\subsection{The Wasserstein-2 Distance and Action in Continuous Space}

\label{subsection:W2-and-action-continuous}

To motivate our discrete construction, we briefly recall the continuous-space definitions.

\paragraph{The Wasserstein-2 distance.}
Let $\mathcal{X}=\mathbb{R}^d$, and let $\mathcal P_2\tp{\mathcal X}$ denote the probability measures on $\mathcal X$ with finite second moments. The Wasserstein-2 distance between $\mu, \nu \in \mathcal P_2\tp{\mathcal X}$ is defined via optimal couplings:
$$    
    W_2(\mu, \nu) = \tp{\inf_{\gamma \in \Gamma(\mu, \nu)} \int_{\+X \times \+X} \|x - y\|^2 \, \dd \gamma(x,y)}^{1/2},
$$
where $\Gamma(\mu, \nu)$ is the set of all couplings of $\mu$ and $\nu$.

\paragraph{The Benamou--Brenier formula.}
As discussed in the introduction, the Benamou--Brenier formula \cite{BB00} provides an equivalent dynamical characterization:
\begin{equation*}
    W_2^2(\mu, \nu) = \inf_{\tp{\pi_t, v_t}_{t \in [0,1]}} \int_0^1 \int_{\mathcal X} \|v_t(x)\|^2 \, \dd\pi_t(x) \, \dd t,
\end{equation*}
where the infimum is over probability curves $(\pi_t)_{t \in [0,1]}$ with $\pi_0 = \mu$, $\pi_1 = \nu$, and velocity fields $v_t:\mathcal X\to\bb R^d$ satisfying the continuity equation
$$
    \partial_t \pi_t + \nabla \cdot (\pi_t v_t) = 0.
$$
Under the usual regularity assumptions, the optimal velocity can be represented as $v_t=\nabla\psi_t$ for a scalar potential $\psi_t$. This formulation shows that the Wasserstein-2 distance measures the minimal kinetic energy required to transport $\mu$ to $\nu$ along an optimal path of probability measures.

\paragraph{Metric derivative and action.}
For a fixed $T>0$ and an absolutely continuous curve of probability measures $\{\pi_t\}_{t \in [0,T]}$ in the metric space $(\mathcal P_2\tp{\mathcal X}, W_2)$ (see Section~\ref{sec:preliminaries} for the definition), the metric derivative at time $t$ is
\[
    |\dot{\pi}|_t = \lim_{h \to 0} \frac{W_2(\pi_{t+h}, \pi_t)}{|h|}.
\]
The Benamou--Brenier formula implies the following pointwise characterization (see \cite[Theorem 8.3.1 and Section 8.4]{AGS05}):
for almost every $t\in[0,T]$, its squared metric derivative is given by
$$
    |\dot{\pi}|_t^2 = \min_{v_t} \quad \int_{\+X} \|v_t(x)\|^2 \dd \pi_t(x),
$$
where the minimum is taken over all velocity fields $v_t:\+X\to\bb R^d$ satisfying the continuity equation $\partial_t \pi_t + \nabla \cdot (\pi_t v_t) = 0$. Moreover, the minimizer lies in the closure in $L^2(\pi_t)$ of the set of gradients $\nabla\psi$ of smooth compactly supported functions $\psi:\+X\to\bb R$.

The \emph{action} of the curve is defined as
$$    
    \+A\tp{\tp{\pi_t}_{t \in \stp{0, T}}} \defeq \int_0^T \abs{\dot{\pi}}_t^2 \, \dd t,
$$
which measures the total kinetic energy required to transport $\pi_0$ to $\pi_T$ along a given path. As shown in \cite{GTC25}, the action is the key geometric quantity controlling the convergence of annealed Langevin dynamics.





\subsection{A Discrete Space Framework}

\label{subsection:the-extension-to-discrete-space}

To transfer this picture to a finite state space $\Omega$, we follow
Maas~\cite{Maas11} and discretize the dynamical formulation above.
Mass moves along the edges of a graph $(\Omega,E)$, the gradient
$\nabla\psi_t(x)$ is replaced by the edge difference $\psi_t(y)-\psi_t(x)$,
and the density weight $\pi_t(x)$ is replaced by a weight on the edge
$\set{x,y}$.
Maas chooses this weight via the logarithmic mean of the endpoint densities,
which makes the Markov evolution the gradient flow of relative entropy.

For our purpose, this gradient-flow structure is not needed. Instead, the edge
weight should reflect the transition rates of the Markov chain run at each
stage of annealing. We therefore fix a connected undirected graph $(\Omega,E)$
and allow a general family of \emph{edge capacity functions}
$(c_\rho)_{\rho\in\+P_{>0}(\Omega)}$, $c_\rho:\Omega\times\Omega\to\bb R_{\ge0}$,
satisfying the following regularity conditions:
\begin{itemize}
    \item $c_{\rho}$ is symmetric, $c_\rho\tp{x,y}>0$ if and only if $\set{x,y}\in E$;
    \item $\sup_{\rho\in \+P_{>0}(\Omega)} \sum_{x,y \in \Omega} c_{\rho}(x,y) <\infty$;
    \item for each fixed $(x,y) \in \Omega \times \Omega$, the mapping $\rho \mapsto c_{\rho}(x,y)$ is continuous with respect to the Euclidean metric on $\mathcal{P}_{>0}(\Omega)$.
\end{itemize}
This framework includes Maas's metric as a special case.

We now define the Wasserstein-2 distance and action in discrete spaces in terms of $c_\rho$.

\begin{definition}[Admissible potential]
    For a probability curve \(\tp{\pi_t}_{t \in \stp{0, T}} \subseteq \+P_{>0}(\Omega)\), we say a scalar function \(\psi_t: \Omega \to \bb R\) is an \emph{admissible potential} of \(\tp{\pi_t}_{t \in \stp{0, T}}\) at time \(t\) if it satisfies the discrete continuity equation:
    \[
    \partial_t \pi_t\tp{x} + \sum_{y \in \Omega} \tp{\psi_t\tp{y} - \psi_t\tp{x}} c_{\pi_t}\tp{x, y}= 0, \quad \forall x \in \Omega.
    \]
\end{definition}

Since $(\Omega,E)$ is connected and $\sum_x\partial_t\pi_t(x)=0$, an admissible
potential exists whenever $\partial_t\pi_t$ exists, and it is unique up to an
additive constant (see \zcref{sec:pf-of-metric-derivative}).

\begin{definition}[Discrete Wasserstein-2 distance]\label{def:discrete-W2}
    For any two distributions $\mu,\nu \in \mathcal{P}_{>0}(\Omega)$, we define
    $$
        W_2(\mu, \nu) := \tp{\inf_{\tp{\pi_t, \psi_t}_{t \in [0,1]}} \int_0^1 \frac{1}{2} \sum_{x, y \in \Omega} \tp{\psi_t\tp{x} - \psi_t\tp{y}}^2 c_{\pi_t}\tp{x,y} \, \dd t}^{1/2},
    $$
    where the infimum is taken over all curves $(\pi_t)_{t \in [0,1]}$ in $\+P_{>0}(\Omega)$ with $\pi_0=\mu$, $\pi_1=\nu$ that are absolutely continuous as maps $[0,1]\to\bb R^\Omega$, and all $(\psi_t)_{t \in [0,1]}$ such that $\psi_t$ is an admissible potential for almost every $t$.
\end{definition}

We prove that $W_2$ is indeed a metric on $\+P_{>0}(\Omega)$ in \zcref{sec:pf-of-discrete-W2}.

\begin{remark}
Unlike~\cite{Maas11}, who defines an extended metric on all of $\mathcal P(\Omega)$,
we work only on the interior $\mathcal P_{>0}(\Omega)$, since all target
distributions in this paper have full support. We therefore do not discuss
completeness of the metric or related properties.
\end{remark}

For an absolutely continuous curve of probability measures $\{\pi_t\}_{t \in [0,T]}$ in the metric space $(\mathcal P_{>0}\tp{\Omega}, W_2)$, the metric derivative at time $t$ is
$$
    \abs{\dot{\pi}}_t = \lim_{h \to 0} \frac{W_2(\pi_{t+h}, \pi_t)}{\abs{h}}.
$$
We can then define the action of a curve of probability measures on a discrete space to measure the cost of transporting mass along the curve.
\begin{definition}[Action of a curve]\label{def:discrete-action}
    For an absolutely continuous curve of probability measures $\{\pi_t\}_{t \in [0,T]}$ in the metric space $(\mathcal P_{>0}\tp{\Omega}, W_2)$, we define its \emph{action} by
    $$
        \+A\tp{\tp{\pi_t}_{t \in \stp{0, T}}} \defeq \int_0^T \abs{\dot{\pi}}_t^2 \, \dd t.
    $$
\end{definition}

By \zcref{lem:ac-euclidean-w2}, a curve in $\+P_{>0}(\Omega)$ is absolutely
continuous in $W_2$ if and only if it is absolutely continuous as a map into
$\bb R^\Omega$. Accordingly, when we say that such a curve is absolutely
continuous, either notion may be meant.

As in the continuous case, we can characterize the metric derivative in terms of admissible potentials. The proof of the following theorem is provided in \zcref{sec:pf-of-metric-derivative}.
\begin{theorem}
    For an absolutely continuous curve of probability measures $\{\pi_t\}_{t \in [0,T]}$ in the metric space $(\mathcal P_{>0}\tp{\Omega}, W_2)$ and almost every $t \in [0,T]$, its squared metric derivative is given by
    $$
        |\dot{\pi}|_t^2 = \frac{1}{2} \sum_{x, y \in \Omega} \tp{\psi_t\tp{x} - \psi_t\tp{y}}^2 c_{\pi_t}\tp{x,y},
    $$
    where $\psi_t$ is any admissible potential at time $t$. Consequently, the action of $\tp{\pi_t}_{t \in [0,T]}$ can be expressed as
    $$
        \+A\tp{\tp{\pi_t}_{t \in \stp{0, T}}} = \int_0^T \frac{1}{2} \sum_{x, y \in \Omega} \tp{\psi_t\tp{x} - \psi_t\tp{y}}^2 c_{\pi_t}\tp{x,y} \dd t.
    $$
    \label{thm:metric-derivative-potential-formulation}
\end{theorem}

By \zcref{thm:metric-derivative-potential-formulation}, the metric derivative
$\abs{\dot\pi}_t$ depends on the capacity family only through $c_{\pi_t}$, so the
action of a curve is determined by the capacities along it. In all annealing
results in this paper, we are given a path $(\pi_s)_{s\in[0,1]}$ and rate kernels
$p_s$ reversible with respect to $\pi_s$, and we assume that the capacity
family agrees with these kernels along the path:
\begin{equation*}
    c_s(x,y)\defeq c_{\pi_s}(x,y)=\pi_s(x)p_s(x,y)=\pi_s(y)p_s(y,x),
    \qquad x\neq y,
\end{equation*}
so that transport is priced by the moves available to the sampler.

\subsection{A Flux Formulation of Action}

We now give an equivalent formulation of the metric derivative in terms of
fluxes, which describe how the change $\partial_t\pi_t$ is routed along the edges
of the graph. Since this is a minimization over fluxes, every admissible flux
gives an upper bound on the action.

\begin{definition}[Admissible flux]
    We say that \(J: \Omega \times \Omega \to \bb R\) is a flux on the graph \(\tp{\Omega, E}\) if it satisfies:
    \begin{itemize}
        \item \(J\tp{x, y} = 0\) whenever \(\set{x, y} \notin E\).
        \item \(J\) is anti-symmetric: \(J\tp{x, y} = -J\tp{y, x}\) for all \(x, y \in \Omega\).
    \end{itemize}
    We interpret \(J\tp{x, y}\) as the net rate at which mass flows from \(x\) to \(y\).
    Let \(\tp{\pi_t}_{t \in [0, T]} \subseteq \+P_{>0}(\Omega)\) be a curve of probability measures. We say that a flux \(J_t: \Omega \times \Omega \to \bb R\) is an admissible flux of \(\tp{\pi_t}_{t \in [0,T]}\) at time \(t\) if it satisfies the discrete continuity equation
    \[
    \partial_t \pi_t\tp{x} + \divergence_{\-{d}} J_t\tp{x} = 0, \quad \forall x \in \Omega.
    \]
    Here, the discrete divergence operator on $\Omega$ is defined by
    \[
    \divergence_{\-{d}} J_t\tp{x} \defeq \sum_{y \in \Omega} J_t\tp{x, y}, \quad x \in \Omega.
    \]
\end{definition}

\begin{proposition}
    For an absolutely continuous curve of probability measures $\{\pi_t\}_{t \in [0,T]}$ in the metric space $(\mathcal P_{>0}\tp{\Omega}, W_2)$ and almost every $t \in [0,T]$,
    \[
    \abs{\dot{\pi}}_t^2 = \min_{\text{admissible fluxes } J_t} \quad \sum_{\set{x,y}\in E} \frac{J_t\tp{x, y}^2}{c_{\pi_t}\tp{x, y}}.
    \]
    Moreover, the minimum is attained uniquely by
    \[
    J_t(x,y)=\tp{\psi_t(y)-\psi_t(x)}c_{\pi_t}(x,y),
    \]
    where $\psi_t$ is an admissible potential at time $t$.
    \label{prop:metric-derivative-flux-formulation}
\end{proposition}

The proof is deferred to \zcref{subsec:appendix-metric-derivative-flux-formulation}.

Consequently, the action of a curve can be expressed in terms of the flux formulation as follows:
\[
        \+A\tp{\tp{\pi_t}_{t \in \stp{0, T}}} = \int_0^T \min_{\text{admissible fluxes } J_t} \quad \sum_{\set{x,y}\in E} \frac{J_t\tp{x, y}^2}{c_{\pi_t}\tp{x, y}} \, \dd t.
\]

\subsection{Weak Action with an $L^1$ Residual}
\label{subsec:weak-action}

By \zcref{prop:metric-derivative-flux-formulation}, the action requires the
entire change $\partial_t\pi_t$ to be routed through the edges, so it can be very
large when some edges have small capacity. It can then be cheaper to leave a
small part of $\partial_t\pi_t$ unrouted. We therefore relax the continuity
equation by allowing a residual term, and define the weak action of a curve with
a residual budget $\varepsilon\geq0$ on the $L^1$ size of this residual. Recall
that $\norm{r}_1=\sum_{x\in\Omega}\abs{r\tp{x}}$ for a signed function
$r:\Omega\to\bb R$.

\begin{definition}[$L^1$ weak action]
    Consider an absolutely continuous curve of probability measures $\{\pi_t\}_{t \in [0,T]}$ in the metric space $(\mathcal P_{>0}\tp{\Omega}, W_2)$. A pair $\tp{J_t,r_t}$ is called an \emph{$L^1$-relaxed flux} at time $t$ if $J_t$ is a flux, $r_t:\Omega\to\bb R$ satisfies $\sum_{x\in\Omega}r_t\tp{x}=0$, and
    \[
        \partial_t\pi_t+\divergence_{\-{d}}J_t=r_t.
    \]
    For a residual budget $\varepsilon\geq0$, define the $L^1$ weak action of the curve by
    \begin{equation*}
    \+A^{\mathrm{w}}_{\varepsilon}\tp{\tp{\pi_t}_{t\in[0,T]}}
        \defeq
        \inf_{\tp{J_t,r_t}_{t\in[0,T]}}
        \int_0^T\sum_{\set{x,y}\in E}
        \frac{J_t\tp{x,y}^2}{c_{\pi_t}\tp{x,y}}\,\dd t,
    \end{equation*}
    where the infimum is over all measurable families $\tp{J_t,r_t}_{t\in[0,T]}$ such that $\tp{J_t,r_t}$ is an $L^1$-relaxed flux for almost every $t$ and
    \[
        \frac12\int_0^T\norm{r_t}_1\,\dd t\leq\varepsilon.
    \]
\end{definition}

In the definition of weak action, only $\partial_t\pi_t-r_t=-\divergence_{\-{d}}J_t$ is transported through the edges and incurs a quadratic cost. When $\varepsilon=0$, the residual vanishes almost everywhere, and hence
\[
    \+A^{\mathrm{w}}_0\tp{\tp{\pi_t}_{t\in[0,T]}}
    =\+A\tp{\tp{\pi_t}_{t\in[0,T]}}.
\]
Thus, the weak action is a relaxation of the exact action defined in \zcref{def:discrete-action}, with the residual budget recording how much of the path velocity is left unrouted. We show in \zcref{section:annealing} how the weak action, together with the residual budget, bounds the error of annealing algorithms.

\section{One-Path Annealing: The Basic Building Block}

\label{section:annealing}

This section proves the pathwise comparison underlying our algorithmic
results.  For one prescribed annealing path, exact action controls the
reverse-KL error of the annealing algorithm, while weak action controls total variation
after allowing a small residual term.  In
\zcref{sec:componentwise-mixture-annealing} we apply this comparison to
separately initialized phase paths to sample from mixture distributions.  Thus the
one-path bounds below serve as the basic building block for the componentwise
theorem used in the applications.

\subsection{The Annealing Algorithm}
\label{subsec:annealing-algorithm}

Throughout this section, fix a target distribution $\pi$ and an annealing path
$\tp{\pi_s}_{s\in\stp{0,1}}$ with $\pi_1=\pi$ that is absolutely continuous in
$\tp{\+P_{>0}(\Omega),W_2}$, and transition matrices
$\tp{P_s}_{s\in\stp{0,1}}$ such that $s\mapsto P_s$ is measurable and each
$P_s$ is reversible with respect to $\pi_s$. Write
$p_s\defeq P_s-\!I_{\abs{\Omega}}$ for the transition rate kernels, and assume
that the edge set $E=\set{\set{x,y}:x\neq y,\ P_s\tp{x,y}>0}$ does not depend
on $s$ and that the graph $\tp{\Omega,E}$ is connected. We assume that the
capacity family $\tp{c_\rho}$ of
\zcref{subsection:the-extension-to-discrete-space} is compatible with these
kernels along the path, in the sense that
\[
    c_s\tp{x,y}\defeq c_{\pi_s}\tp{x,y}=\pi_s\tp{x}p_s\tp{x,y}=\pi_s\tp{y}p_s\tp{y,x},
    \qquad x\neq y,\ s\in\stp{0,1}.
\]

Given a running time $T>0$, \zcref{algo:poisson-annealing} traverses the path
at speed $1/T$. It samples $X\sim\pi_0^{\-{\!{ALG}}}$ and, at each ring time
$t$ of a unit-rate Poisson clock on $\stp{0,T}$, resamples
$X\sim P_{t/T}(X,\cdot)$. Since $p_s=P_s-\!I_{\abs{\Omega}}$, this is exactly
the time-inhomogeneous Markov jump process started from
$\pi_0^{\-{\!{ALG}}}$ with rate kernels $\tilde p_t\defeq p_{t/T}$,
$t\in\stp{0,T}$. Each $\tilde p_t$ is reversible with respect to
$\tilde\pi_t\defeq\pi_{t/T}$, and we write $\tilde c_t\defeq c_{t/T}$ for the
corresponding capacities. The initial law $\pi_0^{\-{\!{ALG}}}$ need not equal
$\pi_0$. The mismatch enters the bounds below only through
$D_{\mathrm{KL}}\tp{\pi_0\|\pi_0^{\-{\!{ALG}}}}$.

For convenience, we recall the complete procedure.

\paragraph{\zcref{algo:poisson-annealing} (recalled).}
\begin{algorithmic}[1]
    \Require Annealing path $\tp{\pi_s}_{s\in[0,1]}$, initial law
    $\pi_0^{\-{\!{ALG}}}$, transition matrices $\tp{P_s}_{s\in[0,1]}$ with
    $P_s$ reversible with respect to $\pi_s$, and running time $T$.
    \State Sample $X\sim\pi_0^{\-{\!{ALG}}}$
    \State Run a unit-rate Poisson clock on $[0,T]$
    \For{each ring time $t$}
        \State Sample $X\sim P_{t/T}(X,\cdot)$
    \EndFor
    \State \Return $X$
\end{algorithmic}

\subsection{Exact- and Weak-Action Comparison}
We now bound the error of \zcref{algo:poisson-annealing} for a fixed running
time $T$. Both bounds in the following lemma come from a single defective
free-energy inequality. With exact transport, the Donsker--Varadhan formula
turns it into a reverse-KL bound. With an $L^1$ residual, testing against
indicator functions and applying Hoeffding's lemma turns it into a
total-variation bound.

\begin{lemma}[Exact- and weak-action error bound]
\label{lem:exact-weak-action-comparison}
    In the setting of \zcref{subsec:annealing-algorithm}, fix a running time
    $T>0$ and let $\nu_T$ be the output distribution of
    \zcref{algo:poisson-annealing}. Then
    \begin{equation}
        D_{\mathrm{KL}}\tp{\pi \| \nu_T}
        \leq D_{\mathrm{KL}}\tp{\pi_0 \| \pi_0^{\-{\!{ALG}}}}
        +\frac{1}{4}
        \+A\tp{\tp{\tilde\pi_t}_{t \in \stp{0, T}}}.
        \label{eq:exact-action-comparison}
    \end{equation}
    Moreover, for every residual budget $\varepsilon\geq0$,
    \begin{equation}
        D_{\mathrm{TV}}\tp{\pi,\nu_T}
        \leq \varepsilon+
        \sqrt{\frac12\left(
        D_{\mathrm{KL}}\tp{\pi_0\|\pi_0^{\-{\!{ALG}}}}
        +\frac14\+A^{\mathrm w}_{\varepsilon}
        \tp{\tp{\tilde\pi_t}_{t\in[0,T]}}
        \right)}.
        \label{eq:weak-action-comparison}
    \end{equation}
\end{lemma}

\begin{proof}
    Put $K\defeq D_{\mathrm{KL}}\tp{\pi_0\|\pi_0^{\-{\!{ALG}}}}$. If $K=+\infty$,
    both conclusions are immediate. We first prove a defective free-energy
    inequality for an arbitrary $L^1$-relaxed flux $\tp{J_t,r_t}_{t\in[0,T]}$
    of the curve. Write
    \[
        \mathcal Q(J)\defeq\int_0^T\sum_{\set{x,y}\in E}
        \frac{J_t\tp{x,y}^2}{\tilde c_t\tp{x,y}}\,\dd t,
        \qquad
        R(r)\defeq\frac12\int_0^T\norm{r_t}_1\,\dd t
    \]
    for its quadratic cost and residual budget, and suppose these quantities
    are finite.

    Let $\tp{X_t}_{t\in[0,T]}$ denote the Markov jump process on
    $\tp{\Omega,E}$ with time-dependent transition rate kernels
    $\tp{\tilde p_t}_{t\in[0,T]}$, and let $Q_{s,t}$ denote its transition
    kernel from time $s$ to time $t$. For any initial distribution $\mu$ on
    $\Omega$, let $\bb P_\mu$ and $\bb E_\mu$ denote the law and expectation
    of this process with $X_0\sim\mu$. By
    \zcref{subsec:annealing-algorithm}, $\nu_T$ is the law of $X_T$ under
    $\bb P_{\pi_0^{\-{\!{ALG}}}}$.

    \paragraph{Step 1: The backward potential and the initial entropy cost.}
    Fix an arbitrary test function $h:\Omega\to\bb R$ and define its
    terminal free-energy gap by
    \begin{equation*}
        \mathcal F_T(h)
        \defeq
        \sum_{x\in\Omega}\pi\tp{x}h\tp{x}
        -\log\sum_{x\in\Omega}
          \nu_T\tp{x}e^{h\tp{x}}.
    \end{equation*}
    This is the common variational quantity behind both conclusions of the
    lemma. Steps~1--3 bound $\mathcal F_T(h)$ by a constant plus a multiple of
    $\operatorname{Osc}\tp{h}$, and Step~4 converts this bound into the two
    estimates. We now propagate $h$ backward through the annealing process.

    Define the backward value function
    \[
        u_t\tp{x}\defeq\tp{Q_{t,T}e^h}\tp{x}
        =\bb E_{t,x}\stp{e^{h\tp{X_T}}},
        \qquad t\in[0,T],\; x\in\Omega,
    \]
    where $\bb E_{t,x}$ denotes expectation for the process started from
    $X_t=x$ at time $t$. Set its logarithmic potential to be
    $\varphi_t\tp{x}\defeq\log u_t\tp{x}$.

    Since $s\mapsto P_s$ is measurable and every off-diagonal rate satisfies
    $0\le\tilde p_t\tp{x,y}\le1$, the function $t\mapsto u_t$ is Lipschitz on
    $[0,T]$ and, for almost every $t$, satisfies the backward Kolmogorov
    equation
    \[
        \partial_tu_t\tp{x}+\sum_y\tilde p_t\tp{x,y}\tp{u_t\tp{y}-u_t\tp{x}}=0,
        \qquad u_T=e^h.
    \]
    Moreover $e^{\min h}\le u_t\le e^{\max h}$, so
    $\min h\le\varphi_t\le\max h$ and $t\mapsto\varphi_t$ is Lipschitz as
    well. Since $\tp{\pi_s}_{s\in\stp{0,1}}$ is absolutely continuous in
    $\tp{\+P_{>0}(\Omega),W_2}$, so is $\tp{\tilde\pi_t}_{t\in\stp{0,T}}$, and
    by \zcref{lem:w2-bounds-euclidean} it is also absolutely continuous as a
    map into $\bb R^\Omega$. Hence the pairing
    $t\mapsto\sum_{x\in\Omega}\tilde\pi_t\tp{x}\varphi_t\tp{x}$ is absolutely
    continuous on $[0,T]$, and for almost every $t$ its derivative is given by
    the product rule.

    At $t=T$, the terminal condition
    $u_T=e^h$ implies $\varphi_T=h$, and hence
    \[
        \sum_{x\in\Omega}\pi\tp{x}h\tp{x}
        =\sum_{x\in\Omega}\tilde\pi_T\tp{x}\varphi_T\tp{x}.
    \]
    At $t=0$, by the law of total expectation,
    \[
        \sum_{x\in\Omega}\pi_0^{\-{\!{ALG}}}\tp{x}e^{\varphi_0\tp{x}}
        =\sum_{x\in\Omega}\pi_0^{\-{\!{ALG}}}\tp{x}\tp{Q_{0,T}e^h}\tp{x}
        =\bb E_{\pi_0^{\-{\!{ALG}}}}\stp{e^{h\tp{X_T}}}
        =\sum_{x\in\Omega}\nu_T\tp{x}e^{h\tp{x}}.
    \]
    Applying the Donsker--Varadhan variational formula for relative entropy with $\mu=\pi_0$, $\nu=\pi_0^{\-{\!{ALG}}}$, and test function $\varphi_0$ gives
    \[
        \sum_{x\in\Omega}\pi_0\tp{x}\varphi_0\tp{x}-\log\sum_{x\in\Omega}\pi_0^{\-{\!{ALG}}}\tp{x}e^{\varphi_0\tp{x}}
        \leq D_{\mathrm{KL}}\tp{\pi_0\|\pi_0^{\-{\!{ALG}}}},
    \]
    or equivalently,
    \[
        \log\sum_{x\in\Omega}\nu_T\tp{x}e^{h\tp{x}}
        \ge\sum_{x\in\Omega}\pi_0\tp{x}\varphi_0\tp{x}-D_{\mathrm{KL}}\tp{\pi_0\|\pi_0^{\-{\!{ALG}}}}.
    \]
    Subtracting this lower bound from the terminal pairing and applying the
    fundamental theorem of calculus to the absolutely continuous pairing, we
    obtain
    \begin{equation}
        \mathcal F_T(h)
        \leq K+\int_0^T\frac{\dd}{\dd t}
        \sum_{x\in\Omega}\tilde\pi_t\tp{x}\varphi_t\tp{x}\,\dd t.
        \label{eq:dv-time-bridge}
    \end{equation}

    \paragraph{Step 2: Transport work and Markov dissipation.}
    For almost every $t$, differentiating the pairing along the curve gives
    \[
        \frac{\dd}{\dd t}\sum_{x\in\Omega}\tilde\pi_t\tp{x}\varphi_t\tp{x}
        =\sum_{x\in\Omega}\partial_t\tilde\pi_t\tp{x}\varphi_t\tp{x}
        +\sum_{x\in\Omega}\tilde\pi_t\tp{x}\partial_t\varphi_t\tp{x}.
    \]
    For the first term, the relaxed continuity equation $\partial_t\tilde\pi_t=-\divergence_{\-{d}}J_t+r_t$ and discrete summation by parts (using the antisymmetry of $J_t$) give
    \[
        \begin{aligned}
            \sum_{x\in\Omega}\partial_t\tilde\pi_t\tp{x}\varphi_t\tp{x}
            &=-\sum_{x\in\Omega}\divergence_{\-{d}}J_t\tp{x}\varphi_t\tp{x}+\sum_{x\in\Omega}r_t\tp{x}\varphi_t\tp{x}\\
            &=\sum_{\set{x,y}\in E}J_t\tp{x,y}\nabla_{x,y}\varphi_t+\sum_{x\in\Omega}r_t\tp{x}\varphi_t\tp{x},
        \end{aligned}
    \]
    where $\nabla_{x,y}\varphi_t\defeq\varphi_t\tp{y}-\varphi_t\tp{x}$ denotes the discrete gradient along edge $\set{x,y}$.
    For the second term, dividing the backward Kolmogorov equation by
    $u_t\tp{x}>0$ yields the discrete Hamilton--Jacobi equation
    \[
        \partial_t\varphi_t\tp{x}
        =-\sum_y\tilde p_t\tp{x,y}\tp{e^{\varphi_t\tp{y}-\varphi_t\tp{x}}-1}.
    \]
    Pairing with $\tilde\pi_t$ and using the reversibility $\tilde c_t\tp{x,y}=\tilde\pi_t\tp{x}\tilde p_t\tp{x,y}=\tilde\pi_t\tp{y}\tilde p_t\tp{y,x}$ to group the two orientations of each edge $\set{x,y}\in E$, we obtain
    \[
        \begin{aligned}
            \sum_{x\in\Omega}\tilde\pi_t\tp{x}\partial_t\varphi_t\tp{x}
            &=-\sum_{\set{x,y}\in E}\tilde c_t\tp{x,y}\stp{\tp{e^{\nabla_{x,y}\varphi_t}-1}+\tp{e^{-\nabla_{x,y}\varphi_t}-1}}\\
            &=-2\sum_{\set{x,y}\in E}\tilde c_t\tp{x,y}\tp{\cosh\tp{\nabla_{x,y}\varphi_t}-1}.
        \end{aligned}
    \]
    Combining both terms yields the instantaneous balance
    \[
        \frac{\dd}{\dd t}\sum_{x\in\Omega}\tilde\pi_t\tp{x}\varphi_t\tp{x}
        =\sum_{\set{x,y}\in E}\stp{J_t\tp{x,y}\nabla_{x,y}\varphi_t-2\tilde c_t\tp{x,y}\tp{\cosh\tp{\nabla_{x,y}\varphi_t}-1}}
        +\sum_{x\in\Omega}r_t\tp{x}\varphi_t\tp{x}.
    \]

    \paragraph{Step 3: The defective free-energy inequality.}
    By \zcref{lem:cosh-quadratic-bound},
    $jz-2c(\cosh z-1)\le jz-cz^2\le\frac{j^2}{4c}$ for all $j,z\in\bb R$
    and $c>0$. Applying this on each edge with $j=J_t\tp{x,y}$,
    $c=\tilde c_t\tp{x,y}$ and $z=\nabla_{x,y}\varphi_t$ gives
    \[
        J_t\tp{x,y}\nabla_{x,y}\varphi_t-2\tilde c_t\tp{x,y}\tp{\cosh\tp{\nabla_{x,y}\varphi_t}-1}
        \le\frac{J_t\tp{x,y}^2}{4\tilde c_t\tp{x,y}}.
    \]
    To control the residual term, let $m_t\defeq\frac12\tp{\max_x\varphi_t\tp{x}+\min_x\varphi_t\tp{x}}$. Since $\sum_xr_t\tp{x}=0$, we can subtract $m_t$ without altering the sum. Furthermore, $\min h\le\varphi_t\le\max h$ by Step~1, so $\operatorname{Osc}\tp{\varphi_t}\le\operatorname{Osc}\tp{h}$. Therefore,
    \[
        \sum_{x\in\Omega}r_t\tp{x}\varphi_t\tp{x}
        =\sum_{x\in\Omega}r_t\tp{x}\tp{\varphi_t\tp{x}-m_t}
        \le\sum_{x\in\Omega}\abs{r_t\tp{x}}\abs{\varphi_t\tp{x}-m_t}
        \le\frac12\norm{r_t}_1\operatorname{Osc}\tp{h}.
    \]
    Combining the edge and residual bounds, we conclude that for almost every $t\in[0,T]$,
    \[
        \frac{\dd}{\dd t}\sum_{x\in\Omega}\tilde\pi_t\tp{x}\varphi_t\tp{x}
        \le\frac14\sum_{\set{x,y}\in E}\frac{J_t\tp{x,y}^2}{\tilde c_t\tp{x,y}}
        +\frac12\norm{r_t}_1\operatorname{Osc}\tp{h}.
    \]

    Integrating this differential inequality over $t\in[0,T]$ gives
    \[
        \sum_{x\in\Omega}\tilde\pi_T\tp{x}\varphi_T\tp{x}-\sum_{x\in\Omega}\pi_0\tp{x}\varphi_0\tp{x}
        \leq\frac14\mathcal Q(J)+R(r)\operatorname{Osc}\tp{h}.
    \]
    Substitution into \eqref{eq:dv-time-bridge} proves the defective free-energy inequality
    \begin{equation}
        \mathcal F_T(h)
        \leq K+\frac14\mathcal Q(J)+R(r)\operatorname{Osc}\tp{h},
        \qquad h:\Omega\to\bb R.
        \label{eq:defective-free-energy-comparison}
    \end{equation}

    \paragraph{Step 4: Exact- and weak-action conclusions.}
    We first consider exact action. If the action is infinite, the claim is
    immediate. Otherwise, since
    $\+A^{\mathrm w}_0=\+A$ (see \zcref{subsec:weak-action}), for any $\eta>0$
    there is an $L^1$-relaxed flux $\tp{J_t,r_t}_{t\in[0,T]}$ with $R(r)=0$ and
    \[
        \mathcal Q(J)\leq
        \+A\tp{\tp{\tilde\pi_t}_{t\in[0,T]}}+\eta.
    \]
    Since $R(r)=0$, the right-hand side of
    \eqref{eq:defective-free-energy-comparison} does not depend on $h$.
    The Donsker--Varadhan formula gives
    \[
        D_{\mathrm{KL}}\tp{\pi\|\nu_T}
        =\sup_{h:\Omega\to\bb R}\mathcal F_T(h)
        \leq K+\frac14\left(
          \+A\tp{\tp{\tilde\pi_t}_{t\in[0,T]}}+\eta
        \right).
    \]
    Letting $\eta\downarrow0$ proves \eqref{eq:exact-action-comparison}.

    We next consider weak action. If
    $\+A^{\mathrm w}_{\varepsilon}
      \tp{\tp{\tilde\pi_t}_{t\in[0,T]}}=+\infty$,
    the claim is immediate. Otherwise, fix $\eta>0$ and choose an
    $L^1$-relaxed flux such that
    \[
        R(r)\leq\varepsilon,
        \qquad
        \mathcal Q(J)\leq
        \+A^{\mathrm w}_{\varepsilon}
        \tp{\tp{\tilde\pi_t}_{t\in[0,T]}}+\eta.
    \]
    Fix an event $A\subseteq\Omega$ and $\lambda>0$, and set
    $h=\lambda\mathbf1_A$. Applying \zcref{lem:hoeffding} to
    $\mathbf1_A\tp{X}$ with $X\sim\nu_T$ gives
    \[
        \log\sum_x\nu_T\tp{x}e^{h\tp{x}}
        \leq\lambda\nu_T\tp{A}+\frac{\lambda^2}{8}.
    \]
    Combining this with \eqref{eq:defective-free-energy-comparison} and
    $\operatorname{Osc}\tp{h}\leq\lambda$ yields
    \[
        \pi\tp{A}-\nu_T\tp{A}
        \leq\varepsilon+
        \frac1\lambda\left(
          K+\frac14\left(
            \+A^{\mathrm w}_{\varepsilon}
            \tp{\tp{\tilde\pi_t}_{t\in[0,T]}}+\eta
          \right)
        \right)+\frac\lambda8.
    \]
    For $M\geq0$,
    $\inf_{\lambda>0}\tp{M/\lambda+\lambda/8}=\sqrt{M/2}$, with the
    value at $M=0$ understood by letting $\lambda\downarrow0$. Optimizing in
    $\lambda$, taking the supremum over $A$, and then letting
    $\eta\downarrow0$ proves \eqref{eq:weak-action-comparison}.
\end{proof}

The bounds in \zcref{lem:exact-weak-action-comparison} are stated for the
rescaled curve on $[0,T]$, so they mix the running time with the geometry of
the path. The next lemma separates the two by showing that both the action and
the weak action scale as $1/T$.

\begin{lemma}[Time rescaling]\label{lemma:time-rescaling-of-action}
    Let $\tp{\pi_s}_{s\in\stp{0,1}}$ be an absolutely continuous curve in
    $\tp{\+P_{>0}(\Omega),W_2}$, let $T>0$, and let
    $\tilde\pi_t\defeq\pi_{t/T}$ for $t\in\stp{0,T}$. Then
    \[
    \+A\tp{\tp{\tilde{\pi}_t}_{t \in \stp{0, T}}} = \frac{1}{T} \+A\tp{\tp{\pi_s}_{s \in \stp{0, 1}}},
    \]
    and, for every \(\varepsilon\geq0\),
    \[
        \+A^{\mathrm w}_{\varepsilon}
        \tp{\tp{\tilde{\pi}_t}_{t\in\stp{0,T}}}
        =\frac1T\+A^{\mathrm w}_{\varepsilon}
        \tp{\tp{\pi_s}_{s\in\stp{0,1}}}.
    \]
\end{lemma}

\begin{proof}
    The rescaled curve $\tp{\tilde\pi_t}_{t\in\stp{0,T}}$ is again absolutely
    continuous in $\tp{\+P_{>0}(\Omega),W_2}$, and
    $\partial_t\tilde\pi_t=T^{-1}\partial_s\pi_s|_{s=t/T}$ for almost every
    $t$. Hence a measurable family $\tp{J_s,r_s}_{s\in\stp{0,1}}$ is an
    $L^1$-relaxed flux of $\tp{\pi_s}_{s\in\stp{0,1}}$ if and only if the
    rescaled family
    \[
        \tilde{J}_t \defeq \frac{1}{T} J_{t / T},
        \qquad
        \tilde r_t\defeq\frac1T r_{t/T},
        \quad t \in \stp{0, T},
    \]
    is an $L^1$-relaxed flux of $\tp{\tilde\pi_t}_{t\in\stp{0,T}}$. Since
    $c_{\tilde\pi_t}=c_{\pi_{t/T}}$, the change of variables $s=t/T$ shows that
    its residual budget is unchanged, while its quadratic cost is divided by
    $T$:
    \[
        \frac12\int_0^T\norm{\tilde r_t}_1\,\dd t
        =\frac12\int_0^1\norm{r_s}_1\,\dd s,
        \qquad
        \int_0^T\sum_{\set{x,y}\in E}\frac{\tilde J_t\tp{x,y}^2}{c_{\tilde\pi_t}\tp{x,y}}\,\dd t
        =\frac1T\int_0^1\sum_{\set{x,y}\in E}\frac{J_s\tp{x,y}^2}{c_{\pi_s}\tp{x,y}}\,\dd s.
    \]
    Taking the infimum in both directions proves the weak-action identity. The exact-action identity is the case \(\varepsilon=0\).
\end{proof}

Combining \zcref{lem:exact-weak-action-comparison} with \zcref{lemma:time-rescaling-of-action} gives
the action-based guarantees for \zcref{algo:poisson-annealing}. The following theorem is a formal version of \zcref{thm:intro-annealing-bound}.
\begin{theorem}
    \label{thm:annealing-ub}
    Let $\pi$ be a target distribution and let $\tp{\pi_s}_{s\in\stp{0,1}}$
    be an annealing path with $\pi_1=\pi$ that is absolutely continuous in
    $\tp{\+P_{>0}(\Omega),W_2}$. Let $\tp{P_s}_{s\in\stp{0,1}}$ be transition
    matrices such that $s\mapsto P_s$ is measurable and each $P_s$ is
    reversible with respect to $\pi_s$. Write
    $p_s\defeq P_s-\!I_{\abs{\Omega}}$, assume that the edge set
    $E=\set{\set{x,y}:x\neq y,\ P_s\tp{x,y}>0}$ does not depend on $s$ and
    that the graph $\tp{\Omega,E}$ is connected, and assume that the capacity
    family of \zcref{subsection:the-extension-to-discrete-space} satisfies
    \[
        c_s\tp{x,y}\defeq c_{\pi_s}\tp{x,y}
        =\pi_s\tp{x}p_s\tp{x,y}=\pi_s\tp{y}p_s\tp{y,x},
        \qquad x\neq y,\ s\in\stp{0,1}.
    \]
    Let $\nu_T$ be the output distribution of \zcref{algo:poisson-annealing}. Then
    \begin{equation}
        D_{\mathrm{KL}}\tp{\pi\|\nu_T}
        \leq
        D_{\mathrm{KL}}\tp{\pi_0\|\pi_0^{\-{\!{ALG}}}}
        +\frac{\+A\tp{\tp{\pi_s}_{s\in\stp{0,1}}}}{4T}.
        \label{eq:continuous-exact-action-bound}
    \end{equation}
    Moreover, for every residual budget $\varepsilon\geq0$,
    \begin{equation}
        D_{\mathrm{TV}}\tp{\pi,\nu_T}
        \leq\varepsilon+
        \sqrt{\frac12\left(
        D_{\mathrm{KL}}\tp{\pi_0\|\pi_0^{\-{\!{ALG}}}}
        +\frac{\+A^{\mathrm w}_{\varepsilon}
        \tp{\tp{\pi_s}_{s\in\stp{0,1}}}}{4T}
        \right)}.
        \label{eq:continuous-weak-action-bound}
    \end{equation}
    The expected number of updates is $T$.
\end{theorem}

\begin{proof}
    Apply \zcref{lem:exact-weak-action-comparison}, whose curve is
    $\tilde\pi_t=\pi_{t/T}$, and then use
    \zcref{lemma:time-rescaling-of-action}. This gives
    \eqref{eq:continuous-exact-action-bound} and
    \eqref{eq:continuous-weak-action-bound} directly.

    Finally, a unit-rate Poisson clock has \(\!{Pois}(T)\) rings on
    \([0,T]\), and the algorithm performs one update at each ring.  Hence its
    expected number of updates is \(T\).
\end{proof}

\section{Componentwise Annealing: The Algorithmic Master Theorem}
\label{sec:componentwise-mixture-annealing}

This section proves the componentwise theorem stated in
\zcref{thm:intro-componentwise-annealing}, which applies when the target is a
mixture of well-separated components. We first fix the setting and explain
the main idea. \zcref{subsec:localized-unrestricted-comparison} then adapts
the one-path error bound of \zcref{lem:exact-weak-action-comparison} to a
single component, \zcref{subsec:componentwise-mixture-comparison} combines
the components into the main theorem, and
\zcref{subsec:componentwise-mixture-applications} explains how the
applications use it.

Let $\pi$ be the target distribution on a finite state space $\Omega$.
Suppose that it admits a representation
\begin{equation*}
    \pi=\sum_{i=1}^k w_i\pi_i,
\end{equation*}
where $w=(w_i)_{i\in[k]}$ is a probability vector and each $\pi_i$ is a
probability law on $\Omega$ whose support is a set $\Omega_i\subseteq\Omega$,
called its phase domain. The phase domains need not be disjoint.

Choose a reference path $(\rho_s)_{s\in[0,1]}$ of strictly positive
probability laws that is absolutely continuous as a map into $\bb R^\Omega$,
such that
\begin{equation*}
    \rho_1(\,\cdot\mid\Omega_i)=\pi_i,
    \qquad i\in[k].
\end{equation*}
Let $(P_s)_{s\in[0,1]}$ be transition matrices such that $s\mapsto P_s$ is
measurable and each $P_s$ is reversible with respect to $\rho_s$, and write
$p_s=P_s-\!I_{\abs{\Omega}}$. Assume that the edge set
$E=\set{\set{x,y}:x\neq y,\ P_s(x,y)>0}$ does not depend on $s$, that the
graph $(\Omega,E)$ is connected, and that each induced graph
$(\Omega_i,E_i)$, where $E_i=\set{\set{x,y}\in E:x,y\in\Omega_i}$, is
connected.

The phase domains induce the component paths
\begin{equation}
    \mu_{i,s}\defeq\rho_s(\,\cdot\mid\Omega_i),
    \label{eq:conditional-component-path}
\end{equation}
so that $\mu_{i,1}=\pi_i$ and $\pi=\sum_iw_i\mu_{i,1}$. The target is matched
only through the conditional identities $\rho_1(\,\cdot\mid\Omega_i)=\pi_i$,
so $\rho_1$ need not equal $\pi$. For instance, when the phase domains are
disjoint, $w_i=\pi(\Omega_i)$ and $\pi_i=\pi(\,\cdot\mid\Omega_i)$ whenever
$w_i>0$, while $\rho_1$ may still assign different weights $\rho_1(\Omega_i)$
to the domains. When the domains overlap, $w_i$ need not equal
$\pi(\Omega_i)$.

We now look at what annealing towards such a mixture target involves.
Consider a mixture path $\bar\mu_s=\sum_iw_i(s)\mu_{i,s}$, where
$w(s)=(w_i(s))_{i\in[k]}$ is an absolutely continuous probability vector with
$w(1)=w$, so that $\bar\mu_1=\pi$. Its derivative is
\begin{equation}
    \partial_s\bar\mu_s
    =\sum_i w_i(s)\partial_s\mu_{i,s}
     +\sum_i\dot w_i(s)\mu_{i,s}.
    \label{eq:component-source-decomposition}
\end{equation}
The first term changes the conditional law inside each component, and the
second term changes the component weights. On the labelled space
$\bigsqcup_i\{i\}\times\Omega_i$, the first term moves mass with the label $i$
fixed, while the second moves mass between labels. When the components are
separated by a low-capacity cut, routing the second term through the edges of
$(\Omega,E)$ can make the action of $(\bar\mu_s)_{s\in[0,1]}$ very large.

We now discuss how to handle the case where this second term is large. A
direct idea is to keep the weights fixed. The fixed-weight path
$\sum_iw_i\mu_{i,s}$ ends at $\pi$ and has no second term. To realize it
approximately, one would draw a latent label $I$ from a probability vector
$\widehat w=(\widehat w_i)_{i\in[k]}$ that approximates $w$, initialize from a
law $\kappa_I$ on $\Omega_I$ close to $\mu_{I,0}$, and then run the dynamics
restricted to
$\Omega_I$, which rejects every transition leaving $\Omega_I$. We take
$\widehat w$ as given, since how to obtain it depends on the application. The
restricted dynamics is reversible with respect to $\mu_{I,s}$, so the argument
of \zcref{lem:exact-weak-action-comparison} would apply on each $\Omega_i$
separately, and only the action or weak action of the component paths would
matter, not that of $(\rho_s)_{s\in[0,1]}$.

This restricted algorithm needs to know the label $I$ and to test membership
in $\Omega_I$. Our theorem does not assume either. The decomposition
$\pi=\sum_iw_i\pi_i$ is used only in the analysis, and the initializer may be
given as a black box, for example from data. We therefore run the
unrestricted dynamics $(P_s)_{s\in[0,1]}$ from the initial law
$\sum_i\widehat w_i\kappa_i$, that is, \zcref{algo:poisson-annealing} with this
initial law. Its transitions may leave $\Omega_i$ and move probability into
other phases, and \zcref{subsec:localized-unrestricted-comparison} bounds the
error caused by these transitions.

\subsection{A localized comparison for unrestricted dynamics}
\label{subsec:localized-unrestricted-comparison}

We first make precise the restricted dynamics described at the beginning of
this section. For $S\subseteq\Omega$, the \emph{restriction} $P_s^S$ of $P_s$
to $S$ rejects every transition from $S$ to $\bar S$. Explicitly, for
$x,y\in S$,
\begin{equation*}
    P_s^S(x,y)=P_s(x,y)\quad(x\neq y),
    \qquad
    P_s^S(x,x)=P_s(x,x)+\sum_{z\notin S}P_s(x,z).
\end{equation*}
Thus restriction preserves all internal transition rates and adds the exit
probability to the holding probability.

If a path $(\mu_s)_{s\in[0,1]}$ supported on $S$ satisfies internal detailed
balance with $p_s$, then $P_s^S$ is reversible with respect to $\mu_s$, and
the argument of \zcref{lem:exact-weak-action-comparison} bounds the error of
the restricted dynamics using a flux supported on edges inside $S$. The
algorithm, however, runs the unrestricted dynamics $P_s$, which may leave $S$.
The following lemma bridges this gap. Its proof, given in
\zcref{subsec:proof-localized-unrestricted-comparison}, follows the
backward-potential argument of \zcref{lem:exact-weak-action-comparison}, with
one additional term from transitions that leave $S$. This term is controlled
by the integrated outward flow $\Lambda$ defined in the lemma, and contributes
$T\Lambda$ to each bound.

\begin{lemma}[Localized exact- and weak-action comparison]
    \label{lem:localized-continuous-annealing-comparison}
    Let $S\subseteq\Omega$ be fixed, let
    $(\mu_s)_{s\in[0,1]}$ be a probability path that is absolutely continuous
    as a map into $\bb R^\Omega$, strictly positive on $S$, and zero outside
    $S$, and let $(P_s)_{s\in[0,1]}$
    be transition matrices on $\Omega$ such that $s\mapsto P_s$ is
    measurable. Write
    $p_s=P_s-\!I_{\abs{\Omega}}$, assume that the edge set
    $E=\set{\set{x,y}:x\neq y,\ P_s(x,y)>0}$ does not depend on $s$, and
    assume internal detailed balance,
    \begin{equation*}
        \mu_s(x)p_s(x,y)=\mu_s(y)p_s(y,x),
        \qquad x,y\in S.
    \end{equation*}
    Suppose that a measurable family of antisymmetric fluxes $J_s$, supported
    on edges contained in $S$, and signed residuals $r_s$ satisfy, for almost
    every $s$,
    \begin{equation*}
        \partial_s\mu_s+\divergence_{\-{d}}J_s=r_s,
        \qquad \sum_{x\in S}r_s(x)=0,
    \end{equation*}
    together with
    \begin{equation*}
        \frac12\int_0^1\|r_s\|_1\,\dd s\leq\varepsilon,
        \qquad
        \int_0^1\sum_{\substack{\{x,y\}\in E\\x,y\in S}}
          \frac{J_s(x,y)^2}{\mu_s(x)p_s(x,y)}\,\dd s\leq A.
    \end{equation*}
    Define the integrated outward flow
    \begin{equation*}
        \Lambda
        \defeq\int_0^1\sum_{x\in S}\sum_{y\notin S}
          \mu_s(x)p_s(x,y)\,\dd s.
    \end{equation*}
    Let $Q_T$ denote the transition kernel of
    \zcref{algo:poisson-annealing} with running time $T>0$, that is, of the
    unrestricted dynamics with rate kernels $p_{t/T}$, $t\in[0,T]$. For any
    initializer $\kappa$ on $\Omega$, put $\nu=\kappa Q_T$ and
    \begin{equation*}
        K=D_{\mathrm{KL}}(\mu_0\|\kappa),
        \qquad d=D_{\mathrm{TV}}(\mu_0,\kappa).
    \end{equation*}
    Then
    \begin{align}
        D_{\mathrm{TV}}(\mu_1,\nu)
        &\leq
        \varepsilon+T\Lambda
        +\sqrt{\frac12\left(K+\frac{A}{4T}\right)},
        \label{eq:localized-kl-initializer-tv-bound}\\
        D_{\mathrm{TV}}(\mu_1,\nu)
        &\leq
        d+\varepsilon+T\Lambda+\sqrt{\frac{A}{8T}}.
        \label{eq:localized-tv-initializer-bound}
    \end{align}
    If $r_s=0$ for almost every $s$, then also
    \begin{equation}
        D_{\mathrm{KL}}(\mu_1\|\nu)
        \leq K+\frac{A}{4T}+T\Lambda.
        \label{eq:localized-exact-kl-bound}
    \end{equation}
\end{lemma}

\subsection{Assembling the componentwise comparison}
\label{subsec:componentwise-mixture-comparison}

Recall that the component path is the conditional law
$\mu_{i,s}=\rho_s(\,\cdot\mid\Omega_i)$ from
\eqref{eq:conditional-component-path}.  Reversibility of $P_s$ with respect
to $\rho_s$ immediately gives \emph{internal} detailed balance:
\begin{equation*}
    \mu_{i,s}(x)p_s(x,y)
    =\mu_{i,s}(y)p_s(y,x),
    \qquad x,y\in\Omega_i.
\end{equation*}
Accordingly, define the internal capacity
\begin{equation*}
    c_{i,s}(x,y)
    \defeq \mu_{i,s}(x)p_s(x,y)
    =\mu_{i,s}(y)p_s(y,x),
    \qquad \{x,y\}\in E_i.
\end{equation*}

\begin{remark}[Why the component paths are conditional laws]
    \label{rem:component-conditional-law}
    Internal detailed balance determines the component path. Suppose that a
    probability law $\sigma$ is strictly positive on $\Omega_i$, vanishes
    outside $\Omega_i$, and satisfies internal detailed balance with $p_s$.
    On every edge $\{x,y\}\in E_i$, comparison with detailed balance for
    $\rho_s$ gives
    \begin{equation*}
        \frac{\sigma(x)}{\rho_s(x)}=\frac{\sigma(y)}{\rho_s(y)}.
    \end{equation*}
    Since $(\Omega_i,E_i)$ is connected, this ratio is constant on
    $\Omega_i$, and normalization gives $\sigma=\mu_{i,s}$. Without
    connectivity, the connected pieces of $\Omega_i$ could be reweighted
    independently.
\end{remark}

For each component, suppose that there is a measurable family of
antisymmetric internal fluxes $J_{i,s}$ and signed residuals $r_{i,s}$ such
that, for almost every $s$,
\begin{equation*}
    \partial_s\mu_{i,s}
    +\divergence_{\-{d}}J_{i,s}=r_{i,s},
    \qquad
    \sum_{x\in\Omega_i}r_{i,s}(x)=0.
\end{equation*}
Here $J_{i,s}$ is supported on $E_i$, and the divergence is taken on
$\Omega_i$.  Assume that this flux--residual pair
satisfies
\begin{equation}
    \frac12\int_0^1\norm{r_{i,s}}_1\,\dd s\leq\varepsilon_i,
    \qquad
    \int_0^1\sum_{\{x,y\}\in E_i}
       \frac{J_{i,s}(x,y)^2}{c_{i,s}(x,y)}\,\dd s
       \leq A_i.
    \label{eq:component-action-bounds}
\end{equation}
Thus $(J_{i,s},r_{i,s})_{s\in[0,1]}$ is an $L^1$-relaxed flux of
$(\mu_{i,s})_{s\in[0,1]}$ on the graph $(\Omega_i,E_i)$ in the sense of
\zcref{subsec:weak-action}. Its residual budget is at most $\varepsilon_i$,
and its quadratic cost, computed with the capacities $c_{i,s}$ of the
restricted dynamics $P_s^{\Omega_i}$, is at most $A_i$.

The error from running the unrestricted dynamics is controlled by the
outward flow
\begin{equation}
    \lambda^{\rm out}_{i,s}
    \defeq
    \sum_{x\in\Omega_i}\sum_{y\notin\Omega_i}
       \mu_{i,s}(x)p_s(x,y),
    \qquad
    \Lambda_i\defeq\int_0^1\lambda^{\rm out}_{i,s}\,\dd s.
    \label{eq:component-exit-flow}
\end{equation}

Let $\widehat w=(\widehat w_i)_{i\in[k]}$ be a probability vector of
approximate component weights, and let $\kappa_1,\dots,\kappa_k$ be
probability laws on $\Omega$. As described at the beginning of this section,
the algorithm is \zcref{algo:poisson-annealing} started from
\begin{equation}
    \kappa^{\-{ALG}}=\sum_i\widehat w_i\kappa_i.
    \label{eq:component-global-initializer}
\end{equation}
Let $Q_T$ be as in \zcref{lem:localized-continuous-annealing-comparison}, let
$\nu_i=\kappa_iQ_T$, and let $\nu_T=\kappa^{\-{ALG}}Q_T=\sum_i\widehat w_i\nu_i$
be the output law.

\begin{remark}[Implicit labels and black-box initializers]
    \label{rem:component-implicit-label}
    For the analysis, represent the marginal initializer
    \eqref{eq:component-global-initializer} as the $\Omega$-marginal of the
    joint law
    \begin{equation*}
        \widehat\kappa(i,x)=\widehat w_i\kappa_i(x),
        \qquad (i,x)\in[k]\times\Omega.
    \end{equation*}
    The algorithm samples only the marginal law
    $\kappa^{\-{ALG}}$ and never observes the label $i$.  In particular, when the
    components overlap, the label need not be determined by the initial
    state.  Neither this joint representation nor the target decomposition
    needs to be an algorithmic input. The theorem assumes only that such a
    representation exists with small componentwise errors.
\end{remark}

\begin{theorem}[Componentwise exact- and weak-action annealing]
    \label{thm:componentwise-weak-action-annealing}
    Assume the setting at the beginning of this section, let
    $(J_{i,s},r_{i,s})$ be flux--residual pairs as above satisfying
    \eqref{eq:component-action-bounds}, let $\Lambda_i$ be given by
    \eqref{eq:component-exit-flow}, and let $\nu_T$ be the output law of the
    algorithm started from \eqref{eq:component-global-initializer}. Define
    \begin{equation*}
        \delta_w\defeq D_{\mathrm{TV}}(w,\widehat w),
        \qquad
        m_i\defeq\min\{w_i,\widehat w_i\},
        \qquad
        K_i\defeq D_{\mathrm{KL}}(\mu_{i,0}\|\kappa_i).
    \end{equation*}
    Then
    \begin{equation}
        D_{\mathrm{TV}}(\pi,\nu_T)
        \leq
        \delta_w
        +\sum_i m_i\min\left\{1,
          \varepsilon_i+T\Lambda_i
          +\sqrt{\frac12\left(
             K_i+\frac{A_i}{4T}
          \right)}\right\}.
        \label{eq:componentwise-weak-action-kl-init}
    \end{equation}
    If only the total-variation initialization errors
    $d_i\defeq D_{\mathrm{TV}}(\mu_{i,0},\kappa_i)$ are available, then
    \begin{equation}
        D_{\mathrm{TV}}(\pi,\nu_T)
        \leq
        \delta_w
        +\sum_i m_i\min\left\{1,
          d_i+\varepsilon_i+T\Lambda_i
          +\sqrt{\frac{A_i}{8T}}\right\}.
        \label{eq:componentwise-weak-action-tv-init}
    \end{equation}
    The better of the two component bounds may be used separately for each
    $i$.

    If $r_{i,s}=0$ for every $i$ and almost every $s$, and
    $w\ll\widehat w$, then there is also the exact-action bound
    \begin{equation}
        D_{\mathrm{KL}}(\pi\|\nu_T)
        \leq
        D_{\mathrm{KL}}(w\|\widehat w)
        +\sum_i w_i\left(
           K_i+\frac{A_i}{4T}+T\Lambda_i
        \right).
        \label{eq:componentwise-exact-action-kl}
    \end{equation}
\end{theorem}

\begin{proof}
    Fix a component \(i\).  Apply
    \zcref{lem:localized-continuous-annealing-comparison} with
    \[
        S=\Omega_i,\qquad
        \mu_s=\mu_{i,s},\qquad
        (J_s,r_s)=(J_{i,s},r_{i,s}),\qquad
        \kappa=\kappa_i.
    \]
    Its hypotheses hold. The path $\mu_{i,s}$ is strictly positive on
    $\Omega_i$ and zero outside it, and it is absolutely continuous as a map
    into $\bb R^\Omega$, since $\rho_s$ is and $s\mapsto\rho_s(\Omega_i)$ is
    absolutely continuous and bounded below by a positive constant on $[0,1]$.
    Internal detailed balance was shown above, and the flux--residual pair
    satisfies \eqref{eq:component-action-bounds}. With $Q_T$ as in that lemma,
    $\nu_i=\kappa_iQ_T$.  The two TV conclusions of that
    lemma give, respectively,
    \begin{align*}
        D_{\mathrm{TV}}(\mu_{i,1},\nu_i)
        &\leq
        \varepsilon_i+T\Lambda_i
        +\sqrt{\frac12\left(K_i+\frac{A_i}{4T}\right)},
        \\
        D_{\mathrm{TV}}(\mu_{i,1},\nu_i)
        &\leq
        d_i+\varepsilon_i+T\Lambda_i+\sqrt{\frac{A_i}{8T}}.
    \end{align*}
    If \(r_{i,s}=0\) for almost every $s$, its exact conclusion gives
    \begin{equation}
        D_{\mathrm{KL}}(\mu_{i,1}\|\nu_i)
        \leq K_i+\frac{A_i}{4T}+T\Lambda_i.
        \label{eq:component-branch-exact-kl}
    \end{equation}

    It remains only to aggregate these branchwise comparisons.  For
    probability laws \(P,Q\) and nonnegative scalars \(a,b\),
    \begin{equation}
        \frac12\|aP-bQ\|_1
        \leq\frac12|a-b|+\min\{a,b\}D_{\mathrm{TV}}(P,Q).
        \label{eq:weighted-mixture-tv-inequality}
    \end{equation}
    Applying \eqref{eq:weighted-mixture-tv-inequality} label by label, the
    first terms sum to \(\delta_w\), while the second terms give the sum
    weighted by \(m_i\).  Clipping each branch error at one proves
    \eqref{eq:componentwise-weak-action-kl-init} and
    \eqref{eq:componentwise-weak-action-tv-init}.  Equivalently, this
    calculation takes place on the labelled space
    \([k]\times\Omega\) and then applies the map
    \((i,x)\mapsto x\), so overlap of the phase domains causes no problem.

    In the exact case, the chain rule for relative entropy on that labelled
    space gives
    \begin{equation*}
        D_{\mathrm{KL}}\left(
          (w_i\mu_{i,1})_i\,\middle\|\,
          (\widehat w_i\nu_i)_i
        \right)
        =D_{\mathrm{KL}}(w\|\widehat w)
         +\sum_i w_iD_{\mathrm{KL}}(\mu_{i,1}\|\nu_i).
    \end{equation*}
    Forgetting the label can only decrease relative entropy.  Substituting
    \eqref{eq:component-branch-exact-kl} proves
    \eqref{eq:componentwise-exact-action-kl}.
\end{proof}

If $\widehat w=w$ and $\kappa_i=\mu_{i,0}$ for every $i$, so that the
algorithm is initialized exactly according to $\sum_iw_i\mu_{i,0}$, then
$\delta_w=0$ and $d_i=0$ for every $i$, so
\eqref{eq:componentwise-weak-action-tv-init} reduces to
\begin{equation*}
    D_{\mathrm{TV}}(\pi,\nu_T)
    \leq
    \sum_i w_i\left(
       \varepsilon_i+T\Lambda_i
       +\sqrt{\frac{A_i}{8T}}
    \right).
\end{equation*}
The action term decreases like $T^{-1/2}$, whereas the outward-flow term
$T\Lambda_i$ grows like $T$. Thus the bound is useful when $T$ can be chosen
so that both $A_i/T$ and $T\Lambda_i$ are small, which requires
$A_i\Lambda_i\ll1$.

\zcref{thm:componentwise-weak-action-annealing} needs no lower bound on
$\min_i w_i$. A missing component is paid for
through $\delta_w$, and a very small component contributes only in proportion
to its weight in the target mixture.

\begin{corollary}[Random or data-dependent initialization]
    \label{cor:data-initialized-componentwise-annealing}
    Let a random dataset or advice variable $D$ determine a marginal
    initializer $\kappa^{\-{ALG}}(D)$, and suppose that, given $D$, the
    annealing dynamics uses randomness independent of $D$.  Suppose also
    that, for the analysis, $\kappa^{\-{ALG}}(D)$ admits a possibly
    unobserved representation
    \begin{equation*}
        \kappa^{\-{ALG}}(D)
        =\sum_i\widehat w_i(D)\kappa_i(D).
    \end{equation*}
    Conditional on every realization of $D$, the bounds in
    \zcref{thm:componentwise-weak-action-annealing} hold deterministically.
    In particular, if with probability at least $1-\delta$ the right-hand
    side of \eqref{eq:componentwise-weak-action-tv-init} is at most $\tau$,
    then, with probability at least $1-\delta$ over $D$, the
    conditional output law is within $\tau$ of $\pi$ in total variation.
\end{corollary}

\begin{remark}[Relation to data-initialized mixture sampling]
    \label{rem:data-initialized-comparison}
    The fixed-target results of \cite{KLV25} and \cite[Section~5]{HMRW25}
    study empirical or advice-based initialization for a mixture whose
    components satisfy Poincar\'e inequalities.  They quantify the amount of
    data needed for the initialization to balance the slow component
    directions, and \cite{KLV25} does so more generally through a
    higher-order spectral gap.

    Our emphasis is complementary.
    \zcref{cor:data-initialized-componentwise-annealing} does not bound the
    size of $D$ or the number of samples needed to construct a suitable
    initializer. It only transfers high-probability initialization bounds
    into a sampling guarantee.  Instead,
    \zcref{thm:componentwise-weak-action-annealing} replaces the
    componentwise Poincar\'e assumption by exact- or weak-action bounds along
    time-dependent component paths.  It allows the initializer to come from
    data, symmetry, or approximate counting, and focuses on the subsequent
    annealing dynamics, including the outward-flow term $T\Lambda_i$.
\end{remark}

\subsection{How the applications instantiate the theorem}
\label{subsec:componentwise-mixture-applications}

To apply \zcref{thm:componentwise-weak-action-annealing}, each application
verifies the following quantities and conditions.
\begin{enumerate}[label=\textnormal{(\Roman*)}]
    \item \emph{Target decomposition:} choose possibly overlapping phase
    domains $\Omega_i$ with connected induced graphs $(\Omega_i,E_i)$ and a
    reference path $(\rho_s)_{s\in[0,1]}$ such that the target is
    $\pi=\sum_iw_i\pi_i$ with $\pi_i=\rho_1(\,\cdot\mid\Omega_i)$.

    \item \emph{Implementable initialization:} compute the weights, or
    determine them by symmetry, so as to bound $\delta_w$, and construct
    $\kappa_i$ so as to bound $d_i$ or $K_i$.

    \item \emph{Reversibility along the reference path:} verify that each
    $P_s$ is reversible with respect to $\rho_s$, that $s\mapsto P_s$ is
    measurable, and that the edge set does not depend on $s$.  Internal
    detailed balance for $\mu_{i,s}$ then follows automatically.

    \item \emph{Internal transport:} build a componentwise flux--residual
    pair and record bounds $A_i$ on its quadratic cost and $\varepsilon_i$ on
    its residual budget.

    \item \emph{Stability under the unrestricted dynamics:} bound the integrated
    outward flow $\Lambda_i$ and choose $T$ so that both
    $T\Lambda_i$ and $A_i/T$ are small.
\end{enumerate}
They can then be substituted directly into
\eqref{eq:componentwise-weak-action-kl-init},
\eqref{eq:componentwise-weak-action-tv-init}, or
\eqref{eq:componentwise-exact-action-kl}.

In several applications, symmetry determines the phase weights exactly and
allows the initializer to reproduce them without estimation or counting. In
particular, $\widehat w=w$ and $\delta_w=0$.  When the fixed-weight mixture
path also equals $\rho_s$ throughout the schedule, one may combine the
phasewise flows, bound the action or weak action of this full path, and apply
the one-path theorem \zcref{thm:annealing-ub} directly. This corresponds to
the case $k=1$ of \zcref{thm:componentwise-weak-action-annealing} with
$\Omega_1=\Omega$ and $\Lambda_1=0$.
Alternatively, one may retain the component decomposition and apply
\zcref{thm:componentwise-weak-action-annealing} with the known weights
$\widehat w=w$.

We use the one-path theorem when symmetry makes the full-space bound simpler,
or when a negligible phase is absorbed into the weak residual.  A nontrivial
component split is useful when initialization removes a costly phase-weight
direction without introducing a large outward flow.  The high-fugacity
application in \zcref{subsec:polymer-informed-annealing} gives a concrete
two-component implementation.

\section[Functional Inequalities and Path Routings
  (a.k.a. How to Bound Action)]{Functional Inequalities and Path Routings\\
  (a.k.a. How to Bound Action)}
\label{sec:functional-inequalities-and-action}

The preceding sections reduce sampling to bounds on the action or weak action
of prescribed distribution paths.  This section gives two reusable ways to
prove such bounds.  \zcref{subsection:transport-variance-inequalities} shows
that a Poincar\'e inequality along the path controls the exact action, and a
weak Poincar\'e inequality controls the weak action.\footnote{The converse
implications fail.  Exact and weak action can allow completely separated
modes of equal, pathwise constant weight, whereas PI and any nontrivial WPI
cannot.}  \zcref{subsec:path-routing-action} instead constructs the flux
directly by routing the path derivative along paths in the state graph.
Finally, \zcref{sec:wpi-poincare-core} gives a static interpretation of a weak
Poincar\'e inequality, namely that after a small exceptional set is removed,
the restricted chain satisfies a true Poincar\'e inequality.  Intuitively, the
exceptional set plays the role of the residual in weak action, which helps
explain why a
Poincar\'e inequality corresponds to exact action and a weak Poincar\'e
inequality to weak action.  This result also partially answers a question of
\cite{HMRW25}.

The arguments of this section apply to a full Gibbs path and also to a
component path whose dynamics is restricted to a phase domain.  In the latter
case, the restriction is used only to obtain bounds of the form
\eqref{eq:component-action-bounds}, and the algorithm in
\zcref{sec:componentwise-mixture-annealing} still runs the unrestricted
dynamics and separately accounts for transitions leaving the phase.

\subsection{Functional Inequalities Imply Action Bounds}
\label{subsection:transport-variance-inequalities}

We first establish a variational formula for the metric derivative.

\begin{lemma}[Variational formula for the metric derivative]
    \label{lem:instantaneous-action-variational-formula}
    Let \(\tp{\pi_s}_{s\in[0,1]}\) be an absolutely continuous curve of probability measures in the metric space $\tp{\+P_{>0}(\Omega),W_2}$. For each $s\in[0,1]$, let $p_s$ be a transition rate kernel reversible with respect to $\pi_s$, and suppose that $(p_s)_{s\in[0,1]}$ is measurable.  Assume that the capacity family of \zcref{subsection:the-extension-to-discrete-space} is compatible with these kernels along the path, that is, $c_{\pi_s}(x,y)=\pi_s(x)p_s(x,y)$ for $x\neq y$ and $s\in[0,1]$. Then, for almost every $s\in[0,1]$,
    \[
        \abs{\dot\pi}_s^2
        =\sup_{\substack{f:\Omega\to\bb R\\
        \+E_{p_s}\tp{f,f}>0}}
        \frac{\inner{\partial_s\log\pi_s}{f}_{\pi_s}^2}
        {\+E_{p_s}\tp{f,f}}.
    \]
\end{lemma}

\begin{proof}
    Fix $s\in[0,1]$ at which $\partial_s\pi_s$ exists and the potential formulation in \zcref{thm:metric-derivative-potential-formulation} holds, and abbreviate $\pi=\pi_s$, $p=p_s$, and $c=c_{\pi_s}$.

    \paragraph{Step 1: The metric derivative as a Dirichlet energy.}
    Since $c(x,y)=\pi(x)p(x,y)$, the discrete continuity equation for a
    potential $\phi$ can be written as
    \[
        \partial_s\pi_s(x)+\pi(x)\+L^p\phi(x)=0,
        \qquad x\in\Omega,
        \qquad\text{or equivalently}\qquad
        -\+L^p\phi=\partial_s\log\pi_s.
    \]
    Thus, by \zcref{thm:metric-derivative-potential-formulation} and
    \zcref{prop:generator-and-dirichlet-form-of-discrete-chains}, the metric
    derivative is the Dirichlet energy of any potential driving the continuity
    equation:
    \[
        \abs{\dot\pi}_s^2=\+E_p\tp{\phi,\phi}
        \qquad\text{whenever } -\+L^p\phi=\partial_s\log\pi_s.
    \]
    Our goal is to represent this energy as a dual variational supremum over test functions.

    \paragraph{Step 2: Well-posedness of the Poisson equation.}
    We next show that such a potential exists and is unique up to an additive constant. Probability conservation gives
    \[
        \E[\pi]{\partial_s\log\pi_s}
        = \sum_{x\in\Omega}\partial_s\pi_s\tp{x}
        = \partial_s(1) = 0.
    \]
    Thus $\partial_s\log\pi_s$ belongs to the mean-zero subspace
    \[
        L_0^2\tp{\pi}
        \defeq
        \set{f:\Omega\to\bb R : \inner{f}{\mathbf 1}_{\pi} = \E[\pi]{f} = 0}.
    \]
    Reversibility implies that $-\+L^p$ is self-adjoint on $L^2\tp{\pi}$ with
    \[
        \begin{aligned}
        \inner{f}{-\+L^p h}_{\pi}
        &= \+E_p\tp{f,h}\\
        &= \frac12\sum_{x,y\in\Omega}c\tp{x,y}
           \tp{f\tp{y}-f\tp{x}}\tp{h\tp{y}-h\tp{x}}\\
        &= \inner{-\+L^p f}{h}_{\pi}.
        \end{aligned}
    \]
    Since $\+L^p\mathbf 1 = 0$, $-\+L^p$ maps $L_0^2\tp{\pi}$ into itself. Because the graph is connected and all edge capacities are strictly positive, $\+E_p\tp{f,f}=0$ forces $f$ to be constant; its mean-zero condition then forces $f=0$. Hence the restriction $-\+L^p:L_0^2\tp{\pi}\to L_0^2\tp{\pi}$ is injective, and by finite dimensionality, an isomorphism. There is thus a unique $\psi\in L_0^2\tp{\pi}$ satisfying the Poisson equation
    \[
        -\+L^p\psi = \partial_s\log\pi_s.
    \]
    By Step~1,
    \[
        \abs{\dot\pi}_s^2 = \+E_p\tp{\psi,\psi}.
    \]

    \paragraph{Step 3: Duality via integration by parts.}
    Now consider any test function $f:\Omega\to\bb R$ with $\+E_p\tp{f,f}>0$. Substituting the Poisson equation $-\+L^p\psi = \partial_s\log\pi_s$ and using self-adjointness, the numerator pairing is precisely the Dirichlet inner product between $\psi$ and $f$:
    \[
        \inner{\partial_s\log\pi_s}{f}_{\pi}
        = \inner{-\+L^p\psi}{f}_{\pi}
        = \+E_p\tp{\psi,f}
        = \sum_{\set{x,y}\in E}c\tp{x,y}\tp{\psi\tp{y}-\psi\tp{x}}\tp{f\tp{y}-f\tp{x}}.
    \]
    Applying the Cauchy--Schwarz inequality to this edge sum gives
    \[
        \inner{\partial_s\log\pi_s}{f}_{\pi}^2
        = \+E_p\tp{\psi,f}^2
        \leq \+E_p\tp{\psi,\psi}\+E_p\tp{f,f}
        = \abs{\dot\pi}_s^2\+E_p\tp{f,f}.
    \]
    Dividing by $\+E_p\tp{f,f}$ and taking the supremum over all test functions with $\+E_p\tp{f,f}>0$ yields
    \[
        \sup_{\substack{f:\Omega\to\bb R\\ \+E_p\tp{f,f}>0}}
        \frac{\inner{\partial_s\log\pi_s}{f}_{\pi}^2}{\+E_p\tp{f,f}}
        \leq \abs{\dot\pi}_s^2.
    \]

    \paragraph{Step 4: Sharpness of the variational bound.}
    It remains to show that the supremum is attained. If $\partial_s\pi_s=0$, then both sides vanish. If $\partial_s\pi_s\neq 0$, the potential $\psi$ is nonconstant, so $\+E_p\tp{\psi,\psi}>0$. Choosing $f=\psi$ as the test function gives
    \[
        \frac{\inner{\partial_s\log\pi_s}{\psi}_{\pi}^2}{\+E_p\tp{\psi,\psi}}
        = \frac{\+E_p\tp{\psi,\psi}^2}{\+E_p\tp{\psi,\psi}}
        = \+E_p\tp{\psi,\psi}
        = \abs{\dot\pi}_s^2.
    \]
    This confirms that the supremum equals $\abs{\dot\pi}_s^2$, completing the proof.
\end{proof}

Recall that the optimal Poincar\'e constant of the pair $(p_s,\pi_s)$ has the variational form
\[
    C_{\-{PI},s}
    =\sup_{\substack{f:\Omega\to\bb R\\
    \+E_{p_s}\tp{f,f}>0}}
    \frac{\Var[\pi_s]{f}}
    {\+E_{p_s}\tp{f,f}}.
\]
The Poincar\'e inequality controls $\Var[\pi_s]{f}$ for every test function
$f$.  By contrast, the variational formula for the metric derivative depends
on $f$ only through its pairing with the fixed path derivative
$\partial_s\log\pi_s$.  In particular, a test function orthogonal to this
derivative makes no contribution to the numerator, even if its Dirichlet
energy is small.  Thus action is adapted to the chosen distribution path and
can be substantially smaller than a worst-case Poincar\'e bound.

\begin{proposition}[Action bounds from Poincar\'e inequalities]
    \label{prop:PI-implies-action}
    Let \(\tp{\pi_s}_{s\in[0,1]}\) be an absolutely continuous curve of probability measures in the metric space $\tp{\+P_{>0}(\Omega), W_2}$.  For each $s\in[0,1]$, let $p_s$ be a transition rate kernel reversible with respect to $\pi_s$, and suppose that $(p_s)_{s\in[0,1]}$ is measurable.  Assume that the capacity family of
    \zcref{subsection:the-extension-to-discrete-space} satisfies
    \[
        c_s\tp{x,y}\defeq c_{\pi_s}\tp{x,y}
        =\pi_s\tp{x}p_s\tp{x,y}
        =\pi_s\tp{y}p_s\tp{y,x},
        \qquad x\neq y,\ s\in[0,1].
    \]
    Suppose that, for almost
    every $s$, the pair $(p_s,\pi_s)$ satisfies a Poincar\'e inequality with a
    measurable constant $C_{\-{PI},s}$.
    Then, for almost every $s\in[0,1]$,
    \[
        \abs{\dot{\pi}}_s^2
        \leq C_{\-{PI},s}
        \norm{\partial_s\log\pi_s}_{L^2\tp{\pi_s}}^2.
    \]
    Consequently,
    \[
        \+A\tp{\tp{\pi_s}_{s\in[0,1]}}
        \leq\int_0^1 C_{\-{PI},s}
        \norm{\partial_s\log\pi_s}_{L^2\tp{\pi_s}}^2\,\dd s.
    \]
    In particular, if $C_{\-{PI},s}\leq C_{\-{PI}}$ along the path, then
    \[
        \+A\tp{\tp{\pi_s}_{s\in[0,1]}}
        \leq C_{\-{PI}}\int_0^1
        \norm{\partial_s\log\pi_s}_{L^2\tp{\pi_s}}^2\,\dd s.
    \]
\end{proposition}

\begin{proof}
    Fix $s$ at which the variational formula in
    \zcref{lem:instantaneous-action-variational-formula} holds. Since
    $\E[\pi_s]{\partial_s\log\pi_s}=0$, Cauchy--Schwarz and the Poincar\'e inequality give, for every $f:\Omega\to\bb R$,
    \[
        \begin{aligned}
        \inner{\partial_s\log\pi_s}{f}_{\pi_s}^2
        &=\inner{\partial_s\log\pi_s}
        {f-\E[\pi_s]{f}}_{\pi_s}^2\\
        &\leq
        \norm{\partial_s\log\pi_s}_{L^2\tp{\pi_s}}^2
        \Var[\pi_s]{f}\\
        &\leq C_{\-{PI},s}
        \norm{\partial_s\log\pi_s}_{L^2\tp{\pi_s}}^2
        \+E_{p_s}\tp{f,f}.
        \end{aligned}
    \]
    Dividing by $\+E_{p_s}\tp{f,f}$ and taking the supremum over functions with positive Dirichlet energy proves the pointwise bound by
    \zcref{lem:instantaneous-action-variational-formula}. Integrating over $[0,1]$ proves the action bound.
\end{proof}

There is an analogous implication from weak Poincar\'e inequalities to weak
action.  Its proof, given in
\zcref{subsection:appendix-transport-inequalities}, separates the path
derivative into a part controlled by the Dirichlet form and an $L^1$
residual.

\begin{proposition}[Weak-action bounds from weak Poincar\'e inequalities]
    \label{prop:WPI-implies-weak-action}
    Let $\tp{\pi_s}_{s\in[0,1]}\subseteq\+P_{>0}\tp{\Omega}$ be an
    absolutely continuous curve in the metric space
    $\tp{\+P_{>0}(\Omega), W_2}$ such that
    \[
        \int_0^1\norm{\partial_s\log\pi_s}_{L^2\tp{\pi_s}}^2\,\dd s
        <\infty.
    \]
    For each $s$, let $p_s$ be a transition rate kernel reversible with
    respect to $\pi_s$, and suppose that $(p_s)_{s\in[0,1]}$ is measurable.
    Assume that the capacity family of
    \zcref{subsection:the-extension-to-discrete-space} satisfies
    $c_s\tp{x,y}\defeq c_{\pi_s}\tp{x,y}=\pi_s\tp{x}p_s\tp{x,y}$ for
    $x\neq y$ and $s\in[0,1]$.  Suppose that, for some
    $C_{\-{WPI}},\delta_{\-{WPI}}\geq0$ and almost every $s\in[0,1]$, the pair
    $(p_s,\pi_s)$ satisfies the WPI
    \begin{equation}
        \Var[\pi_s]{f}
        \leq C_{\-{WPI}}\cdot\+E_{p_s}\tp{f,f}
        +\delta_{\-{WPI}}\,\operatorname{Osc}\tp{f}^2,
        \qquad f:\Omega\to\bb R.
        \label{eq:wpi-for-l1-weak-action}
    \end{equation}
    Then, for almost every $s$, there exists an $L^1$-relaxed flux
    $\tp{J_s,r_s}$ such that
    \begin{align}
        \sum_{\set{x,y}\in E}
        \frac{J_s\tp{x,y}^2}{c_s\tp{x,y}}
        &\leq C_{\-{WPI}}
        \norm{\partial_s\log\pi_s}_{L^2\tp{\pi_s}}^2,
        \label{eq:wpi-good-flux-cost}\\
        \frac12\norm{r_s}_1
        &\leq\sqrt{\delta_{\-{WPI}}}\,
        \norm{\partial_s\log\pi_s}_{L^2\tp{\pi_s}}.
        \label{eq:wpi-residual-cost}
    \end{align}
    Consequently, setting
    \[
        \varepsilon
        \defeq
        \sqrt{\delta_{\-{WPI}}}\int_0^1
        \norm{\partial_s\log\pi_s}_{L^2\tp{\pi_s}}\,\dd s
    \]
    gives
    \begin{equation}
        \+A^{\mathrm{w}}_{\varepsilon}
        \tp{\tp{\pi_s}_{s\in[0,1]}}
        \leq C_{\-{WPI}}\int_0^1
        \norm{\partial_s\log\pi_s}_{L^2\tp{\pi_s}}^2\,\dd s.
        \label{eq:wpi-integrated-weak-action}
    \end{equation}
\end{proposition}

\subsection{Direct Action Bounds from Path Routings}
\label{subsec:path-routing-action}

Functional inequalities bound action by first controlling every test
function.  When suitable paths in the state graph are known explicitly, we
can instead use them to construct a flux for the particular derivative of
the distribution path.  The following elementary observation is the common
starting point.  Fix an arbitrary orientation of every edge.  For an oriented
path $\gamma$, let $\gamma^-$ and $\gamma^+$ denote its initial and terminal
vertices, and let $\chi_\gamma(e)\in\set{-1,0,1}$ record whether $\gamma$
uses $e$ against its fixed orientation, does not use it, or uses it in its
fixed orientation.

\begin{lemma}[Path routings produce fluxes]
    \label{lem:path-superposition-action}
    Let $\tp{\pi_s}_{s\in[0,1]}\subseteq\+P_{>0}(\Omega)$ be absolutely
    continuous, and write $c_s=c_{\pi_s}$.  Let $(r_s)_{s\in[0,1]}$ be a measurable
    family of signed functions of total mass zero, and set, for almost every
    $s$,
    $$
        b_s=\partial_s\pi_s-r_s.
    $$
    Suppose there are measurable coefficients $a_s(\gamma)\in\bb R$, indexed
    by the oriented simple paths in $\tp{\Omega,E}$, such that, for almost
    every $s$ and every $x\in\Omega$,
    \begin{equation}
        \sum_{\gamma:\,\gamma^-=x}a_s(\gamma)
        -\sum_{\gamma:\,\gamma^+=x}a_s(\gamma)
        =-b_s(x).
        \label{eq:path-superposition-balance}
    \end{equation}
    Define the absolute load on an edge by
    $$
        L_s(e)=\sum_{\gamma:\,e\in\gamma}\abs{a_s(\gamma)}.
    $$
    Then, for every
    $\varepsilon\geq\frac12\int_0^1\norm{r_s}_1\,\dd s$,
    \begin{equation*}
        \+A^{\mathrm w}_\varepsilon
        \tp{\tp{\pi_s}_{s\in[0,1]}}
        \leq
        \int_0^1\sum_{e\in E}\frac{L_s(e)^2}{c_s(e)}\,\dd s.
    \end{equation*}
\end{lemma}

\begin{proof}
    Put $J_s(e)=\sum_\gamma a_s(\gamma)\chi_\gamma(e)$.  If the fixed
    orientation of $e$ is from $x$ to $y$, interpret this scalar as
    $J_s(x,y)=J_s(e)$ and $J_s(y,x)=-J_s(e)$.  Internal vertices of each path
    have zero divergence,
    while its two endpoints contribute
    $a_s(\gamma)$ and $-a_s(\gamma)$.  Thus
    \eqref{eq:path-superposition-balance} gives
    $\divergence_{\-{d}}J_s=-b_s$, or equivalently
    $\partial_s\pi_s+\divergence_{\-{d}}J_s=r_s$.
    Moreover, $\abs{J_s(e)}\leq L_s(e)$.  The claim follows by inserting this
    relaxed flux into the definition of weak action.
\end{proof}

Canonical paths give one systematic choice of the coefficients in the
preceding lemma.  The statement below includes a chosen residual, since the
same construction can route only the part of the derivative that remains
after that residual is removed.

\begin{proposition}[Action bounds from canonical paths]
    \label{prop:canonical-path-action}
    Let $\tp{\pi_s}_{s\in[0,1]}$ and $c_s$ be as in
    \zcref{lem:path-superposition-action}.  For each unordered pair of
    distinct states $\set{x,y}\subseteq\Omega$, choose, measurably in $s$, a
    simple path $\gamma_{xy,s}$ joining $x$ and $y$.  Define its congestion by
    \begin{equation}
        \Gamma_s(e)
        =\frac1{c_s(e)}
        \sum_{\substack{\set{x,y}\subseteq\Omega:\\
                        x\neq y,\ e\in\gamma_{xy,s}}}
        \pi_s(x)\pi_s(y)\abs{\gamma_{xy,s}},
        \qquad
        \Gamma_s=\max_{e\in E}\Gamma_s(e).
        \label{eq:canonical-path-action-congestion}
    \end{equation}
    Let $(r_s)_{s\in[0,1]}$ be any measurable family of signed functions of
    total mass zero, set $b_s=\partial_s\pi_s-r_s$, and write
    $h_s=b_s/\pi_s$.  Then every
    $\varepsilon\geq\frac12\int_0^1\norm{r_s}_1\,\dd s$ satisfies
    \begin{equation}
        \+A^{\mathrm w}_\varepsilon
        \tp{\tp{\pi_s}_{s\in[0,1]}}
        \leq
        \int_0^1\Gamma_s\norm{h_s}_{L^2\tp{\pi_s}}^2\,\dd s.
        \label{eq:canonical-path-weak-action}
    \end{equation}
    In particular, taking $r_s=0$ gives
    \begin{equation}
        \+A\tp{\tp{\pi_s}_{s\in[0,1]}}
        \leq
        \int_0^1\Gamma_s
        \norm{\partial_s\log\pi_s}_{L^2\tp{\pi_s}}^2\,\dd s.
        \label{eq:canonical-path-exact-action}
    \end{equation}
\end{proposition}

\begin{proof}
    Fix $s$ and abbreviate $\pi=\pi_s$, $b=b_s$, and $h=h_s$.  Since
    $\sum_xb(x)=0$, we have $\E[\pi]{h}=0$.  Give each unordered pair one of
    its two orientations.  If the chosen orientation is from $x$ to $y$,
    send the signed amount
    $$
        \pi(x)\pi(y)\tp{h(y)-h(x)}
    $$
    along $\gamma_{xy,s}$.  For a pair $\set{x,y}$, its contribution to the
    divergence at $x$ is
    $\pi(x)\pi(y)\tp{h(y)-h(x)}$, regardless of which orientation was chosen.
    The endpoint contributions at $x$ therefore sum to
    $$
        \sum_{y\neq x}\pi(x)\pi(y)\tp{h(y)-h(x)}
        =-\pi(x)h(x)=-b(x).
    $$
    Hence these routed amounts satisfy
    \eqref{eq:path-superposition-balance}.  Their absolute edge load obeys,
    by Cauchy--Schwarz,
    \begin{align*}
        L_s(e)^2
        &\leq
        \tp{\sum_{\substack{\set{x,y}\subseteq\Omega:\,x\neq y,\\
                            e\in\gamma_{xy,s}}}
              \pi(x)\pi(y)\abs{\gamma_{xy,s}}}
        \tp{\sum_{\substack{\set{x,y}\subseteq\Omega:\,x\neq y,\\
                            e\in\gamma_{xy,s}}}
              \frac{\pi(x)\pi(y)\tp{h(y)-h(x)}^2}
                   {\abs{\gamma_{xy,s}}}}.
    \end{align*}
    Divide by $c_s(e)$, sum over $e$, and use
    \eqref{eq:canonical-path-action-congestion}.  Each pair contributes once
    for every edge of its path, and therefore
    \begin{align*}
        \sum_{e\in E}\frac{L_s(e)^2}{c_s(e)}
        &\leq \Gamma_s
        \sum_{\substack{\set{x,y}\subseteq\Omega:\,x\neq y}}
        \pi(x)\pi(y)\tp{h(y)-h(x)}^2\\
        &=\Gamma_s\Var[\pi]{h}
        =\Gamma_s\norm{h}_{L^2\tp{\pi}}^2.
    \end{align*}
    The second line uses a sum over unordered pairs, so no factor of two is
    present.  Integrating and applying
    \zcref{lem:path-superposition-action} proves
    \eqref{eq:canonical-path-weak-action}; the exact case follows by taking
    $r_s=0$.
\end{proof}

The same proof allows each commodity to be split among several paths.  More
precisely, let $Q_{xy,s}$ be a probability distribution on paths joining
$x$ and $y$, measurable in $s$, route the preceding signed amount according
to this distribution, and define
$$
    \Gamma_s^{\mathrm{mc}}(e)
    =\frac1{c_s(e)}
    \sum_{\substack{\set{x,y}\subseteq\Omega:\,x\neq y}}
    \pi_s(x)\pi_s(y)
    \sum_{\gamma:\,e\in\gamma}
    Q_{xy,s}(\gamma)\abs{\gamma},
    \qquad
    \Gamma_s^{\mathrm{mc}}
    =\max_{e\in E}\Gamma_s^{\mathrm{mc}}(e).
$$
Then \eqref{eq:canonical-path-weak-action} and
\eqref{eq:canonical-path-exact-action} hold with $\Gamma_s^{\mathrm{mc}}$ in
place of $\Gamma_s$.  This is the usual fractional multicommodity-flow
version of the canonical-path bound.

For comparison, when $c_s(x,y)=\pi_s(x)p_s(x,y)$ for a rate kernel $p_s$
reversible with respect to $\pi_s$, the congestion $\Gamma_s$ in
\eqref{eq:canonical-path-action-congestion} also gives the standard
Poincar\'e bound $C_{\-{PI},s}\leq\Gamma_s$ for $(p_s,\pi_s)$; see, for
example, \cite[Chapter~13]{LP17}.  Thus, in the exact case,
\zcref{prop:canonical-path-action} gives the same numerical estimate as
first proving a Poincar\'e inequality and then applying
\zcref{prop:PI-implies-action}.  The direct proof constructs the flux
explicitly and allows a prescribed part of the path derivative to be left in
the residual.

A different useful routing sends every contribution to one root.  Unlike the
all-pairs construction, its load is expressed directly in terms of the part
of the derivative being transported.

\begin{proposition}[Action bounds from rooted paths]
    \label{prop:rooted-path-action}
    Let $\tp{\pi_s}_{s\in[0,1]}$, $c_s$, $(r_s)_{s\in[0,1]}$, and $b_s$ be as
    in \zcref{lem:path-superposition-action}, and suppose that $b_s$ is
    supported on $G_s\subseteq\Omega$.  Choose a
    root $z_s\in G_s$ and, for every $x\in G_s\setminus\set{z_s}$, a simple
    path $\gamma_{x,s}\subseteq G_s$ from $x$ to $z_s$.  Assume these choices
    are measurable in $s$, and define
    $$
        L_s(e)=
        \sum_{\substack{x\in G_s\setminus\set{z_s}:\\
                        e\in\gamma_{x,s}}}
        \abs{b_s(x)}.
    $$
    Then, for every
    $\varepsilon\geq\frac12\int_0^1\norm{r_s}_1\,\dd s$,
    \begin{equation}
        \+A^{\mathrm w}_\varepsilon
        \tp{\tp{\pi_s}_{s\in[0,1]}}
        \leq
        \int_0^1\sum_{e\in E}\frac{L_s(e)^2}{c_s(e)}\,\dd s.
        \label{eq:rooted-path-action-bound}
    \end{equation}
    Taking $r_s=0$ gives the corresponding exact-action bound.
\end{proposition}

\begin{proof}
    For each $x\in G_s\setminus\set{z_s}$, orient $\gamma_{x,s}$ from $x$ to
    $z_s$ and assign it the coefficient $-b_s(x)$.  The balance condition
    \eqref{eq:path-superposition-balance} holds away from the root by
    construction, and at the root it follows from $\sum_xb_s(x)=0$.  The
    resulting absolute edge load is $L_s(e)$, so the claim follows from
    \zcref{lem:path-superposition-action}.
\end{proof}

\subsection{From WPI to PI on a High-Probability Restriction}
\label{sec:wpi-poincare-core}

A WPI yields a weak-action bound dynamically through
\zcref{prop:WPI-implies-weak-action}.  There is also a useful static
interpretation.  A small exceptional part of the state space can be removed
so that the original dynamics, restricted to the remaining high-probability
set, satisfies a true Poincar\'e inequality.  The exceptional mass is the
static counterpart of the residual that weak action leaves untransported.

A Poincar\'e inequality is a property of a distribution together with a
Dirichlet form, rather than of the distribution alone, so the point is that
the restricted chain keeps the original local transitions.  Throughout this
subsection, let $p$ be an irreducible transition rate kernel on a finite state space $\Omega$ with $\abs{\Omega}\geq2$, reversible with respect to the distribution $\pi$, and write
\[
    c\tp{x,y} \defeq \pi\tp{x}p\tp{x,y},
    \qquad
    c\tp{A,B} \defeq \sum_{x\in A,\,y\in B}c\tp{x,y}.
\]
The following trimming lemma shows that, geometrically, WPI implies good conductance outside a small exceptional set.

\begin{lemma}[Trimming lemma]\label{lem:trimming-lemma}
    Suppose that $(p,\pi)$ satisfies a WPI with parameters $C_{\-{WPI}},\delta_{\-{WPI}}\geq0$, where $\delta_{\-{WPI}}\leq 1/16$.  Let $h>0$ satisfy $C_{\-{WPI}}h\leq 1/16$.  Then there exists a nonempty set $G\subseteq\Omega$ such that
    \begin{equation}
        \pi\tp{\bar G}\leq 8\delta_{\-{WPI}}
        \label{eq:trimmed-mass}
    \end{equation}
    and
    \begin{equation}
        c\tp{S,G\setminus S}\geq h\pi\tp{S}
        \label{eq:restriction-conductance-lower-bound}
    \end{equation}
    for every nonempty $S\subseteq G$ satisfying $\pi\tp{S}\leq\pi\tp{G}/2$.
\end{lemma}

\begin{proof}
    We iteratively remove sparse sides of cuts.  Set $R_0=\varnothing$ and $G_0=\Omega$.  Given $G_k=\Omega\setminus R_k$, stop if~\eqref{eq:restriction-conductance-lower-bound} holds with $G=G_k$.  Otherwise, choose a nonempty set $S_k\subseteq G_k$ such that
    \begin{equation}
        \pi\tp{S_k}\leq\frac{\pi\tp{G_k}}{2}
        \qquad\text{and}\qquad
        c\tp{S_k,G_k\setminus S_k}<h\pi\tp{S_k},
        \label{eq:sparse-cut-at-step-k}
    \end{equation}
    and set $R_{k+1}=R_k\cup S_k$ and $G_{k+1}=G_k\setminus S_k$.  Each step removes at least one state and leaves at least one state, so the procedure terminates.

    We first record the boundary invariant
    \begin{equation}
        c\tp{R_k,G_k}\leq h\pi\tp{R_k}.
        \label{eq:removed-boundary-invariant}
    \end{equation}
    It is immediate for $k=0$.  If it holds at step $k$, then reversibility, nonnegativity of the capacities, and~\eqref{eq:sparse-cut-at-step-k} give
    \begin{align*}
        c\tp{R_{k+1},G_{k+1}}
        &=c\tp{R_k,G_k\setminus S_k}
          +c\tp{S_k,G_k\setminus S_k}\\
        &\leq c\tp{R_k,G_k}+h\pi\tp{S_k}\\
        &\leq h\pi\tp{R_k\cup S_k}.
    \end{align*}
    Thus~\eqref{eq:removed-boundary-invariant} follows by induction.

    It remains to show that the removed mass never exceeds $8\delta_{\-{WPI}}$.  Write $r_k\defeq\pi\tp{R_k}$.  For the indicator of $R_k$, we have
    \[
        \Var[\pi]{\mathbf 1_{R_k}}=r_k\tp{1-r_k},
        \qquad
        \operatorname{Osc}\tp{\mathbf 1_{R_k}}^2\leq1.
    \]
    Moreover, since $G_k=\bar R_k$ and $c$ is symmetric,
    \[
        \begin{aligned}
        \+E_p\tp{\mathbf 1_{R_k},\mathbf 1_{R_k}}
        &=\frac12\sum_{x,y\in\Omega}
          \tp{\mathbf 1_{R_k}\tp{y}-\mathbf 1_{R_k}\tp{x}}^2c\tp{x,y}\\
        &=\frac12\tp{c\tp{R_k,G_k}+c\tp{G_k,R_k}}
         =c\tp{R_k,G_k}.
        \end{aligned}
    \]
    Applying~\eqref{eq:weak-poincare-inequality} and then using~\eqref{eq:removed-boundary-invariant} and $C_{\-{WPI}}h\leq1/16$ yields
    \begin{equation}
        r_k\tp{1-r_k}
        \leq C_{\-{WPI}}\,c\tp{R_k,G_k}+\delta_{\-{WPI}}
        \leq \frac{r_k}{16}+\delta_{\-{WPI}}.
        \label{eq:removed-mass-wpi}
    \end{equation}
    We prove $r_k\leq8\delta_{\-{WPI}}$ by induction.  The claim is clear at $k=0$.  Assuming it at step $k$, the choice of the smaller side in~\eqref{eq:sparse-cut-at-step-k} gives
    \[
        r_{k+1}
        =r_k+\pi\tp{S_k}
        \leq r_k+\frac{1-r_k}{2}
        =\frac{1+r_k}{2}
        \leq\frac12+4\delta_{\-{WPI}}
        \leq\frac34.
    \]
    Here $r_{k+1}\leq3/4$ gives $r_{k+1}(1-r_{k+1})\geq r_{k+1}/4$.  If $r_{k+1}>8\delta_{\-{WPI}}$, then $\delta_{\-{WPI}}<r_{k+1}/8$, so~\eqref{eq:removed-mass-wpi} would imply
    \[
        \frac{r_{k+1}}{4}
        \leq r_{k+1}\tp{1-r_{k+1}}
        \leq\frac{r_{k+1}}{16}+\delta_{\-{WPI}}
        <\frac{3r_{k+1}}{16},
    \]
    a contradiction.  Hence $r_{k+1}\leq8\delta_{\-{WPI}}$.  At termination, $G=G_k$ satisfies both~\eqref{eq:trimmed-mass} and~\eqref{eq:restriction-conductance-lower-bound}.
\end{proof}

For a nonempty set $G\subseteq\Omega$, the \emph{restriction of $p$ to $G$} is the transition rate kernel $p^G$ on $G$ obtained by retaining the rates between states in $G$ and suppressing transitions from $G$ to $\bar G$:
\[
    \forall x\neq y,\ p^G\tp{x,y}=p\tp{x,y},
    \qquad
    p^G\tp{x,x}=-\sum_{y\in G\setminus\set{x}}p\tp{x,y}.
\]
We write $\pi^G\defeq\pi\tp{\cdot\mid G}$ for the conditional distribution on $G$.  The restricted kernel is reversible with respect to $\pi^G$ because, for distinct $x,y\in G$,
\[
    \pi^G\tp{x}p^G\tp{x,y}
    =\frac{\pi\tp{x}p\tp{x,y}}{\pi\tp{G}}
    =\frac{\pi\tp{y}p\tp{y,x}}{\pi\tp{G}}
    =\pi^G\tp{y}p^G\tp{y,x}.
\]
In particular, when $p$ describes Glauber dynamics, $p^G$ retains the same
legal local updates inside $G$ and suppresses updates that would leave $G$.

\zcref{lem:trimming-lemma} identifies a high-probability set.  The following
theorem shows that the original dynamics restricted to this set satisfies a
true Poincar\'e inequality.

\begin{theorem}[WPI yields PI on a high-probability restriction]
    \label{thm:poincare-core}
    Let $p$ be an irreducible transition rate kernel on a finite state space $\Omega$ with $\abs{\Omega}\geq2$, reversible with respect to $\pi$.  Suppose that $(p,\pi)$ satisfies a WPI with parameters $C_{\-{WPI}}>0$ and $\delta_{\-{WPI}}\in[0,1/16]$.  Let
    \[
        q_{\max}\defeq\max_{x\in\Omega}\sum_{y\neq x}p\tp{x,y}
    \]
    be the maximal exit rate of the original chain, and define
    \[
        h\defeq\min\set{\frac{q_{\max}}{2},
            \frac{1}{16C_{\-{WPI}}}}.
    \]
    Then there exists a nonempty $G\subseteq\Omega$ such that, writing $\pi^G=\pi\tp{\cdot\mid G}$,
    \begin{equation}
        D_{\mathrm{TV}}\tp{\pi,\pi^G}
        =\pi\tp{\bar G}
        \leq8\delta_{\-{WPI}},
        \label{eq:tv-close-poincare-core}
    \end{equation}
    and the restriction $p^G$ of the original chain to $G$ satisfies
    \begin{equation}
        \Var[\pi^G]{f}
        \leq \frac{2q_{\max}}{h^2}\cdot \+E_{p^G}\tp{f,f}
        \leq\tp{\frac{8}{q_{\max}}
          +512q_{\max} C_{\-{WPI}}^2}\cdot \+E_{p^G}\tp{f,f}
        \label{eq:poincare-on-core}
    \end{equation}
    for every $f:G\to\bb R$. 
\end{theorem}

\begin{proof}[Proof sketch]
    Apply \zcref{lem:trimming-lemma} with the stated value of $h$.  It gives a
    nonempty set $G$ satisfying \eqref{eq:trimmed-mass} and whose restricted
    conductance is at least $h$.  If $G$ is not a singleton, normalize the
    restricted rate kernel as
    \[
        P^G=\!I_G+\frac{1}{q_{\max}}p^G.
    \]
    This reversible Markov kernel has conductance at least $h/q_{\max}$, so
    Cheeger's inequality gives
    $\operatorname{gap}(P^G)\geq h^2/(2q_{\max}^2)$.  Since
    $p^G=q_{\max}(P^G-\!I_G)$, the corresponding Poincar\'e constant for
    $p^G$ is at most $2q_{\max}/h^2$, which gives
    \eqref{eq:poincare-on-core}.  Finally,
    $D_{\mathrm{TV}}(\pi,\pi^G)=\pi(\bar G)$ gives
    \eqref{eq:tv-close-poincare-core}.  The details are given in
    \zcref{subsec:appendix-wpi-to-pi-restriction}.
\end{proof}

This result provides a static viewpoint for understanding what is relaxed when we compare WPI with PI. The exceptional and potentially bad part allowed in WPI intuitively corresponds to the residual allowed in weak action. This explains why WPI can naturally lead to a weak-action bound and be used to analyze annealing algorithms.

\zcref{thm:poincare-core} also partially resolves an open problem in \cite{HMRW25}. The authors of \cite{HMRW25} showed that a distribution that is close in total variation to another distribution satisfying a true Poincar\'e inequality satisfies a WPI, and asked whether the converse holds:
\begin{quote}
    \emph{If a distribution satisfies a weak Poincar\'e inequality, is it TV-close to a distribution satisfying a true Poincar\'e inequality?}~\cite[Problem~1.16]{HMRW25}
\end{quote}
For finite reversible chains, this problem is partially answered by
\zcref{thm:poincare-core}, allowing restriction to a high-probability subset
and to the corresponding restricted dynamics.

\section{Instantiating the Framework I: Mean-Field Ising and Potts}
\label{sec:mean-field-phase-transitions}

We begin the model-specific part of the paper with two mean-field examples.
In both models, projection onto color counts reduces the original exponential
configuration space to a polynomial-size space, and we will show that bounds
on the projected space suffice to control the action of the original Gibbs
path.  Throughout this section, the number of colors $q$ is treated as a fixed
constant.

The two examples illustrate different reasons why slow worst-case mixing need
not prevent annealing.  In mean-field Ising, spin-flip symmetry keeps the two
low-temperature modes at equal weight throughout the Gibbs path.  Folding this
symmetry removes the slow direction, and canonical paths on the folded space
give polynomial exact action.  Mean-field Potts requires three different
arguments in different temperature ranges: a polynomial Poincar\'e inequality
gives exact action before metastability appears, weak action leaves the losing
well unrouted in the two metastability windows, and direct routing again
gives exact action after the disordered well disappears.  Both results use the
one-path theorem from \zcref{section:annealing}.

At the Potts coexistence point, the ordered and disordered phases both have
order-one mass, so neither can be left in the residual.  Our Potts result
excludes this point, which is also delicate for other
methods~\cite{BR16,BGZ26}.

\zcref{subsec:mean-field-folded-reduction} reduces both models to the folded
count space, and
\zcref{subsec:mean-field-ising-exact,subsec:mean-field-potts-noncritical}
treat Ising and Potts, respectively.

\subsection{A common reduction to the folded count space}
\label{subsec:mean-field-folded-reduction}

Both applications below construct fluxes on the polynomial-size space of color
counts rather than directly on the exponential configuration space.  These
fluxes are only an analytic device.  The sampler itself runs
\zcref{algo:poisson-annealing} on spin configurations, with the
$(q-1)$-block heat-bath kernels defined below.  It never simulates a chain on
the folded count space.  This subsection shows that analyzing the count space
loses no action, and records the two estimates needed later: a lower bound
on the capacities between adjacent folded states and an upper bound on the
path derivative of the folded law.

For $q\geq2$ and $n\geq q$, let $\Omega=[q]^n$ and define the mean-field Potts
distribution
\begin{equation}
    \mu_\beta(\sigma)
    =\frac1{Z_{n,q}(\beta)}
    \exp\left(\frac\beta n\sum_{a=1}^q M_a(\sigma)^2\right),
    \qquad
    M_a(\sigma)=\abs{\set{i:\sigma_i=a}}.
    \label{eq:mean-field-potts}
\end{equation}
We use the $(q-1)$-block heat-bath chain.  It chooses a uniformly
random block of $q-1$ sites and resamples the block from its conditional Gibbs
law.  Write $P_\beta$ for one update, $p_\beta=P_\beta-\!I$ for the rate kernel
(that is, the Poisson clock rate is $\lambda_t\equiv1$, so updates occur at
unit rate), and
\[
    c_\beta(\sigma,\sigma')
    =\mu_\beta(\sigma)p_\beta(\sigma,\sigma'),
    \qquad \sigma\neq\sigma',
\]
for the capacity, with $c_\beta(\sigma,\sigma)=0$.  This block size makes the
folded moves used below available as single transitions.  It is constant when
$q$ is fixed; for $q=2$, this is the usual single-site heat-bath chain.

The energy in \eqref{eq:mean-field-potts} depends only on the count vector
\[
    \*M(\sigma)=\tp{M_1(\sigma),\ldots,M_q(\sigma)}.
\]
For a vector $v\in\bb R^q$, let $\operatorname{sort}(v)$ denote the vector
obtained by rearranging its coordinates in nonincreasing order.  We therefore
use the reduction
\begin{equation*}
    \sigma
    \longmapsto \*M(\sigma)
    \longmapsto \operatorname{sort}\tp{\*M(\sigma)}.
\end{equation*}
The first map forgets the locations of the colors and retains only their
counts; the second identifies count vectors that differ only by a permutation
of the color labels.  The resulting folded space is
\[
    \widehat\Omega_{n,q}
    =\set{\*m\in\bb Z_{\geq0}^q:
      m_1\geq\cdots\geq m_q,\ \sum_{a=1}^qm_a=n},
\]
and $|\widehat\Omega_{n,q}|\leq(n+1)^{q-1}$.  Define the folded law and capacity
by
\[
    \widehat\mu_\beta(\*m)
    =\sum_{\sigma\in\Omega:\
      \operatorname{sort}\tp{\*M(\sigma)}=\*m}
      \mu_\beta(\sigma)
\]
and, for distinct $\*m,\*m'\in\widehat\Omega_{n,q}$,
\[
    \widehat c_\beta(\*m,\*m')
    =\sum_{\substack{\sigma,\sigma'\in\Omega:\
      \operatorname{sort}\tp{\*M(\sigma)}=\*m,\\
      \operatorname{sort}\tp{\*M(\sigma')}=\*m'}}
      c_\beta(\sigma,\sigma'),
\]
with $\widehat c_\beta(\*m,\*m)=0$.
We call two elements of $\widehat\Omega_{n,q}$ adjacent when their folded
capacity is positive; this adjacency relation does not depend on $\beta$.

The two lemmas below are proved in
\zcref{subsec:appendix-mean-field-projection}.  The first says that this
reduction loses no action; its proof transfers a folded flux back to the
original space.

\begin{lemma}[Folding preserves weak action]
    \label{lem:mean-field-projected-geometry}
    Let $(\beta_s)_{s\in[0,1]}$ be absolutely continuous.  For every residual
    budget $\varepsilon\geq0$,
    \begin{equation}
        \+A^{\mathrm w}_{\varepsilon}
        \tp{\tp{\mu_{\beta_s}}_{s\in[0,1]}}
        =
        \+A^{\mathrm w}_{\varepsilon}
        \tp{\tp{\widehat\mu_{\beta_s}}_{s\in[0,1]}}.
        \label{eq:folding-preserves-weak-action}
    \end{equation}
\end{lemma}

The next estimates compare the folded capacity with the endpoint masses and
bound how quickly the probability of a folded count vector changes along a
temperature path.  We use
their general form for Potts; the Ising argument below uses a sharper
$q=2$ estimate.

\begin{lemma}[Folded capacity and path derivative]
    \label{lem:mean-field-folded-estimates}
    For adjacent $\*m,\*m'\in \widehat\Omega_{n,q}$,
    \begin{equation}
        \widehat c_\beta\tp{\*m,\*m'}
        \geq
        \frac{
        \tp{\widehat\mu_\beta\tp{\*m}
        \wedge\widehat\mu_\beta\tp{\*m'}}^2}
        {(q!)^2n^{2q-2}}.
        \label{eq:folded-potts-capacity}
    \end{equation}
    Moreover, if $(\beta_s)_{s\in[0,1]}$ is absolutely continuous, then for
    almost every $s$,
    \begin{equation}
        \partial_s\widehat\mu_{\beta_s}\tp{\*m}
        =\frac{\beta_s'}n\widehat\mu_{\beta_s}\tp{\*m}
        \left(\sum_{a=1}^qm_a^2-
        \mathbb E_{\widehat\mu_{\beta_s}}\sum_{a=1}^qm_a^2\right).
        \label{eq:folded-potts-score}
    \end{equation}
    In particular,
    \begin{equation}
        \abs{\partial_s\widehat\mu_{\beta_s}\tp{\*m}}
        \leq n\abs{\beta_s'}
        \widehat\mu_{\beta_s}\tp{\*m}.
        \label{eq:folded-potts-score-bound}
    \end{equation}
\end{lemma}

In the next two subsections, we construct paths and fluxes on
$\widehat\Omega_{n,q}$ and use the two lemmas above to bound their cost.
\zcref{lem:mean-field-projected-geometry} then turns the resulting bound into
an action bound for the original Gibbs path on $[q]^n$, to which the one-path
annealing theorem applies.

\subsection{Mean-field Ising: symmetry removes the bottleneck}
\label{subsec:mean-field-ising-exact}

For $q=2$, map the two colors to spins in $\set{\pm1}$.  Up to a constant in
the Hamiltonian, \eqref{eq:mean-field-potts} becomes the mean-field Ising law
\begin{equation}
    \mu^{\mathrm{Ising}}_\beta(\sigma)
    =\frac1{Z_n^{\mathrm{Ising}}(\beta)}
    \exp\left(\frac\beta{2n}M(\sigma)^2\right),
    \qquad M(\sigma)=\sum_{i=1}^n\sigma_i.
    \label{eq:mean-field-ising}
\end{equation}
For every fixed $\beta>1$, the magnetization has two well-separated peaks and
the terminal single-site chain has exponentially slow worst-case mixing as
$n\to\infty$~\cite{LLP10}.  The map $M\mapsto-M$ exchanges the two peaks, so
their weights remain equal along the entire temperature path.  Folding to
$|M|$ removes precisely this slow direction.

Write $\widehat\mu^{\mathrm{Ising}}_\beta$ and
$\widehat c^{\mathrm{Ising}}_\beta$ for the folded magnetization law and capacity.
Its state space is
\begin{equation*}
    \widehat{\mathcal M}_n
    =\set{|m|:m\in\set{-n,-n+2,\ldots,n}},
\end{equation*}
with edges between consecutive values.  This graph is a path, so use the
unique path $\gamma_{uv}$ between every pair
$\set{u,v}\subseteq\widehat{\mathcal M}_n$ as its canonical path.  For a
folded edge $e$, define
\begin{equation*}
    \Gamma^{\mathrm{Ising}}_\beta(e)
    =\frac1{\widehat c^{\mathrm{Ising}}_\beta(e)}
      \sum_{\substack{\{u,v\}:\\e\in\gamma_{uv}}}
      \widehat\mu^{\mathrm{Ising}}_\beta(u)
      \widehat\mu^{\mathrm{Ising}}_\beta(v)|\gamma_{uv}|.
\end{equation*}

\begin{proposition}[Canonical-path bounds for folded Ising]
    \label{prop:folded-ising-canonical-congestion}
    Uniformly for $\beta\geq0$,
    \begin{equation}
        \max_e\Gamma^{\mathrm{Ising}}_\beta(e)\leq n^3.
        \label{eq:folded-ising-canonical-congestion}
    \end{equation}
    Moreover, for every absolutely continuous schedule $(\beta_s)_{s\in[0,1]}$,
    \begin{equation}
        \norm{\partial_s\log\widehat\mu^{\mathrm{Ising}}_{\beta_s}}
        _{L^2\tp{\widehat\mu^{\mathrm{Ising}}_{\beta_s}}}^2
        \leq\frac{n^2}{16}\abs{\beta_s'}^2
        \qquad\text{for almost every }s.
        \label{eq:folded-ising-score-bound}
    \end{equation}
\end{proposition}

The first bound follows from unimodality of the folded magnetization law and
the estimate
\begin{equation*}
    \widehat c^{\mathrm{Ising}}_\beta(m,m+2)
    \geq\frac1{2n}
    \tp{\widehat\mu^{\mathrm{Ising}}_\beta(m)\wedge
    \widehat\mu^{\mathrm{Ising}}_\beta(m+2)}.
\end{equation*}
The score bound follows by differentiating \eqref{eq:mean-field-ising}.
The details are given in \zcref{subsec:appendix-mean-field-ising}.

Applying the canonical-path action bound from
\zcref{prop:canonical-path-action} on the folded space, and then using
\eqref{eq:folding-preserves-weak-action}, gives the exact action bound.

\begin{proposition}[Mean-field Ising action]
    \label{prop:mean-field-ising-exact-action}
    For every absolutely continuous schedule
    $(\beta_s)_{s\in[0,1]}\subseteq\bb R_{\geq0}$,
    \begin{equation}
        \+A\tp{\tp{\mu^{\mathrm{Ising}}_{\beta_s}}_{s\in[0,1]}}
        \leq
        \frac{n^5}{16}\int_0^1\abs{\beta_s'}^2\,\dd s.
        \label{eq:mean-field-ising-action}
    \end{equation}
\end{proposition}

For the linear schedule $\beta_s=s\beta$, the initial law
$\mu^{\mathrm{Ising}}_0$ is uniform and can be sampled exactly.  Thus the one-path
theorem \zcref{thm:annealing-ub} immediately gives the algorithmic conclusion.

\begin{theorem}[All-temperature mean-field Ising annealing]
    \label{thm:mean-field-ising-annealing}
    For every $\beta\geq0$ and $\zeta\in(0,1)$, continuous-time annealed
    single-site heat bath, initialized from $\mu^{\mathrm{Ising}}_0$ and run along
    $\beta_s=s\beta$, has an output law $\nu_T$ satisfying
    \[
        D_{\mathrm{KL}}\tp{\mu^{\mathrm{Ising}}_\beta\|\nu_T}\leq\zeta
    \]
    upon taking
    \[
        T=\max\left\{1,\frac{n^5\beta^2}{64\zeta}\right\}
        =O\tp{1+\frac{n^5\beta^2}{\zeta}}.
    \]
    The expected number of single-site updates is $T$.
\end{theorem}

Censoring the fixed-temperature dynamics to nonnegative magnetization is
another way to use this symmetry~\cite{LLP10,DLP09}.  In our analysis, there
is no need to divide the state space into positive- and
negative-magnetization components.  Such a division would introduce an
artificial boundary near $M=0$, whereas folding simply records that spin-flip
symmetry keeps the two mode weights equal.  Thus the continuous phase
transition at $\beta=1$ creates no singularity in the action bound.

\subsection{Mean-field Potts away from coexistence: reading the landscape}
\label{subsec:mean-field-potts-noncritical}

We now fix $q\geq3$.  At $\beta=0$, all colorings are equally likely and a
typical coloring uses every color on roughly $n/q$ vertices.  Increasing
$\beta$ rewards colorings in which more vertices carry the same color, and at
low temperature one color becomes dominant.  The Potts landscape records the
competition between the number of balanced colorings and the larger Gibbs
weight of concentrated colorings.

To describe this competition, imagine drawing a coloring $\sigma$ from
$\mu_\beta$ and then forgetting where the individual colors occur.  We retain
only the fraction of vertices of each color.  Thus, if
$\*m=\*M(\sigma)=(m_1,\ldots,m_q)$ is the count vector, we record
\begin{equation*}
    x=\frac{\*m}{n}\in\Delta_q,
    \qquad
    \Delta_q=\set{x\in\bb R_{\geq0}^q:\sum_{a=1}^qx_a=1}.
\end{equation*}
We call $x$ the \emph{empirical color profile}: its coordinate $x_a$ is simply
the proportion of vertices carrying color $a$.  The question behind this
subsection is which profiles are likely under $\mu_\beta$, how this answer
changes with $\beta$, and whether the likely regions are separated by a
low-probability bottleneck.  The answer is encoded in the Potts landscape
described next.

To answer this question, one must count the total probability mass of a
profile, rather than only the Gibbs weight of one configuration.  A count
vector $\*m$ is realized by $n!/\prod_am_a!$ different colorings, and all of
them have the same Gibbs weight.  Hence their total unnormalized mass is
\begin{equation*}
    W_\beta(\*m)
    =\underbrace{\frac{n!}{\prod_{a=1}^q m_a!}}
       _{\substack{\text{number of colorings}\\
          \text{with counts }\*m}}
     \mathbin{\cdot}
     \underbrace{\exp\tp{\frac\beta n\sum_{a=1}^qm_a^2}}
       _{\substack{\text{Gibbs weight of}\\
          \text{each such coloring}}},
\end{equation*}
and therefore
\begin{equation*}
    \bb P_{\sigma\sim\mu_\beta}\tp{\*M(\sigma)=\*m}
    =\frac{W_\beta(\*m)}{\sum_{\*m'}W_\beta(\*m')}.
\end{equation*}
This is the quantity that the free-energy calculation estimates.

For $\*m=nx$, Stirling's formula turns the multinomial coefficient into an
entropy term and gives
\begin{equation*}
    \frac1n\log W_\beta(nx)
    =\underbrace{-\sum_{a=1}^qx_a\log x_a}_{\text{entropy}}
     +\underbrace{\beta\sum_{a=1}^qx_a^2}_{\text{interaction}}
     +O_q\tp{\frac{\log n}{n}}
    =\Phi_\beta(x)+O_q\tp{\frac{\log n}{n}},
\end{equation*}
where $0\log0=0$.  Thus $\Phi_\beta(x)$ is the leading log-mass per vertex of
all configurations with profile $x$.  A fixed advantage
$\Phi_\beta(x)-\Phi_\beta(y)=\delta>0$ produces a probability ratio of order
$e^{n\delta}$, up to polynomial factors.  The derivation and its uniform
finite-$n$ form are given in
\zcref{subsec:appendix-potts-coexistence-profiles}.

The two terms have a direct interpretation.  At $\beta=0$, only the number
of colorings matters, and there
are overwhelmingly more colorings with nearly equal color proportions.  The
corresponding balanced profile is
\begin{equation*}
    x^{\mathrm{dis}}=\tp{1/q,\ldots,1/q};
\end{equation*}
we call it \emph{disordered} because no color is favored.  As $\beta$ grows,
the interaction term instead rewards putting more vertices in the same
color.  After naming the favored color as color $1$, the candidate profiles
have the form
\begin{equation*}
    x(a)=\tp{a,\frac{1-a}{q-1},\ldots,
      \frac{1-a}{q-1}},
    \qquad a\in[1/q,1].
\end{equation*}
A local maximum with $a>1/q$ represents an \emph{ordered} well: one color has
become macroscopically more common than the others.  Choosing the favored
color gives $q$ symmetry-related wells, which color folding identifies with
one another.

In this language, a global maximum of $\Phi_\beta$ describes a dominant
phase, a lower local maximum describes a metastable phase, and a lower
stationary profile between them describes the barrier.  We work with
$\Phi_\beta$, which is maximized; equivalently, the rate function
$I_\beta=-\Phi_\beta$, up to an irrelevant additive constant, is minimized.

\paragraph{What happens as $\beta$ increases.}
As the inverse temperature increases, the landscape passes a lower ordered
spinodal, the coexistence point, and an upper disordered spinodal:
\begin{equation*}
    \beta_{\mathrm{sp}}^{\mathrm{ord}}<\beta_{\mathsf c}<
    \beta_{\mathrm{sp}}^{\mathrm{dis}}=q/2.
\end{equation*}
At the lower threshold $\beta_{\mathrm{sp}}^{\mathrm{ord}}$, a degenerate
non-disordered stationary point appears.  Immediately above this threshold it
splits into an ordered local maximum and a neighboring saddle.  At
$\beta_{\mathsf c}$, the ordered and disordered maxima have the same leading
log-weight and exchange global dominance.  At the upper threshold
$\beta_{\mathrm{sp}}^{\mathrm{dis}}$, the disordered local maximum loses
stability.  This progression is summarized in \zcref{fig:potts-landscape}.

Our result is organized around $\beta_{\mathsf c}$, not around the other two
thresholds: it covers targets on either side of coexistence, but not
coexistence itself.  The two metastable intervals are where weak action is
useful.  The lower threshold $\beta_{\mathrm{sp}}^{\mathrm{ord}}$ is included
to orient the reader in the landscape, while
$\beta_{\mathrm{sp}}^{\mathrm{dis}}$ marks the range in which our construction
strengthens from weak action to exact action.

\begin{figure}[H]
    \centering
    \begin{tikzpicture}[
        >=Stealth,x=1cm,y=1cm,
        landscape/.style={draw=black!70,line width=0.9pt},
        disordered/.style={fill=blue!72!black},
        ordered/.style={fill=orange!90!black},
        smalllabel/.style={font=\scriptsize,align=center},
        regimelabel/.style={font=\tiny,align=center}
    ]
        \fill[blue!10] (0,-0.08) rectangle (3,0.08);
        \fill[orange!12] (3,-0.08) rectangle (6.5,0.08);
        \fill[red!9] (6.5,-0.08) rectangle (9.5,0.08);
        \fill[green!10] (9.5,-0.08) rectangle (13.2,0.08);
        \draw[->,thick,black!75] (-0.25,0)--(13.45,0)
            node[right] {$\beta$};
        \node[below] at (0,-0.10) {$0$};

        \foreach \x/\lab in {
            3/{\beta_{\mathrm{sp}}^{\mathrm{ord}}},
            6.5/{\beta_{\mathsf c}},
            9.5/{\beta_{\mathrm{sp}}^{\mathrm{dis}}=q/2}}
        {
            \draw[thick,black!75] (\x,-0.14)--(\x,0.14);
            \node[above] at (\x,0.18) {$\lab$};
        }

        \node[regimelabel] at (1.5,1.10)
            {disordered\\only};
        \node[regimelabel] at (4.70,1.10)
            {disordered dominant\\ordered metastable};
        \node[regimelabel] at (8.05,1.10)
            {ordered dominant\\disordered metastable};
        \node[regimelabel] at (11.35,1.10)
            {ordered\\only};
        \node[regimelabel,below] at (6.5,-0.18)
            {equal phase depths};

        \node[smalllabel,rotate=90] at (-0.12,-2.28)
            {$I_\beta=-\Phi_\beta$};
        \def\ybase{-2.35}

        \begin{scope}[shift={(1.5,\ybase)}]
            \draw[black!45,thin] (-1.05,0)--(1.05,0);
            \draw[landscape,smooth] plot[domain=-0.9:0.9,samples=40]
                (\x,{1.75*\x*\x});
            \fill[disordered] (0,0) circle (1.5pt);
            \node[smalllabel,below] at (0,-0.08) {dis.};
        \end{scope}

        \begin{scope}[shift={(4.75,\ybase)}]
            \draw[black!45,thin] (-1.10,0)--(1.10,0);
            \draw[landscape,smooth] plot[domain=-1.05:1.05,samples=60]
                (\x,{3.0*(\x*\x*\x*\x*\x*\x
                    -1.5*\x*\x*\x*\x+0.6*\x*\x)});
            \fill[disordered] (0,0) circle (1.5pt);
            \node[smalllabel,below] at (0,-0.08) {dis.};
            \pgfmathsetmacro{\xs}{0.851}
            \pgfmathsetmacro{\ys}{3.0*(\xs*\xs*\xs*\xs*\xs*\xs
                -1.5*\xs*\xs*\xs*\xs+0.6*\xs*\xs)}
            \fill[ordered] (-\xs,\ys) circle (1.5pt);
            \fill[ordered] (\xs,\ys) circle (1.5pt);
            \node[smalllabel,below] at (\xs,{\ys-0.08}) {ord.};
        \end{scope}

        \begin{scope}[shift={(8.0,\ybase)}]
            \draw[black!45,thin] (-1.10,0)--(1.10,0);
            \draw[landscape,smooth] plot[domain=-1.05:1.05,samples=60]
                (\x,{3.75*(\x*\x*\x*\x*\x*\x
                    -1.5*\x*\x*\x*\x+0.5*\x*\x)+0.15});
            \fill[disordered] (0,0.15) circle (1.5pt);
            \node[smalllabel,above] at (0,0.18) {dis.};
            \pgfmathsetmacro{\xs}{0.886}
            \pgfmathsetmacro{\ys}{3.75*(\xs*\xs*\xs*\xs*\xs*\xs
                -1.5*\xs*\xs*\xs*\xs+0.5*\xs*\xs)+0.15}
            \fill[ordered] (-\xs,\ys) circle (1.5pt);
            \fill[ordered] (\xs,\ys) circle (1.5pt);
            \node[smalllabel,below] at (\xs,{\ys-0.08}) {ord.};
        \end{scope}

        \begin{scope}[shift={(11.35,\ybase)}]
            \draw[black!45,thin] (-1.10,0)--(1.10,0);
            \draw[landscape,smooth] plot[domain=-1.05:1.05,samples=60]
                (\x,{1.6*(\x*\x-0.8)*(\x*\x-0.8)});
            \pgfmathsetmacro{\xs}{0.894}
            \fill[ordered] (-\xs,0) circle (1.5pt);
            \fill[ordered] (\xs,0) circle (1.5pt);
            \node[smalllabel,below] at (\xs,-0.08) {ord.};
            \fill[blue!25] (0,{1.6*0.64}) circle (1.4pt);
            \node[smalllabel,above] at (0,{1.6*0.64+0.03})
                {dis. unstable};
        \end{scope}
    \end{tikzpicture}
    \caption{The four regimes of the mean-field Potts landscape for $q\geq3$.
    The curves schematically show the rate function
    $I_\beta=-\Phi_\beta$ along a one-dimensional projection, so lower wells
    correspond to exponentially larger profile mass.  Blue marks the
    disordered profile and orange marks ordered profiles; the side wells
    schematically represent the $q$ possible choices of the favored color.
    At $\beta_{\mathsf c}$ the two kinds of phase have equal depth.}
    \label{fig:potts-landscape}
\end{figure}

The four regimes have a simple meaning.  Below
$\beta_{\mathrm{sp}}^{\mathrm{ord}}$, the disordered profile is the only local
maximum.  On
$(\beta_{\mathrm{sp}}^{\mathrm{ord}},\beta_{\mathsf c})$, an ordered well
exists but is metastable; the disordered phase still dominates.  On
$(\beta_{\mathsf c},\beta_{\mathrm{sp}}^{\mathrm{dis}})$, the ordered wells
dominate while the disordered well remains metastable.  Above
$\beta_{\mathrm{sp}}^{\mathrm{dis}}$, the disordered profile is no longer a
local maximum.  At $\beta_{\mathsf c}$ itself, both kinds of phase have
order-one mass and a fixed free-energy barrier separates them.

We use only this qualitative phase classification below.  Its standard
supporting details are recorded in
\zcref{subsec:appendix-potts-coexistence-profiles}; see also
\cite[Section~2.2]{CDL+12}.

\paragraph{What the annealing path uses.}
For each prescribed target inverse temperature
$\beta\neq\beta_{\mathsf c}$, our goal is to design an annealing algorithm
whose terminal distribution approximates $\mu_\beta$.  The annealing path is
therefore allowed to depend on the target $\beta$; we are not seeking one
schedule that simultaneously handles all temperatures.  The four physical
regimes nevertheless require only two types of path, because the direction is
determined solely by which side of coexistence contains the target.  For a
target $0\leq\beta<\beta_{\mathsf c}$, we start from the uniform law and use
the forward path
\begin{equation*}
    \beta_s=s\beta.
\end{equation*}
For a target $\beta>\beta_{\mathsf c}$, choose
$\beta_{\mathrm{hi}}\geq
\max\{\beta_{\mathrm{sp}}^{\mathrm{dis}},\beta\}$ and use the reverse path
\begin{equation*}
    \beta_s=(1-s)\beta_{\mathrm{hi}}+s\beta.
\end{equation*}
Thus neither path crosses coexistence.  The two spinodals only determine what
the nondominant region looks like and whether it must be included in the weak
residual.  More precisely, the four regimes are handled as follows.

\begin{enumerate}[leftmargin=*,label=(\roman*)]
    \item \emph{Disordered-only regime.}
    If $0\leq\beta\leq\beta_{\mathrm{sp}}^{\mathrm{ord}}$, use the
    forward path.  Theorems~1 and~2 of Cuff et al.~\cite{CDL+12} give a
    single-site heat-bath mixing time of order at most $n^{4/3}$, uniformly
    along this path and including at the lower spinodal; see
    \zcref{subsec:appendix-potts-noncritical-action}.  The resulting
    Poincar\'e bound carries over to the block chain by the standard
    heat-bath comparison.  We therefore apply
    \zcref{prop:PI-implies-action} directly: the action is exact and there is
    no residual.

    \item \emph{Ordered-metastable regime.}
    If
    $\beta_{\mathrm{sp}}^{\mathrm{ord}}<\beta<\beta_{\mathsf c}$, use the
    same forward path.  The disordered phase remains dominant throughout.
    Once the path passes the lower spinodal, an ordered metastable well
    appears; we include its change in the weak residual.

    \item \emph{Disordered-metastable regime.}
    If
    $\beta_{\mathsf c}<\beta<\beta_{\mathrm{sp}}^{\mathrm{dis}}$, use the
    reverse path.  The ordered phase remains dominant throughout.  Once the
    path passes below the upper spinodal, the disordered well appears as a
    metastable well; we include its change in the weak residual.

    \item \emph{Ordered-only regime.}
    If $\beta\geq\beta_{\mathrm{sp}}^{\mathrm{dis}}$, use the same
    reverse path.  The disordered well is no longer stable, and the change can
    be routed on the entire folded state space.  The residual is zero, so the
    resulting bound is an exact-action bound.
\end{enumerate}

The lower bound displayed above for $\beta_{\mathrm{hi}}$ is enough for the
action estimate.  For an implemented sampler, the initialization result below
chooses $\beta_{\mathrm{hi}}$ still larger so that the starting law is easy to
approximate.  At $\beta=\beta_{\mathsf c}$, neither phase has small mass, so
none of the four treatments applies.

In either case, define the length of the temperature path by
\begin{equation*}
    \mathsf B\defeq\int_0^1|\beta_s'|\,\dd s
    =\begin{cases}
       \beta,&\beta<\beta_{\mathsf c},\\
       \beta_{\mathrm{hi}}-\beta,&\beta>\beta_{\mathsf c}.
     \end{cases}
\end{equation*}
Along either path, the target is the point closest to coexistence.  The
free-energy separation from the nondominant region is therefore smallest at
the target.  This is why a bound of order
\begin{equation*}
    \exp\tp{-c_qn\abs{\beta-\beta_{\mathsf c}}}
\end{equation*}
controls the residual along the entire path.  For a requested TV error
$\tau$, our theorem will consequently require only
$n\abs{\beta-\beta_{\mathsf c}}\gtrsim_q\log(1/\tau)$.  Thus, although
coexistence is the only fixed limiting inverse temperature excluded, at
finite $n$ the residual determines an accuracy-dependent window around it.

The Poincar\'e argument settles case~(i).  The next lemma supplies the geometric
input for the remaining three cases.  The core is centered at
$x^{\mathrm{dis}}$ below coexistence and at the folded ordered maximizer above
coexistence.

\begin{lemma}[A dominant core for mean-field Potts]
    \label{lem:potts-phase-core}
    There exist constants $c_q,C_q,K_q>0$ with the following properties.  For
    every $\theta\in[0,q/2)$ with $\theta\neq\beta_{\mathsf c}$, there is a set
    $\mathcal C_\theta\subseteq\widehat\Omega_{n,q}$ and a state
    $\*m_\theta^\star\in\mathcal C_\theta$ such that
    \begin{equation*}
        \widehat\mu_\theta\tp{\mathcal C_\theta^{\mathsf c}}
        \leq
        C_q\exp\tp{-c_qn\abs{\theta-\beta_{\mathsf c}}}.
    \end{equation*}
    Every $\*m\in\mathcal C_\theta$ has a path in the folded state graph,
    \[
        \gamma_{\*m,\theta}
        =\tp{\*m=\*m_0,\*m_1,\ldots,
          \*m_\ell=\*m_\theta^\star},
        \qquad \ell\leq n,
    \]
    and every state $\*y$ on this path satisfies
    \begin{equation}
        \widehat\mu_\theta(\*y)
        \geq K_q^{-1}\widehat\mu_\theta(\*m).
        \label{eq:potts-core-ascent-path}
    \end{equation}
    When $\theta\geq q/2$, one may take
    $\mathcal C_\theta=\widehat\Omega_{n,q}$; the corresponding paths have
    length at most $2n$ and satisfy
    \eqref{eq:potts-core-ascent-path} with $K_q=q!e^{q-1}$.
\end{lemma}

The proof combines the free-energy gap, a uniform lattice Laplace estimate,
and paths that move each folded count vector toward the dominant maximizer without a
large loss of probability.  It is given in
\zcref{subsec:appendix-potts-phase-core}.

\paragraph{Action bounds.}
In the two metastability regimes, we include the path derivative outside
$\mathcal C_{\beta_s}$ in the weak residual and route the remaining derivative
to $\*m_{\beta_s}^\star$ along the paths in
\zcref{lem:potts-phase-core}.  In the disordered-only regime, we instead use
the Poincar\'e argument above.

\begin{proposition}[Noncritical mean-field Potts action bounds]
    \label{prop:potts-noncritical-weak-action}
    Fix $q\geq3$ and a target
    $\beta\in\bb R_{\geq0}\setminus\set{\beta_{\mathsf c}}$, and use the
    corresponding forward or reverse path above.  There are constants
    $c_q,C_q>0$ and a polynomial $Q_q(n)\leq C_qn^{6q}$ such that, for some
    residual budget $\varepsilon_{\mathrm w}$,
    \begin{equation}
        \+A^{\mathrm w}_{\varepsilon_{\mathrm w}}
        \tp{\tp{\mu_{\beta_s}}_{s\in[0,1]}}
        \leq Q_q(n)\int_0^1\abs{\beta_s'}^2\,\dd s
        =Q_q(n)\mathsf B^2,
        \label{eq:potts-weak-action-bound}
    \end{equation}
    and
    \begin{equation}
        \varepsilon_{\mathrm w}
        \leq C_q\exp\tp{-c_qn\abs{\beta-\beta_{\mathsf c}}}.
        \label{eq:potts-weak-residual}
    \end{equation}
    If $0\leq\beta\leq\beta_{\mathrm{sp}}^{\mathrm{ord}}$, the forward path
    satisfies the stronger exact-action bound
    $\varepsilon_{\mathrm w}=0$.  The same is true for the reverse path when
    $\beta\geq q/2$.
\end{proposition}

The exact bound in the disordered-only regime follows from the Poincar\'e
argument above.  For the phase-core construction, on the complement of the
core we set the residual equal to the path derivative and add the opposite
total mass at the root.  The residual is therefore bounded by the
losing-phase mass and
\eqref{eq:folded-potts-score-bound}.  Inside the core, the rooted-path result
\zcref{prop:rooted-path-action}, together with
\zcref{eq:folded-potts-capacity,eq:potts-core-ascent-path}, controls the
transport cost in the weak-action bound.  Finally, folding preserves both the
action and the residual.
The full calculation is given in
\zcref{subsec:appendix-potts-noncritical-action}.

\paragraph{Initialization and sampling.}
The forward path starts from the uniform law $\mu_0$.  For the reverse path,
the law at a sufficiently large $\beta_{\mathrm{hi}}$ is concentrated on the
$q$ monochromatic configurations, which gives an efficient approximate
initializer with full support.
More explicitly, if $\Omega_{\mathrm{mono}}$ denotes these $q$
configurations, we mix $\operatorname{Unif}(\Omega_{\mathrm{mono}})$ with
$\operatorname{Unif}([q]^n)$; the second term supplies full support.

\begin{lemma}[Forward and reverse initialization]
    \label{lem:potts-noncritical-initialization}
    For each fixed $q\geq3$, there is a constant $C_q>0$ such that the
    following holds.  The forward path can be initialized exactly from
    $\mu_0$.  For the reverse path and every $\zeta\in(0,1)$, one can choose
    \begin{equation*}
        \beta_{\mathrm{hi}}
        \geq\max\left\{q/2,\beta,
          n\log q+\log\frac{C_q}{\zeta}\right\}
    \end{equation*}
    and choose $\delta_{\mathrm{init}}\in(0,1/2]$ so that the efficiently
    sampleable law
    \begin{equation*}
        \nu_0=(1-\delta_{\mathrm{init}})
              \operatorname{Unif}(\Omega_{\mathrm{mono}})
              +\delta_{\mathrm{init}}\operatorname{Unif}([q]^n)
    \end{equation*}
    satisfies
    \begin{equation*}
        D_{\mathrm{KL}}\tp{\mu_{\beta_{\mathrm{hi}}}\|\nu_0}
        \leq\zeta.
    \end{equation*}
\end{lemma}

The reverse-initialization estimate is proved in
\zcref{subsec:appendix-potts-algorithm}.

Combining the action and initialization estimates with the one-path
theorem gives the following result.

\begin{theorem}[Mean-field Potts annealing away from coexistence]
    \label{thm:potts-noncritical-annealing}
    Fix $q\geq3$, a target
    $\beta\in\bb R_{\geq0}\setminus\set{\beta_{\mathsf c}}$, and a target TV
    error $\tau\in(0,1)$.  Let $c_q,C_q>0$ and $Q_q$ be chosen so that
    \zcref{prop:potts-noncritical-weak-action,
    lem:potts-noncritical-initialization} both hold.  Assume that at least one
    of the following holds:
    \begin{equation*}
        \beta\leq\beta_{\mathrm{sp}}^{\mathrm{ord}},
        \qquad \beta\geq q/2,
        \qquad\text{or}\qquad
        n\abs{\beta-\beta_{\mathsf c}}
        \geq\frac1{c_q}\log\frac{4C_q}{\tau}.
    \end{equation*}
    Choose the path and initialization as follows.  If
    $\beta<\beta_{\mathsf c}$, use $\beta_s=s\beta$ and initialize exactly
    from $\mu_0$.  If $\beta>\beta_{\mathsf c}$, set
    \begin{equation*}
        \beta_{\mathrm{hi}}
        =\max\left\{q/2,\beta,
          n\log q+\log\frac{64C_q}{\tau^2}\right\},
    \end{equation*}
    use the reverse path
    $\beta_s=(1-s)\beta_{\mathrm{hi}}+s\beta$, and use the efficiently
    sampleable full-support initialization from
    \zcref{lem:potts-noncritical-initialization}.  In either case, there is
    \begin{equation*}
        T=O\tp{1+\frac{Q_q(n)\mathsf B^2}{\tau^2}}
    \end{equation*}
    such that continuous-time annealing using the
    $(q-1)$-block heat-bath chain has an output law $\nu_T$ satisfying
    \begin{equation*}
        D_{\mathrm{TV}}\tp{\mu_\beta,\nu_T}\leq\tau.
    \end{equation*}
    The expected number of block updates is $T$.  In particular, for every
    fixed $\beta\neq\beta_{\mathsf c}$ and inverse-polynomial $\tau$, the
    running time is polynomial in $n$ and $1/\tau$.
\end{theorem}

When $\beta\leq\beta_{\mathrm{sp}}^{\mathrm{ord}}$ or $\beta\geq q/2$, the
corresponding action bound is exact.  The exact part of
\zcref{thm:annealing-ub} therefore also gives reverse-KL error at most
$\zeta$ with $T=O(Q_q(n)\mathsf B^2/\zeta)$, using exact initialization for
the forward path and initialization error at most $\zeta/2$ for the reverse
path.

For accuracies smaller than the losing-phase mass, a zero-residual
spanning-tree routing gives a weaker exact-action bound.  The resulting
refinement, which is polynomial in $(n,1/\tau)$ for every fixed
$\beta\neq\beta_{\mathsf c}$ but has degree depending on
$|\beta-\beta_{\mathsf c}|^{-1}$, is recorded in
\zcref{subsec:appendix-potts-exact-fallback}.

\section{Instantiating the Framework II: Bipartite Spin Systems}
\label{sec:weak-action-two-spin}

We now apply the one-path and componentwise annealing theorems to spin systems
on bipartite graphs.  Throughout this section, the algorithm uses the
total-rate-one single-site heat-bath chain for the distribution under
consideration.  The first two subsections treat the hard-core model, while the
third also generalizes to an antiferromagnetic two-spin model.

The first three subsections correspond to three forms of the target law.  In
\zcref{subsec:symmetric-low-temperature-bis}, the target is close to an equal
mixture of two phases.  When side symmetry keeps the two weights equal
throughout the path, we average the two within-phase fluxes and apply the
one-path theorem from \zcref{section:annealing}.  When the target weights are
only approximately equal, we initialize the two phases with equal weights and
apply \zcref{thm:componentwise-weak-action-annealing}.  In
\zcref{subsec:polymer-informed-annealing}, the target is close to an unequal
mixture whose weights can be computed, and the componentwise theorem applies
again.  In \zcref{subsec:fixed-gap-two-spin-model}, one phase contains all but
a small amount of mass.  The small complement is included in the weak
residual, and the one-path theorem applies.  Finally,
\zcref{subsec:complete-biclique-separation} uses complete bicliques to show
that, in this third geometry, weak action can be polynomial even when exact
action is exponential.  A figure at the beginning of each of the first three
subsections shows the corresponding geometry.

For a graph $G=(V,E)$, let $\mathcal I(G)$ denote its independent sets.  At
\emph{fugacity} (or \emph{activity}) $\lambda>0$, the hard-core law is
\begin{equation*}
    \mu_{G,\lambda}(I)
    =\frac{\lambda^{|I|}}{Z_G(\lambda)},
    \qquad
    Z_G(\lambda)=\sum_{J\in\mathcal I(G)}\lambda^{|J|},
    \qquad I\in\mathcal I(G).
\end{equation*}
Thus larger $\lambda$ favors independent sets with more occupied vertices.
For $I\in\mathcal I(G)$, let $\delta_I$ denote the point mass at the
configuration $I$.

The single-site heat-bath chain chooses a uniformly random vertex.  If one of
its neighbors is occupied, the chosen vertex is made unoccupied; otherwise it
is occupied with probability $\lambda/(1+\lambda)$ and made unoccupied with
probability $1/(1+\lambda)$.  This chain is reversible with respect to
$\mu_{G,\lambda}$; denote its transition matrix by $P_{G,\lambda}$, and set
\begin{equation*}
    p_{G,\lambda}=P_{G,\lambda}-\!I,
    \qquad
    c_{G,\lambda}(I,J)
    =\mu_{G,\lambda}(I)p_{G,\lambda}(I,J)
\end{equation*}
for distinct independent sets $I,J$, with zero capacity on the diagonal.

For a bipartite graph $G=(L,R,E)$ and $S\subseteq V(G)$, let $N_G(S)$
denote the set of neighbors of $S$.  We call $G$ a bipartite
$\alpha$-expander if
\begin{equation*}
    |N_G(S)|\geq(1+\alpha)|S|
\end{equation*}
whenever $S$ is contained in one side of the bipartition and has size at most
half the size of that side.  This definition applies to both balanced and
unbalanced bipartite graphs.

\subsection{Exact and approximate phase balance on balanced expanders}
\label{subsec:symmetric-low-temperature-bis}

\begin{figure}[H]
    \centering
    \begin{tikzpicture}[
        phase/.style={draw,very thick,rounded corners=15pt,
            minimum height=1.05cm,align=center},
        residual/.style={draw=gray!70,densely dashed,rounded corners=9pt,
            fill=gray!8,minimum height=0.72cm,align=center}
    ]
        \node[phase,fill=blue!13,minimum width=2.55cm] at (-1.55,0)
            {$L$-phase\\weight $\frac12$};
        \node[phase,fill=orange!18,minimum width=2.55cm] at (1.55,0)
            {$R$-phase\\weight $\frac12$};
        \node[residual,minimum width=1.25cm] at (4.05,0)
            {TV error\\$\leq\varepsilon$};
    \end{tikzpicture}
    \caption{The geometry in this subsection.  Up to TV error
    $\varepsilon$, the target is a balanced mixture of its $L$- and
    $R$-dominant phase laws.  Exact or approximate balance supplies the two
    weights before annealing.}
    \label{fig:balanced-two-phase-geometry}
\end{figure}

\paragraph{Exact balance.}
Let $G=(L,R,E)$ be a bipartite $\alpha$-expander with
$|L|=|R|=N/2$ and maximum degree at most a fixed $\Delta$.  Assume that
there is a graph automorphism
$\sigma:V(G)\to V(G)$ such that $\sigma^2=\mathrm{id}$ and
$\sigma(L)=R$.  In other words, $\sigma$ preserves adjacency, exchanges the
two sides, and applying it twice fixes every vertex.  We call this condition
\emph{side symmetry}.  Expansion implies
that $L$ and $R$ are the only independent sets of size $N/2$.  Hence the
high-fugacity hard-core law has two principal phases, one near each ground
state, and side symmetry gives them equal weight.

A simple two-layer construction produces this symmetry.  Given a graph $H$,
make two copies
$L=V(H)\times\{0\}$ and $R=V(H)\times\{1\}$, and replace every edge
$\{u,v\}$ of $H$ by the two edges
$(u,0)(v,1)$ and $(v,0)(u,1)$.  The map
$\sigma(v,i)=(v,1-i)$ is an involutive automorphism exchanging the two sides.
The construction applied to the Petersen graph is shown in
\zcref{fig:two-layer-bipartite-graph}.
\begin{figure}[H]
    \centering
    \begin{tikzpicture}[
        basevertex/.style={circle,draw,thick,fill=white,
            minimum size=3.8mm,inner sep=0pt},
        leftvertex/.style={circle,draw,thick,fill=blue!16,
            minimum size=3.8mm,inner sep=0pt},
        rightvertex/.style={circle,draw,thick,fill=orange!22,
            minimum size=3.8mm,inner sep=0pt},
        graphedge/.style={draw,black!72,line width=0.6pt},
        inneredge/.style={draw,black!55,line width=0.45pt}
    ]
        \begin{scope}[shift={(-1.05,0.1)}]
            \foreach \i in {0,...,4} {
                \node[basevertex] (po\i) at ({90+72*\i}:1.02) {};
                \node[basevertex] (pi\i) at ({90+72*\i}:0.44) {};
            }
            \foreach \i in {0,...,4} {
                \pgfmathtruncatemacro{\j}{mod(\i+1,5)}
                \pgfmathtruncatemacro{\k}{mod(\i+2,5)}
                \draw[graphedge] (po\i)--(po\j);
                \draw[inneredge] (pi\i)--(pi\k);
                \draw[graphedge] (po\i)--(pi\i);
            }
        \end{scope}
        \node at (-1.05,-1.18) {Petersen graph $H$};

        \draw[-{Latex[length=3mm]},very thick]
            (0.85,0.15)--(2.4,0.15);
        \node[align=center] at (1.625,0.72) {two-layer\\construction};

        \begin{scope}[shift={(4.75,0.2)}]
            \foreach \i in {0,...,9} {
                \ifodd\i\relax
                    \node[rightvertex] (do\i) at ({90+36*\i}:1.5) {};
                    \node[leftvertex]  (di\i) at ({90+36*\i}:0.8) {};
                \else
                    \node[leftvertex]  (do\i) at ({90+36*\i}:1.5) {};
                    \node[rightvertex] (di\i) at ({90+36*\i}:0.8) {};
                \fi
            }
            \foreach \i in {0,...,9} {
                \pgfmathtruncatemacro{\j}{mod(\i+1,10)}
                \pgfmathtruncatemacro{\k}{mod(\i+3,10)}
                \draw[graphedge] (do\i)--(do\j);
                \draw[inneredge] (di\i)--(di\k);
                \draw[graphedge] (do\i)--(di\i);
            }
        \end{scope}
        \node at (4.75,-1.55) {bipartite graph $\widetilde H$};
    \end{tikzpicture}
    \caption{The Petersen graph and the bipartite graph $\widetilde H$ produced
    by the two-layer construction.  Both graphs are $3$-regular.  Blue and
    orange indicate the two sides of $\widetilde H$; the half-turn exchanges
    the colors and realizes the automorphism $\sigma$.}
    \label{fig:two-layer-bipartite-graph}
\end{figure}

Consider the reverse log-linear path
\begin{equation}
    \mathsf B=\log(\lambda_0/\lambda_*),
    \qquad
    \lambda_s=\lambda_0e^{-s\mathsf B},
    \qquad s\in[0,1],
    \qquad \lambda_0\geq\lambda_*,
    \label{eq:symmetric-reverse-path}
\end{equation}
and write $\pi_s=\mu_{G,\lambda_s}$.  We take the whole path in the
high-fugacity regime, and abbreviate
$p_s=p_{G,\lambda_s}$ and $c_s=c_{G,\lambda_s}$.  As
$\lambda_0\to\infty$, the initial distribution approaches the easily sampled
mixture
$$
    \rho_{\mathrm{gs}}
    =\frac12\delta_L+\frac12\delta_R.
$$
Here $\delta_L$ and $\delta_R$ are the point masses at the two independent
sets $L$ and $R$, respectively: all vertices on one side are occupied and all
vertices on the other side are unoccupied.
The next proposition collects the model-specific estimates used below.  Its proof in
\zcref{subsec:appendix-symmetric-phase-bound} combines the polymer
representation of \cite{JKP20} with the restricted-chain bounds of
\cite{CGGPSV21}.

\begin{proposition}[Symmetric two-phase weak-action bound]
    \label{prop:symmetric-two-phase-bound}
    Fix $\alpha\in(0,1)$ and an integer $\Delta\geq3$.  There is a
    threshold $\lambda_{\min}=\lambda_{\min}(\alpha,\Delta)$ such that the
    following holds.  Let $G=(L,R,E)$ be a bipartite $\alpha$-expander on
    $N$ vertices with $|L|=|R|$, maximum degree at most $\Delta$, and an
    involutive automorphism that exchanges $L$ and $R$.  Fix
    $\lambda_0\geq\lambda_*\geq\lambda_{\min}$ and use the reverse path
    \eqref{eq:symmetric-reverse-path}.  For every inverse-polynomial
    $\xi\in(0,1/4]$ and all sufficiently large $N$, there is a residual
    budget $\varepsilon_{\mathrm{sym}}\geq0$ such that
    \begin{equation}
        \+A^{\mathrm w}_{\varepsilon_{\mathrm{sym}}}
        \tp{\tp{\pi_s}_{s\in[0,1]}}
        \leq
        Q^{\mathrm{sym}}_{\alpha,\Delta}
        (N,1/\xi,\lambda_0)\mathsf B^2,
        \qquad
        \varepsilon_{\mathrm{sym}}
        \leq C_{\alpha,\Delta}N\mathsf B\xi,
        \label{eq:symmetric-two-phase-bound}
    \end{equation}
    where $Q^{\mathrm{sym}}_{\alpha,\Delta}$ is a polynomial.
\end{proposition}

The automorphism $\sigma$ exchanges the two phases and preserves the
hard-core weight $\lambda_s^{|I|}$, so the two phases have equal probability
throughout the schedule.  The flux constructed in the proof is the average of
two fluxes that move probability only within their respective phases; it does
not move probability from one phase to the other.  The residual is the part
of $\partial_s\pi_s$ not represented by these fluxes, and
$\varepsilon_{\mathrm{sym}}$ controls its total size.  The action bound is
polynomial in $N$ and $1/\xi$ whenever $\lambda_0$ is polynomially bounded.

\begin{theorem}[Weak-action annealing for symmetric high-fugacity hard core]
    \label{thm:symmetric-bis-weak-action}
    Fix $\alpha\in(0,1)$ and an integer $\Delta\geq3$.  There is a threshold
    $\lambda_{\min}=\lambda_{\min}(\alpha,\Delta)$ with the following
    property.  Let $G=(L,R,E)$ be a bipartite $\alpha$-expander on $N$
    vertices with $|L|=|R|$, maximum degree at most $\Delta$, and an
    involutive automorphism $\sigma$ that exchanges $L$ and $R$.  Fix
    $\lambda_*\geq\lambda_{\min}$.  For all sufficiently large $N$ and every
    inverse-polynomial $\tau\in(0,1)$, one may choose a polynomially bounded
    $\lambda_0\geq\lambda_*$, an efficiently sampleable full-support
    initialization, and a
    polynomial running time $T$ such that continuous-time unrestricted
    single-site heat-bath dynamics along \eqref{eq:symmetric-reverse-path}
    outputs a law $\nu_T$ satisfying
    $$
        D_{\mathrm{TV}}\tp{\mu_{G,\lambda_*},\nu_T}
        \leq\tau.
    $$
    Its expected number of updates is $T$.
\end{theorem}

The proof is given in \zcref{subsec:appendix-symmetric-phase-bound}.

Although this is a one-path application, the target still has two phases.  We
average the two within-phase fluxes because symmetry keeps their relative
weights fixed throughout the path.  We next relax exact balance: if only the
target weights are close to $(1/2,1/2)$, we use the componentwise theorem.

\paragraph{Approximate balance.}
An automorphism exchanging the two sides is sufficient for their phase
weights to be exactly equal, but it is not necessary.  The next proposition
shows that high girth can make the weights approximately equal even without
such an automorphism.

We define the two phase weights directly.  Let $G^2$ be the graph on $V(G)$
in which two vertices are adjacent when their distance in $G$ is at most two.
Call a subset of either bipartition class \emph{sparse} if every connected
component that it induces in $G^2$ has size at most half that class.  Define
\begin{equation*}
    \Omega_L
    =\{I\in\mathcal I(G):I\cap R\text{ is sparse}\},
    \qquad
    \Omega_R
    =\{I\in\mathcal I(G):I\cap L\text{ is sparse}\}.
\end{equation*}
Thus $\Omega_L$ is the phase in which the occupied vertices on the side
opposite $L$ occur only in small connected groups; $\Omega_R$ is defined
symmetrically.  The two phase domains cover $\mathcal I(G)$ and may overlap
\cite[Section~4.1]{JKP20}.  Their unnormalized hard-core weights are
\begin{equation}
    W^L_\lambda
    \defeq\sum_{I\in\Omega_L}\lambda^{|I|},
    \qquad
    W^R_\lambda
    \defeq\sum_{I\in\Omega_R}\lambda^{|I|}.
    \label{eq:hard-core-phase-partition-functions}
\end{equation}
These are the partition functions of the $L$-dominant and $R$-dominant
phases, respectively.

\begin{proposition}[High girth gives approximate phase balance]
    \label{prop:high-girth-phase-balance}
    Let $G=(L,R,E)$ be a balanced $d$-regular bipartite $\alpha$-expander on
    $N$ vertices.  Fix $\lambda>0$ and $a>0$, and suppose
    \begin{equation}
        2e^{2+a}d^4\leq(1+\lambda)^\alpha.
        \label{eq:high-girth-kp-margin}
    \end{equation}
    If $1\leq\ell\leq |L|/2$ and
    $2\ell<\operatorname{girth}(G)$, then
    \begin{equation*}
        \left|\log\frac{W^L_\lambda}{W^R_\lambda}\right|
        \leq 2Ne^{-a\ell},
        \qquad
        \left|
          \frac{W^L_\lambda}{W^L_\lambda+W^R_\lambda}-\frac12
        \right|
        \leq \frac N2e^{-a\ell}.
    \end{equation*}
\end{proposition}

The proof is given in
\zcref{subsec:appendix-high-girth-phase-balance}.  It represents the logarithm
of each phase weight as a sum of connected local terms.  When such a term
involves fewer than $\ell$ vertices, the girth assumption makes the relevant
part of $G$ a tree.  Since $G$ is regular and balanced, every such tree occurs
equally often from the two sides, so these terms cancel when the two phase
weights are compared.  Condition
\eqref{eq:high-girth-kp-margin} ensures that the total contribution of all
remaining terms is at most $Ne^{-a\ell}$ on either side.  Thus high girth gives
the cancellation, while the quantitative fugacity condition controls the
remainder.  In particular, if
$\operatorname{girth}(G)\geq c_{\rm g}\log N$, then taking
$\ell=(c_{\rm g}/2-o(1))\log N$ makes the weight error at most
\begin{equation*}
    N^{\,1-a c_{\rm g}/2+o(1)}.
\end{equation*}
Thus inverse-polynomial balance requires the additional quantitative condition
$a c_{\rm g}/2>1$; for any fixed $\kappa>0$, accuracy $N^{-\kappa}$ follows
when $a c_{\rm g}/2>\kappa+1$.

At the target fugacity $\lambda_*$, define the normalized mixture coefficients
\begin{equation}
    w_L=\frac{W^L_{\lambda_*}}
      {W^L_{\lambda_*}+W^R_{\lambda_*}},
    \qquad
    w_R=1-w_L.
    \label{eq:polymer-target-phase-weights}
\end{equation}
Because $\Omega_L$ and $\Omega_R$ may overlap, these are coefficients in the
two-phase approximation; they need not equal the event probabilities
$\mu_{G,\lambda_*}(\Omega_L)$ and
$\mu_{G,\lambda_*}(\Omega_R)$.  Write $w=(w_L,w_R)$, and suppose that the
algorithm chooses the two
initial phases with probabilities $(1/2,1/2)$.  Then
\begin{equation*}
    \delta_w=D_{\mathrm{TV}}\tp{w,(1/2,1/2)}
    =\left|w_L-\frac12\right|
    \leq \frac N2e^{-a\ell}.
\end{equation*}
Thus the componentwise theorem adds at most
$Ne^{-a\ell}/2$ to the TV bound.  \zcref{prop:high-girth-phase-balance}
controls only this error in the phase
weights.  The other input to the componentwise theorem---polynomial
within-phase action and exponentially small outward flow---is established for
the same phase domains in \zcref{prop:polymer-component-bounds} in the next
subsection.  Combining these two propositions gives the following theorem;
its proof is in \zcref{subsec:appendix-approximate-balance-annealing}.  Notice
that only the target weights must be close to $1/2$; the intermediate Gibbs
weights need not be.  Unlike the exact-balance result, this argument does not
assume an automorphism exchanging the two sides.
\begin{theorem}[Annealing from approximate phase balance]
    \label{thm:approximate-balance-annealing}
    Fix $\alpha\in(0,1)$, an integer $\Delta\geq3$, and $a>0$.  There is a
    threshold $\lambda_{\min}=\lambda_{\min}(\alpha,\Delta)$ such that the
    following holds for every target fugacity
    $\lambda_*\geq\lambda_{\min}$.  Let
    $G=(L,R,E)$ be a balanced $d$-regular bipartite $\alpha$-expander on $N$
    vertices, where $d\leq\Delta$.  Let $\tau\in(0,1)$ be inverse polynomial,
    and suppose that there is an integer $1\leq\ell\leq |L|/2$ such that
    \begin{equation*}
        2e^{2+a}d^4\leq(1+\lambda_*)^\alpha,
        \qquad
        2\ell<\operatorname{girth}(G),
        \qquad
        Ne^{-a\ell}\leq\frac\tau4.
    \end{equation*}
    For all sufficiently large $N$, one may choose a sufficiently large,
    polynomially bounded
    $\lambda_0\geq\lambda_*$ and a polynomial running time $T$ such that the
    following algorithm succeeds.  Choose $L$ or $R$ with probability $1/2$
    each as the initial independent set, and then run the unrestricted
    single-site heat-bath dynamics along the reverse schedule
    \eqref{eq:symmetric-reverse-path} for time $T$.  Its output law $\nu_T$
    satisfies
    \begin{equation*}
        D_{\mathrm{TV}}(\mu_{G,\lambda_*},\nu_T)\leq\tau,
    \end{equation*}
    and the expected number of updates is $T$.
\end{theorem}

\subsection{Computable unequal phase weights}
\label{subsec:polymer-informed-annealing}

\begin{figure}[H]
    \centering
    \begin{tikzpicture}[
        phase/.style={draw,very thick,rounded corners=15pt,
            minimum height=1.05cm,align=center},
        residual/.style={draw=gray!70,densely dashed,rounded corners=9pt,
            fill=gray!8,minimum height=0.72cm,align=center}
    ]
        \node[phase,fill=blue!13,minimum width=3.0cm] at (-1.75,0)
            {$L$-phase\\weight $w_L$};
        \node[phase,fill=orange!18,minimum width=1.95cm] at (1.45,0)
            {$R$-phase\\weight $w_R$};
        \node[residual,minimum width=1.25cm] at (3.85,0)
            {TV error\\$\leq\varepsilon$};
    \end{tikzpicture}
    \caption{The geometry in this subsection.  Up to TV error $\varepsilon$,
    the target is a mixture of its $L$- and $R$-dominant phase laws with
    weights $w_L$ and $w_R$, where $w_L+w_R=1$.  The weights are computed
    before annealing.}
    \label{fig:computed-two-phase-geometry}
\end{figure}

We retain the phase domains $\Omega_L,\Omega_R$ and their unnormalized weights
$W^L_\lambda,W^R_\lambda$ from
\eqref{eq:hard-core-phase-partition-functions}, but we no longer assume that
the graph is balanced or that the two weights are close.  Fix
$\alpha\in(0,1)$ and
$\Delta\geq3$, and let $G=(L,R,E)$ be a bipartite $\alpha$-expander on $N$
vertices with maximum degree at most $\Delta$.  At the target fugacity
$\lambda_*$, retain the normalized coefficients $w_L,w_R$ from
\eqref{eq:polymer-target-phase-weights}.
There is no reason for $w_L$ and $w_R$ to equal $1/2$.  Instead, the algorithm
of Jenssen, Keevash, and Perkins~\cite{JKP20} approximates
$W^L_{\lambda_*}$ and $W^R_{\lambda_*}$.  If both approximations are within
factors $e^{\pm\gamma}$, normalizing them gives estimates
$\widehat w_L,\widehat w_R$ whose total-variation distance from $(w_L,w_R)$ is
at most $e^{2\gamma}-1$.

It remains to control the dynamics inside each phase.  Jenssen, Keevash, and
Perkins show that the hard-core law is close to the mixture of the two
conditional phase laws and that configurations with a large connected
deviation have small mass~\cite[Section~4]{JKP20}.  Chen et al. show that,
after removing these configurations, the single-site chain within each phase
mixes in polynomial time
\cite[Section~5.3.2 and Theorems~25 and~29]{CGGPSV21}.  The next proposition states
the action, residual, initialization, and outward-flow bounds that follow from
these results.
\begin{proposition}[High-fugacity phase bounds imply componentwise action]
    \label{prop:polymer-component-bounds}
    Fix $\alpha\in(0,1)$ and $\Delta\geq3$.  There is a threshold
    $\lambda_{\min}=\lambda_{\min}(\alpha,\Delta)$ such that the following
    holds.  Let $G=(L,R,E)$ be a bipartite $\alpha$-expander on $N$ vertices
    with maximum degree at most $\Delta$, and fix
    $\lambda_0\geq\lambda_*\geq\lambda_{\min}$ with $\lambda_0$
    polynomially bounded.  Use the phase domains above and the reverse path
    \eqref{eq:symmetric-reverse-path}.  Let $\xi>0$ be inverse polynomial.
    The fixed phase domains $\Omega_L,\Omega_R$ are connected under
    single-site updates.  Define
    $\mu_{i,s}=\mu_{G,\lambda_s}(\,\cdot\mid\Omega_i)$.  For the
    flux--residual pair below, let $A_i$ denote its integrated quadratic cost
    and $\varepsilon_i$ its integrated $L^1$ residual budget, as in
    \eqref{eq:component-action-bounds}.  Let $\Lambda_i$ denote the integrated
    stationary outward flow in \eqref{eq:component-exit-flow}.  Set
    \begin{equation*}
        \pi_{\rm mix}
        \defeq w_L\mu_{L,1}+w_R\mu_{R,1}.
    \end{equation*}
    Then the target approximation and outward flow satisfy
    \begin{equation}
        D_{\mathrm{TV}}(\mu_{G,\lambda_*},\pi_{\rm mix})
        \leq e^{-\Omega_{\alpha,\Delta}(N)},
        \qquad
        \Lambda_L+\Lambda_R
        \leq e^{-\Omega_{\alpha,\Delta}(N)}.
        \label{eq:polymer-target-and-leakage-bounds}
    \end{equation}
    Moreover, for all sufficiently large $N$ and each $i\in\{L,R\}$, there
    is a flux--residual pair with
    \begin{equation}
        A_i\leq Q_{\alpha,\Delta}(N,1/\xi,\lambda_0)\mathsf B^2,
        \qquad
        \varepsilon_i\leq C_{\alpha,\Delta}N\mathsf B\xi,
        \qquad
        d_i\leq C_{\alpha,\Delta}\left(
          \frac N{\lambda_0}+\frac N{(1+\lambda_0)^\alpha}
        \right),
        \label{eq:polymer-component-bounds}
    \end{equation}
    where
    $d_L=D_{\mathrm{TV}}(\mu_{L,0},\delta_L)$,
    $d_R=D_{\mathrm{TV}}(\mu_{R,0},\delta_R)$, and
    $Q_{\alpha,\Delta}$ is a polynomial.
\end{proposition}

The proof is given in
\zcref{subsec:appendix-polymer-component-bounds}.  The cited results are used
as follows.  The phase decomposition defines $w$, and the counting algorithm
provides $\widehat w$.
The bound on large connected deviations controls the weak residual and the
outward flow.  The Poincar\'e bound for the restricted chain controls the
within-phase action.

The fixed-weight mixture path used in the componentwise comparison is
\begin{equation}
    \bar\mu_s
    =w_L\mu_{L,s}+w_R\mu_{R,s},
    \label{eq:polymer-fixed-weight-reference-path}
\end{equation}
where the weights are those of the \emph{target} fugacity $\lambda_*$ and
are held fixed for every $s$.  In applying
\zcref{thm:componentwise-weak-action-annealing}, the reference path is
$\rho_s=\mu_{G,\lambda_s}$ and the mixture target is
$\bar\mu_1=\pi_{\rm mix}$.  The approximation bound in
\eqref{eq:polymer-target-and-leakage-bounds} then converts the guarantee for
$\pi_{\rm mix}$ into one for the Gibbs target $\mu_{G,\lambda_*}$.
In general, $\bar\mu_s$ need not equal $\rho_s$, even at $s=1$.
The derivative of the mixture path contains only the within-phase terms
$w_i\partial_s\mu_{i,s}$ and no term proportional to the derivative of an
intermediate phase weight.

\begin{theorem}[Unrestricted annealing with computed phase weights]
    \label{thm:computed-phase-weight-annealing}
    Fix $\alpha\in(0,1)$ and $\Delta\geq3$.  There is a threshold
    $\lambda_{\min}=\lambda_{\min}(\alpha,\Delta)$ such that the following
    holds for every target fugacity $\lambda_*\geq\lambda_{\min}$.  Let
    $G=(L,R,E)$ be a bipartite $\alpha$-expander on $N$ vertices with maximum
    degree at most $\Delta$.  For all sufficiently large $N$ and every
    inverse-polynomial $\tau\in(0,1)$, one can choose a polynomially bounded
    $\lambda_0\geq\lambda_*$ and
    $T=\operatorname{poly}_{\alpha,\Delta}(N,1/\tau)$ so that the following
    algorithm uses polynomial time:
    \begin{enumerate}
        \item approximate $W^L_{\lambda_*}$ and $W^R_{\lambda_*}$ in
        polynomial time closely enough that the normalized estimates satisfy
        $D_{\mathrm{TV}}(\widehat w,w)\leq\tau/8$;
        \item choose the $L$- or $R$-ground state with these
        probabilities; and
        \item discard the phase label and run continuous-time unrestricted
        single-site heat-bath dynamics along
        \eqref{eq:symmetric-reverse-path} for time $T$.
    \end{enumerate}
    Its output law $\nu_T$ satisfies
    \begin{equation*}
        D_{\mathrm{TV}}(\mu_{G,\lambda_*},\nu_T)\leq\tau,
    \end{equation*}
    and its expected number of updates is $T$.
\end{theorem}

The proof is given in
\zcref{subsec:appendix-polymer-component-bounds}.

The restricted chains are used only to prove the within-phase bounds.  The
chain studied by Chen et al.~\cite[Section~5.3.2 and Theorems~25
and~29]{CGGPSV21} rejects moves that leave a chosen high-probability
subset of one phase.  Our algorithm instead chooses an initial phase using
$\widehat w$, discards the label, and runs the ordinary unrestricted
single-site heat-bath chain.  Since the initialization already has the desired
phase weights, the chain need not move probability between the phases.  The
terms $T\Lambda_i$ account for updates that leave the phase domains.

This application illustrates the information needed by componentwise
annealing in a two-phase system: each phase must admit low-cost internal
transport with little outward flow, and the target phase weights must be
computed before annealing begins.  Here the approximate counting algorithm
provides these weights, and the initialization uses them before the
unrestricted dynamics begins.  The same argument applies to other two-phase
systems whenever these three properties can be verified.

\subsection{One dominant phase on unbalanced bipartite graphs}
\label{subsec:fixed-gap-two-spin-model}

\begin{figure}[H]
    \centering
    \begin{tikzpicture}[
        phase/.style={draw,very thick,rounded corners=16pt,
            minimum height=1.12cm,align=center},
        residual/.style={draw=gray!70,densely dashed,rounded corners=9pt,
            fill=gray!8,minimum height=0.72cm,align=center}
    ]
        \node[phase,fill=blue!13,minimum width=5.25cm] at (-0.70,0)
            {$L$-dominant phase core\\mass at least $1-\varepsilon$};
        \node[residual,minimum width=1.35cm] at (3.70,0)
            {remaining tail\\$\leq\varepsilon$};
    \end{tikzpicture}
    \caption{The geometry in this subsection.  One phase carries all but
    $\varepsilon$ of the mass.  We transport the change inside this phase and
    include the remaining tail in the weak residual.}
    \label{fig:one-phase-geometry}
\end{figure}

Unlike the two-phase settings in \zcref{subsec:symmetric-low-temperature-bis,
subsec:polymer-informed-annealing}, no phase weight needs to be supplied at
initialization.  A fixed degree imbalance makes the $L$-dominant phase carry
almost all the mass; the small complement is included in the weak residual.

Let $G=(L,R,E)$ be a finite, simple, nonempty $(d,D)$-biregular bipartite
graph satisfying
\begin{equation*}
    \deg(u)=d\quad(u\in L),
    \qquad
    \deg(v)=D\quad(v\in R),
    \qquad
    \frac{D}{d}\geq1+\delta
\end{equation*}
for a fixed $\delta>0$.  Since $d|L|=D|R|$, the side $L$ is larger.  Write
$N=|L|+|R|$.

For an independent set $I=A\cup B$, where $A\subseteq L$ and $B\subseteq R$,
define
\begin{equation}
    \begin{aligned}
        \pi_\lambda(A,B)
        &=\frac{\lambda^{|B|}}{Z_G(1,\lambda)},
        &\lambda&\in[0,1],\\
        Z_G(1,\lambda)
        &=\sum_{\substack{A'\subseteq L,\,B'\subseteq R\\
                           A'\cup B'\in\mathcal I(G)}}
          \lambda^{|B'|}.
    \end{aligned}
    \label{eq:fixed-gap-hard-core-law}
\end{equation}
Here and below, we use the convention $0^0=1$.
The two arguments of $Z_G(1,\lambda)$ record the fugacities on $L$ and $R$.
Thus the fugacity is one on $L$ and $\lambda$ on $R$.  Let $P_\lambda$ denote
the total-rate-one single-site heat-bath transition matrix for $\pi_\lambda$,
and set
\begin{equation*}
    p_\lambda=P_\lambda-\!I,
    \qquad
    c_\lambda(I,J)=\pi_\lambda(I)p_\lambda(I,J)
\end{equation*}
for distinct independent sets $I,J$, and set $c_\lambda(I,I)=0$.  The target
$\pi_1$ is the uniform
distribution on $\mathcal I(G)$.

Given $\lambda_0\in(0,1)$, set
\begin{equation}
    \mathsf B=\log(1/\lambda_0),
    \qquad
    \lambda_s=\lambda_0e^{s\mathsf B},
    \qquad s\in[0,1],
    \label{eq:fixed-gap-log-path}
\end{equation}
and write $\pi_s=\pi_{\lambda_s}$, $p_s=p_{\lambda_s}$, and
$c_s=c_{\lambda_s}$.

Recall that $G^2$ joins vertices whose distance in $G$ is at most two.  For
$B\subseteq R$, call the connected components of $G^2[B]$ \emph{contours}.
For $m_*\geq1$, let $\mathcal A_{m_*}$ be the configurations for which every
contour $\gamma$ satisfies $|N_G(\gamma)|\leq m_*$.

\begin{proposition}[A high-mass phase core]
    \label{prop:fixed-gap-phase-core-bound}
    For every fixed $\delta>0$, there are constants
    $D_0=D_0(\delta)$ and $c=c(\delta)>0$ such that, whenever
    $D\geq D_0$, uniformly for $\lambda\in(0,1]$ and $m_*\geq1$,
    writing $\pi_\lambda^{\mathcal A}
    =\pi_\lambda(\,\cdot\mid\mathcal A_{m_*})$ and
    $p_\lambda^{\mathcal A}$ for the restriction of $p_\lambda$ to
    $\mathcal A_{m_*}$, we have
    \begin{align}
        \pi_\lambda(\mathcal A_{m_*}^{\mathsf c})
        &\leq N e^{-cm_*},
        \label{eq:fixed-gap-tail}\\
        \Var[\pi_\lambda^{\mathcal A}]{f}
        &\leq 32N^2m_*^3 2^{m_*}
        \+E_{p_\lambda^{\mathcal A}}\tp{f,f},
        \qquad f:\mathcal A_{m_*}\to\bb R.
        \label{eq:fixed-gap-restricted-pi}
    \end{align}
\end{proposition}

The first bound says that $\mathcal A_{m_*}$ contains almost all the mass; the
second says that heat-bath dynamics mixes rapidly inside it when $m_*$ is
logarithmic.  The proof uses the rooted counting bound of Jenssen, Perkins,
and Potukuchi~\cite[Lemma~4]{JPP23} and is given in
\zcref{subsec:appendix-fixed-gap-bound}.

\begin{proposition}[Weak action inside the dominant phase]
    \label{prop:fixed-gap-weak-action}
    Under the assumptions of
    \zcref{prop:fixed-gap-phase-core-bound}, for every $m_*\geq1$, there
    exists a residual budget $\varepsilon_*\geq0$ such that
    \begin{equation}
        \+A^{\mathrm w}_{\varepsilon_*}
        \tp{\tp{\pi_s}_{s\in[0,1]}}
        \leq
        32N^2m_*^3 2^{m_*}|R|^2\mathsf B^2,
        \qquad
        \varepsilon_*\leq
        N|R|\mathsf B e^{-cm_*}.
        \label{eq:fixed-gap-weak-action-bound}
    \end{equation}
\end{proposition}

The flux transports the path derivative inside $\mathcal A_{m_*}$; the
contribution from its small complement is included in the residual.

At $\lambda=0$, the law has $B=\varnothing$ and includes every vertex of $L$
independently with probability $1/2$.  This gives a simple reference law for
initializing the path from a sufficiently small positive $\lambda_0$.

\begin{lemma}[Initialization]
    \label{lem:fixed-gap-initialization}
    For every $\zeta\in(0,1)$, one can choose
    $\lambda_0^{-1}=\operatorname{poly}(N,1/\zeta)$ and an efficiently sampleable
    full-support law $\nu_0$ such that
    \begin{equation}
        D_{\mathrm{KL}}(\pi_{\lambda_0}\|\nu_0)\leq\zeta.
        \label{eq:fixed-gap-initial-kl}
    \end{equation}
\end{lemma}

\begin{theorem}[Annealing the uniform hard-core model with one dominant phase]
    \label{thm:fixed-gap-sampling}
    For every fixed $\delta>0$, there is
    $D_0=D_0(\delta)$ such that the following holds for every finite
    simple $(d,D)$-biregular graph with
    $D/d\geq1+\delta$ and $D\geq D_0$.  For every $\tau\in(0,1)$, one can
    choose $\lambda_0\in(0,1)$, an efficiently sampleable full-support initialization, and
    \begin{equation*}
        T=\operatorname{poly}_{\delta}(N,1/\tau)
    \end{equation*}
    so that continuous-time unrestricted single-site heat-bath dynamics along
    \eqref{eq:fixed-gap-log-path} outputs a law $\nu_T$ satisfying
    \begin{equation*}
        D_{\mathrm{TV}}\tp{\pi_1,\nu_T}\leq\tau.
    \end{equation*}
    The expected number of updates is $T$.
\end{theorem}

The proofs of the two propositions, the initialization lemma, and the theorem
are given in \zcref{subsec:appendix-fixed-gap-bound}.  This is the one-phase
counterpart of the preceding two applications: no phase label is sampled,
because the degree imbalance leaves only one phase above the requested
accuracy.

A closely related one-phase geometry appears in the work of Cannon and
Perkins~\cite{CP20}.  For the uniform hard-core model on a $(d,D)$-biregular
graph, their sampling result applies when $D\geq7d\log d$.  In this regime,
every instance with $d\geq2$ has $D/d\geq7\log2>4$, so it satisfies the
degree conditions of the theorem above with $\delta=3$ once $D\geq D_0(3)$.
Both
arguments are centered on configurations with no occupied vertices in $R$ and
describe deviations by connected occupied sets in $R$.  Cannon and Perkins
use this geometry to obtain a convergent polymer expansion.  Truncating that
expansion gives their counting algorithm and sampler.  Here the one-phase
geometry is instead used to bound weak action for unrestricted annealing.

\paragraph{Soft-core extension.}
For arbitrary $A\subseteq L$ and $B\subseteq R$, define
\begin{equation*}
    \begin{aligned}
        \pi_{\lambda,\eta}(A,B)
        &=\frac{\lambda^{|B|}\eta^{e(A,B)}}
                {Z_G(1,\lambda;\eta)},
        &\lambda&\in[0,1],\quad \eta\in[0,1),\\
        Z_G(1,\lambda;\eta)
        &=\sum_{A'\subseteq L,\,B'\subseteq R}
          \lambda^{|B'|}\eta^{e(A',B')}.
    \end{aligned}
\end{equation*}
Here $e(A,B)$ is the number of edges between $A$ and $B$.  The semicolon in
$Z_G(1,\lambda;\eta)$ separates the interaction parameter from the two
fugacities.
This agrees with \eqref{eq:fixed-gap-hard-core-law} when $\eta=0$.
For $\eta>0$, it is an antiferromagnetic two-spin system in which adjacent
occupied vertices receive a penalty $\eta$.  Let $P_{\lambda,\eta}$ denote its
total-rate-one single-site heat-bath transition matrix, and set
\begin{equation*}
    p_{\lambda,\eta}=P_{\lambda,\eta}-\!I,
    \qquad
    c_{\lambda,\eta}(x,y)
    =\pi_{\lambda,\eta}(x)p_{\lambda,\eta}(x,y)
\end{equation*}
for distinct configurations $x,y$, with zero capacity on the diagonal.

For general bounded-degree graphs, Li, Lu, and Yin give a correlation-decay
FPTAS throughout the tree-uniqueness region~\cite{LLY13}.  In the tree
non-uniqueness region, Cai et al. establish broad \#BIS-hardness results even
on bounded-degree bipartite graphs~\cite{CGGGJSV16}.  The result below instead
uses degree imbalance to retain the one-phase geometry of the hard-core case.

\begin{proposition}[Phase bounds for a small soft-core interaction]
    \label{prop:fixed-gap-soft-core-bounds}
    For every fixed $\delta>0$, there is $\eta_0=\eta_0(\delta)>0$ such that
    the following holds for every fixed $\eta\in(0,\eta_0)$.  After allowing
    $D_0$, $c$, and the initialization polynomial to depend on
    $(\delta,\eta)$, the conclusions of
    \zcref{prop:fixed-gap-phase-core-bound,prop:fixed-gap-weak-action,
    lem:fixed-gap-initialization} remain valid with
    $(\pi_\lambda,p_\lambda,c_\lambda)$ replaced by
    $(\pi_{\lambda,\eta},p_{\lambda,\eta},c_{\lambda,\eta})$.
\end{proposition}

The only new estimate is the contour weight.  After summing over the
$L$-spins, softening the hard constraint increases the weight of a contour
$\gamma$ by at most $(1+\eta)^{|N_G(\gamma)|}$.  For sufficiently small
$\eta$, the exponential decay supplied by the degree imbalance still
dominates this factor.  The remaining phase-core, weak-action, and
initialization arguments are unchanged up to constants.  The proof is given
in \zcref{subsec:appendix-fixed-gap-bound}.

\begin{theorem}[Soft-core extension]
    \label{thm:fixed-gap-soft-core-sampling}
    For every fixed $\delta>0$, there is $\eta_0=\eta_0(\delta)>0$ such that
    the following holds for every fixed $\eta\in(0,\eta_0)$.  There is
    $D_0=D_0(\delta,\eta)$ such that, for every finite simple
    $(d,D)$-biregular graph with $D/d\geq1+\delta$ and $D\geq D_0$, and every
    $\tau\in(0,1)$, one can choose $\lambda_0\in(0,1)$, an efficiently sampleable
    full-support initialization, and
    \begin{equation*}
        T=\operatorname{poly}_{\delta,\eta}(N,1/\tau)
    \end{equation*}
    so that continuous-time unrestricted single-site heat-bath dynamics
    along \eqref{eq:fixed-gap-log-path}, with kernels
    $(P_{\lambda_s,\eta})_{s\in[0,1]}$, outputs a law $\nu_T$ satisfying
    \begin{equation*}
        D_{\mathrm{TV}}\tp{\pi_{1,\eta},\nu_T}\leq\tau.
    \end{equation*}
    The expected number of updates is $T$.
\end{theorem}

This follows from \zcref{prop:fixed-gap-soft-core-bounds} and the one-path
annealing theorem.  Details are included in
\zcref{subsec:appendix-fixed-gap-bound}.

\subsection{Weak action can be polynomial while exact action is exponential}
\label{subsec:complete-biclique-separation}

The preceding subsection gives polynomial weak action along the one-sided path
\eqref{eq:fixed-gap-log-path}.  We now show that the exact action of the same
path can be exponential.  Consider the complete biclique $G=K_{D,d}$ with
$|L|=D$ and $|R|=d$.  This is a $(d,D)$-biregular graph in the notation of the
preceding subsection.  Assume that $D/d\to\rho>1$, and take
$\lambda_0\leq e^{-1/d}$ so that the path crosses the interval
$\lambda\in[e^{-1/d},1]$.

The relevant small set consists of independent sets with no occupied vertex
in $L$.  Its probability changes along the path, but the stationary heat-bath
flow across its boundary is much smaller than this rate of change.  This
forces any exact flux to pay a large cost.

\begin{proposition}[Hard-core exact action on complete bicliques]
    \label{prop:complete-biclique-hard-core-action}
    Under the assumptions above, the action of the path
    $(\pi_{\lambda_s})_{s\in[0,1]}$, computed using the capacities
    $(c_{\lambda_s})_{s\in[0,1]}$, satisfies
    \begin{equation}
        \+A\tp{\tp{\pi_{\lambda_s}}_{s\in[0,1]}}
        \geq
        \mathsf B d\,2^{2d-D-O(1)}.
        \label{eq:hard-core-complete-biclique-action}
    \end{equation}
\end{proposition}

The proof is given in \zcref{subsec:appendix-complete-biclique-action}.

Now fix $\rho\in(1,2)$ and $\delta\in(0,\rho-1)$.  For all sufficiently large
$d$, the graph $K_{D,d}$ satisfies the assumptions of
\zcref{prop:fixed-gap-weak-action}.  Choose $\lambda_0$ as in
\zcref{lem:fixed-gap-initialization}.  This already gives
$\lambda_0\leq e^{-1/d}$.  Then $\lambda_0^{-1}$ is polynomially bounded,
$\mathsf B=O(\log N)$, and taking $m_*$ logarithmic makes the weak action
polynomial with an inverse-polynomial residual budget.  In contrast, since
$D/d\to\rho<2$,
\eqref{eq:hard-core-complete-biclique-action} is exponential in $d$.  This
justifies the intuition that weak action avoids transporting the changing mass
of the small set through this small boundary capacity.  Instead, that
contribution is included in the residual.

\section*{AI Disclosure}

This work was originally conceived prior to the recent rapid surge in frontier AI mathematical capabilities. Inspired by the continuous-space transport framework of Guo, Tao, and Chen~\cite{GTC25}, the authors sought to construct a discrete analogue for finite-state simulated annealing. In the initial phase of this project, the authors independently established the foundational theoretical architecture, including the discrete Benamou--Brenier action, the discrete continuity equations, and the mean-field applications. During this period, generative AI was utilized solely as an auxiliary tool for preliminary discussions and prose drafting, with all core conceptual breakthroughs, formulations, and structural insights originating entirely from the authors.

In the subsequent phase, while exploring concrete applications of the framework in dialogue with GPT-5.6 Sol, the model directly contributed to deriving the action bound for the unbalanced bipartite spin system presented in \zcref{subsec:fixed-gap-two-spin-model}. In addition, the model provided valuable assistance in executing intermediate algebraic calculations, verifying derivations across other examples, and polishing the manuscript's prose. Claude Opus 5.5 was used in the final revision to polish the exposition. The authors have independently audited and rigorously verified all mathematical proofs, computations, and arguments, and take full responsibility for the contents and intellectual integrity of this paper.
\bibliographystyle{alpha}
\bibliography{refs}

\appendix
\section{Omitted Proofs}

\subsection{Proof of
\zcref{lem:localized-continuous-annealing-comparison}}
\label{subsec:proof-localized-unrestricted-comparison}

\begin{proof}
    Rescale the flux--residual pair from $[0,1]$ to $[0,T]$ by setting
    \begin{equation*}
        \tilde\mu_t=\mu_{t/T},\qquad
        \tilde J_t=T^{-1}J_{t/T},\qquad
        \tilde r_t=T^{-1}r_{t/T},\qquad
        \tilde c_t(x,y)=\mu_{t/T}(x)p_{t/T}(x,y)
    \end{equation*}
    for $x,y\in S$. Then
    \begin{align}
        \int_0^T\sum_{\substack{\{x,y\}\in E\\x,y\in S}}
          \frac{\tilde J_t(x,y)^2}{\tilde c_t(x,y)}\,\dd t
        &\leq\frac{A}{T},
        \notag\\
        R_T\defeq\frac12\int_0^T\|\tilde r_t\|_1\,\dd t
        &\leq\varepsilon,
        \label{eq:localized-rescaled-bounds}\\
        \int_0^T\sum_{x\in S}\sum_{y\notin S}
          \tilde\mu_t(x)p_{t/T}(x,y)\,\dd t
        &=T\Lambda.
        \notag
    \end{align}

    Let $Q_{s,t}$ denote the transition kernel from time $s$ to time $t$ of
    the unrestricted annealing process with rate kernel $p_{t/T}$. Fix an
    arbitrary $h:\Omega\to\bb R$ and define
    \begin{equation*}
        u_t(x)=\tp{Q_{t,T}e^h}(x),
        \qquad
        \varphi_t(x)=\log u_t(x).
    \end{equation*}
    As in the proof of \zcref{lem:exact-weak-action-comparison}, the backward
    equation holds for almost every $t$, and
    $e^{\min h}\le u_t\le e^{\max h}$. Hence
    \begin{equation}
        \partial_t\varphi_t(x)
        +\sum_y p_{t/T}(x,y)
          \left(e^{\varphi_t(y)-\varphi_t(x)}-1\right)=0,
        \qquad
        \operatorname{Osc}(\varphi_t)\leq\operatorname{Osc}(h).
        \label{eq:localized-backward-potential}
    \end{equation}
    Moreover,
    \begin{equation*}
        \sum_x\kappa(x)e^{\varphi_0(x)}
        =\sum_x\nu(x)e^{h(x)}.
    \end{equation*}
    The Donsker--Varadhan formula \eqref{eq:donsker-varadhan} at the initial
    time yields,
    with the right-hand side understood in the extended sense,
    \begin{equation*}
        \sum_x\mu_0(x)\varphi_0(x)
        -\log\sum_x\nu(x)e^{h(x)}\leq K.
    \end{equation*}
    Hence, with
    \begin{equation*}
        \mathcal F_T^S(h)
        \defeq
        \sum_x\mu_1(x)h(x)-\log\sum_x\nu(x)e^{h(x)},
    \end{equation*}
    we have
    \begin{equation}
        \mathcal F_T^S(h)
        \leq K+
        \int_0^T\frac{\dd}{\dd t}
          \sum_{x\in S}\tilde\mu_t(x)\varphi_t(x)\,\dd t.
        \label{eq:localized-initial-entropy-bridge}
    \end{equation}

    We now differentiate the pairing. The relaxed continuity equation,
    discrete summation by parts, \eqref{eq:localized-backward-potential}, and
    internal detailed balance give, for almost every $t$,
    \begin{align*}
        \frac{\dd}{\dd t}
        \sum_{x\in S}\tilde\mu_t(x)\varphi_t(x)
        ={}&
        \sum_{\substack{\{x,y\}\in E\\x,y\in S}}
        \left[
          \tilde J_t(x,y)\tp{\varphi_t(y)-\varphi_t(x)}
          -2\tilde c_t(x,y)
           \tp{\cosh\tp{\varphi_t(y)-\varphi_t(x)}-1}
        \right]\\
        &+\sum_{x\in S}\tilde r_t(x)\varphi_t(x)\\
        &+\sum_{x\in S}\sum_{y\notin S}
          \tilde\mu_t(x)p_{t/T}(x,y)
          \tp{1-e^{\varphi_t(y)-\varphi_t(x)}}.
    \end{align*}
    The internal edge term is at most
    $\tilde J_t(x,y)^2/(4\tilde c_t(x,y))$, while the zero-mass
    property of $\tilde r_t$ gives
    \begin{equation*}
        \sum_{x\in S}\tilde r_t(x)\varphi_t(x)
        \leq\frac12\|\tilde r_t\|_1\operatorname{Osc}(h).
    \end{equation*}

    It remains to control the new boundary term. For every $z\in\bb R$,
    \begin{equation*}
        1-e^z\leq1,
        \qquad
        1-e^z\leq(-z)_+.
    \end{equation*}
    The second inequality follows from $e^z\geq1+z$. Since
    $(-\tp{\varphi_t(y)-\varphi_t(x)})_+
      \leq\operatorname{Osc}(\varphi_t)$, each exit summand is therefore at
    most its exit rate times
    $\min\set{1,\operatorname{Osc}(h)}$.
    Integrating the differential inequality and using
    \eqref{eq:localized-rescaled-bounds} in
    \eqref{eq:localized-initial-entropy-bridge} proves the localized defective
    free-energy inequality
    \begin{equation}
        \mathcal F_T^S(h)
        \leq
        K+\frac{A}{4T}
        +R_T\operatorname{Osc}(h)
        +T\Lambda\min\set{1,\operatorname{Osc}(h)}.
        \label{eq:localized-defective-free-energy-comparison}
    \end{equation}

    We first derive the weak-action TV bound. If $K=+\infty$, the first bound
    is immediate. Otherwise, fix $B\subseteq\Omega$ and $\theta>0$, and take
    $h=\theta\mathbf1_B$. Applying \zcref{lem:hoeffding} to
    $\mathbf1_B\tp{X}$ with $X\sim\nu$ gives
    \begin{equation*}
        \mathcal F_T^S(\theta\mathbf1_B)
        \geq
        \theta\tp{\mu_1(B)-\nu(B)}-\frac{\theta^2}{8}.
    \end{equation*}
    Since $R_T\leq\varepsilon$ and
    $\min\{1,\operatorname{Osc}(h)\}\leq\theta$, comparison with
    \eqref{eq:localized-defective-free-energy-comparison} yields
    \begin{equation*}
        \mu_1(B)-\nu(B)
        \leq
        \varepsilon+T\Lambda
        +\frac1\theta\left(K+\frac{A}{4T}\right)
        +\frac\theta8.
    \end{equation*}
    Since $\inf_{\theta>0}\tp{M/\theta+\theta/8}=\sqrt{M/2}$ for $M\geq0$, with
    the value at $M=0$ understood by letting $\theta\downarrow0$, optimizing in
    $\theta$ and then taking the supremum over $B$ proves
    \eqref{eq:localized-kl-initializer-tv-bound}.

    For the second TV bound, let $\bar\nu=\mu_0Q_T$. Applying the preceding
    argument with the matched initializer $\mu_0$ gives
    \begin{equation*}
        D_{\mathrm{TV}}(\mu_1,\bar\nu)
        \leq\varepsilon+T\Lambda+\sqrt{\frac{A}{8T}}.
    \end{equation*}
    Total variation contracts under $Q_T$, so
    $D_{\mathrm{TV}}(\bar\nu,\nu)\leq D_{\mathrm{TV}}(\mu_0,\kappa)=d$. The triangle inequality
    proves \eqref{eq:localized-tv-initializer-bound}.

    Finally, suppose $r_s=0$ for almost every $s$. If $K=+\infty$, the exact
    assertion is immediate. Otherwise $R_T=0$, and
    \eqref{eq:localized-defective-free-energy-comparison}, together with
    $\min\{1,\operatorname{Osc}(h)\}\leq1$, gives uniformly in $h$
    \begin{equation*}
        \mathcal F_T^S(h)
        \leq K+\frac{A}{4T}+T\Lambda.
    \end{equation*}
    Taking the supremum over $h:\Omega\to\bb R$ and using the
    Donsker--Varadhan formula \eqref{eq:donsker-varadhan} proves
    \eqref{eq:localized-exact-kl-bound}.
\end{proof}

\subsection{Proof of the Metric Properties of \zcref{def:discrete-W2}}\label{sec:pf-of-discrete-W2}

Before proving the metric properties, we first establish a fundamental lemma, which guarantees that convergence in $W_2$ implies convergence in $L_1$.

\begin{lemma}\label{lem:w2-bounds-euclidean}
    Let $C_{\max} \defeq \sup_{\rho \in \+P_{>0}(\Omega)} \sum_{x,y \in \Omega} c_\rho(x,y) < \infty$. For any $\mu, \nu \in \+P_{>0}(\Omega)$, their $L_1$ distance is bounded by their $W_2$ distance:
    $$
        \|\mu - \nu\|_1 \le \sqrt{2 C_{\max}} W_2(\mu, \nu).
    $$
\end{lemma}
\begin{proof}
    Consider any curve $\tp{\rho_s}_{s \in [0,1]}$ connecting $\rho_0 = \mu$ and $\rho_1 = \nu$, driven by the admissible potential $\phi_s$. We have
    $$
        \nu(x) - \mu(x) = \int_0^1 \partial_s \rho_s(x) \, \dd s = \int_0^1 \sum_{y \in \Omega} \tp{\phi_s(x) - \phi_s(y)} c_{\rho_s}(x,y) \, \dd s.
    $$
    Applying the triangle inequality:
    $$
        \|\mu - \nu\|_1 = \sum_{x \in \Omega} |\mu(x) - \nu(x)| \le \int_0^1 \sum_{x, y \in \Omega} |\phi_s(x) - \phi_s(y)| c_{\rho_s}(x,y) \, \dd s.
    $$
    For the integrand, we apply the Cauchy-Schwarz inequality:
    \begin{align*}
        \sum_{x, y} |\phi_s(x) - \phi_s(y)| c_{\rho_s}(x,y) 
        &\le \tp{ \sum_{x, y} \tp{\phi_s(x) - \phi_s(y)}^2 c_{\rho_s}(x,y) }^{\frac{1}{2}} \tp{ \sum_{x, y} c_{\rho_s}(x,y) }^{\frac{1}{2}}.
    \end{align*}
    Bounding the second term by $\sqrt{C_{\max}}$, we obtain
    $$
        \|\mu - \nu\|_1 \le \sqrt{C_{\max}}  \int_0^1 \tp{ \sum_{x, y} \tp{\phi_s(x) - \phi_s(y)}^2 c_{\rho_s}(x,y) }^{\frac{1}{2}} \dd s.
    $$
    Applying Jensen's inequality, we have
    $$
        \|\mu - \nu\|_1 \le  \sqrt{2C_{\max}} \tp{ \frac{1}{2} \int_0^1 \sum_{x, y} \tp{\phi_s(x) - \phi_s(y)}^2 c_{\rho_s}(x,y)  \dd s}^{\frac{1}{2}}.
    $$
    Taking the infimum over all valid curves $\rho_s$ complete the proof.
\end{proof}

The converse comparison holds locally, which gives the following
characterization of absolutely continuous curves.

\begin{lemma}\label{lem:ac-euclidean-w2}
    A curve $\tp{\pi_t}_{t\in[0,T]}\subseteq\+P_{>0}(\Omega)$ is absolutely
    continuous in $\tp{\+P_{>0}(\Omega),W_2}$ if and only if it is absolutely
    continuous as a map $[0,T]\to\bb R^\Omega$.
\end{lemma}

\begin{proof}
    If the curve is absolutely continuous in $W_2$, then
    \zcref{lem:w2-bounds-euclidean} shows that it is absolutely continuous as a
    map into $\bb R^\Omega$. Conversely, suppose that it is absolutely
    continuous as a map into $\bb R^\Omega$. For $\rho\in\+P_{>0}(\Omega)$,
    let $(L_\rho\psi)(x)=\sum_y\tp{\psi(x)-\psi(y)}c_\rho(x,y)$. The image of
    the curve is compact in $\+P_{>0}(\Omega)$, and so is its convex hull
    $\widehat K$. Since $(\Omega,E)$ is connected and the capacities are
    strictly positive on $E$ and continuous in $\rho$, the smallest eigenvalue
    of $L_\rho$ on the zero-sum subspace is bounded below by some
    $\lambda_*>0$ uniformly in $\rho\in\widehat K$. For $s<t$, put
    $\Delta=\pi_t-\pi_s$. The segment $\rho_r=(1-r)\pi_s+r\pi_t$,
    $r\in[0,1]$, lies in $\widehat K$, and $\psi_r=L_{\rho_r}^+\Delta$ is an
    admissible potential for it, where the pseudoinverse is taken on the
    zero-sum subspace. Its cost at time $r$ is
    \[
        \frac12\sum_{x,y\in\Omega}\tp{\psi_r(x)-\psi_r(y)}^2c_{\rho_r}(x,y)
        =\langle\Delta,L_{\rho_r}^+\Delta\rangle
        \le\lambda_*^{-1}\norm{\Delta}_2^2.
    \]
    Hence
    \[
        W_2\tp{\pi_s,\pi_t}
        \le\lambda_*^{-1/2}\norm{\pi_t-\pi_s}_2
        \le\lambda_*^{-1/2}\int_s^t\norm{\partial_r\pi_r}_2\,\dd r,
    \]
    so the curve is absolutely continuous in $W_2$.
\end{proof}

We next prove that the Wasserstein-2 like metric defined in \zcref{def:discrete-W2} satisfies the properties of a metric on $\+P_{>0}(\Omega)$.

\begin{proposition}
    The Wasserstein-2 like metric defined in \zcref{def:discrete-W2} is a metric on $\+P_{>0}(\Omega)$.
\end{proposition}

\begin{proof}
We first check that the distance is finite.  Given $\mu,\nu\in\+P_{>0}(\Omega)$, take the straight-line curve $\rho_s=(1-s)\mu+s\nu$.  Define
\[
    (L_{\rho_s}\psi)\tp{x}
    =\sum_y\tp{\psi\tp{x}-\psi\tp{y}}c_{\rho_s}\tp{x,y}.
\]
Because the fixed graph $E$ is connected, $L_{\rho_s}$ is positive definite on the zero-sum subspace.  Hence there is a unique zero-mean solution of
$L_{\rho_s}\psi_s=\nu-\mu$.  Along the compact line segment traced by $\rho_s$, continuity and strict positivity of the edge capacities imply that these solutions and their quadratic costs are uniformly bounded.  Thus $(\rho_s,\psi_s)$ is an admissible pair with finite cost, so $W_2(\mu,\nu)<\infty$.

It remains to verify non-negativity, identity of indiscernibles, symmetry, and the triangle inequality.

Non-negativity obviously holds by the definition in \zcref{def:discrete-W2} and the properties of $c_\rho$.

For identity of indiscernibles, if $\mu = \nu$, we can construct a constant curve $\pi_t \equiv \mu$ for all $t \in [0,1]$. By setting the potential $\psi_t \equiv 0$, the continuity equation $\partial_t \pi_t + \sum (\psi_t(y) - \psi_t(x)) c_{\pi_t}(x,y) = 0$ holds trivially. Hence $W_2(\mu, \mu) = 0$.
For the other direction, assume $W_2(\mu, \nu) = 0$. 
By \zcref{lem:w2-bounds-euclidean}, we have $\|\mu - \nu\|_1 \le \sqrt{2 C_{\max}} W_2(\mu, \nu) = 0$. This implies $\mu = \nu$.


For the symmetry property, given an arbitrary curve $(\pi_t)_{t \in [0,1]}$ from $\mu$ to $\nu$ with admissible potential $\psi_t$, we define the time-reversed curve $\tilde{\pi}_t = \pi_{1-t}$ and reversed potential $\tilde{\psi}_t = -\psi_{1-t}$.
It is easy to verify that $\tilde{\psi}_t$ is an admissible potential for $\tilde{\pi}_t$. The cost remains unchanged because $(-\tilde{\psi}_t(x) - (-\tilde{\psi}_t(y)))^2 = (\tilde{\psi}_t(y) - \tilde{\psi}_t(x))^2$. Taking the infimum over all curves and admissible potentials yields symmetry.

For the triangle inequality, let $L_1 = W_2(\mu, \nu)$ and $L_2 = W_2(\nu, \sigma)$. For any $\epsilon > 0$, there exist curves $(\pi^1, \psi^1)$ and $(\pi^2, \psi^2)$ such that their corresponding integrals are strictly bounded by $(L_1+\epsilon)^2$ and $(L_2+\epsilon)^2$, respectively.
We concatenate these paths with a constant-speed time rescaling. Define $\lambda = \frac{L_1+\epsilon}{L_1+L_2+2\epsilon}$. We construct a new curve $(\tilde{\pi}_t)_{t \in [0,1]}$ from $\mu$ to $\sigma$:
\[
    \tilde{\pi}_t = 
    \begin{cases} 
        \pi^1_{t/\lambda} & \text{for } t \in [0, \lambda] \\
        \pi^2_{(t-\lambda)/(1-\lambda)} & \text{for } t \in (\lambda, 1]
    \end{cases}.
\]
Since the speed of the first segment is scaled by $1/\lambda$, we may take $\tilde{\psi}_t = \frac{1}{\lambda} \psi^1_{t/\lambda}$ as an admissible potential. The corresponding integral over the first segment becomes:
\[
    \int_0^\lambda \frac{1}{2} \sum_{x,y} \tp{\frac{1}{\lambda}\tp{\psi^1_{t/\lambda}(x) - \psi^1_{t/\lambda}(y)}}^2 c_{\tilde{\pi}_t}(x,y) \dd t.
\]
Letting $s = t/\lambda$, $\dd t = \lambda \dd s$, this simplifies to: 
\[
    \frac{1}{\lambda^2} \int_0^1 \frac{1}{2} \sum_{x,y} \tp{\psi^1_s(x) - \psi^1_s(y)}^2 c_{\pi^1_s} \cdot \lambda \dd s  \le \frac{(L_1+\epsilon)^2}{\lambda}.
\]
Similarly, the corresponding integral over the second segment $(\lambda, 1]$ is bounded by $\frac{(L_2+\epsilon)^2}{1-\lambda}$. 
Summing the two segments gives 
\[
    W_2^2(\mu, \sigma) \le \frac{(L_1+\epsilon)^2}{\lambda} + \frac{(L_2+\epsilon)^2}{1-\lambda} = (L_1+L_2+2\epsilon)^2.
\]
Taking the square root and letting $\epsilon \to 0$ yields the triangle inequality. This completes the proof.
\end{proof}

\subsection{Proof of Theorem~\ref{thm:metric-derivative-potential-formulation}}\label{sec:pf-of-metric-derivative}

\begin{proof}[Proof of Theorem~\ref{thm:metric-derivative-potential-formulation}]
The proof consists of two parts: first, we show that admissible potentials exist and are unique up to an additive constant, so that the right-hand side does not depend on the choice of admissible potential; second, we establish the equality with the metric derivative using a local limit argument and the continuity of the capacities.

At any time $t$, define the Laplacian operator $L_{\pi_t} : \bb R^{|\Omega|} \to \bb R^{|\Omega|}$ associated with the edge capacity $c_{\pi_t}(x,y)$ as
$$
    (L_{\pi_t} \psi)(x) \defeq \sum_{y \in \Omega} \tp{\psi(x) - \psi(y)} c_{\pi_t}(x,y).
$$
We can regard $\psi$ and $L_{\pi_t}$ as a vector and a matrix, respectively. The discrete continuity equation can thus be rewritten as $\partial_t \pi_t = L_{\pi_t} \psi_t$. 

To establish the existence of an admissible potential $\psi_t$, we examine the kernel and the image of $L_{\pi_t}$. The kernel is defined as $\ker(L_{\pi_t}) = \{\psi \in \bb R^{|\Omega|} : L_{\pi_t} \psi = 0\}$. For any $\psi \in \ker(L_{\pi_t})$, the associated quadratic form vanishes:
$$
    \langle \psi, L_{\pi_t} \psi \rangle  = \sum_{x, y \in \Omega} c_{\pi_t}(x,y) \tp{\psi(x) - \psi(y)}\psi(x) = \frac{1}{2} \sum_{x, y \in \Omega} c_{\pi_t}(x,y) \tp{\psi(x) - \psi(y)}^2 = 0.
$$
Since $c_{\pi_t}(x,y) \ge 0$, this implies $\psi(x) = \psi(y)$ for all edges $(x,y)$ where $c_{\pi_t}(x,y) > 0$. By our assumption, the underlying graph induced by $c_{\pi_t}$ is connected. Consequently, $\psi$ must be a constant vector, which means $\ker(L_{\pi_t}) = \text{span}\{\mathbf{1}\}$, where $\mathbf{1}$ is the all-ones vector.

By the fundamental theorem of linear algebra, since $L_{\pi_t}$ is a symmetric matrix, its image space $\text{Im}(L_{\pi_t})$ is the orthogonal complement of its kernel, i.e., $\text{Im}(L_{\pi_t}) = \ker(L_{\pi_t})^\perp$. Notice that because $\pi_t \in \+P_{>0}(\Omega)$ is a probability distribution for all $t$, the sum of its components is always $1$. Taking the time derivative yields $\sum_{x \in \Omega} \partial_t \pi_t(x) = 0$. This sum is exactly the Euclidean inner product $\langle \partial_t \pi_t, \mathbf{1} \rangle = 0$. This implies that $\partial_t \pi_t$ is orthogonal to the kernel, and thus $\partial_t \pi_t \in \text{Im}(L_{\pi_t})$. By the definition of the image space, this guarantees the existence of at least one solution $\psi_t$ to the linear system $L_{\pi_t} \psi_t = \partial_t \pi_t$, ensuring that an admissible potential always exists.

Furthermore, the solution $\psi_t$ is unique up to an additive constant. If $\psi_t^{(1)}$ and $\psi_t^{(2)}$ are two admissible potentials satisfying $L_{\pi_t} \psi_t = \partial_t \pi_t$, by linearity we have $L_{\pi_t} \tp{\psi_t^{(1)} - \psi_t^{(2)}} = 0$. This means their difference belongs to the kernel, so $\psi_t^{(1)} - \psi_t^{(2)} = c \mathbf{1}$ for some constant $c \in \bb R$. In other words, any two admissible potentials differ only by a globally constant shift. 

The right-hand side of \zcref{thm:metric-derivative-potential-formulation} is exactly the quadratic form associated with the Laplacian:
$$
    \|\psi_t\|_{\pi_t}^2 \defeq \frac{1}{2} \sum_{x, y \in \Omega} \tp{\psi_t(x) - \psi_t(y)}^2 c_{\pi_t}(x,y) = \langle \psi_t, L_{\pi_t} \psi_t \rangle.
$$
Because this objective depends only on the differences $\psi_t(x) - \psi_t(y)$, the value of $\|\psi_t\|_{\pi_t}^2$ is identical for all admissible potentials. Therefore, the right-hand side of \zcref{thm:metric-derivative-potential-formulation} does not depend on the choice of admissible potential.

We now establish the equality $|\dot{\pi}|_t^2 = \|\psi_t\|_{\pi_t}^2$.  By \zcref{lem:w2-bounds-euclidean}, every $W_2$-absolutely continuous curve is also Euclidean absolutely continuous.  Fix a time $t$ at which its Euclidean derivative and metric derivative both exist.  For small $h$, put
\[
    \Delta_h=\pi_{t+h}-\pi_t,
    \qquad
    \rho_s^{(h)}=\pi_t+s\Delta_h,
    \quad s\in[0,1].
\]
The straight segment lies in $\+P_{>0}(\Omega)$.  Let
$\phi_s^{(h)}=L_{\rho_s^{(h)}}^+\Delta_h$, where the pseudoinverse is taken on the zero-sum subspace.  Then
$L_{\rho_s^{(h)}}\phi_s^{(h)}=\Delta_h=\partial_s\rho_s^{(h)}$, so this is an admissible pair.  The fixed support and continuity of the capacities imply that $\rho\mapsto L_\rho^+$ is continuous on the positive simplex.  Consequently,
\[
    \frac{W_2^2\tp{\pi_t,\pi_{t+h}}}{h^2}
    \leq
    \int_0^1
    \left\langle\frac{\Delta_h}{h},
    L_{\rho_s^{(h)}}^+\frac{\Delta_h}{h}\right\rangle\,\dd s
    \longrightarrow
    \left\langle\partial_t\pi_t,
    L_{\pi_t}^+\partial_t\pi_t\right\rangle
    =\|\psi_t\|_{\pi_t}^2.
\]
This proves the upper bound without requiring square-integrability of the speed.

For the lower bound, for any distribution $\rho \in \+P_{>0}(\Omega)$ and any scalar functions $f, \phi: \Omega \to \bb R$, we have
\begin{equation*}
    \langle f, L_\rho \phi \rangle = \sum_{x, y \in \Omega} f(x)\tp{\phi(x) - \phi(y)} c_\rho(x,y)  = \frac{1}{2} \sum_{x, y \in \Omega} \tp{f(x) - f(y)}\tp{\phi(x) - \phi(y)} c_\rho(x,y).
\end{equation*}
By the Cauchy-Schwarz inequality, we have
\begin{equation} \label{eq:cs-graph}
    \abs{\langle f, L_\rho \phi \rangle} \le \|f\|_\rho \|\phi\|_\rho.
\end{equation}

Fix $t$ and consider a small time increment $h > 0$. By the definition of $W_2$, for this specific $h$, we can construct a path $\tp{\rho^{(h)}_s}_{s \in [0, 1]}$ connecting $\rho^{(h)}_0 = \pi_t$ and $\rho^{(h)}_1 = \pi_{t+h}$, driven by the admissible potential $\phi^{(h)}_s$, such that its action is near-optimal in the following sense:
\begin{equation} \label{eq:near-optimal-path}
    \int_0^1 \|\phi^{(h)}_s\|_{\rho^{(h)}_s}^2 \, \dd s \le W_2^2(\pi_t, \pi_{t+h}) + h^3.
\end{equation}

Let $f: \Omega \to \bb R$ be an arbitrary test function such that $\|f\|_{\pi_t}\neq 0$. We have
\[
    \langle f, \pi_{t+h} - \pi_t \rangle = \int_0^1 \langle f, \partial_s \rho^{(h)}_s \rangle \, \dd s = \int_0^1 \langle f, L_{\rho^{(h)}_s} \phi^{(h)}_s \rangle \, \dd s.
\]
Applying inequality \eqref{eq:cs-graph} to the integrand, followed by the Cauchy-Schwarz inequality over the time integral $s \in [0,1]$, we bound the inner product:
\begin{align}
    \abs{\langle f, \pi_{t+h} - \pi_t \rangle} &\le \int_0^1 \|f\|_{\rho^{(h)}_s} \|\phi^{(h)}_s\|_{\rho^{(h)}_s} \, \dd s \nonumber \\
    &\le \tp{ \sup_{s \in [0,1]} \|f\|_{\rho^{(h)}_s} } \int_0^1 \|\phi^{(h)}_s\|_{\rho^{(h)}_s} \, \dd s \nonumber \\
    &\le \tp{ \sup_{s \in [0,1]} \|f\|_{\rho^{(h)}_s} } \tp{ \int_0^1 \|\phi^{(h)}_s\|_{\rho^{(h)}_s}^2 \, \dd s }^{\frac{1}{2}}. \label{eq:bound-with-sup}
\end{align}

For any arbitrary $s \in [0,1]$, we can reparameterize the segment of the path from $r=0$ to $r=s$ to normal time $\tau \in [0,1]$ via $\tilde{\rho}_\tau \defeq \rho^{(h)}_{s\tau}$ with an admissible potential $\tilde{\phi}_\tau = s \phi^{(h)}_{s\tau}$. By the definition of $W_2$, 
\[
    W_2^2(\pi_t, \rho^{(h)}_s) \le \int_0^1 \|s \phi^{(h)}_{s\tau}\|_{\rho^{(h)}_{s\tau}}^2 \dd\tau = s \int_0^s \|\phi^{(h)}_r\|_{\rho^{(h)}_r}^2 \dd r \le \int_0^1 \|\phi^{(h)}_r\|_{\rho^{(h)}_r}^2 \dd r.
\]
Coupling this with \eqref{eq:near-optimal-path}, we establish a uniform bound over all $s \in [0,1]$:
\[
    \sup_{s \in [0,1]} W_2^2(\pi_t, \rho^{(h)}_s) \le W_2^2(\pi_t, \pi_{t+h}) + h^3.
\]
Because the base curve $\tp{\pi_t}$ is absolutely continuous, its distance $W_2(\pi_t, \pi_{t+h}) \to 0$ as $h \to 0$. Therefore,
\begin{equation} \label{eq:uniform-w2-limit}
    \lim_{h \to 0} \tp{ \sup_{s \in [0,1]} W_2(\pi_t, \rho^{(h)}_s) } = 0.
\end{equation}

As a consequence of Lemma~\ref{lem:w2-bounds-euclidean}, the $L_1$ distance (and consequently the Euclidean distance) is dominated by the $W_2$ distance. Therefore, \eqref{eq:uniform-w2-limit} guarantees that the sequence of paths $\tp{\rho^{(h)}_s}$ converges uniformly to $\pi_t$. That is,
\[
    \lim_{h \to 0} \tp{ \sup_{s \in [0,1]} \|\rho^{(h)}_s - \pi_t\|_2 } = 0.
\]

Because the capacity mapping $\rho \mapsto c_\rho(x,y)$ is continuous, the mapping $\rho \mapsto \|f\|_\rho$ is also continuous on $\mathcal{P}_{>0}(\Omega)$. Since $\pi_t$ lies in the strictly positive interior $\mathcal{P}_{>0}(\Omega)$ and $\rho^{(h)}_s$ converges to $\pi_t$ uniformly, for sufficiently small $h$, the entire path $\tp{\rho^{(h)}_s}_{s \in [0,1]}$ is contained within a compact neighborhood of $\pi_t$ inside $\mathcal{P}_{>0}(\Omega)$. From the Heine–Cantor theorem, on this compact set, the mapping $\rho \mapsto \|f\|_\rho$ is uniformly continuous, which yields
\[
    \lim_{h \to 0} \tp{ \sup_{s \in [0,1]} \abs{ \|f\|_{\rho^{(h)}_s} - \|f\|_{\pi_t} } } = 0.
\]
Using the property of suprema that $\abs{\sup_{s} g(s) - L} \le \sup_{s} \abs{g(s) - L}$, this uniform convergence of the difference directly yields the limit of the supremum:
\begin{equation} \label{eq:limit-f-norm}
    \lim_{h \to 0} \tp{ \sup_{s \in [0,1]} \|f\|_{\rho^{(h)}_s} } = \|f\|_{\pi_t}.
\end{equation}

With this limit rigorously established, we return to the change of distribution bound \eqref{eq:bound-with-sup}. Substituting the near-optimal action \eqref{eq:near-optimal-path} and dividing both sides by $h > 0$ yields:
\[
    \abs{\left\langle f, \frac{\pi_{t+h} - \pi_t}{h} \right\rangle} \le \tp{ \sup_{s \in [0,1]} \|f\|_{\rho^{(h)}_s} } \sqrt{\frac{W_2^2(\pi_t, \pi_{t+h})}{h^2} + h}.
\]
We now take the limit as $h \to 0$. On the right side, utilizing \eqref{eq:limit-f-norm} and the definition of the metric derivative $\abs{\dot{\pi}}_t = \lim_{h \to 0} \frac{W_2(\pi_t, \pi_{t+h})}{\abs{h}}$, we obtain:
\[
    \abs{\langle f, \partial_t \pi_t \rangle} \le \|f\|_{\pi_t} \abs{\dot{\pi}}_t.
\]
Assuming $\|f\|_{\pi_t} \neq 0$, squaring both sides and rearranging yields:
\begin{equation} \label{eq:lower-bound-f}
    \abs{\dot{\pi}}_t^2 \ge \frac{\langle f, \partial_t \pi_t \rangle^2}{\|f\|_{\pi_t}^2} \quad \text{for a.e. } t \in [0,T].
\end{equation}

If $\|\psi_t\|_{\pi_t}=0$, then $\partial_t\pi_t=0$ and the upper bound already gives $\abs{\dot\pi}_t=0$.  Otherwise, choose the test function $f$ to be the potential $\psi_t$ satisfying $L_{\pi_t} \psi_t = \partial_t \pi_t$. Substituting this gives:
\[
    \langle \psi_t, \partial_t \pi_t \rangle = \langle \psi_t, L_{\pi_t} \psi_t \rangle = \|\psi_t\|_{\pi_t}^2.
\]
Plugging this into \eqref{eq:lower-bound-f} provides the final lower bound:
\[
    \abs{\dot{\pi}}_t^2 \ge \frac{\tp{\|\psi_t\|_{\pi_t}^2}^2}{\|\psi_t\|_{\pi_t}^2} = \|\psi_t\|_{\pi_t}^2.
\]

Coupling this lower bound with the explicitly constructed upper bound $\abs{\dot{\pi}}_t^2 \le \|\psi_t\|_{\pi_t}^2$, the strict equality holds:
\[
    \abs{\dot{\pi}}_t^2 = \|\psi_t\|_{\pi_t}^2 = \frac{1}{2} \sum_{x, y \in \Omega} \tp{\psi_t(x) - \psi_t(y)}^2 c_{\pi_t}(x,y) \quad \text{for a.e. } t \in [0,T].
\]
This completes the proof.

\end{proof}

\subsection{Proof of \zcref{prop:metric-derivative-flux-formulation}}
\label{subsec:appendix-metric-derivative-flux-formulation}

\begin{proof}
    For any admissible potential \(\psi_t\), define the associated flux \(J_t\) as 
    \[
    J_t\tp{x, y} = \tp{\psi_t\tp{y} - \psi_t\tp{x}} c_{\pi_t}\tp{x, y}.
    \]
    Then \(J_t\) is anti-symmetric:
    \[
    J_t\tp{x, y} = -J_t\tp{y, x}, \quad \forall x, y \in \Omega,
    \]
    and satisfies the discrete continuity equation:
    \[
    \partial_t \pi_t\tp{x} +  \sum_{y \in \Omega} J_t\tp{x, y} = 0, \quad \forall x \in \Omega.
    \]
    Moreover, the objective values coincide:
    \[
    \sum_{\set{x,y}\in E} \frac{J_t\tp{x, y}^2}{c_{\pi_t}\tp{x, y}} = \frac{1}{2} \sum_{x, y \in \Omega} \tp{\psi_t\tp{y} - \psi_t\tp{x}}^2 c_{\pi_t}\tp{x, y}.
    \]
    Therefore, every admissible potential \(\psi_t\) induces an admissible flux \(J_t\) with the same cost. In particular, since an admissible potential exists, the set of admissible fluxes is nonempty.

    Now consider the optimization with respect to flux \(J_t\). Since the set of admissible fluxes is nonempty, the objective is a positive-definite quadratic function, and the constraints are linear, there exists a unique optimal solution \(J_t^*\), and there exist Lagrange multipliers \(\lambda_t^*\) for the divergence constraints. Define the Lagrangian as
    \[
    \+L\tp{J_t, \lambda_t} \defeq \sum_{\set{x,y}\in E} \frac{J_t\tp{x, y}^2}{c_{\pi_t}\tp{x, y}} + \sum_{x \in \Omega} \lambda_t\tp{x} \tp{\partial_t \pi_t\tp{x} +  \sum_{y \in \Omega} J_t\tp{x, y}}.
    \]
    Then the stationarity condition implies
    \[
    \frac{2J_t^*\tp{x, y}}{c_{\pi_t}\tp{x, y}} + \lambda_t^*\tp{x} - \lambda_t^*\tp{y} = 0, \quad \forall \set{x, y} \in E.
    \]
    Setting $\psi_t^*=\lambda_t^*/2$, we equivalently have
    \[
    J_t^*\tp{x, y} = \tp{\psi_t^*\tp{y} - \psi_t^*\tp{x}} c_{\pi_t}\tp{x, y}, \quad \forall \set{x, y} \in E.
    \]
    Since \(J_t^*\) satisfies the continuity equation, \(\psi_t^*\) satisfies the potential formulation constraints:
    \[
    \partial_t \pi_t\tp{x} + \sum_{y \in \Omega} \tp{\psi_t^*\tp{y} - \psi_t^*\tp{x}} c_{\pi_t}\tp{x, y} = 0, \quad \forall x \in \Omega.
    \]
    Hence \(\psi_t^*\) is an admissible potential, and by \zcref{thm:metric-derivative-potential-formulation},
    \[
    \begin{aligned}
        \min_{\text{admissible fluxes } J_t} \quad \sum_{\set{x,y}\in E} \frac{J_t\tp{x, y}^2}{c_{\pi_t}\tp{x, y}} &= \sum_{\set{x,y}\in E} \frac{J_t^*\tp{x, y}^2}{c_{\pi_t}\tp{x, y}} \\
        &= \frac{1}{2} \sum_{x, y \in \Omega} \tp{\psi_t^*\tp{y} - \psi_t^*\tp{x}}^2 c_{\pi_t}\tp{x, y} \\
        &= \abs{\dot{\pi}}_t^2. \\
    \end{aligned}
    \]
    This proves the identity. Finally, since admissible potentials differ only by additive constants, $J_t^*\tp{x,y}=\tp{\psi_t\tp{y}-\psi_t\tp{x}}c_{\pi_t}\tp{x,y}$ for every admissible potential $\psi_t$, and $J_t^*$ is the unique minimizer. This completes the proof.
\end{proof}

\subsection{Proof of \zcref{prop:WPI-implies-weak-action}}

\label{subsection:appendix-transport-inequalities}

\begin{proof}[Proof of \zcref{prop:WPI-implies-weak-action}]
    We prove the slightly more general version on an interval $[0,T]$.
    Write $p_t$ for the chosen rate kernel reversible with respect to
    $\pi_t$, and write
    $c_t(x,y)=c_{\pi_t}(x,y)=\pi_t(x)p_t(x,y)$ for $x\neq y$, with zero
    capacity on the diagonal.
    Fix a time $t$ at which $\partial_t\pi_t$ exists and the WPI
    \eqref{eq:wpi-for-l1-weak-action} holds, and abbreviate
    $\pi=\pi_t$, $p=p_t$, and
    \[
        b\defeq\partial_t\pi_t,
        \qquad
        g\defeq\partial_t\log\pi_t.
    \]
    Since $\pi\in\+P_{>0}\tp{\Omega}$, we have $b=\pi g$ pointwise.  Moreover,
    \[
        \E[\pi]{g}
        =\sum_{x\in\Omega}b\tp{x}
        =\partial_t\sum_{x\in\Omega}\pi_t\tp{x}
        =0.
    \]
    Hence, for every $f:\Omega\to\bb R$,
    \begin{align}
        \abs{\sum_{x\in\Omega}b\tp{x}f\tp{x}}
        &=\abs{\inner{g}{f-\E[\pi]{f}}_{\pi}}\notag\\
        &\leq\norm{g}_{L^2\tp{\pi}}\sqrt{\Var[\pi]{f}}\notag\\
        &\leq\norm{g}_{L^2\tp{\pi}}
        \sqrt{C_{\-{WPI}}\cdot\+E_p\tp{f,f}
        +\delta_{\-{WPI}}\operatorname{Osc}\tp{f}^2}\notag\\
        &\leq\sqrt{C_{\-{WPI}}}\norm{g}_{L^2\tp{\pi}}
        \sqrt{\+E_p\tp{f,f}}
        +\sqrt{\delta_{\-{WPI}}}\norm{g}_{L^2\tp{\pi}}
        \operatorname{Osc}\tp{f}.
        \label{eq:wpi-source-dual-bound}
    \end{align}
    Here the first inequality is Cauchy--Schwarz in $L^2\tp{\pi}$, the second is the WPI in~\eqref{eq:wpi-for-l1-weak-action}, and the last uses $\sqrt{s+t}\leq\sqrt{s}+\sqrt{t}$ for $s,t\geq0$.

    We next turn this bound on test functions into a decomposition of the signed vector $b$.  Let
    \[
        V\defeq\bb R^\Omega/\operatorname{span}\set{\mathbf 1}.
    \]
    Its dual is naturally identified with
    \[
        V^*=\set{a\in\bb R^\Omega:\sum_{x\in\Omega}a\tp{x}=0}.
    \]
    Indeed, the pairing $\langle a,f\rangle\defeq\sum_xa\tp{x}f\tp{x}$ is unchanged when a constant is added to $f$ precisely when $\sum_xa\tp{x}=0$.  We already showed that $b$ has zero sum, so $b\in V^*$.
    On $V$, define
    \[
        q\tp{f}\defeq\sqrt{\+E_p\tp{f,f}},
        \qquad
        \omega\tp{f}\defeq\operatorname{Osc}\tp{f}.
    \]
    Both quantities are unchanged when a constant is added to $f$ and are seminorms on $\bb R^\Omega$.  Moreover, $\omega\tp{f}=0$ exactly when $f$ is constant, and connectedness of the common underlying graph implies the same statement for $q\tp{f}$.  Hence $q$ and $\omega$ are norms on the quotient space $V$.  Their dual norms on $V^*$ are
    \[
        q^*\tp{a}\defeq\sup_{q\tp{f}\leq1}\abs{\langle a,f\rangle},
        \qquad
        \omega^*\tp{a}\defeq\sup_{\omega\tp{f}\leq1}\abs{\langle a,f\rangle}.
    \]
    Set
    \[
        \alpha\defeq\sqrt{C_{\-{WPI}}}\norm{g}_{L^2\tp{\pi}},
        \qquad
        \beta\defeq\sqrt{\delta_{\-{WPI}}}\norm{g}_{L^2\tp{\pi}}.
    \]
    Denote the unit balls of $q^*$ and $\omega^*$ by $B_{q^*}$ and $B_{\omega^*}$, respectively.  These balls are compact and convex because $V^*$ is finite-dimensional.  Therefore their scaled Minkowski sum
    \[
        K\defeq\alpha B_{q^*}+\beta B_{\omega^*}\subseteq V^*
    \]
    is also compact and convex.  For every $f\in V$, its support function is
    \begin{align}
        \sup_{a\in K}\langle a,f\rangle
        &=\sup_{\substack{v\in B_{q^*}\\w\in B_{\omega^*}}}
          \langle\alpha v+\beta w,f\rangle\notag\\
        &=\alpha\sup_{v\in B_{q^*}}\langle v,f\rangle
          +\beta\sup_{w\in B_{\omega^*}}\langle w,f\rangle\notag\\
        &=\alpha q\tp{f}+\beta\omega\tp{f}.
        \label{eq:wpi-minkowski-support-function}
    \end{align}
    The second equality holds because $v$ and $w$ can be optimized independently and $\alpha,\beta\geq0$.  To justify the last equality, the definition of the dual norm gives
    \[
        \langle v,f\rangle
        \leq\abs{\langle v,f\rangle}
        \leq q^*\tp{v}q\tp{f}
        \leq q\tp{f},
        \qquad v\in B_{q^*}.
    \]
    For the reverse inequality, suppose $f\neq0$ and put $h_0=f/q\tp{f}$, which lies on the boundary of the $q$-unit ball.  The finite-dimensional supporting-hyperplane theorem gives a nonzero $\widetilde v_f\in V^*$ such that
    \[
        \langle\widetilde v_f,h\rangle
        \leq\langle\widetilde v_f,h_0\rangle
        \qquad\text{whenever }q\tp{h}\leq1.
    \]
    Because the unit ball is symmetric, its maximum of $\langle\widetilde v_f,h\rangle$ equals its maximum of $\abs{\langle\widetilde v_f,h\rangle}$, namely $q^*\tp{\widetilde v_f}$.  Hence
    $\langle\widetilde v_f,h_0\rangle=q^*\tp{\widetilde v_f}>0$.  After normalizing
    $v_f=\widetilde v_f/q^*\tp{\widetilde v_f}$, we obtain
    \[
        q^*\tp{v_f}=1,
        \qquad
        \langle v_f,f\rangle
        =q\tp{f}\langle v_f,h_0\rangle
        =q\tp{f}.
    \]
    Thus the upper bound is attained.  For $f=0$, the equality is immediate, and therefore
    \[
        \sup_{v\in B_{q^*}}\langle v,f\rangle=q\tp{f}.
    \]
    Applying the same argument to $\omega$ proves
    $\sup_{w\in B_{\omega^*}}\langle w,f\rangle=\omega\tp{f}$, completing the justification of~\eqref{eq:wpi-minkowski-support-function}.

    In this notation,~\eqref{eq:wpi-source-dual-bound} reads
    \[
        \abs{\langle b,f\rangle}
        \leq\alpha q\tp{f}+\beta\omega\tp{f}
        =\sup_{a\in K}\langle a,f\rangle,
        \qquad f\in V.
    \]
    In particular, $\langle b,f\rangle$ is bounded above by the support function of $K$ for every $f\in V$.  If $b\notin K$, the finite-dimensional strict separation theorem would give an $f\in V$ such that
    \[
        \langle b,f\rangle
        >\sup_{a\in K}\langle a,f\rangle,
    \]
    contradicting the preceding bound.  Therefore $b\in K$.  By the definition of the Minkowski sum, there are $v\in B_{q^*}$ and $w\in B_{\omega^*}$ such that $b=\alpha v+\beta w$.  Setting $u_t\defeq\alpha v$ and $r_t\defeq\beta w$ gives
    \begin{equation}
        b=u_t+r_t,
        \qquad
        q^*\tp{u_t}\leq\sqrt{C_{\-{WPI}}}\norm{g}_{L^2\tp{\pi}},
        \qquad
        \omega^*\tp{r_t}\leq\sqrt{\delta_{\-{WPI}}}\norm{g}_{L^2\tp{\pi}}.
        \label{eq:wpi-dual-source-decomposition}
    \end{equation}

    We now identify the two dual norms.  First, for any $u\in V^*$, define its minimum routing cost by
    \[
        \mathcal C_{\pi}\tp{u}
        \defeq
        \min_{\text{flux }J:\,u+\divergence_{\-{d}}J=0}
        \sum_{\set{x,y}\in E}\frac{J\tp{x,y}^2}{c_{\pi}\tp{x,y}}.
    \]
    The constraint is feasible because $u$ has zero sum and the graph is connected.  Indeed, choose a spanning tree rooted at $o$.  For every nonroot vertex $x$, let $S_x$ be the subtree rooted at $x$ and set
    \[
        J\tp{x,\operatorname{parent}\tp{x}}
        \defeq-\sum_{z\in S_x}u\tp{z}.
    \]
    Extend $J$ antisymmetrically and set it to zero off the tree.  Then $\divergence_{\-{d}}J\tp{x}=-u\tp{x}$ at every nonroot vertex; the same identity holds at the root because $\sum_xu\tp{x}=0$.  Thus $u+\divergence_{\-{d}}J=0$.  Since all capacities are positive, the objective is a coercive, strictly convex quadratic function of the edge variables.  It therefore has a unique minimizer on the nonempty affine constraint set.

    If $J$ is any feasible flux, then $u=-\divergence_{\-{d}}J$.  Antisymmetry of $J$ gives the discrete summation-by-parts identity
    \[
        \begin{aligned}
        \langle u,f\rangle
        &=-\sum_{x\in\Omega}f\tp{x}
          \sum_{y\in\Omega}J\tp{x,y}\\
        &=\frac12\sum_{x,y\in\Omega}
          J\tp{x,y}\tp{f\tp{y}-f\tp{x}}.
        \end{aligned}
    \]
    Fixing one orientation $\tp{x,y}$ for each unordered edge and applying Cauchy--Schwarz to the corresponding edge sum yields
    \[
        \begin{aligned}
        \abs{\langle u,f\rangle}
        &=\abs{\sum_{\set{x,y}\in E}
          \frac{J\tp{x,y}}{\sqrt{c_{\pi}\tp{x,y}}}
          \sqrt{c_{\pi}\tp{x,y}}
          \tp{f\tp{y}-f\tp{x}}}\\
        &\leq
        \tp{\sum_{\set{x,y}\in E}
        \frac{J\tp{x,y}^2}{c_{\pi}\tp{x,y}}}^{1/2}
        \sqrt{\+E_p\tp{f,f}}.
        \end{aligned}
    \]
    Since $q\tp{f}=\sqrt{\+E_p\tp{f,f}}$, taking the supremum over $q\tp{f}\leq1$ gives
    \[
        q^*\tp{u}
        \leq
        \tp{\sum_{\set{x,y}\in E}
        \frac{J\tp{x,y}^2}{c_{\pi}\tp{x,y}}}^{1/2}.
    \]
    This holds for every feasible $J$.  Minimizing the right-hand side and then squaring shows
    \[
        q^*\tp{u}^2\leq\mathcal C_{\pi}\tp{u}.
    \]
    For the reverse inequality, the surjectivity just used to establish feasibility allows us to introduce a Lagrange multiplier $2\psi\in V$ for the constraint $u+\divergence_{\-{d}}J=0$.  The Lagrangian is
    \[
        \sum_{\set{x,y}\in E}\frac{J\tp{x,y}^2}{c_{\pi}\tp{x,y}}
        +2\langle\psi,u+\divergence_{\-{d}}J\rangle.
    \]
    Fix one orientation $\tp{x,y}$ for each edge.  Since
    \[
        \langle\psi,\divergence_{\-{d}}J\rangle
        =\sum_{\set{x,y}\in E}J\tp{x,y}
        \tp{\psi\tp{x}-\psi\tp{y}},
    \]
    stationarity with respect to the edge variable $J\tp{x,y}$ gives
    \[
        \frac{2J\tp{x,y}}{c_{\pi}\tp{x,y}}
        +2\tp{\psi\tp{x}-\psi\tp{y}}=0.
    \]
    Thus the unique minimizer has the form
    \[
        J^*\tp{x,y}
        =c_{\pi}\tp{x,y}\tp{\psi\tp{y}-\psi\tp{x}}.
    \]
    Since $u+\divergence_{\-{d}}J^*=0$, summation by parts and the formula for $J^*$ give
    \[
        \begin{aligned}
        \langle u,\psi\rangle
        &=-\langle\divergence_{\-{d}}J^*,\psi\rangle\\
        &=\sum_{\set{x,y}\in E}J^*\tp{x,y}
          \tp{\psi\tp{y}-\psi\tp{x}}\\
        &=\sum_{\set{x,y}\in E}
          c_{\pi}\tp{x,y}\tp{\psi\tp{y}-\psi\tp{x}}^2.
        \end{aligned}
    \]
    On the other hand,
    \[
        \mathcal C_{\pi}\tp{u}
        =\sum_{\set{x,y}\in E}
          \frac{J^*\tp{x,y}^2}{c_{\pi}\tp{x,y}}
        =\sum_{\set{x,y}\in E}
          c_{\pi}\tp{x,y}\tp{\psi\tp{y}-\psi\tp{x}}^2
        =q\tp{\psi}^2.
    \]
    If $\mathcal C_{\pi}\tp{u}>0$, then $q\tp{\psi}>0$, and the test function $\psi/q\tp{\psi}$ in the definition of $q^*$ gives
    \[
        q^*\tp{u}
        \geq\frac{\abs{\langle u,\psi\rangle}}{q\tp{\psi}}
        =\sqrt{\mathcal C_{\pi}\tp{u}}.
    \]
    If the cost is zero, the minimizing flux is $J^*=0$, and its constraint forces $u=0$, so the same inequality is immediate.  Combining this with the opposite inequality proved above yields
    \begin{equation}
        \mathcal C_{\pi}\tp{u}=q^*\tp{u}^2.
        \label{eq:dirichlet-dual-is-routing-cost}
    \end{equation}
    Choose the minimizing flux $J_t$ for $u_t$.  The first bound in~\eqref{eq:wpi-dual-source-decomposition} and~\eqref{eq:dirichlet-dual-is-routing-cost} give
    \[
        \sum_{\set{x,y}\in E}
        \frac{J_t\tp{x,y}^2}{c_{\pi}\tp{x,y}}
        \leq C_{\-{WPI}}\norm{g}_{L^2\tp{\pi}}^2,
    \]
    which is~\eqref{eq:wpi-good-flux-cost}.

    Second, for every $r\in V^*$,
    \[
        \omega^*\tp{r}
        =\sup_{\operatorname{Osc}\tp{f}\leq1}
          \abs{\sum_{x\in\Omega}r\tp{x}f\tp{x}}
        =\frac12\norm{r}_1.
    \]
    To prove this identity, first suppose that $\operatorname{Osc}\tp{f}\leq1$ and set $m=\tp{\max f+\min f}/2$.  Then $\abs{f\tp{x}-m}\leq1/2$ for every $x$.  Since $r$ has zero sum,
    \[
        \begin{aligned}
        \abs{\sum_xr\tp{x}f\tp{x}}
        &=\abs{\sum_xr\tp{x}\tp{f\tp{x}-m}}\\
        &\leq\sum_x\abs{r\tp{x}}\abs{f\tp{x}-m}\\
        &\leq\frac12\norm{r}_1.
        \end{aligned}
    \]
    This proves $\omega^*\tp{r}\leq\norm{r}_1/2$.  If $r\neq0$, let $P\defeq\set{x:r\tp{x}>0}$ and take $f=\mathbf 1_P$, which has oscillation one.  The zero-sum identity implies
    \[
        \sum_{x\in P}r\tp{x}
        =-\sum_{x\notin P}r\tp{x}
        =\frac12\sum_{x\in\Omega}\abs{r\tp{x}}
        =\frac12\norm{r}_1.
    \]
    Thus this $f$ attains the upper bound; when $r=0$, the claimed identity is immediate.  The second bound in~\eqref{eq:wpi-dual-source-decomposition} therefore gives
    \[
        \frac12\norm{r_t}_1
        \leq\sqrt{\delta_{\-{WPI}}}\norm{g}_{L^2\tp{\pi}},
    \]
    which is~\eqref{eq:wpi-residual-cost}.  Finally,
    \[
        \partial_t\pi_t+\divergence_{\-{d}}J_t
        =u_t+r_t+\divergence_{\-{d}}J_t
        =r_t,
    \]
    and $r_t\in V^*$ has zero sum.  Hence $\tp{J_t,r_t}$ is an $L^1$-relaxed flux at this time.

    It remains to choose these pointwise objects measurably in $t$.  Let $N\subseteq[0,T]$ be a null set outside which the derivative of the path exists and the WPI \eqref{eq:wpi-for-l1-weak-action} holds.  Choose a measurable representative $b_t$ of $\partial_t\pi_t$ on $[0,T]\setminus N$, and set $b_t=0$ on $N$.  Define
    \[
        g_t\tp{x}\defeq\frac{b_t\tp{x}}{\pi_t\tp{x}},
        \qquad x\in\Omega.
    \]
    Since every $\pi_t\tp{x}$ is positive, $g_t$ is measurable; outside $N$, the ordinary chain rule gives $g_t=\partial_t\log\pi_t$.  This definition also ensures that $b_t=\pi_t g_t$ for every $t$.

    We first record an explicit coordinate representation of the time-dependent dual norm.  Identify $V^*$ with the fixed Euclidean subspace
    \[
        H_0\defeq\set{f\in\bb R^\Omega:\sum_{x\in\Omega}f\tp{x}=0}.
    \]
    Every equivalence class in $V$ has a unique representative in $H_0$.  For $f\in H_0$, write
    \[
        q_t\tp{f}
        \defeq\tp{\sum_{\set{x,y}\in E}c_{\pi_t}\tp{x,y}
        \tp{f\tp{y}-f\tp{x}}^2}^{1/2},
        \qquad
        q_t^*\tp{u}
        \defeq\sup_{q_t\tp{f}\leq1}\abs{\langle u,f\rangle},
        \quad u\in H_0.
    \]
    These are precisely the time-dependent versions of the norms $q$ and $q^*$ used in the pointwise argument.
    For each $t$, define the weighted graph Laplacian $A_t:\bb R^\Omega\to\bb R^\Omega$ by
    \[
        A_t f\tp{x}
        \defeq-\sum_{y\in\Omega}c_{\pi_t}\tp{x,y}
        \tp{f\tp{y}-f\tp{x}}.
    \]
    With respect to the Euclidean inner product $\langle f,h\rangle=\sum_xf\tp{x}h\tp{x}$, symmetry of the capacities gives
    \[
        \langle f,A_t h\rangle
        =\sum_{\set{x,y}\in E}c_{\pi_t}\tp{x,y}
          \tp{f\tp{y}-f\tp{x}}
          \tp{h\tp{y}-h\tp{x}}
        =\langle A_t f,h\rangle.
    \]
    Thus $A_t$ is self-adjoint.  Since $A_t\mathbf 1=0$, self-adjointness gives $\langle A_t f,\mathbf 1\rangle=\langle f,A_t\mathbf 1\rangle=0$, so $A_t$ maps $H_0$ into itself.  If $f\in H_0$ and $A_t f=0$, then
    \[
        0=\langle f,A_t f\rangle
        =\sum_{\set{x,y}\in E}c_{\pi_t}\tp{x,y}
          \tp{f\tp{y}-f\tp{x}}^2.
    \]
    All edge capacities are positive and the graph is connected, so $f$ is constant; since $f\in H_0$, it follows that $f=0$.  Define the restricted linear operator
    \[
        A_t^{(0)}:H_0\longrightarrow H_0,
        \qquad
        A_t^{(0)}f\defeq A_t f.
    \]
    The preceding argument shows that $A_t^{(0)}$ is injective.  Its domain and codomain have the same finite dimension, so it is invertible.

    Fix once and for all an orthonormal basis $e_1,\ldots,e_{\abs{\Omega}-1}$ of $H_0$, and let
    \[
        \widehat A_t
        \defeq\tp{\langle e_i,A_t e_j\rangle}_{i,j=1}^{\abs{\Omega}-1}
    \]
    be the matrix representing $A_t^{(0)}$ in this basis.  For $u\in H_0$, write $\widehat u\defeq\tp{\langle e_i,u\rangle}_{i=1}^{\abs{\Omega}-1}$.  The routing-cost identity~\eqref{eq:dirichlet-dual-is-routing-cost} and the minimizing-potential calculation above give
    \begin{equation}
        q_t^*\tp{u}^2
        =\left\langle u,\tp{A_t^{(0)}}^{-1}u\right\rangle
        =\widehat u^{\mathsf T}\widehat A_t^{-1}\widehat u.
        \label{eq:time-dependent-dual-norm-coordinate-formula}
    \end{equation}
    Indeed, if $\psi=\tp{A_t^{(0)}}^{-1}u$, then the unique minimizing flux is $J\tp{x,y}=c_{\pi_t}\tp{x,y}\tp{\psi\tp{y}-\psi\tp{x}}$, and hence
    \[
        q_t^*\tp{u}^2
        =\sum_{\set{x,y}\in E}c_{\pi_t}\tp{x,y}
          \tp{\psi\tp{y}-\psi\tp{x}}^2
        =\langle\psi,A_t\psi\rangle
        =\langle u,\psi\rangle.
    \]

    We now make this decomposition measurable.  Set
    \[
        \alpha_t\defeq\sqrt{C_{\-{WPI}}}\norm{g_t}_{L^2\tp{\pi_t}},
        \qquad
        \beta_t\defeq\sqrt{\delta_{\-{WPI}}}\norm{g_t}_{L^2\tp{\pi_t}},
    \]
    and define
    \[
        F_t\defeq\set{\tp{u,r}\in H_0\times H_0:
        b_t=u+r,\ q_t^*\tp{u}\leq\alpha_t,\
        \omega^*\tp{r}\leq\beta_t}.
    \]
    For $t\notin N$, the pointwise argument above shows that $F_t$ is nonempty.  For $t\in N$, we have $b_t=g_t=0$, so $\tp{0,0}\in F_t$.  Each $F_t$ is convex.  It is also compact: it is closed, and its two norm constraints place $u$ and $r$ in bounded subsets of the finite-dimensional space $H_0$.

    We verify the required dependence on $t$.  The path $t\mapsto\pi_t$ and the functions $t\mapsto b_t,g_t,\alpha_t,\beta_t$ are measurable.  By the standing continuity assumption on the capacity family, every $t\mapsto c_{\pi_t}\tp{x,y}$ is measurable, and therefore every entry of $\widehat A_t$ is measurable.  Since $\widehat A_t$ is invertible, its inverse can be written as $\operatorname{adj}\tp{\widehat A_t}/\det\tp{\widehat A_t}$; hence every entry of $\widehat A_t^{-1}$ is measurable.  Formula~\eqref{eq:time-dependent-dual-norm-coordinate-formula} now shows that $\tp{t,u}\mapsto q_t^*\tp{u}$ is jointly measurable and, for each fixed $t$, continuous in $u$.  The norm $\omega^*$ is independent of $t$ and continuous.  Consequently,
    \[
        \set{\tp{t,u,r}\in[0,T]\times H_0\times H_0:
        \tp{u,r}\in F_t}
    \]
    is measurable.  Together with the nonempty compact values of $F_t$, this is the measurable-graph condition required by the finite-dimensional measurable minimum theorem used below.

    On each $F_t$, minimize the continuous function
    \[
        \tp{u,r}\longmapsto\norm{u}_2^2+\norm{r}_2^2.
    \]
    A minimizer exists because $F_t$ is compact, and it is unique because $F_t$ is convex and the objective is strictly convex.  The measurable minimum theorem therefore gives measurable functions $t\mapsto u_t$ and $t\mapsto r_t$ whose pair is this unique minimizer.

    Finally, define
    \[
        \widehat u_t\defeq\tp{\langle e_i,u_t\rangle}_{i=1}^{\abs{\Omega}-1},
        \qquad
        \widehat\psi_t\defeq\widehat A_t^{-1}\widehat u_t,
        \qquad
        \psi_t\defeq\sum_{i=1}^{\abs{\Omega}-1}
        \widehat\psi_t\tp{i}e_i.
    \]
    The preceding measurability statements show that $t\mapsto\psi_t$ is measurable, and the definitions give $\psi_t\in H_0$ and $A_t\psi_t=u_t$.  Define the flux
    \[
        J_t\tp{x,y}
        \defeq c_{\pi_t}\tp{x,y}
        \tp{\psi_t\tp{y}-\psi_t\tp{x}}.
    \]
    This flux is antisymmetric because the capacities are symmetric.  Moreover,
    \[
        \divergence_{\-{d}}J_t\tp{x}
        =\sum_yc_{\pi_t}\tp{x,y}
          \tp{\psi_t\tp{y}-\psi_t\tp{x}}
        =-A_t\psi_t\tp{x}=-u_t\tp{x},
    \]
    so the minimizing-potential calculation above shows that it is the unique minimum-energy flux for $u_t$.  Its defining formula and the measurability of $c_{\pi_t}$ and $\psi_t$ show that $t\mapsto J_t$ is measurable.

    We have therefore obtained a measurable family of $L^1$-relaxed fluxes satisfying~\eqref{eq:wpi-good-flux-cost} and~\eqref{eq:wpi-residual-cost} for almost every $t$.  By Cauchy--Schwarz in time,
    \[
        \int_0^T\norm{\partial_t\log\pi_t}_{L^2\tp{\pi_t}}\,\dd t
        \leq\sqrt{T}
        \tp{\int_0^T
        \norm{\partial_t\log\pi_t}_{L^2\tp{\pi_t}}^2\,\dd t}^{1/2}
        <\infty.
    \]
    Define
    \[
        \varepsilon_T
        \defeq\sqrt{\delta_{\-{WPI}}}\int_0^T
        \norm{\partial_t\log\pi_t}_{L^2\tp{\pi_t}}\,\dd t.
    \]
    Integrating the residual bound gives
    \[
        \frac12\int_0^T\norm{r_t}_1\,\dd t
        \leq
        \sqrt{\delta_{\-{WPI}}}\int_0^T
        \norm{\partial_t\log\pi_t}_{L^2\tp{\pi_t}}\,\dd t
        =\varepsilon_T,
    \]
    while integrating the flux-cost bound gives
    \[
        \int_0^T\sum_{\set{x,y}\in E}
        \frac{J_t\tp{x,y}^2}{c_{\pi_t}\tp{x,y}}\,\dd t
        \leq
        C_{\-{WPI}}\int_0^T
        \norm{\partial_t\log\pi_t}_{L^2\tp{\pi_t}}^2\,\dd t.
    \]
    Thus this measurable relaxed flux family is feasible for the residual
    budget $\varepsilon_T$, and its integrated transport cost is bounded by
    the right-hand side above.  Taking the infimum over all feasible relaxed
    flux families proves the corresponding result on $[0,T]$.  Setting
    $T=1$ yields \eqref{eq:wpi-integrated-weak-action}.
\end{proof}

\subsection{Proof of \zcref{thm:poincare-core}}
\label{subsec:appendix-wpi-to-pi-restriction}

\begin{proof}
    The choice of $h$ satisfies $C_{\-{WPI}}h\leq1/16$, so the trimming lemma
    gives a nonempty set $G$ satisfying \eqref{eq:trimmed-mass} and
    \eqref{eq:restriction-conductance-lower-bound}.  If $G$ is a singleton,
    \eqref{eq:poincare-on-core} is immediate.  We may therefore assume
    $\abs{G}\geq2$ and consider the transition matrix
    \[
        P^G\defeq \!I_G+\frac{1}{q_{\max}}p^G.
    \]
    This is a Markov kernel because its off-diagonal entries are nonnegative,
    its rows sum to one, and
    \[
        P^G\tp{x,x}
        =1-\frac{1}{q_{\max}}
          \sum_{y\in G\setminus\set{x}}p\tp{x,y}
        \geq0.
    \]
    Since $p^G$ is reversible with respect to $\pi^G$, so is $P^G$.
    Its conductance is at least $h/q_{\max}$.  Indeed, for every admissible
    $S\subseteq G$, the definitions of $\pi^G$ and $P^G$ give
    \[
        \sum_{x\in S,\,y\in G\setminus S}\pi^G\tp{x}P^G\tp{x,y}
        =\frac{1}{q_{\max}\pi\tp{G}}
          \sum_{x\in S,\,y\in G\setminus S}\pi\tp{x}p\tp{x,y}
        =\frac{c\tp{S,G\setminus S}}{q_{\max}\pi\tp{G}}.
    \]
    Consequently, the factors $\pi\tp{G}$ cancel, and
    \[
        \frac{\sum_{x\in S,\,y\in G\setminus S}\pi^G\tp{x}P^G\tp{x,y}}
             {\pi^G\tp{S}}
        =\frac{c\tp{S,G\setminus S}}{q_{\max}\pi\tp{S}}
        \geq\frac{h}{q_{\max}}.
    \]
    Cheeger's inequality for finite reversible chains
    \cite[Section~13.2]{LP17} therefore gives
    \[
        \operatorname{gap}\tp{P^G}
        \geq\frac{h^2}{2q_{\max}^2}.
    \]
    Since $p^G=q_{\max}\tp{P^G-\!I_G}$,
    \[
        \+E_{p^G}\tp{f,f}
        =q_{\max}\inner{f}{\tp{\!I_G-P^G}f}_{\pi^G}.
    \]
    Therefore, the variational characterization of the spectral gap gives
    \[
        \Var[\pi^G]{f}
        \leq\frac{1}{\operatorname{gap}\tp{P^G}}
        \inner{f}{\tp{\!I_G-P^G}f}_{\pi^G}
        \leq\frac{2q_{\max}}{h^2}\+E_{p^G}\tp{f,f},
    \]
    which is the first inequality in \eqref{eq:poincare-on-core}.  Furthermore,
    \[
        \frac{2q_{\max}}{h^2}
        =\max\set{\frac{8}{q_{\max}},
                  512q_{\max}C_{\-{WPI}}^2}
        \leq\frac{8}{q_{\max}}+512q_{\max}C_{\-{WPI}}^2,
    \]
    proving the second inequality.

    Finally, viewing $\pi^G$ as a distribution on $\Omega$ supported on $G$,
    we have
    \[
        \begin{aligned}
        D_{\mathrm{TV}}\tp{\pi,\pi^G}
        &=\frac12\sum_{x\in G}
          \pi\tp{x}\tp{\frac{1}{\pi\tp{G}}-1}
          +\frac12\sum_{x\in\bar G}\pi\tp{x}\\
        &=\frac12\pi\tp{\bar G}+\frac12\pi\tp{\bar G}
         =\pi\tp{\bar G}.
        \end{aligned}
    \]
    Combining this identity with \eqref{eq:trimmed-mass} proves
    \eqref{eq:tv-close-poincare-core}.
\end{proof}

\section{Omitted Proofs for the Mean-Field Applications}

\subsection{Projection, symmetry, and folded capacities}
\label{subsec:appendix-mean-field-projection}

\begin{proof}[Proof of \zcref{lem:mean-field-projected-geometry}]
We prove the symmetry statement in a form that applies to both projections used in \zcref{sec:mean-field-phase-transitions}.  Let a finite group act on $\Omega$ so that $\Pi(gx)=\Pi(x)$ for every $g$ and $x$, and suppose its action on each fiber of a map $\Pi:\Omega\to\overline\Omega$ is transitive.  Suppose that $\pi_s$ and the rate kernel $p_s$ are invariant under this action.  Write $F_a=\Pi^{-1}(a)$ and define
$$
    \overline\pi_s(a)=\sum_{x\in F_a}\pi_s(x),
    \qquad
    \overline c_s(a,b)
    =\sum_{x\in F_a,\,y\in F_b}c_s(x,y).
$$
Invariance gives, for $x\in F_a$,
\begin{equation}
    \pi_s(x)=\frac{\overline\pi_s(a)}{|F_a|},
    \qquad
    \sum_{y\in F_b}c_s(x,y)
    =\frac{\overline c_s(a,b)}{|F_a|}.
    \label{eq:appendix-orbit-lumpability}
\end{equation}

Given a relaxed flux $(J_s,r_s)$ on $\Omega$, define
$$
    \overline J_s(a,b)
    =\sum_{x\in F_a,\,y\in F_b}J_s(x,y),
    \qquad
    \overline r_s(a)=\sum_{x\in F_a}r_s(x).
$$
Summing the continuity equation over a fiber shows that $(\overline J_s,\overline r_s)$ is relaxed for $\overline\pi_s$.  Moreover,
$$
    \frac12\norm{\overline r_s}_1
    \leq\frac12\norm{r_s}_1,
$$
and Cauchy--Schwarz on every pair of fibers gives
$$
    \frac{\overline J_s(a,b)^2}{\overline c_s(a,b)}
    \leq
    \sum_{x\in F_a,\,y\in F_b}
    \frac{J_s(x,y)^2}{c_s(x,y)}.
$$
Thus projection does not increase weak action.

Conversely, extend a projected relaxed flux back to the original space by
$$
    r_s(x)=\frac{\overline r_s(a)}{|F_a|},
    \qquad x\in F_a,
$$
and, for $a\neq b$,
$$
    J_s(x,y)
    =\frac{c_s(x,y)}{\overline c_s(a,b)}
      \overline J_s(a,b),
    \qquad x\in F_a,\ y\in F_b.
$$
Set the flux on edges within one fiber to zero.  Equation \eqref{eq:appendix-orbit-lumpability} shows that the resulting continuity equation is the projected equation divided by $|F_a|$.  The residual norms are equal, and another direct calculation gives
$$
    \sum_{x\in F_a,\,y\in F_b}
    \frac{J_s(x,y)^2}{c_s(x,y)}
    =\frac{\overline J_s(a,b)^2}{\overline c_s(a,b)}.
$$
Taking infima in both directions proves \eqref{eq:folding-preserves-weak-action}.  Site permutations generate the fibers of the count map, while color permutations generate the fibers of the sorting map; the Gibbs law and block heat-bath kernel are invariant under both actions.
\end{proof}

\begin{proof}[Proof of \zcref{lem:mean-field-folded-estimates}]
We first verify the capacity estimate.  Let $\overline\mu_\beta$ be the law of the unsorted count vector.  Fix adjacent count vectors $\*m,\*m'$ and a block $I$ that realizes this transition.  If $\*m''$ is the count vector obtained from any competing assignment on $I$, then $\*m'$ and $\*m''$ are obtained by adding $q-1$ colors to the same outside count vector.  Expanding the two multinomial coefficients therefore gives
$$
    \frac{\binom n{\*m'}}{\binom n{\*m''}}
    \leq n^{q-1}.
$$
Moreover, at most $(q-1)!$ assignments on $I$ give the same count increment.  Thus the conditional probability of the target block assignment is at least
$$
    \frac{\overline\mu_\beta(\*m')/\binom n{\*m'}}
    {\sum_{\text{assignments on }I}
      \overline\mu_\beta(\*m'')/\binom n{\*m''}}
    \geq
    \frac{\overline\mu_\beta(\*m')}
    {n^{q-1}(q-1)!\sum_{\*m''}\overline\mu_\beta(\*m'')}
    \geq
    \frac{\overline\mu_\beta(\*m')}{n^{q-1}(q-1)!}.
$$
For every microscopic configuration with count vector $\*m$, adjacency
supplies at least one realizing block and assignment.  Using also
$$
    \binom n{q-1}\leq\frac{n^{q-1}}{(q-1)!}
$$
in the block-selection probability, and then summing over the count-vector
$\*m$ fiber, yields
\begin{equation*}
    \overline c_\beta(\*m,\*m')
    \geq n^{-2q+2}
       \overline\mu_\beta(\*m)
       \overline\mu_\beta(\*m').
\end{equation*}
Folding only adds capacities, and every color orbit has size at most $q!$.  Hence
$$
    \widehat c_\beta(\*m,\*m')
    \geq\frac{
      \widehat\mu_\beta(\*m)
      \widehat\mu_\beta(\*m')}
      {(q!)^2n^{2q-2}},
$$
which implies \eqref{eq:folded-potts-capacity}.

Finally, before folding, the unnormalized count-vector weight is
$$
    W_\beta(\*m)
    =\frac{n!}{\prod_am_a!}
      \exp\left(\frac\beta n\sum_am_a^2\right).
$$
The orbit multiplicity introduced by folding is independent of $\beta$.  Differentiating the normalized weight therefore gives \eqref{eq:folded-potts-score}.  Since $0\leq\sum_am_a^2\leq n^2$, \eqref{eq:folded-potts-score-bound} follows.
\end{proof}

\subsection{The folded mean-field Ising chain}
\label{subsec:appendix-mean-field-ising}

\begin{proof}[Proof of \zcref{prop:folded-ising-canonical-congestion}]
Let $\overline\mu^{\mathrm{Ising}}_\beta$ be the magnetization law on
$$
    \set{-n,-n+2,\ldots,n-2,n}.
$$
Write $\widehat{\mathcal M}_n$ for the corresponding nonnegative state
space obtained by folding $m$ and $-m$ together.
For $m\geq0$, the adjacent log-ratio is
\begin{equation*}
    f_\beta(m)
    =\log\frac{\overline\mu^{\mathrm{Ising}}_\beta(m+2)}
                    {\overline\mu^{\mathrm{Ising}}_\beta(m)}
    =\frac{2\beta(m+1)}n
      +\log\frac{n-m}{n+m+2}.
\end{equation*}
As a function of real $m\in[0,n)$,
$$
    f_\beta''(m)
    =-\frac1{(n-m)^2}+\frac1{(n+m+2)^2}<0.
$$
If $\beta\leq1-1/n$, then $f_\beta(0)\leq0$ and $f_\beta'(0)<0$, so the sequence is decreasing.  If $1-1/n<\beta<1$ and $f_\beta(0)\leq0$, the expansion
$$
    \frac{\log(1+2/n)}{2/n}
    <1-\frac1n+\frac4{3n^2}
$$
implies $f_\beta'(0)\leq0$ for $n\geq4$; the cases $n\leq3$ are immediate.  If instead $f_\beta(0)>0$, strict concavity and $f_\beta(m)\to-\infty$ give a single sign change.  The same conclusion holds for $\beta\geq1$, since then $f_\beta(0)>0$.  Thus the nonnegative magnetization sequence is unimodal in every case.  Multiplying positive magnetizations by the spin-flip orbit size two preserves unimodality after folding by $m\mapsto|m|$.

For completeness, summing microscopic heat-bath capacities between the two magnetization fibers gives
$$
    \overline c^{\mathrm{Ising}}_\beta(m,m+2)
    =\frac{n-m}{2n}
      \frac{\overline\mu^{\mathrm{Ising}}_\beta(m)
            \overline\mu^{\mathrm{Ising}}_\beta(m+2)}
      {\frac{n-m}{n+m+2}\overline\mu^{\mathrm{Ising}}_\beta(m)
       +\overline\mu^{\mathrm{Ising}}_\beta(m+2)}.
$$
Consequently, after folding,
\begin{equation}
    \widehat c^{\mathrm{Ising}}_\beta(m,m+2)
    \geq\frac1{2n}
       \tp{\widehat\mu^{\mathrm{Ising}}_\beta(m)
       \wedge\widehat\mu^{\mathrm{Ising}}_\beta(m+2)}.
    \label{eq:appendix-folded-ising-capacity}
\end{equation}

For each unordered pair of distinct folded magnetizations $\set{x,y}$, let
$\gamma_{xy}$ be the unique path joining them.  Fix an edge
$e=\set{u,u+2}$ and let
$$
    S_e^-=\set{m:m\leq u},
    \qquad
    S_e^+=\set{m:m\geq u+2}.
$$
The canonical path $\gamma_{xy}$ traverses $e$ exactly when one endpoint is
in each of these two sets, and every such path has length at most $n/2$.
Consequently,
\begin{align*}
    &\sum_{\substack{\set{x,y}\subseteq\widehat{\mathcal M}_n:\\
          x\neq y,\ e\in\gamma_{xy}}}
       \widehat\mu^{\mathrm{Ising}}_\beta(x)
       \widehat\mu^{\mathrm{Ising}}_\beta(y)\abs{\gamma_{xy}}\\
    &\hspace{3cm}\leq
       \frac n2\widehat\mu^{\mathrm{Ising}}_\beta(S_e^-)
       \widehat\mu^{\mathrm{Ising}}_\beta(S_e^+).
\end{align*}
Whichever side of $e$ does not contain a mode lies on a monotone tail of the
folded law.  It has at most $n$ states, and hence unimodality gives
$$
    \widehat\mu^{\mathrm{Ising}}_\beta(S_e^-)
    \widehat\mu^{\mathrm{Ising}}_\beta(S_e^+)
    \leq n\tp{\widehat\mu^{\mathrm{Ising}}_\beta(u)
       \wedge\widehat\mu^{\mathrm{Ising}}_\beta(u+2)}.
$$
Dividing by \eqref{eq:appendix-folded-ising-capacity} proves
$$
    \max_e\Gamma^{\mathrm{Ising}}_\beta(e)
    \leq n^3,
$$
which is \eqref{eq:folded-ising-canonical-congestion}.

For an absolutely continuous schedule $s\mapsto\beta_s$, differentiating
the folded law gives
$$
    \partial_s\log\widehat\mu^{\mathrm{Ising}}_{\beta_s}(m)
    =\frac{\beta_s'}{2n}
      \tp{m^2-\mathbb E_{\widehat\mu^{\mathrm{Ising}}_{\beta_s}}m^2}.
$$
Because $m^2\in[0,n^2]$, its variance is at most $n^4/4$.  Therefore
\begin{equation*}
    \norm{\partial_s\log\widehat\mu^{\mathrm{Ising}}_{\beta_s}}
      _{L^2\tp{\widehat\mu^{\mathrm{Ising}}_{\beta_s}}}^2
    \leq\frac{n^2}{16}\abs{\beta_s'}^2.
\end{equation*}
This is \eqref{eq:folded-ising-score-bound}.
\end{proof}

\begin{proof}[Proof of \zcref{prop:mean-field-ising-exact-action}]
Apply \zcref{prop:canonical-path-action} on the folded space, using
\eqref{eq:folded-ising-canonical-congestion} and
\eqref{eq:folded-ising-score-bound}.  The exact-action part gives
$$
    \+A\tp{\tp{\widehat\mu^{\mathrm{Ising}}_{\beta_s}}_{s\in[0,1]}}
    \leq\frac{n^5}{16}\int_0^1\abs{\beta_s'}^2\,\dd s.
$$
The symmetry identity \eqref{eq:folding-preserves-weak-action}, with
zero residual budget, transfers this bound to the original Ising space and proves
\eqref{eq:mean-field-ising-action}.
\end{proof}

\begin{proof}[Proof of \zcref{thm:mean-field-ising-annealing}]
For the linear schedule $\beta_s=s\beta$, initialize exactly from the uniform
law $\mu^{\mathrm{Ising}}_0$.  Thus
$D_{\mathrm{KL}}(\mu^{\mathrm{Ising}}_0\|\pi_0^{\-{\!{ALG}}})=0$, while
\eqref{eq:mean-field-ising-action} gives action at most $n^5\beta^2/16$.
The exact-action conclusion of \zcref{thm:annealing-ub} therefore yields
$$
    D_{\mathrm{KL}}\tp{\mu^{\mathrm{Ising}}_\beta\|\nu_T}
    \leq\frac{n^5\beta^2}{64T}.
$$
Taking
$T=\max\{1,n^5\beta^2/(64\zeta)\}$ proves the claim, and
\zcref{thm:annealing-ub} also identifies $T$ as the expected number of
single-site updates.
\end{proof}

\subsection{The Potts coexistence profiles}
\label{subsec:appendix-potts-coexistence-profiles}

We derive the values of the ordered profile and the coexistence inverse temperature used in the main text.  A count vector $\*m=nx$ contains $n!/\prod_am_a!$ configurations, all with the same interaction energy.  Its unnormalized Gibbs weight is therefore
$$
    W_\beta(\*m)
    =\frac{n!}{\prod_am_a!}
      \exp\tp{\frac\beta n\sum_{a=1}^qm_a^2}.
$$
Using $\log(k!)=k\log k-k+O(\log(k+1))$ for $k\geq0$, with
$0\log0=0$, we have
$$
\begin{aligned}
    \log\frac{n!}{\prod_am_a!}
    &=n\log n-\sum_{a=1}^qm_a\log m_a+O_q(\log n)\\
    &=-n\sum_{a=1}^qx_a\log x_a+O_q(\log n),\\
    \frac\beta n\sum_{a=1}^qm_a^2
    &=\beta n\sum_{a=1}^qx_a^2.
\end{aligned}
$$
Combining the multiplicity and Gibbs-weight terms gives
$$
    \log W_\beta(\*m)
    =n\Phi_\beta(x)+O_q(\log n).
$$
Thus $\Phi_\beta(x)$ gives the leading term $n\Phi_\beta(x)$ in the log-weight of a type.  A phase is represented, at this exponential scale, by a region around a global maximizer of $\Phi_\beta$.  Two phases \emph{coexist} when their distinct profiles attain the same global maximum value.  It is the full function $\Phi_\beta$, including both its entropy and interaction terms, that must agree; the interaction energies alone need not be equal.

The function $\Phi_\beta$ is unchanged when its coordinates are permuted.  The profile preserving all color symmetries is
$$
    x^{\mathrm{dis}}=\tp{1/q,\ldots,1/q}.
$$
An ordered profile instead selects one dominant color.  After calling it color $1$, the symmetry among the other $q-1$ colors leads to the one-parameter family
$$
    x(a)=\tp{a,\frac{1-a}{q-1},\ldots,
      \frac{1-a}{q-1}},
    \qquad a\in[1/q,1].
$$
The standard phase classification \cite[Section~2.2]{CDL+12} shows that the ordered global maximizers have this form.  The calculation below identifies the member of this family that coexists with $x^{\mathrm{dis}}$.

Restrict the free energy to this family by setting $F_\beta(a)=\Phi_\beta(x(a))$.  Explicitly,
$$
    F_\beta(a)
    =-a\log a-(1-a)\log\frac{1-a}{q-1}
      +\beta\tp{a^2+\frac{(1-a)^2}{q-1}}.
$$
A stationary point is a value of $a$ at which an infinitesimal change of the dominant-color proportion does not change the free energy to first order.  Thus it satisfies
\begin{equation}
    F_\beta'(a)
    =\log\frac{1-a}{(q-1)a}
      +\frac{2\beta(qa-1)}{q-1}=0.
    \label{eq:appendix-potts-stationary-profile}
\end{equation}
At coexistence, the ordered stationary point must also have the same free energy as the disordered point:
\begin{equation}
    F_\beta(a)=F_\beta(1/q).
    \label{eq:appendix-potts-equal-free-energy}
\end{equation}

To solve these two equations, let
$$
    r=\frac{(q-1)a}{1-a}>1,
    \qquad
    a=\frac{r}{r+q-1}.
$$
The ratio $r$ compares the dominant-color proportion with each of the other color proportions.  Equation \eqref{eq:appendix-potts-stationary-profile} becomes
\begin{equation}
    2\beta=\frac{(r+q-1)\log r}{r-1}.
    \label{eq:appendix-potts-stationary-ratio}
\end{equation}
Substituting this expression for $\beta$ into \eqref{eq:appendix-potts-equal-free-energy} gives
\begin{equation}
    \log\frac{r+q-1}{q}
    =\frac{(q+1)r+q-1}{2q(r+q-1)}\log r.
    \label{eq:appendix-potts-coexistence-ratio}
\end{equation}
The ordered solution singled out by the phase classification is
$$
    r_{\mathsf c}=(q-1)^2.
$$
Indeed, for this value $r_{\mathsf c}+q-1=q(q-1)$, and both sides of \eqref{eq:appendix-potts-coexistence-ratio} equal $\log(q-1)$.  Recovering $a$ and then $\beta$ from \eqref{eq:appendix-potts-stationary-ratio} gives
$$
    a_{\mathsf c}=\frac{q-1}{q},
    \qquad
    \beta_{\mathsf c}
    =\frac{q-1}{q-2}\log(q-1).
$$
Equivalently, direct substitution gives
$$
    F_\beta(a_{\mathsf c})-F_\beta(1/q)
    =-\frac{q-2}{q}\log(q-1)
      +\beta\frac{(q-2)^2}{q(q-1)},
$$
which vanishes exactly at $\beta=\beta_{\mathsf c}$.  Thus the ordered coexistence profile is
$$
    x^{\mathrm{ord}}=x(a_{\mathsf c})
    =\tp{\frac{q-1}{q},\frac1{q(q-1)},\ldots,
      \frac1{q(q-1)}}.
$$
The cited phase classification supplies the global part of the conclusion: below $\beta_{\mathsf c}$, $x^{\mathrm{dis}}$ is the unique global maximizer; at $\beta_{\mathsf c}$, $x^{\mathrm{dis}}$ and the $q$ color permutations of $x^{\mathrm{ord}}$ are precisely the global maximizers; and above $\beta_{\mathsf c}$, the ordered profiles are the global maximizers.

We first check that, up to a permutation of colors, every local maximizer of
$\Phi_\beta$ is either $x^{\mathrm{dis}}$ or a point $x(a)$ with $a>1/q$.  A
boundary point is not a local maximizer: if $x_i=0$, moving mass $t$ from a
positive coordinate to coordinate $i$ changes $\Phi_\beta$ by
$-t\log t+O(t)$, which is positive for small $t>0$.  At an interior
stationary point, the Lagrange condition gives
$$
    2\beta x_i-\log x_i=\lambda,
    \qquad i=1,\ldots,q,
$$
for some $\lambda\in\bb R$.  The function $x\mapsto2\beta x-\log x$ is
strictly convex, so the coordinates take at most two values; when $\beta=0$
it is strictly decreasing, and all coordinates are equal.  If they take
two values $a>b$, then $a>1/(2\beta)>b$, because the two solutions lie on
opposite sides of the minimizer $1/(2\beta)$.  If at least two coordinates
equal $a$, say $x_i=x_j=a$, then the Hessian
$\operatorname{diag}(-1/x_1,\ldots,-1/x_q)+2\beta I$ gives the second
derivative $2(2\beta-1/a)>0$ in the direction $e_i-e_j$, so $x$ is not a
local maximizer.  Hence a local maximizer has all coordinates equal or, after
permuting colors, has the form $(a,b,\ldots,b)$ with $a>b$, which is
$x(a)$ with $a>1/q$.

We next justify the stationary-branch picture used in the main text.  Write
$$
    \beta(r)=\frac{(r+q-1)\log r}{2(r-1)},
    \qquad r>1,
$$
for the right-hand side of
\eqref{eq:appendix-potts-stationary-ratio}, and set
$$
    H_q(r)=(r-1)(r+q-1)-qr\log r.
$$
A direct differentiation gives
$$
    \beta'(r)=\frac{H_q(r)}{2r(r-1)^2},
    \qquad
    H_q'(r)=2(r-1)-q\log r,
    \qquad
    H_q''(r)=2-\frac qr.
$$
The derivative $H_q'$ first decreases and then increases, starts from zero
at $r=1$, and tends to infinity.  It consequently has one zero in
$(1,\infty)$.  Since $H_q(1)=0$, $H_q$ initially decreases, and
$H_q(r)\to\infty$, the function $H_q$ also has exactly one zero
$r_{\mathrm{sp}}>1$.  Hence $\beta(r)$ decreases up to
$r_{\mathrm{sp}}$ and increases thereafter.  The lower ordered spinodal is
therefore
$$
    \beta_{\mathrm{sp}}^{\mathrm{ord}}
    =\beta(r_{\mathrm{sp}}).
$$
Since
$\beta(q-1)=\beta((q-1)^2)=\beta_{\mathsf c}$, the unique minimum lies
strictly between these two values of $r$, and thus
$\beta_{\mathrm{sp}}^{\mathrm{ord}}<\beta_{\mathsf c}$.

This turning point also separates the saddle and locally stable ordered
branches.  Differentiating once more gives
$$
    F_\beta''(a)
    =-\frac1a-\frac1{1-a}+\frac{2q\beta}{q-1}.
$$
At the stationary profile
$a(r)=r/(r+q-1)$, substitution of $\beta=\beta(r)$ yields
$$
    F_{\beta(r)}''(a(r))
    =-\frac{r+q-1}{(q-1)r(r-1)}H_q(r).
$$
In a tangent direction that only redistributes mass between two minority
colors, the Hessian has eigenvalue
$$
    (r+q-1)\left(\frac{\log r}{r-1}-1\right)<0.
$$
These minority-redistribution directions, together with the direction
$x'(a)$, span the tangent space of the simplex.  It follows that the branch
with $1<r<r_{\mathrm{sp}}$ is a saddle, whereas the branch with
$r>r_{\mathrm{sp}}$ is an ordered local maximum.  Finally,
$\lim_{r\downarrow1}\beta(r)=q/2$; this is where the saddle branch meets the
disordered branch.

We now specialize this classification to coexistence.  At
$\beta=\beta_{\mathsf c}$, both $a=1/q$ and $a=(q-1)/q$ have
$$
    F_{\beta_{\mathsf c}}''(a)
    =\frac{q}{q-1}\tp{2\beta_{\mathsf c}-q}<0,
$$
so a small change of $a$ lowers the free energy: these two points are local maxima along the family $x(a)$.  Here $\beta_{\mathsf c}<q/2$ follows, on putting $p=q-1>1$, from the elementary inequality $\log p<(p-p^{-1})/2$.

The intermediate value $a=1/2$ also satisfies \eqref{eq:appendix-potts-stationary-profile} at $\beta_{\mathsf c}$.  At this point,
$$
    F_{\beta_{\mathsf c}}''(1/2)
    =-4+\frac{2q}{q-2}\log(q-1)>0,
$$
where the inequality follows from $\log p>2(p-1)/(p+1)$ for $p>1$.  Hence increasing or decreasing $a$ slightly raises the free energy.  To examine other directions, recall that the Hessian matrix of $\Phi_\beta$ is
$$
    \nabla^2\Phi_\beta(x)
    =\operatorname{diag}(-1/x_1,\ldots,-1/x_q)+2\beta I.
$$
At $x^{\mathrm{sad}}$, a direction that increases one minority color and decreases another has curvature with sign
$$
    -2(q-1)+2\beta_{\mathsf c}<0.
$$
The last inequality follows from $\beta_{\mathsf c}<q-1$, which in turn follows from $\log p<p-1$ for $p=q-1>1$.  Thus the profile
$$
    x^{\mathrm{sad}}=x(1/2)
    =\tp{\frac12,\frac1{2(q-1)},\ldots,
      \frac1{2(q-1)}}
$$
has an increasing direction and a decreasing direction.  This is precisely what is meant here by a saddle point.

Finally,
$$
    \Phi_{\beta_{\mathsf c}}(x^{\mathrm{dis}})
    =\log q+\frac{\beta_{\mathsf c}}q,
    \qquad
    \Phi_{\beta_{\mathsf c}}(x^{\mathrm{sad}})
    =\log\tp{2\sqrt{q-1}}
      +\frac{q\beta_{\mathsf c}}{4(q-1)}.
$$
Their difference is
$$
    \Phi_{\beta_{\mathsf c}}(x^{\mathrm{dis}})
      -\Phi_{\beta_{\mathsf c}}(x^{\mathrm{sad}})
    =\log\frac{q}{2\sqrt{q-1}}
      -\frac{q-2}{4q}\log(q-1)
    =b_q.
$$
Because $x^{\mathrm{dis}}$ and $x^{\mathrm{ord}}$ have equal free energy,
this is also $\Phi_{\beta_{\mathsf c}}(x^{\mathrm{ord}})-
\Phi_{\beta_{\mathsf c}}(x^{\mathrm{sad}})$.  The phase classification
shows that the difference is strictly positive.  Since
$\log W_\beta(\*m)=n\Phi_\beta(\*m/n)+O_q(\log n)$, a count vector at the
saddle profile is therefore lighter than one at either coexisting maximizer by
the factor $e^{-nb_q}$, up to polynomial corrections.

\subsection{Free-energy geometry and Potts phase cores}
\label{subsec:appendix-potts-phase-core}

We prove \zcref{lem:potts-phase-core} in three steps.  The coexistence calculation in \zcref{subsec:appendix-potts-coexistence-profiles} describes the relevant maximizers.  The argument below derives the two finite-$n$ properties needed for transport.

\paragraph{Uniform free-energy gaps.}
The cited phase classification says that $x^{\mathrm{dis}}$ is the unique
global maximizer of $\Phi_\theta$ for $\theta<\beta_{\mathsf c}$, that the
$q$ permutations of a one-large-coordinate point
$x^{\mathrm{ord}}(\theta)$ are the only global maximizers for
$\theta>\beta_{\mathsf c}$, and that these $q+1$ points coexist at
$\beta_{\mathsf c}$.  At coexistence
$x^{\mathrm{ord}}(\beta_{\mathsf c})=x^{\mathrm{ord}}$ as defined in the
main text.  These maxima are interior and nondegenerate throughout the
compact interval $[0,q/2]$ on the side where they are dominant.  To see this
directly, restrict
$$
    \nabla^2\Phi_\theta(x)
    =\operatorname{diag}(-1/x_1,\ldots,-1/x_q)+2\theta I
$$
to the tangent space of the simplex.  At $x^{\mathrm{dis}}$ it equals
$(-q+2\theta)I$ and is negative for $\theta\in[0,\beta_{\mathsf c}]$ because
$\beta_{\mathsf c}<q/2$.  At a one-large-coordinate stationary point
$(a,b,\ldots,b)$, the two eigenspaces are the differences among the small coordinates and the symmetric line.  Their signs are respectively those of
$-1/b+2\theta$ and
$$
    -\frac1a-\frac1{1-a}
      +\frac{2q\theta}{q-1}.
$$
To verify that both are strictly negative on the stable ordered branch, put $r=a/b$.  The stationary equation and $a+(q-1)b=1$ give
$$
    b=\frac1{r+q-1},
    \qquad
    2\theta=\frac{(r+q-1)\log r}{r-1}.
$$
At coexistence, the explicit ordered maximizer $x^{\mathrm{ord}}$ has
$r=r_0=(q-1)^2$.  The stable ordered branch singled out by the cited phase
classification is its continuation toward $r\geq r_0$; the calculation below
shows that $\theta(r)$ is strictly increasing there, so
$r\in[r_0,\infty)$ gives a one-to-one parametrization of this branch.
The first eigenvalue is
$$
    (r+q-1)\left(\frac{\log r}{r-1}-1\right)<0.
$$
Negativity of the second eigenvalue on this branch is equivalent to
$$
    H(r)\defeq(r-1)(r+q-1)-qr\log r>0.
$$
At $r_0$,
$$
    H(r_0)=q(q-1)
    \tp{q(q-2)-2(q-1)\log(q-1)}>0,
$$
where the last inequality is equivalent to $\beta_{\mathsf c}<q/2$.  Moreover,
$H'(r_0)=2q\tp{q-2-\log(q-1)}>0$ and $H''(r)=2-q/r>0$ for $r\geq r_0$, so $H(r)>0$ throughout the stable branch.  Differentiating the stationary relation gives
$$
    \theta'(r)=\frac{H(r)}{2r(r-1)^2}>0,
$$
which verifies the asserted parametrization.  Compactness makes these bounds uniform.  By continuity, the losing ordered branch remains nondegenerate in a fixed neighborhood immediately below coexistence.

Near coexistence, let $x^{\mathrm{ord}}(\theta)$ denote the smooth ordered
stationary branch.  The envelope theorem gives
$$
    \frac{\dd}{\dd\theta}
    \tp{\Phi_\theta(x^{\mathrm{ord}}(\theta))
      -\Phi_\theta(x^{\mathrm{dis}})}
    =\norm{x^{\mathrm{ord}}(\theta)}_2^2
      -\norm{x^{\mathrm{dis}}}_2^2.
$$
At $\theta=\beta_{\mathsf c}$ the right-hand side is
$$
    \frac{(q-2)^2}{q(q-1)}>0.
$$
The implicit-function theorem and continuity therefore give
\begin{equation}
    c_q\abs{\theta-\beta_{\mathsf c}}
    \leq
    \abs{\Phi_\theta(x^{\mathrm{ord}}(\theta))
      -\Phi_\theta(x^{\mathrm{dis}})}
    \leq
    C_q\abs{\theta-\beta_{\mathsf c}}
    \label{eq:appendix-potts-linear-free-energy-gap}
\end{equation}
in a fixed neighborhood of $\beta_{\mathsf c}$.  Choose disjoint neighborhoods of the disordered and ordered stationary orbits.  By the phase classification and compactness, after making these neighborhoods smaller if necessary, every point outside them has free energy at least a fixed constant below the global maximum.  Away from coexistence, the same is true of the entire nondominant neighborhood.  Thus the only gap that can vanish on $[0,q/2]$ is the one in \eqref{eq:appendix-potts-linear-free-energy-gap}.

\paragraph{Finite-$n$ phase masses.}
For an unsorted type $\*m=nx$, write
$$
    W_\theta(\*m)
    =\frac{n!}{\prod_am_a!}
      \exp\tp{\frac\theta n\sum_am_a^2}.
$$
On every compact subset of the simplex interior, Stirling's formula gives uniformly
\begin{equation}
    W_\theta(\*m)
    =(2\pi n)^{-(q-1)/2}
      \tp{\prod_ax_a}^{-1/2}
      e^{n\Phi_\theta(x)}
      \tp{1+O_q(n^{-1})}.
    \label{eq:appendix-potts-uniform-stirling}
\end{equation}
We spell out the uniform lattice estimate used below.  Choose a fixed neighborhood of coexistence on which both the disordered and stable ordered branches are interior and nondegenerate, and choose disjoint compact neighborhoods of their orbits.  There are constants $r,c,C>0$, independent of $\theta$ in this interval, such that for every relevant maximum $x_*(\theta)$ and every tangent vector $u$ with $\|u\|_2\leq r$,
$$
    -C\|u\|_2^2
    \leq
    \Phi_\theta(x_*(\theta)+u)-\Phi_\theta(x_*(\theta))
    \leq
    -c\|u\|_2^2.
$$
This follows from the uniform negative-definiteness established above and Taylor's theorem; shrinking $r$ makes the cubic remainder smaller than either quadratic bound.  Divide the neighborhood into lattice annuli
$j/\sqrt n\leq\|u\|_2<(j+1)/\sqrt n$.  The $j$th annulus contains at most
$C_qn^{(q-1)/2}(j+1)^{q-2}$ count vectors, whose weights are at most
$C_qn^{-(q-1)/2}e^{n\Phi_\theta(x_*)}e^{-c j^2}$.  Summing over $j$ gives the uniform upper bound
$$
    \sum_{\*m/n\in U_*}W_\theta(\*m)
    \leq C_qe^{n\Phi_\theta(x_*)}.
$$
The lattice ball $\|u\|_2\leq n^{-1/2}$ contains
$\Theta_q(n^{(q-1)/2})$ count vectors, and the lower quadratic bound and
\eqref{eq:appendix-potts-uniform-stirling} give the matching lower bound
$$
    \sum_{\*m/n\in U_*}W_\theta(\*m)
    \geq c_qe^{n\Phi_\theta(x_*)}.
$$

We will also use the corresponding asymptotic at coexistence.  Let $B$ be the $q\times(q-1)$ matrix with columns $e_i-e_q$, so that $B$ identifies $\mathbb R^{q-1}$ with the tangent space in the same coordinates as the type lattice.  The preceding annulus bound supplies a summable majorant, and the Riemann sum of the quadratic Taylor expansion gives, uniformly for $\theta_n=\beta_{\mathsf c}+O(1/n)$,
\begin{equation*}
    \sum_{\*m/n\in U_*}W_{\theta_n}(\*m)
    =A_{\theta_n}(x_*)
      e^{n\Phi_{\theta_n}(x_*)}\tp{1+o(1)},
    \qquad
    A_\theta(x_*)
    =\frac{\tp{\prod_ax_{*,a}}^{-1/2}}
    {\sqrt{\det\tp{-B^\mathsf T
       \nabla^2\Phi_\theta(x_*)B}}}.
\end{equation*}
Here $x_*$ is the corresponding stationary branch at $\theta_n$.  The formula follows directly by multiplying the Stirling prefactor in \eqref{eq:appendix-potts-uniform-stirling} by the Gaussian integral in these lattice coordinates.  In particular, $A_\theta(x_*)$ is continuous and bounded above and below by positive constants.

Apply the two-sided estimate to the dominant and losing maxima in the fixed coexistence neighborhood.  Equation \eqref{eq:appendix-potts-linear-free-energy-gap} bounds the ratio of their partition sums by
$C_qe^{-c_qn|\theta-\beta_{\mathsf c}|}$.  Outside their fixed neighborhoods, compactness gives a uniform free-energy deficit.  The elementary multinomial bound
$W_\theta(nx)\leq e^{n\Phi_\theta(x)}$ and the at most $(n+1)^{q-1}$ count
vectors then give relative mass at most
$C_qn^{q-1}e^{-c_qn}$.  Away from the coexistence neighborhood, compactness gives the same constant-order deficit for the entire nondominant region.  Hence the complement of a sufficiently small neighborhood $\mathcal C_\theta$ of the dominant maximum satisfies
\begin{equation*}
    \widehat\mu_\theta(\mathcal C_\theta^{\mathsf c})
    \leq
    C_qe^{-c_qn\abs{\theta-\beta_{\mathsf c}}}
    +C_qn^{q-1}e^{-c_qn}.
\end{equation*}
After decreasing $c_q$ and increasing $C_q$, the second term is absorbed by
the first uniformly for $\theta\in[0,q/2]$.  Folding changes neither this
estimate nor its constants by more than a factor $q!$.

\paragraph{Paths within the dominant core.}
We first assume that $n\geq N_q$ for a sufficiently large constant $N_q$; the finitely many remaining values of $n$ are treated below.  Choose $\mathcal C_\theta$ to be a smaller convex neighborhood on which the Hessian of $\Phi_\theta$, restricted to the simplex, is uniformly negative.  This neighborhood has a fixed radius, so it contains lattice points once $n\geq N_q$; let $\*m_\theta^\star$ maximize $\Phi_\theta$ over these lattice points.  Given a count vector $\*m\in\mathcal C_\theta$, choose nonnegative integers $N_{ij}$ describing transfers from its surplus coordinates to the deficit coordinates of $\*m_\theta^\star$.  Their sum is
$$
    D=\frac12\norm{\*m-\*m_\theta^\star}_1\leq n.
$$

We use the following elementary interlacing fact.  Suppose finitely many move classes $a$ have multiplicities $N_a$, with $\sum_aN_a=D$.  Give the $\ell$th move of class $a$ the deadline $(\ell-1/2)/N_a$, order all moves by deadline, and break ties by a fixed rule.  If $C_a(k)$ is the number of class-$a$ moves among the first $k$, then
\begin{equation}
    \abs{C_a(k)-\frac{kN_a}{D}}
    \leq 2+2|\{b:N_b>0\}|.
    \label{eq:appendix-balanced-interlacing}
\end{equation}
Indeed, at the deadline $t$ of the $k$th move, equally spaced midpoints give $|C_a(k)-tN_a|\leq2$ for every $a$, including the possible tie at $t$.  Summing over the move classes gives $|k-tD|\leq2|\{b:N_b>0\}|$, and \eqref{eq:appendix-balanced-interlacing} follows.

Apply this construction to the nonzero $N_{ij}$.  There are at most $q^2$ move classes.  Each coordinate is a signed sum of the counts of its incident move classes, so \eqref{eq:appendix-balanced-interlacing} implies that, after $k$ moves, every coordinate differs by at most a constant depending only on $q$ from the point at time $k/D$ on the segment from $\*m$ to $\*m_\theta^\star$.  Choose the core and the concave neighborhood with a fixed positive buffer.  Since the discrepancy in normalized coordinates is $O_q(1/n)$, the entire discrete path remains in the concave neighborhood once $N_q$ is large enough.

Each transfer is a one-spin count move.  It is an edge of the projected $(q-1)$-block chain because the remaining $q-2$ spins in the chosen block may retain their colors.  After sorting, its endpoints are therefore either equal or adjacent in the folded chain; omit any repeated folded state.  Concavity and the maximality of $\*m_\theta^\star$ imply that the continuous segment has free energy at least $\Phi_\theta(\*m/n)$.  On the compact interior neighborhood, the gradient of $\Phi_\theta$ is uniformly bounded, so the $O_q(1/n)$ discrepancy in normalized coordinates changes $n\Phi_\theta$ by at most $O_q(1)$.  Equation \eqref{eq:appendix-potts-uniform-stirling} has uniformly bounded amplitude ratios there, and a color orbit has size between $1$ and $q!$.  Every folded count vector $y$ on the path therefore satisfies
$$
    \widehat\mu_\theta(y)
    \geq K_q^{-1}\widehat\mu_\theta(\*m),
$$
and the path has length $D\leq n$.

For the finitely many $n$ with $q\leq n<N_q$, we take
$\mathcal C_\theta=\widehat\Omega_{n,q}$, so that
$\widehat\mu_\theta(\mathcal C_\theta^{\mathsf c})=0$,
and let $\*m_\theta^\star$ be a maximizer of $\widehat\mu_\theta$.  Each
$\*m$ is joined to $\*m_\theta^\star$ by
$\frac12\norm{\*m-\*m_\theta^\star}_1\leq n$ one-spin transfers, projected to
the folded space as above.  For each such $n$, every $\widehat\mu_\theta(y)$
is positive and continuous in $\theta\in[0,q/2]$, so the ratios
$\widehat\mu_\theta(\*m)/\widehat\mu_\theta(y)$ are bounded.  After enlarging
$K_q$, \eqref{eq:potts-core-ascent-path} therefore holds for these $n$ as
well.  This proves the
first part of \zcref{lem:potts-phase-core}.

It remains to cover the unbounded range $\theta\geq q/2$, on which the compact-interior argument above is not uniform.  The following direct construction treats this range directly.

All roots and paths above may be selected measurably by lexicographic tie-breaking, since the folded space is finite.  The set $\mathcal C_\theta$ is used only to split the already computed path derivative at each fixed $\theta$; no derivative of its indicator is introduced.

\begin{lemma}[Ordered-sector routing]
    \label{lem:appendix-potts-high-beta-routing}
    If $\theta\geq q/2$, every folded count vector has a path of length at
    most $2n$ to a common symmetric-line maximizer, and its unnormalized
    weight decreases by at most a factor $e^{q-1}$ along that path.
\end{lemma}

\begin{proof}
    Fix two coordinates with sum $N$, difference $d\geq0$, and put $\gamma=\theta N/n$.  Up to a constant their conditional weights are
    $$
        w_\gamma(d)
        =\binom N{(N+d)/2}
         \exp\left(\frac{\gamma d^2}{2N}\right),
    $$
    with adjacent log-ratio
    $$
        f_\gamma(d)
        =\frac{2\gamma(d+1)}N
         +\log\frac{N-d}{N+d+2}.
    $$
    Since $f_\gamma''(d)<0$ and
    $f_\gamma(0)\geq2/N-\log(1+2/N)\geq0$ when $\gamma\geq1$, the nonnegative set of $f_\gamma$ is an interval containing the left endpoint.  Thus the weight sequence first increases and then decreases, with either part possibly empty.  When $\gamma<1$, put $u=(d+1)/(N+1)$.  The inequality
    $\log((1-u)/(1+u))\leq-2u$ gives
    $$
        f_\gamma(d)
        \leq\frac{2(d+1)}N
             -\frac{2(d+1)}{N+1}
        \leq\frac2N.
    $$
    Consequently, moving a pair toward balance for at most $N/2$ transfers loses a total factor smaller than $e$.

    We now give the path.  Start from a sorted representative, keep its coordinate labels during the construction, put $h=n/(2\theta)$, and set $m_{q+1}=0$.  At the end, project every visited count vector to its sorted representative.  Permuting coordinates does not change its weight, and two consecutive projected count vectors are either equal or joined by an edge of the folded chain; we omit equal consecutive vectors.  Thus this final projection cannot increase the path length or decrease any weight.  First process $a=q,q-1,\ldots,2$.  If $m_a\geq h$, then the pair $(m_1,m_a)$ has $\gamma=\theta(m_1+m_a)/n\geq1$.  Move particles from $m_a$ to $m_1$ only while the pair weight does not decrease, stopping when the smaller coordinate reaches a mode or the current value of $m_{a+1}$.  The smaller coordinate at a mode is below $h$: if it were $y\geq h$, with larger coordinate $x$, then
    $$
        \log\frac{w(d+2)}{w(d)}
        =\frac{2\theta(d+1)}n-\log\frac{x+1}{y}
        >\frac{2\theta(d+1)}n-\frac{d+1}{y}\geq0,
    $$
    contradicting maximality.  Processing from $q$ down to $2$ and stopping at the current value of $m_{a+1}$ prevents two minority coordinates from crossing.  Since $m_{q+1}=0<h$, descending induction shows that every processed minority coordinate is below $h$.  Thus the sorted order is preserved, this stage never lowers the weight, and it uses at most the initial total minority mass, which is at most $(1-1/q)n$.

    Second, keep $m_1$ fixed and balance $m_2,\ldots,m_q$ until they differ by at most one.  Their total is fixed, so their target values are the floor and ceiling of their average.  In one batch, choose an excess coordinate and a deficit coordinate and transfer particles until at least one reaches its target; never use a completed coordinate again.  There are at most $q-2$ batches.  Every selected pair has sum below $2h$, and hence $\gamma<1$.  The preceding pair estimate shows that each batch loses a factor smaller than $e$.  Thus this stage loses at most $e^{q-2}$, and certainly at most $e^{q-1}$.  It uses at most the total minority mass, which is below $(q-1)h\leq(1-1/q)n$.

    The minorities now take only the values $t$ and $t+1$.  If $l$ of them equal $t$, write
    $$
        m^{(l)}=\tp{A_l,
        \underbrace{t+1,\ldots,t+1}_{q-1-l},
        \underbrace{t,\ldots,t}_{l}},
        \qquad
        A_l=n-(q-1-l)(t+1)-lt,
    $$
    and let $u_l=W_\theta(m^{(l)})$.  For every internal legal index,
    $$
        \frac{u_{l-1}u_{l+1}}{u_l^2}
        =e^{2\theta/n}\frac{A_l}{A_l+1}.
    $$
    If $t\leq\lfloor n/q\rfloor-1$, then $A_l\geq n/q$, so the last display is at least
    $e^{q/n-\log(1+q/n)}>1$.  Thus $(u_l)$ is log-convex and one of its two symmetric-line endpoints is no lighter than the current type.  If $t=\lfloor n/q\rfloor$, the legal indices form an interval ending at $q-1$, and, writing $d=A_l-t$,
    $$
        \log\frac{u_{l+1}}{u_l}
        =\frac{2\theta d}{n}
          -\log\tp{1+\frac d{t+1}}
        \geq d\tp{\frac qn-\frac1{t+1}}\geq0.
    $$
    The selected endpoint differs in at most $q-1$ spins.  Include those spins in one $(q-1)$-block and, if necessary, fill the block with spins whose colors are left unchanged; the selected endpoint is therefore joined to the current type by one projected block edge without lowering the weight.

    It remains to choose a common point on the symmetric line.  Set
    $$
        p_t=W_\theta(n-(q-1)t,t,\ldots,t),
        \qquad F(t)=\log\frac{p_{t-1}}{p_t}.
    $$
    Direct calculation gives
    $$
        F(t)=\frac{\theta(q-1)}n(2n+q-2qt)
        +(q-1)\log t
        -\sum_{i=1}^{q-1}\log\tp{n-(q-1)t+i},
    $$
    and
    $$
        F'(t)=-\frac{2\theta q(q-1)}n+\frac{q-1}{t}
        +(q-1)\sum_{i=1}^{q-1}
          \frac1{n-(q-1)t+i},
    $$
    $$
        F'''(t)=\frac{2(q-1)}{t^3}
        +2(q-1)^3\sum_{i=1}^{q-1}
          \frac1{\tp{n-(q-1)t+i}^3}>0.
    $$
    The two left-endpoint limits follow from the $\log t$ and $1/t$ terms.  At the right endpoint, $\theta\geq q/2$ and $\log(1+x)\leq x$ give
    $$
        F(n/q)
        \geq\frac{q^2(q-1)}{2n}
          -\sum_{i=1}^{q-1}\frac{iq}{n}=0,
    $$
    while
    $$
        F'(n/q)
        <-\frac{q^2(q-1)}n
          +\frac{q(q-1)}n
          +\frac{q(q-1)^2}n=0.
    $$
    Thus $F(0+)=-\infty$, $F'(0+)=+\infty$, $F(n/q)\geq0$, and $F'(n/q)<0$.  Since $F'$ is strictly convex, its nonpositive sublevel set is an interval containing the right endpoint; the endpoint signs therefore imply that $F'$ has exactly one zero.  Hence $F$ first increases and then decreases.  Because it starts at $-\infty$ and ends nonnegative, it changes sign exactly once.  The sequence $(p_t)$ is therefore unimodal.  Consecutive count vectors on the symmetric line differ in exactly $q-1$ spins and hence are joined by one projected block edge.  Fixing one of its maximizers as the common root $\*m_\theta^\star$, the last stage does not lower the weight and has length at most $n/q$.

    The total length is at most
    $$
        2\tp{1-1/q}n+1+n/q\leq2n,
    $$
    where the last inequality uses $n\geq q$.
\end{proof}

The folded probability differs from the normalized unnormalized weight by an orbit factor in $[1,q!]$.  Enlarging $K_q$ by $q!$ completes the proof of \zcref{lem:potts-phase-core}.

\subsection{Action bounds along the noncritical Potts paths}
\label{subsec:appendix-potts-noncritical-action}

\begin{proof}[Proof of \zcref{prop:potts-noncritical-weak-action}]
We first treat the disordered-only regime, where the claim is an
exact-action bound obtained from a Poincar\'e inequality.  Write
$b=\beta_{\mathrm{sp}}^{\mathrm{ord}}$; in the normalization
\eqref{eq:mean-field-potts}, this is the spinodal $\beta_s(q)$ of
\cite{CDL+12}.  Let $t_{\mathrm{mix}}(n,\beta)$ denote the mixing time, to
total-variation distance $1/4$, of the discrete-time single-site heat-bath
chain studied there, and let $C^{\mathrm{single}}_{\mathrm{PI}}(n,\beta)$ be
the Poincar\'e constant of the corresponding rate kernel $P-\!I$.  We show
that
\begin{equation}
    \sup_{0\leq\beta\leq b}
    C^{\mathrm{single}}_{\mathrm{PI}}(n,\beta)
    \leq C_qn^{4/3}
    \qquad\text{for all }n\geq q.
    \label{eq:appendix-potts-uniform-poincare}
\end{equation}

The single-site heat-bath matrix is an average of orthogonal projections, so
it is positive semidefinite.  Hence
$C^{\mathrm{single}}_{\mathrm{PI}}(n,\beta)$ is its relaxation time, and the
standard comparison between relaxation and mixing times
\cite[Chapter~12]{LP17} gives
\[
    C^{\mathrm{single}}_{\mathrm{PI}}(n,\beta)
    \leq1+\frac{t_{\mathrm{mix}}(n,\beta)}{\log2}.
\]
For each fixed $n$, the spectral gap depends continuously on $\beta$ and is
positive, so $C^{\mathrm{single}}_{\mathrm{PI}}(n,\cdot)$ is bounded on
$[0,b]$.  It therefore suffices to prove
\eqref{eq:appendix-potts-uniform-poincare} for all sufficiently large $n$.
Suppose that this fails.  Then there are $n_k\to\infty$ and
$\beta_k\in[0,b]$ such that
$t_{\mathrm{mix}}(n_k,\beta_k)/n_k^{4/3}\to\infty$.  Passing to a
subsequence, we may assume that $\beta_k\to\beta_*\in[0,b]$.

If $\beta_*<b$, then $\beta_k\leq b-\delta$ for some $\delta>0$ and all
large $k$.  Theorem~1 of \cite{CDL+12} is stated for fixed $\beta$, but the
two inputs of its proof are uniform on $[0,b-\delta]$.  Let $D_\beta(x)$ be the
one-dimensional drift bound of \cite[Section~3]{CDL+12}, which is continuous in
$(\beta,x)$ and negative for $\beta<b$ and $x\in(1/q,1]$.  For every fixed
$\rho>0$, compactness gives
\[
    \sup_{\substack{0\leq\beta\leq b-\delta\\
                    1/q+\rho\leq x\leq1}}
    D_\beta(x)<0,
\]
and the contraction rate near the uniform profile satisfies
$1-2\beta/q\geq1-2b/q>0$.  The first bound controls the time needed to enter
a neighborhood of the uniform profile, and the second controls the
contraction inside this neighborhood.  With these two inputs fixed, the proof
of \cite[Proposition~3.3(3)]{CDL+12} and the coupling argument of
\cite[Section~4]{CDL+12} give constants that are uniform on $[0,b-\delta]$,
and hence
$t_{\mathrm{mix}}(n,\beta)\leq C_{q,\delta}\,n\log n$ uniformly there.  This
contradicts the choice of $\beta_k$.

If $\beta_*=b$, put $\xi_k=b-\beta_k\geq0$, so that $\xi_k\to0$.  Passing to
a further subsequence, either $n_k^{2/3}\xi_k\to\infty$ or
$n_k^{2/3}\xi_k$ stays bounded.  Extend $(\xi_k)$ to a sequence
$(\xi(n))_{n\geq1}$ by setting $\xi(n)=n^{-1/3}$ in the first case and
$\xi(n)=0$ in the second case for every $n\notin\set{n_k}$.  Theorem~2 of
\cite{CDL+12} then applies to $\beta(n)=b-\xi(n)$.  In the first case it
gives
\[
    t_{\mathrm{mix}}(n,\beta(n))
    =\tp{1+o(1)}
    \tp{\alpha_1(\beta(n),q)\,n\log n+\alpha_2(q)\,n\,\xi(n)^{-1/2}},
\]
where $\alpha_1(\beta,q)=[2(1-2\beta/q)]^{-1}$ is bounded on $[0,b]$
because $b<q/2$, and
$n\,\xi(n)^{-1/2}=n^{4/3}\tp{n^{2/3}\xi(n)}^{-1/2}=o(n^{4/3})$.  In the
second case it gives $t_{\mathrm{mix}}(n,\beta(n))=O(n^{4/3})$.  Either
way, there is a constant $C$, independent of $k$, such that
$t_{\mathrm{mix}}(n_k,\beta_k)\leq Cn_k^{4/3}$ for all large $k$, again a
contradiction.  This proves \eqref{eq:appendix-potts-uniform-poincare}.

We next pass to the $(q-1)$-block chain.  For a block $B$, the heat-bath
update on $B$ has Dirichlet form
$\mathbb E_{\mu_\beta}\operatorname{Var}_{\mu_\beta}(f\mid\sigma_{B^{\mathsf c}})$.
For every $i\in B$, the law of total variance gives
\[
    \operatorname{Var}_{\mu_\beta}(f\mid\sigma_{B^{\mathsf c}})
    \geq
    \mathbb E_{\mu_\beta}\left[
      \operatorname{Var}_{\mu_\beta}(f\mid\sigma_{[n]\setminus\set{i}})
      \,\middle|\,\sigma_{B^{\mathsf c}}\right].
\]
Averaging over $i\in B$ and then over the uniformly random block shows that
the Dirichlet form of the block chain is at least that of the single-site
chain.  Hence \eqref{eq:appendix-potts-uniform-poincare} also holds for the
block chain.

Finally,
$\partial_s\log\mu_{\beta_s}(\sigma)
=(\beta_s'/n)\tp{\sum_aM_a(\sigma)^2-\mathbb E_{\mu_{\beta_s}}\sum_aM_a^2}$
and $n^2/q\leq\sum_aM_a^2\leq n^2$, so
$\norm{\partial_s\log\mu_{\beta_s}}_{L^2(\mu_{\beta_s})}^2
\leq n^2\abs{\beta_s'}^2$.  Applying \zcref{prop:PI-implies-action} with the
constant Poincar\'e bound $C_qn^{4/3}$ along the forward path gives
\[
    \+A\tp{\tp{\mu_{\beta_s}}_{s\in[0,1]}}
    \leq C_qn^{10/3}\int_0^1\abs{\beta_s'}^2\,\dd s,
\]
which is absorbed into $Q_q(n)$.  This proves the exact-action statement in
the disordered-only regime.

The rest of the proof concerns the phase-core construction.  It gives weak
action in the two metastability regimes and, when the core is the whole
folded space, exact action in the ordered-only regime.

Fix a time $s$ at which
$\beta_s'$ and $\partial_s\widehat\mu_{\beta_s}$ exist.  Let
$\mathcal C_{\beta_s}$, $\*m_{\beta_s}^\star$, and
$\gamma_{x,\beta_s}$ be the core, root, and folded-state-space paths from
\zcref{lem:potts-phase-core}; the lexicographic tie-breaking in its proof
makes these choices measurable.  Define
$$
    r_s(x)=
    \begin{cases}
        \partial_s\widehat\mu_{\beta_s}(x),
            &x\in\mathcal C_{\beta_s}^{\mathsf c},\\
        -\displaystyle\sum_{u\in\mathcal C_{\beta_s}^{\mathsf c}}
            \partial_s\widehat\mu_{\beta_s}(u),
            &x=\*m_{\beta_s}^\star,\\
        0,&\text{otherwise}.
    \end{cases}
$$
The two nonzero parts have opposite total mass, so $\sum_xr_s(x)=0$.
Moreover, the triangle inequality and
\eqref{eq:folded-potts-score-bound} give
\begin{align}
    \frac12\norm{r_s}_1
    &\leq
      \sum_{x\in\mathcal C_{\beta_s}^{\mathsf c}}
      \abs{\partial_s\widehat\mu_{\beta_s}(x)}\notag\\
    &\leq n\abs{\beta_s'}
      \widehat\mu_{\beta_s}
      \tp{\mathcal C_{\beta_s}^{\mathsf c}}.
    \label{eq:potts-residual-pointwise}
\end{align}
Thus
$b_s=\partial_s\widehat\mu_{\beta_s}-r_s$ is supported on
$\mathcal C_{\beta_s}$ and has total mass zero.

We apply \zcref{prop:rooted-path-action} with
$G_s=\widehat\Omega_{n,q}$ and root $z_s=\*m_{\beta_s}^\star$.  For
$x\in\mathcal C_{\beta_s}$, we use the path $\gamma_{x,\beta_s}$ from
\zcref{lem:potts-phase-core}, after erasing loops if necessary; loop erasure
only removes states, so it preserves \eqref{eq:potts-core-ascent-path} and
does not increase the length.  These paths need not stay inside
$\mathcal C_{\beta_s}$.  For $x\notin\mathcal C_{\beta_s}$, we choose an
arbitrary simple path to the root, again by lexicographic tie-breaking.
Since $b_s(x)=0$ for such $x$, these paths carry no load.  Hence only
starting points in the core contribute to the sums below.

For a folded edge $e$, define the path mass
$$
    M_s(e)
    =\sum_{\substack{x\in\mathcal C_{\beta_s}
             \setminus\set{\*m_{\beta_s}^\star}:\\
             e\in\gamma_{x,\beta_s}}}
       \widehat\mu_{\beta_s}(x).
$$
If $e=\set{y,y'}$ belongs to $\gamma_{x,\beta_s}$, then both endpoints lie
on that path, and \eqref{eq:potts-core-ascent-path} gives
$$
    \widehat\mu_{\beta_s}(x)
    \leq K_q\tp{\widehat\mu_{\beta_s}(y)
      \wedge\widehat\mu_{\beta_s}(y')}.
$$
There are at most
$\abs{\widehat\Omega_{n,q}}\leq(n+1)^{q-1}$ possible starting points.
Therefore
\begin{equation}
    M_s(e)
    \leq K_q(n+1)^{q-1}
      \tp{\widehat\mu_{\beta_s}(y)
      \wedge\widehat\mu_{\beta_s}(y')}.
    \label{eq:potts-core-mass-load}
\end{equation}

The rooted-path load is
$$
    L_s(e)
    =\sum_{\substack{x\in\mathcal C_{\beta_s}
             \setminus\set{\*m_{\beta_s}^\star}:\\
             e\in\gamma_{x,\beta_s}}}
      \abs{b_s(x)}.
$$
The residual vanishes at every such starting point.  Hence
\eqref{eq:folded-potts-score-bound} and
\eqref{eq:potts-core-mass-load} imply
\begin{align*}
    L_s(e)
    &=\sum_{\substack{x\in\mathcal C_{\beta_s}
             \setminus\set{\*m_{\beta_s}^\star}:\\
             e\in\gamma_{x,\beta_s}}}
      \abs{\partial_s\widehat\mu_{\beta_s}(x)}\\
    &\leq n\abs{\beta_s'}M_s(e)\\
    &\leq K_qn\abs{\beta_s'}(n+1)^{q-1}
      \tp{\widehat\mu_{\beta_s}(y)
      \wedge\widehat\mu_{\beta_s}(y')}.
\end{align*}
On the other hand, \eqref{eq:folded-potts-capacity} gives
$$
    \widehat c_{\beta_s}(e)
    \geq
    \frac{\tp{\widehat\mu_{\beta_s}(y)
      \wedge\widehat\mu_{\beta_s}(y')}^2}
    {(q!)^2n^{2q-2}}.
$$
Dividing the squared load by this capacity cancels the probability factor:
$$
    \frac{L_s(e)^2}{\widehat c_{\beta_s}(e)}
    \leq(q!)^2K_q^2(n+1)^{2q-2}n^{2q}
      \abs{\beta_s'}^2.
$$
The folded graph has at most
$\abs{\widehat\Omega_{n,q}}^2\leq(n+1)^{2q-2}$ edges.  Consequently,
\begin{align}
    \sum_e\frac{L_s(e)^2}{\widehat c_{\beta_s}(e)}
    &\leq C_q(n+1)^{4q-4}n^{2q}
       \abs{\beta_s'}^2\notag\\
    &\leq C_qn^{6q}\abs{\beta_s'}^2.
    \label{eq:potts-core-flux-cost}
\end{align}

Set
$\varepsilon_{\mathrm w}=\frac12\int_0^1\norm{r_s}_1\,\dd s$.
The rooted-path bound \eqref{eq:rooted-path-action-bound}, followed by
\eqref{eq:potts-core-flux-cost}, gives
$$
    \+A^{\mathrm w}_{\varepsilon_{\mathrm w}}
      \tp{\tp{\widehat\mu_{\beta_s}}_{s\in[0,1]}}
    \leq C_qn^{6q}\int_0^1\abs{\beta_s'}^2\,\dd s.
$$
The symmetry identity \eqref{eq:folding-preserves-weak-action} gives
the same bound on the original Potts space.  Enlarging coefficients if
necessary, take $Q_q(n)\leq C_qn^{6q}$; for either linear schedule in the
theorem, $\int_0^1\abs{\beta_s'}^2\,\dd s=\mathsf B^2$.  This proves
\eqref{eq:potts-weak-action-bound}.

It remains to estimate the residual budget.  Along either schedule,
\zcref{lem:potts-phase-core} and
\eqref{eq:potts-residual-pointwise} give
$$
    \frac12\norm{r_s}_1
    \leq C_qn\abs{\beta_s'}
      e^{-c_qn\abs{\beta_s-\beta_{\mathsf c}}};
$$
when $\beta_s\geq q/2$, the left-hand side is zero.  Each schedule is
monotone, remains on one side of $\beta_{\mathsf c}$, and has target
$\beta$ as its closest point to coexistence.  Changing variables from $s$
to $t=\abs{\beta_s-\beta_{\mathsf c}}$ therefore yields
\begin{align*}
    \varepsilon_{\mathrm w}
    &\leq C_qn
      \int_{\abs{\beta-\beta_{\mathsf c}}}^{\infty}
      e^{-c_qnt}\,\dd t\\
    &\leq C_qe^{-c_qn\abs{\beta-\beta_{\mathsf c}}},
\end{align*}
which proves \eqref{eq:potts-weak-residual}.  Finally, if
$\beta\geq q/2$, then the reverse schedule from $\beta_{\mathrm{hi}}$ to
$\beta$ stays in the regime where
$\mathcal C_{\beta_s}=\widehat\Omega_{n,q}$.  Thus $r_s=0$ throughout, and
the same construction gives the stated exact-action bound.
\end{proof}

\subsection{Initialization}
\label{subsec:appendix-potts-algorithm}

\begin{proof}[Proof of \zcref{lem:potts-noncritical-initialization}]
The forward path begins at $\mu_0$, the uniform law on $[q]^n$, so its
initialization is exact.  For the reverse path, let $\Omega_{\mathrm{mono}}$ be the $q$
monochromatic configurations and put
$$
    \nu_0=(1-\delta_{\mathrm{init}})
          \operatorname{Unif}(\Omega_{\mathrm{mono}})
          +\delta_{\mathrm{init}}\operatorname{Unif}([q]^n).
$$
This law has full support and is efficiently sampleable.  Every nonground
configuration satisfies
$$
    \sum_{j=1}^qM_j^2
    \leq(n-1)^2+1=n^2-2n+2.
$$
Hence, at inverse temperature $\beta_{\mathrm{hi}}$, its total Gibbs mass
outside $\Omega_{\mathrm{mono}}$, denoted by $\delta_0$, is at most
$$
    \delta_0
    \leq q^{n-1}
      \exp\tp{-2\beta_{\mathrm{hi}}(1-1/n)}.
$$
On $\Omega_{\mathrm{mono}}$, symmetry and
$\nu_0\geq(1-\delta_{\mathrm{init}})
\operatorname{Unif}(\Omega_{\mathrm{mono}})$ bound the KL contribution by
$\log(1/(1-\delta_{\mathrm{init}}))$.  Outside
$\Omega_{\mathrm{mono}}$, the likelihood ratio is at most
$q^n/\delta_{\mathrm{init}}$.  Therefore
$$
    D_{\mathrm{KL}}
      \tp{\mu_{\beta_{\mathrm{hi}}}\|\nu_0}
    \leq\log\frac1{1-\delta_{\mathrm{init}}}
      +\delta_0\tp{n\log q+\log\frac1{\delta_{\mathrm{init}}}}.
$$
Choose $\delta_{\mathrm{init}}=c\zeta$ for a sufficiently small constant
$c>0$ (and decrease it to at most $1/2$).  If
$$
    \beta_{\mathrm{hi}}
    \geq\max\left\{q/2,\beta,
      n\log q+\log\frac{C_q}{\zeta}\right\},
$$
then, after increasing $C_q$, both terms in the preceding display sum to at
most $\zeta$.  This proves the lemma.
\end{proof}

\begin{proof}[Proof of \zcref{thm:potts-noncritical-annealing}]
For the forward path, use the exact uniform initializer.  For the reverse
path, apply \zcref{lem:potts-noncritical-initialization} with
$\zeta=\tau^2/64$, so that in either case
$$
    D_{\mathrm{KL}}
      \tp{\mu_{\beta_0}\|\pi_0^{\-{\!{ALG}}}}
    \leq\frac{\tau^2}{64}.
$$
The hypothesis on the distance from coexistence and
\eqref{eq:potts-weak-residual} give
$\varepsilon_{\mathrm w}\leq\tau/4$.  Applying the weak-action conclusion
of \zcref{thm:annealing-ub} with \eqref{eq:potts-weak-action-bound} yields
\begin{equation*}
    D_{\mathrm{TV}}\tp{\mu_\beta,\nu_T}
    \leq\frac\tau4+
    \sqrt{\frac12\left(
       \frac{\tau^2}{64}
       +\frac{Q_q(n)\mathsf B^2}{4T}
    \right)}.
\end{equation*}
Choosing
$T=\max\{1,CQ_q(n)\mathsf B^2/\tau^2\}$ for a sufficiently large universal
constant $C$ makes the right-hand side at most $\tau$.  The expected number
of block updates is $T$ by \zcref{thm:annealing-ub}.

If $\beta\geq q/2$, the residual vanishes.  Applying instead the exact-action
conclusion of the same theorem gives
$$
    D_{\mathrm{KL}}\tp{\mu_\beta\|\nu_T}
    \leq
    D_{\mathrm{KL}}
      \tp{\mu_{\beta_{\mathrm{hi}}}\|\pi_0^{\-{\!{ALG}}}}
    +\frac{Q_q(n)\mathsf B^2}{4T}.
$$
Thus initialization error at most $\zeta/2$ and
$T\geq Q_q(n)\mathsf B^2/(2\zeta)$ give reverse-KL error at most $\zeta$.
\end{proof}

\subsection{An exact-action fallback at extreme accuracy}
\label{subsec:appendix-potts-exact-fallback}

For $0\leq\theta\leq q/2$, every unnormalized count vector has weight at least
one, while
$$
    Z_{n,q}(\theta)
    \leq q^ne^{qn/2}.
$$
Consequently every folded count vector has probability at least $e^{-C_qn}$,
and \eqref{eq:folded-potts-capacity} gives a minimum folded edge capacity at
least $e^{-C_qn}n^{-2q+2}$.  Route the exact path derivative on an arbitrary
spanning tree.  Since \eqref{eq:folded-potts-score-bound} implies
$$
    \sum_x\abs{\partial_s\widehat\mu_{\beta_s}(x)}
    \leq n\abs{\beta_s'},
$$
every tree-edge flux is at most this quantity.  The tree has at most $(n+1)^{q-1}$ edges, so
\begin{equation*}
    \+A
    \leq e^{C_qn}\int_0^1\abs{\beta_s'}^2\,\dd s
\end{equation*}
on a bounded-temperature segment.

Fix $\Delta_\beta=\abs{\beta-\beta_{\mathsf c}}>0$.  If
$\tau\geq4C_qe^{-c_qn\Delta_\beta}$, use the polynomial weak-action bound.
Otherwise
$$
    n\leq\frac1{c_q\Delta_\beta}\log\frac{4C_q}{\tau},
$$
so $e^{C_qn}\leq(C_q/\tau)^{C_q/(c_q\Delta_\beta)}$.  The exact fallback is therefore polynomial in $1/\tau$, with degree depending on $1/\Delta_\beta$.

For completeness, the exact path used in the second case is as follows.  If $\beta<\beta_{\mathsf c}$, the whole forward path $0\to\beta$ lies in $[0,q/2]$, so the spanning-tree construction applies directly.  If $\beta_{\mathsf c}<\beta<q/2$, split the reverse path at $q/2$: the segment $\beta_{\mathrm{hi}}\to q/2$ has the polynomial exact-action bound in \zcref{prop:potts-noncritical-weak-action}, and the bounded segment $q/2\to\beta$ uses the spanning-tree routing above.  Parametrize the two pieces on two fixed subintervals of $[0,1]$; this changes their action bounds only by an absolute factor.  Finally, if $\beta\geq q/2$, the entire reverse path already has a polynomial exact-action bound.  Thus in every noncritical case the claimed polynomial in $(n,1/\tau)$ follows without leaving any part of the path unspecified.  This proves the refinement stated after \zcref{thm:potts-noncritical-annealing}.

\section{Omitted Proofs for the Bipartite Applications}

\subsection{Proofs for exact phase balance}
\label{subsec:appendix-symmetric-phase-bound}

We prove \zcref{prop:symmetric-two-phase-bound}.  Put $m=N/2=|L|=|R|$,
write $H(I)=|I|$, and recall that $\pi_s=\mu_{G,\lambda_s}$.  The
phase representation of \cite{JKP20}, as recorded in
\cite[Section~5.3.2 and Lemma~31]{CGGPSV21}, approximates the hard-core
partition function by the sum of two polymer partition functions centered at
$L$ and $R$.  Configurations belonging to both phase descriptions are counted
twice, so we first bound this overlap.  In the polymer description centered at
$L$, each connected deviation component contains at most $m/4$ missing
$L$-vertices.  The occupied $R$-vertices in that component have all their
neighbors among those missing vertices.  They number at most $m/2$, since
otherwise expansion of an $m/2$-subset would already require more than $m/4$
such neighbors.  Expansion can therefore be applied to them component by
component, giving
$$
    m-|I\cap L|\geq(1+\alpha)|I\cap R|.
$$
The description centered at $R$ gives the symmetric inequality
$m-|I\cap R|\geq(1+\alpha)|I\cap L|$.
Adding them shows that every independent set represented in both phases has size at most $2m/(2+\alpha)$.  Since $Z_G(\lambda)\geq2\lambda^m$, its stationary mass is at most
\begin{equation}
    b_*(m)
    \defeq
    4^m\lambda_*^{-\alpha m/(2+\alpha)}
    =e^{-\Omega_{\alpha,\Delta}(m)}
    \label{eq:appendix-symmetric-overlap}
\end{equation}
after increasing the high-fugacity threshold.

For $i\in\{0,1\}$, let $G_{i,s}$ be the configuration-space image of the truncated hard-core polymer model $\Omega^i_{M_s}$ of \cite[Theorem~29]{CGGPSV21}, where
\[
    M_s=\left\lfloor
    \frac{3(2+\alpha)}{2\alpha\log\lambda_s}
    \log\frac{8m}{\xi}
    \right\rfloor.
\]
The cutoff is piecewise constant and satisfies
$M_s=O_{\alpha,\Delta}(\log(m/\xi))$.  By \cite[Lemma~32]{CGGPSV21}, the sum of the two truncated partition functions is a multiplicative $e^{\pm\xi}$ approximation to $Z_G(\lambda_s)$.  Equivalently, with
\[
    m_s=\mathbf1_{G_{0,s}}+\mathbf1_{G_{1,s}},
    \qquad
    S_s=\mathbb E_{\pi_s}m_s,
\]
we have $|S_s-1|\leq2\xi$ when $\xi\leq1$.  If $a_s=\pi_s(\bar G_{0,s}\cap\bar G_{1,s})$ and $b_s=\pi_s(G_{0,s}\cap G_{1,s})$, then $S_s=1-a_s+b_s$.  The truncated overlap is contained in the overlap in \eqref{eq:appendix-symmetric-overlap}.  Since $\xi\leq1/4$, and since for every inverse-polynomial $\xi$ the latter bound is at most $\xi/2$ for all sufficiently large $m$, we obtain
\begin{equation}
    a_s+b_s\leq3\xi,
    \qquad
    S_s\geq\frac12,
    \label{eq:appendix-symmetric-coverage}
\end{equation}
uniformly along the path.

Define the duplicated, normalized law
\[
    \bar\pi_s=\frac{m_s\pi_s}{S_s}.
\]
Since $|1-S_s|\leq a_s+b_s$,
\begin{equation}
    \norm{\pi_s-\bar\pi_s}_1
    \leq
    \norm{\pi_s-m_s\pi_s}_1
    +\norm{m_s\pi_s-\bar\pi_s}_1
    \leq2(a_s+b_s)\leq6\xi.
    \label{eq:appendix-symmetric-l1}
\end{equation}
Side symmetry exchanges $G_{0,s}$ and $G_{1,s}$.  Thus, writing
$\mu_{i,s}=\pi_s(\cdot\mid G_{i,s})$,
\begin{equation}
    \bar\pi_s=\frac12\mu_{0,s}+\frac12\mu_{1,s},
    \qquad
    \mathbb E_{\mu_{0,s}}H
    =\mathbb E_{\mu_{1,s}}H.
    \label{eq:appendix-symmetric-branch-split}
\end{equation}

If $M_s=0$, each restricted set $G_{i,s}$ consists only of its ground state, so it satisfies a PI with constant one.  Suppose henceforth that $M_s\geq1$.  The proof of \cite[Theorem~29]{CGGPSV21} applies its Theorem~25 with the local weight-ratio parameter there set to $\lambda_s$.  In our notation, the resulting restricted-chain mixing bound is
$$
    O\tp{M_s\lambda_s^{M_s+1}(2m)^2
      \log(2m)\log(\lambda_s/\varepsilon)}.
$$
Since
$$
    M_s=O_{\alpha,\Delta}
    \tp{\frac{\log(m/\xi)}{\log\lambda_s}},
$$
$\lambda_s^{M_s+1}$, and hence the full bound at fixed $\varepsilon$, is polynomial in $m$, $1/\xi$, and $\lambda_0$, uniformly along the path.  The legal addition and deletion rates of the cited Metropolis chain are, respectively, $1/(4m)$ and $1/(4m\lambda_s)$ in our total-rate-one normalization, whereas heat bath has rates $\lambda_s/(2m(1+\lambda_s))$ and $1/(2m(1+\lambda_s))$.  Thus every heat-bath conductance is at least the corresponding Metropolis conductance.

We also record the elementary conversion from mixing to a PI constant.  If a finite reversible transition matrix $P$ satisfies $\max_x D_{\mathrm{TV}}\tp{P^t(x,\cdot),\pi}\leq1/4$, let $f$ be a mean-zero eigenfunction with eigenvalue $\lambda_2\in[0,1)$, normalize $\|f\|_\infty=1$, and evaluate at a maximizer of $|f|$.  Then $\lambda_2^t\leq1/2$, which implies
$$
    \frac1{1-\lambda_2}
    \leq1+\frac{t}{\log2}.
$$
If $\lambda_2<0$, this PI bound is automatic.  Conductance domination can only increase the spectral gap.  The cited mixing estimate therefore gives
\begin{equation*}
    \sup_{s,i}C_{\-{PI}}\tp{
       P_{G,\lambda_s}\text{ restricted to }G_{i,s}}
    \leq Q^{\mathrm{PI}}_{\alpha,\Delta}(m,1/\xi,\lambda_0)
\end{equation*}
for a polynomial $Q^{\mathrm{PI}}_{\alpha,\Delta}$.

Set
\[
    g_{i,s}
    =-\mathsf B\mu_{i,s}
       \tp{H-\mathbb E_{\mu_{i,s}}H},
    \qquad
    g_s=\frac12\tp{g_{0,s}+g_{1,s}}.
\]
By \eqref{eq:appendix-symmetric-branch-split},
$g_s=-\mathsf B\bar\pi_s(H-\mathbb E_{\bar\pi_s}H)$.
The pointwise version of \zcref{prop:PI-implies-action}, applied in each branch, provides fluxes $J_{i,s}$ with divergence $-g_{i,s}$ and cost at most
$Q^{\mathrm{PI}}_{\alpha,\Delta}(m,1/\xi,\lambda_0)m^2\mathsf B^2$, because $H\in[0,2m]$ and hence $\Var{H}\leq m^2$.
Set $J_s=(J_{0,s}+J_{1,s})/2$.  The annealing capacity inside $G_{i,s}$ is $\pi_s(G_{i,s})$ times its conditional capacity, and
$\pi_s(G_{i,s})=S_s/2$.  Therefore \eqref{eq:appendix-symmetric-coverage} and $(u+v)^2\leq2u^2+2v^2$ give
\begin{equation}
    \sum_e\frac{J_s(e)^2}{c_s(e)}
    \leq4Q^{\mathrm{PI}}_{\alpha,\Delta}
      (m,1/\xi,\lambda_0)m^2\mathsf B^2.
    \label{eq:appendix-symmetric-flux-cost}
\end{equation}

Finally, set $r_s=\partial_s\pi_s-g_s$.  For probability measures $\rho,\nu$ and a statistic taking values in an interval of length $R$,
\[
    \norm{\rho\tp{H-\mathbb E_\rho H}
          -\nu\tp{H-\mathbb E_\nu H}}_1
    \leq2R\norm{\rho-\nu}_1.
\]
Applying this with $R=2m$ and \eqref{eq:appendix-symmetric-l1} yields
\begin{equation}
    \frac12\norm{r_s}_1
    \leq12m\mathsf B\xi.
    \label{eq:appendix-symmetric-residual}
\end{equation}
The cutoff changes only finitely many times.  The quantities above are defined pointwise from the true derivative $\partial_s\pi_s$, rather than by differentiating $\bar\pi_s$, so these changes introduce no additional term.  Since both $\partial_s\pi_s$ and $g_s$ have total mass zero, so does $r_s$, and $(J_s,r_s)$ is an $L^1$-relaxed flux for almost every $s$.  Taking $\varepsilon_{\mathrm{sym}}=\frac12\int_0^1\norm{r_s}_1\,\dd s$, equations \eqref{eq:appendix-symmetric-flux-cost} and \eqref{eq:appendix-symmetric-residual} prove \eqref{eq:symmetric-two-phase-bound}, after absorbing the factor $4m^2$ into $Q^{\mathrm{sym}}_{\alpha,\Delta}$ and using $N=2m$.

For completeness, an efficiently sampleable full-support initialization can be obtained without sampling a polymer model.  Let $\chi$ be the law obtained by choosing a uniform subset of $V(G)$ and returning it when it is independent, and returning the empty set otherwise.  Every independent set has $\chi$-mass at least $2^{-2m}$.  For a polynomially small mixing weight $\theta>0$, set
$$
    \nu_0=(1-\theta)\rho_{\mathrm{gs}}+\theta\chi.
$$

We first bound the error of the ground-state mixture.  Expansion and independence imply that every independent set occupies at most half of at least one side.  Fix such a side and let $S$ be its occupied set.  If $|S|=k\leq m/2$, expansion leaves at most $m-(1+\alpha)k$ available vertices on the other side.  Relative to the weight $\lambda_0^m$ of one ground state, all configurations with this fixed $S$ therefore have total weight at most
$$
    \lambda_0^{k-m}
    (1+\lambda_0)^{m-(1+\alpha)k}
    \leq
    e^{m/\lambda_0}\lambda_0^{-\alpha k}.
$$
The nonground configurations with $S=\varnothing$ contribute at most
$e^{m/\lambda_0}-1$.  Summing over $S$ and over the two choices of the minority side gives, whenever
$\lambda_0\geq2m$ and $\lambda_0^\alpha\geq2m$,
\begin{equation*}
    \delta_0\defeq
    D_{\mathrm{TV}}\tp{\pi_0,\rho_{\mathrm{gs}}}
    \leq
    C_\alpha\left(
       \frac m{\lambda_0}
       +\frac m{\lambda_0^\alpha}
    \right).
\end{equation*}
Indeed, the two sums are bounded by constant multiples of
$(1+\lambda_0^{-1})^m-1$ and
$(1+\lambda_0^{-\alpha})^m-1$, respectively.  If the mixing weight is $\theta$, symmetry makes the two ground-state masses equal, and splitting the relative entropy over the ground states and their complement gives
$$
    D_{\mathrm{KL}}(\pi_0\|\nu_0)
    \leq
    \log\frac1{1-\theta}
    +\delta_0\tp{2m\log2+\log\frac1\theta}.
$$
Choosing $\theta$ and $\delta_0$ polynomially smaller than $\tau^2$ therefore yields
$$
    D_{\mathrm{KL}}(\pi_0\|\nu_0)
    \leq\tau^2/16
$$
with $\lambda_0=\operatorname{poly}_{\alpha}(m,1/\tau)$, as used in \zcref{thm:symmetric-bis-weak-action}.

\begin{proof}[Proof of \zcref{thm:symmetric-bis-weak-action}]
    The initialization above gives a polynomially large $\lambda_0$ and an
    efficiently sampleable full-support law $\nu_0$ such that
    \begin{equation*}
        D_{\mathrm{KL}}(\pi_0\|\nu_0)
        \leq\frac{\tau^2}{16}.
    \end{equation*}
    For this $\lambda_0$, choose $\xi$ so that the residual bound in
    \eqref{eq:symmetric-two-phase-bound} is at most $\tau/4$.  Since weak
    action is nonincreasing in the residual budget, the proposition then
    gives
    \begin{equation*}
        \+A^{\mathrm w}_{\tau/4}
        \tp{\tp{\pi_s}_{s\in[0,1]}}
        \leq Q^{\mathrm{sym}}_{\alpha,\Delta}
        (N,1/\xi,\lambda_0)\mathsf B^2.
    \end{equation*}
    Applying \eqref{eq:continuous-weak-action-bound} gives
    \begin{equation*}
        D_{\mathrm{TV}}(\pi_1,\nu_T)
        \leq\frac\tau4+
        \sqrt{\frac12\left(
          \frac{\tau^2}{16}
          +\frac{Q^{\mathrm{sym}}_{\alpha,\Delta}
             (N,1/\xi,\lambda_0)\mathsf B^2}{4T}
        \right)}.
    \end{equation*}
    Taking
    \begin{equation*}
        T=C\frac{Q^{\mathrm{sym}}_{\alpha,\Delta}
          (N,1/\xi,\lambda_0)\mathsf B^2}{\tau^2}
    \end{equation*}
    for a sufficiently large absolute constant $C$ proves the claim.  This
    choice is polynomial in $N$ and $1/\tau$.
\end{proof}

\subsection{Proof of the high-girth phase-balance bound}
\label{subsec:appendix-high-girth-phase-balance}

\begin{proof}[Proof of \zcref{prop:high-girth-phase-balance}]
    We now introduce the hard-core polymer representation used in the proof
    \cite[Section~4.1]{JKP20}.  Let $G^2$ join vertices whose distance in $G$
    is at most two.  An $R$-polymer is a nonempty set $\gamma\subseteq R$
    such that $G^2[\gamma]$ is connected and $|\gamma|\leq |R|/2$.  Its
    weight is
    \begin{equation*}
        w_\gamma
        =\frac{\lambda^{|\gamma|}}
          {(1+\lambda)^{|N_G(\gamma)|}},
    \end{equation*}
    and two $R$-polymers are compatible when their distance in $G$ is greater
    than two.  Let $\mathfrak G_R$ be the collections of pairwise compatible
    $R$-polymers, including the empty collection, and define
    \begin{equation*}
        \Xi^R_\lambda(G)
        =\sum_{\Gamma\in\mathfrak G_R}
          \prod_{\gamma\in\Gamma}w_\gamma.
    \end{equation*}
    Define $L$-polymers, $\mathfrak G_L$, and $\Xi^L_\lambda(G)$ symmetrically.

    To relate these quantities to the phase weights in
    \eqref{eq:hard-core-phase-partition-functions}, fix
    $\Gamma\in\mathfrak G_R$ and put
    $B=\bigcup_{\gamma\in\Gamma}\gamma$.  Compatibility makes the neighbor
    sets $N_G(\gamma)$ disjoint.  The total weight of independent sets with
    $I\cap R=B$ is therefore
    \begin{equation*}
        \lambda^{|B|}(1+\lambda)^{|L|-|N_G(B)|}
        =(1+\lambda)^{|L|}\prod_{\gamma\in\Gamma}w_\gamma.
    \end{equation*}
    Since the polymers in $\Gamma$ are exactly the $G^2$-connected components
    of $B$, summing over $\Gamma$ gives
    \begin{equation*}
        W^L_\lambda=(1+\lambda)^{|L|}\Xi^R_\lambda(G),
        \qquad
        W^R_\lambda=(1+\lambda)^{|R|}\Xi^L_\lambda(G).
    \end{equation*}

    Expansion gives
    \begin{equation}
        w_\gamma\leq(1+\lambda)^{-\alpha|\gamma|}.
        \label{eq:high-girth-polymer-decay}
    \end{equation}

    We first verify the quantitative cluster-expansion tail.  Take the decay
    function $g(\gamma)=a|\gamma|$.  The graph $G^2$ has maximum degree at
    most $d^2$, so the number of $G^2$-connected sets of size $t$ containing
    a fixed vertex is at most $(ed^2)^t$.  By
    \eqref{eq:high-girth-polymer-decay}, the rooted sum in the
    Koteck\'y--Preiss condition is at most
    \begin{equation*}
        \sum_{t\geq1}
        \left[e^{2+a}d^2(1+\lambda)^{-\alpha}\right]^t
        \leq \frac1{d^2},
    \end{equation*}
    where the last inequality follows from
    \eqref{eq:high-girth-kp-margin}.  There are at most $d^2|\gamma|$
    possible roots within $G$-distance two of a fixed polymer $\gamma$.
    Hence the Koteck\'y--Preiss condition holds with this decay function.

    Let $T^L_{<\ell}$ and $T^R_{<\ell}$ be the cluster expansions of
    $\log\Xi^L_\lambda$ and $\log\Xi^R_\lambda$, truncated to clusters
    $\Gamma$ of total polymer size
    $|\Gamma|=\sum_{\gamma\in\Gamma}|\gamma|<\ell$.  The exponential
    cluster bound of \cite[Theorem~6 and Fact~7]{JKP20}, applied with
    $g(\Gamma)=a|\Gamma|$, yields
    \begin{equation}
        \left|\log\Xi^L_\lambda-T^L_{<\ell}\right|
        \leq Ne^{-a\ell},
        \qquad
        \left|\log\Xi^R_\lambda-T^R_{<\ell}\right|
        \leq Ne^{-a\ell}.
        \label{eq:high-girth-cluster-tail}
    \end{equation}

    It remains to compare the two truncated sums.  For a cluster $\Gamma$ on
    one side, let $U$ be the union of its polymer supports.  Connectivity of
    the cluster incompatibility graph makes $U$ connected in $G^2$.  Moreover,
    $|U|\leq|\Gamma|<\ell$.  Any cycle in the incidence graph
    $G[U\cup N_G(U)]$ would have length at most $2|U|<2\ell$, contradicting
    the girth assumption.  Thus this incidence graph is a bipartite tree.

    Retain on this tree the ordered polymer supports, including repetitions.
    This marked tree determines every polymer activity and the cluster
    incompatibility graph, and hence determines the full cluster summand,
    including its Ursell factor.  In a $d$-regular graph, the number of
    embeddings of any such marked tree rooted at a prescribed support vertex
    depends only on $d$ and the marked tree.  Since $|L|=|R|$, counting rooted
    embeddings on the two sides and then dividing by the number of possible
    roots in the support gives the same number of unrooted cluster occurrences.
    Therefore
    \begin{equation}
        T^L_{<\ell}=T^R_{<\ell}.
        \label{eq:high-girth-truncated-cancellation}
    \end{equation}

    Because the graph is balanced, the factors $(1+\lambda)^{|L|}$ and
    $(1+\lambda)^{|R|}$ in the two phase partition functions agree.
    Combining \eqref{eq:high-girth-cluster-tail} and
    \eqref{eq:high-girth-truncated-cancellation} proves
    \begin{equation*}
        \left|\log\frac{W^L_\lambda}{W^R_\lambda}\right|
        =\left|\log\frac{\Xi^R_\lambda}{\Xi^L_\lambda}\right|
        \leq2Ne^{-a\ell}.
    \end{equation*}
    Finally, if $x=\log(W^L_\lambda/W^R_\lambda)$, then
    \begin{equation*}
        \left|\frac{W^L_\lambda}{W^L_\lambda+W^R_\lambda}-\frac12\right|
        =\frac12\left|\tanh\frac x2\right|
        \leq\frac{|x|}{4},
    \end{equation*}
    which proves the second bound.
\end{proof}

\subsection{Proofs for computable unequal phase weights}
\label{subsec:appendix-polymer-component-bounds}

\begin{proof}
    \medskip\noindent\textbf{Phase domains and target approximation.}
    Following
    \cite[Section~4.1]{JKP20}, call a subset of one bipartition class
    \emph{sparse} when each of its $G^2$-connected components has size at
    most half that class, and define the fixed phase domains
    \begin{equation*}
        \Omega_L
        =\{I\in\mathcal I(G):I\cap R\text{ is sparse}\},
        \qquad
        \Omega_R
        =\{I\in\mathcal I(G):I\cap L\text{ is sparse}\}.
    \end{equation*}
    They cover the independent-set space and may overlap.  Along a reverse
    log-fugacity schedule from a polynomially large $\lambda_0$ down to
    $\lambda_*$, let
    $\mu_{R,s}$ and $\mu_{L,s}$ be the hard-core
    laws conditioned on these domains.  Summing over the unrestricted side
    shows that their normalizing constants are exactly the two quantities in
    \eqref{eq:hard-core-phase-partition-functions}.  A phase law is therefore
    the conditional law required by
    \eqref{eq:conditional-component-path}.

    Each phase domain is connected under legal single-site updates.  Indeed,
    from a configuration in phase $i$, first delete all occupied vertices on
    the opposite side; these deletions preserve sparsity.  Then add the
    missing vertices on the dominant side, which is legal because the
    opposite side is empty.  This path stays in $\Omega_i$ and reaches the
    $i$th ground state.  Reversibility of the unrestricted heat-bath kernel
    with respect to the hard-core law therefore gives the unique internally
    reversible law on each phase domain, as required by the componentwise
    theorem.

    The two domains cover $\mathcal I(G)$ and double-count only their
    overlap.  Equivalently,
    $W^L_{\lambda}+W^R_{\lambda}$ is the true
    partition function $Z_\lambda$ plus the total weight $R_\lambda$ of the
    overlap.  Hence the normalized phase mixture differs from the true hard-core
    law by at most $R_\lambda/Z_\lambda$ in total variation.  The proof of
    \cite[Lemma~20]{JKP20} identifies this overlap contribution and bounds it
    by $e^{-\Omega_{\alpha,\Delta}(N)}$ times the true partition function.
    Define the normalized phase mixture
    \begin{equation*}
        \pi_{\rm mix}
        \defeq w_L\mu_{L,1}+w_R\mu_{R,1},
        \qquad
        D_{\mathrm{TV}}\tp{\mu_{G,\lambda_*},\pi_{\rm mix}}
        \leq e^{-\Omega_{\alpha,\Delta}(N)}.
    \end{equation*}
    \medskip\noindent\textbf{Boundary leakage.}
    Deletions preserve sparsity.  If adding a single vertex to, say,
    $I\cap R$ destroys sparsity, then before the addition the
    components joined by that vertex have total size at least
    $|R|/2-O(1)$.  Thus the current defect set has linear size,
    because bipartite expansion implies
    $\min\{|L|,|R|\}\geq N/3$.

    For a sparse defect set $A$, its phase weight relative to the empty
    defect is
    \begin{equation*}
        \frac{\lambda^{|A|}}{(1+\lambda)^{|N(A)|}}
        \leq(1+\lambda)^{-\alpha|A|},
    \end{equation*}
    where expansion is applied component by component.  The exit boundary has
    defect size at least $c_\alpha N-O_\Delta(1)$.  Summing the last bound over
    at most $2^N$ defect sets, and enlarging
    $\lambda_{\min}$ so that the energetic exponent dominates this
    counting factor, shows that the phase mass of the exit boundary is
    $e^{-\Omega_{\alpha,\Delta}(N)}$, uniformly for
    $\lambda\geq\lambda_*$.  Since the total jump rate is at most one, this
    also bounds the outward flow.  Consequently,
    \begin{equation*}
        \Lambda_L+\Lambda_R
        \leq e^{-\Omega_{\alpha,\Delta}(N)}.
    \end{equation*}

    \medskip\noindent\textbf{Componentwise action bounds.}
    We next construct the componentwise flux--residual pairs.  For
    $i\in\{L,R\}$, let $V_i$ be the dominant side and
    $V_{\bar i}$ the opposite side, and define the full deviation set
    \begin{equation*}
        D_i(I)=(V_i\setminus I)\cup(I\cap V_{\bar i}).
    \end{equation*}
    Let $G_{i,s}\subseteq\Omega_i$ consist of configurations for which every
    connected component of the induced graph $G[D_i(I)]$ has size at most
    \begin{equation}
        M_s=\max\left\{1,
          \left\lceil
          \frac{C_{\alpha,\Delta}\log(N/\xi)}{\log\lambda_s}
          \right\rceil\right\}.
        \label{eq:jkp-component-core-cutoff}
    \end{equation}
    This dependence on $\log\lambda_s$ is important: it makes
    $\lambda_s^{M_s+1}$ polynomial even when $\lambda_0$ is polynomially
    large.

    For all sufficiently large $N$, $M_s<\min_j|V_j|/4$.  Thus a
    configuration whose full-deviation components have size at most $M_s$
    automatically belongs to the $i$th sparse phase, and $G_{i,s}$ is exactly
    the restricted state space $\Omega^i_{M_s}(G)$ used in
    \cite[Theorem~29]{CGGPSV21}.  For a connected full-deviation polymer
    $\gamma$ in the $i$th phase,
    expansion gives
    \begin{equation*}
        \mathsf w_s(\gamma)
        \leq
        \lambda_s^{-\alpha|\gamma|/(2+\alpha)}.
    \end{equation*}
    Restoring the ground-state spins on one full-deviation component maps an
    $i$-phase configuration injectively to another $i$-phase configuration;
    sparsity of the opposite-side defect is hereditary under this operation.
    Consequently, even after conditioning on $\Omega_i$, the probability that
    a specified component $\gamma$ occurs is at most $\mathsf w_s(\gamma)$.
    The standard connected-set count for the bounded-degree graph $G$,
    followed by a union bound over roots and sizes larger than $M_s$, gives
    \begin{equation}
        b_{i,s}\defeq\mu_{i,s}(G_{i,s}^{\mathsf c})\leq\xi
        \qquad(i\in\{L,R\}),
        \label{eq:jkp-component-core-tail}
    \end{equation}
    after increasing the high-fugacity threshold.  Notice that the phase
    partition function contains its ground-state term, so this estimate is
    uniform in the phase weight and does not require a lower bound on
    $w_i$.

    The proof of \cite[Theorem~29]{CGGPSV21} applies its Theorem~25 with its
    local weight-ratio parameter set to $\lambda_s$ and gives, at fixed mixing
    accuracy,
    \begin{equation*}
        T_{\rm mix}
        =O\tp{M_s\lambda_s^{M_s+1}N^2
          \log N\log\lambda_s}.
    \end{equation*}
    Converting this bound to a spectral gap and comparing the cited
    restricted Metropolis conductances with the larger heat-bath conductances
    gives
    \begin{equation}
        \sup_{i,s}C_{\-{PI}}\tp{
          P_{G,\lambda_s}\text{ restricted to }G_{i,s}}
        \leq Q^{\mathrm{PI}}_{\alpha,\Delta}(N,1/\xi,\lambda_0),
        \label{eq:jkp-component-core-pi}
    \end{equation}
    where $Q^{\mathrm{PI}}_{\alpha,\Delta}$ is a polynomial.
    The conversion and conductance normalization are recorded explicitly in
    \zcref{subsec:appendix-symmetric-phase-bound}.  The only apparently
    nonpolynomial factor in the comparison bound is
    $\lambda_s^{M_s+1}$, which is polynomial by
    \eqref{eq:jkp-component-core-cutoff}.

    Put $\widehat\mu_{i,s}=\mu_{i,s}(\cdot\mid G_{i,s})$.  On every interval
    on which $M_s$ is constant, decompose
    \begin{equation*}
        \mu_{i,s}
        =(1-b_{i,s})\widehat\mu_{i,s}
          +\mu_{i,s}\mathbf1_{G_{i,s}^{\mathsf c}}
    \end{equation*}
    and transport the part
    $(1-b_{i,s})\partial_s\widehat\mu_{i,s}$ inside $G_{i,s}$.  The
    conditional score is
    \begin{equation*}
        \partial_s\log\widehat\mu_{i,s}(I)
        =-\mathsf B\tp{|I|-\mathbb E_{\widehat\mu_{i,s}}|I|},
        \qquad
        \mathsf B=\log(\lambda_0/\lambda_*),
    \end{equation*}
    and has range at most $N\mathsf B$.  If $\widehat J_{i,s}$ is the PI
    flux for $\partial_s\widehat\mu_{i,s}$, use
    $J_{i,s}=(1-b_{i,s})\widehat J_{i,s}$.  Since the conditional capacity is
    $\widehat c_{i,s}=c_{i,s}/(1-b_{i,s})$ on $G_{i,s}$, this scaling
    multiplies the conditional quadratic cost by $1-b_{i,s}\leq1$.

    The remaining part of the path derivative is
    \begin{equation*}
        r_{i,s}
        =\partial_s\tp{\mu_{i,s}\mathbf1_{G_{i,s}^{\mathsf c}}}
          -(\partial_sb_{i,s})\widehat\mu_{i,s}.
    \end{equation*}
    The unconditioned phase score is bounded in absolute value by
    $N\mathsf B$, and hence
    \begin{equation*}
        \frac12\norm{r_{i,s}}_1
        \leq N\mathsf B\,b_{i,s}.
    \end{equation*}
    The pointwise version of
    \zcref{prop:PI-implies-action}, together with
    \eqref{eq:jkp-component-core-tail} and
    \eqref{eq:jkp-component-core-pi}, gives
    \begin{equation*}
        A_i\leq Q_{\alpha,\Delta}(N,1/\xi,\lambda_0)\mathsf B^2,
        \qquad
        \varepsilon_i\leq C_{\alpha,\Delta}N\mathsf B\xi.
    \end{equation*}
    where $Q_{\alpha,\Delta}$ absorbs the factor $N^2$ into the polynomial
    $Q^{\mathrm{PI}}_{\alpha,\Delta}$.
    The integer $M_s$ changes only finitely many times.  The flux and residual
    are defined from the smooth phase derivative on the intervals between these
    exceptional times, which form a null set and create no additional term in
    the integrated action bound.

    \medskip\noindent\textbf{Initialization.}
    The ground state of phase $i$ is the independent set $V_i$ defined above,
    so $\delta_{V_i}$ is $\delta_L$ or $\delta_R$.  The exact
    partition-ratio identity is
    \begin{equation*}
        \frac{W^i_{\lambda_0}}{\lambda_0^{|V_i|}}
        =\tp{1+\lambda_0^{-1}}^{|V_i|}
          \Xi^{\bar i}_{\lambda_0}.
    \end{equation*}
    Every polymer in $\Xi^{\bar i}_{\lambda_0}$ has activity at most
    $(1+\lambda_0)^{-\alpha|\gamma|}$.  The rooted connected-set count in the
    bounded-degree graph $G^2$ therefore yields, after increasing the
    high-fugacity threshold,
    \begin{equation*}
        \log\Xi^{\bar i}_{\lambda_0}
        \leq C_{\alpha,\Delta}
          \frac{N}{(1+\lambda_0)^\alpha}.
    \end{equation*}
    Since the ground-state atom has phase probability
    $\lambda_0^{|V_i|}/W^i_{\lambda_0}$ and
    $1-x^{-1}\leq\log x$ for $x\geq1$, it follows that
    \begin{equation*}
        d_i
        =D_{\mathrm{TV}}\tp{\mu_{i,0},\delta_{V_i}}
        \leq
        C_{\alpha,\Delta}\left(
          \frac N{\lambda_0}
          +\frac N{(1+\lambda_0)^\alpha}
        \right).
    \end{equation*}
    The first term accounts for holes in the dominant side and the second
    for nonempty opposite-side polymers.  Thus a polynomially large
    $\lambda_0$ makes every $d_i$ as small as required.
    Combining the target approximation, outward-flow, action, residual, and
    initialization estimates above proves
    \eqref{eq:polymer-target-and-leakage-bounds} and
    \eqref{eq:polymer-component-bounds}.
\end{proof}

\begin{proof}[Proof of \zcref{thm:computed-phase-weight-annealing}]
    Run the phase-weight approximation with accuracy $\gamma=O(\tau)$.  Then
    $\delta_w\leq e^{2\gamma}-1\leq\tau/8$ after adjusting the constant.
    Initialize component $L$ from $\kappa_L=\delta_L$ and component $R$ from
    $\kappa_R=\delta_R$.  By \zcref{prop:polymer-component-bounds}, we may
    choose a polynomially large $\lambda_0$ so that $d_i\leq\tau/8$, and then
    choose an inverse-polynomial $\xi$ so that
    $\varepsilon_i\leq\tau/8$.

    Choose $T=\operatorname{poly}_{\alpha,\Delta}(N,1/\tau)$ so that
    $\sqrt{A_i/(8T)}\leq\tau/8$ for both components.  Since the outward flow
    is exponentially small, the same polynomial $T$ also satisfies
    $T(\Lambda_L+\Lambda_R)\leq\tau/8$.  For all sufficiently large $N$, the
    target-approximation term in
    \eqref{eq:polymer-target-and-leakage-bounds} is also at most $\tau/8$.
    Apply the componentwise theorem with reference path
    $\rho_s=\mu_{G,\lambda_s}$ and target $\pi=\pi_{\rm mix}$.
    Substitution into \eqref{eq:componentwise-weak-action-tv-init} gives
    $D_{\mathrm{TV}}(\pi_{\rm mix},\nu_T)\leq7\tau/8$.  The target
    approximation in \eqref{eq:polymer-target-and-leakage-bounds} and the
    triangle inequality then give
    $D_{\mathrm{TV}}(\mu_{G,\lambda_*},\nu_T)\leq\tau$.
\end{proof}

\subsection{Proof of annealing from approximate phase balance}
\label{subsec:appendix-approximate-balance-annealing}

\begin{proof}[Proof of \zcref{thm:approximate-balance-annealing}]
    Let $w=(w_L,w_R)$ be the normalized phase weights at fugacity
    $\lambda_*$, and let $\widehat w=(1/2,1/2)$ be the weights used by the
    initializer.  By \zcref{prop:high-girth-phase-balance} and the assumed
    bound on $Ne^{-a\ell}$,
    \begin{equation*}
        \delta_w=D_{\mathrm{TV}}(w,\widehat w)
        \leq\frac N2e^{-a\ell}
        \leq\frac\tau8.
    \end{equation*}

    Apply \zcref{prop:polymer-component-bounds}.  Initialize the two branches from
    $\kappa_L=\delta_L$ and $\kappa_R=\delta_R$.  First choose
    $\lambda_0=\operatorname{poly}_{\alpha,\Delta}(N,1/\tau)$ large enough
    that the initialization errors $d_L,d_R$ are at most $\tau/8$.  Then
    choose an inverse-polynomial $\xi$ small enough that the residual errors
    are at most $\tau/8$.  The resulting action bounds are polynomial, so a
    polynomial $T$ makes their contribution in
    \eqref{eq:componentwise-weak-action-tv-init} at most $\tau/8$.

    The integrated outward flow and the distance between the phase mixture
    and $\mu_{G,\lambda_*}$ are exponentially small in $N$.  Since $T$ is
    polynomial, for all sufficiently large $N$ their contributions are also
    at most $\tau/8$.  Apply the componentwise theorem with reference path
    $\rho_s=\mu_{G,\lambda_s}$ and target $\pi=\pi_{\rm mix}$.
    Substituting these estimates into
    \eqref{eq:componentwise-weak-action-tv-init} and then using the triangle
    inequality proves the claimed TV bound.  The marginal initializer is
    exactly $(\delta_L+\delta_R)/2$, and all subsequent updates use the same
    unrestricted heat-bath chain, as stated.
\end{proof}

\subsection{Proofs for the fixed-gap unbalanced application}
\label{subsec:appendix-fixed-gap-bound}

Let $C_{\mathrm{count}},c_{\mathrm{count}}>0$ be the absolute constants in
\cite[Lemma~4]{JPP23}.  For fixed $\delta>0$, choose
$\eta_0=\eta_0(\delta)>0$ so that
\begin{equation*}
    \log_2(1+\eta_0)
    <c_{\mathrm{count}}\frac{\delta}{1+\delta}.
\end{equation*}
The next three proofs treat $\eta\in[0,\eta_0)$ simultaneously.  Their
hard-core cases prove
\zcref{prop:fixed-gap-phase-core-bound,prop:fixed-gap-weak-action,
lem:fixed-gap-initialization}; together, their soft-core cases prove
\zcref{prop:fixed-gap-soft-core-bounds}.

\begin{proof}[Proof of \zcref{prop:fixed-gap-phase-core-bound}]
    We prove the proposition and, at the same time, the soft-core estimate
    needed for \zcref{thm:fixed-gap-soft-core-sampling}.
    Fix $\eta\in[0,\eta_0)$ and put
    \begin{equation*}
        \kappa_\delta=\frac{\delta}{1+\delta},
        \qquad
        \ell_\eta=\log_2(1+\eta).
    \end{equation*}
    Fix $B\subseteq R$ and write $k_u(B)=|N_G(u)\cap B|$.  Summing the
    $L$-spins exactly gives
$$
    \sum_{A\subseteq L}\eta^{e(A,B)}
    =\prod_{u\in L}(1+\eta^{k_u(B)}).
$$
After removing the constant factor $2^{|L|}$, the weight of $B$ is
\begin{equation}
    \lambda^{|B|}
    \prod_{u\in N_G(B)}
       \frac{1+\eta^{|N_G(u)\cap B|}}{2}.
    \label{eq:soft-core-marginal-weight}
\end{equation}
If $\gamma_1,\ldots,\gamma_q$ are the components of $G^2[B]$, their
$L$-neighborhoods are disjoint, so this weight factorizes into a positive
polymer gas.  A contour with $g=|N_G(\gamma)|$ has activity at most
\begin{equation*}
    \lambda^{|\gamma|}\left(\frac{1+\eta}{2}\right)^g
    =\lambda^{|\gamma|}2^{-(1-\ell_\eta)g}.
\end{equation*}

For a contour $\gamma$, set $A=N_G(\gamma)$ and let $C=[\gamma]=\{v\in R:N_G(v)\subseteq A\}$ be its closure.  Every $v\in C$ sends all its $D$ edges into $A$, while the sum of degrees in $A$ is $d|A|$.  Thus, with $g=|A|$, $a=|C|$, and $t=g-a$,
\begin{equation}
    dg=Da+|E(A,R\setminus C)|,
    \qquad
    t\geq\delta a,
    \qquad
    t\geq\kappa_\delta g.
    \label{eq:appendix-closure-excess}
\end{equation}
Moreover, every nonempty contour has $g\geq D$.

The rooted expanding-set estimate of \cite{JPP23} is stated for a balanced regular graph.  We apply it through the following closure-preserving completion.  Put $\ell=|L|$, take $D$ disjoint copies of $G$, and set $q_{\mathrm{def}}=\ell(D-d)$.  Attach a distinct new $R$-vertex to every deficit slot at every old $L$-vertex; this creates $q_{\mathrm{def}}D$ new $R$-vertices.  Partition them into $q_{\mathrm{def}}$ groups of size $D$, create $q_{\mathrm{def}}$ corresponding groups of $D-1$ new $L$-vertices, and place a copy of $K_{D,D-1}$ between every pair of corresponding groups.  The resulting graph is simple and $D$-regular.  Its two sides have the same size because
$$
    \bigl(D|R|+q_{\mathrm{def}}D\bigr)
    -\bigl(D|L|+q_{\mathrm{def}}(D-1)\bigr)
    =D|R|-d|L|=0.
$$
No old $R$-vertex receives a new edge, while every new $R$-vertex has a new $L$-neighbor.  Hence the neighborhood, closure, and $2$-link connectivity of every contour in a designated old copy are unchanged.  This completion is used only for the combinatorial count.

More precisely, the cited estimate gives, for all sufficiently large $D$, the
bound below whenever
$$
    t\geq C_{\mathrm{count}}\frac{\log^2D}{D}g.
$$
Since \eqref{eq:appendix-closure-excess} gives $t\geq\kappa_\delta g$, this hypothesis holds for every contour once $D$ is large enough depending on $\delta$.  We obtain
\begin{equation}
    \#\{\gamma\ni v: |N_G(\gamma)|=g,
             |[\gamma]|=a\}
    \leq2^{g-c_{\mathrm{count}}t}.
    \label{eq:appendix-jpp-rooted-count}
\end{equation}

Put
$$
    a_0=c_{\mathrm{count}}
       -\frac{\ell_\eta}{\kappa_\delta}>0.
$$
For fixed $(g,a)$, the product of the count in
\eqref{eq:appendix-jpp-rooted-count} and the activity bound is at most
$$
    2^{\ell_\eta g-c_{\mathrm{count}}t}
    \leq2^{-a_0t},
$$
where the last inequality uses $g\leq t/\kappa_\delta$.  Also
$|\gamma|\leq a\leq g\leq t/\kappa_\delta$, and, for fixed $t$, there are at most $t/\kappa_\delta+1$ possible pairs $(g,a)$ with $t=g-a$.  Since every nonempty contour has
$t\geq\kappa_\delta D$, summing the resulting geometric series gives, for $j\in\{0,1\}$,
\begin{equation}
    \max_{v\in R}
    \sum_{\gamma\ni v}
    |\gamma|^j w_\lambda(\gamma)
    2^{a_0t(\gamma)/2}
    \leq C_{\delta,\eta}e^{-c_{\delta,\eta}D}
    \label{eq:appendix-rooted-moment}
\end{equation}
uniformly in $\lambda\in[0,1]$.  In a positive polymer gas, a specified polymer occurs with probability at most its activity: after fixing that polymer, the remaining compatible partition function is at most the unrestricted one.  If a configuration is outside $\mathcal A_{m_*}$, some contour has $g>m_*$ and hence $t>\kappa_\delta m_*$.  A union bound over a root of that contour and \eqref{eq:appendix-rooted-moment} therefore give
$$
    \pi_{\lambda,\eta}(\mathcal A_{m_*}^{\mathsf c})
    \leq C_{\delta,\eta}e^{-c_{\delta,\eta}D}N
       2^{-a_0\kappa_\delta m_*/2}.
$$
Increasing $D_0$ makes $C_{\delta,\eta}e^{-c_{\delta,\eta}D}\leq1$, proving \eqref{eq:fixed-gap-tail} after converting the remaining power of two to an exponential.

It remains to prove a PI for the original heat-bath chain restricted to the core.  If $\lambda=0$, its stationary support has $B=\varnothing$, and the claim follows directly from the product-chain estimate used below.  We therefore assume $\lambda>0$ and make three comparisons.  Let $\mathcal P_*$ be the set of contours with $g(\gamma)\leq m_*$.  If $\mathcal P_*$ is empty, the core again has the single marginal state $B=\varnothing$, so the same product-chain bound proves the claim.  We may therefore assume that $\mathcal P_*$ is nonempty.  First consider the heat-bath chain on compatible subfamilies of $\mathcal P_*$.  It chooses $\gamma\in\mathcal P_*$ uniformly and resamples its presence conditional on all other polymers.  Equip polymer configurations with the metric
$$
    d(\Gamma,\Gamma')
    =\sum_{\gamma\in\Gamma\mathbin\triangle\Gamma'}|\gamma|.
$$
Suppose that $\Gamma$ and $\Gamma'$ differ only by a polymer $\alpha$.  Under the common-update coupling, choosing $\alpha$ removes the discrepancy.  A new discrepancy can otherwise be created only by inserting a polymer incompatible with $\alpha$, and its insertion probability is at most its activity.  Incompatibility with $\alpha$ forces the new polymer to contain a vertex of $N_{G^2}[\alpha]$.  Since
$|N_{G^2}[\alpha]|\leq(1+D^2)|\alpha|$, \eqref{eq:appendix-rooted-moment} with $j=1$ implies, after increasing $D_0$, that
$$
    \sum_{\gamma:\,\gamma\not\sim\alpha}
       |\gamma|w_\lambda(\gamma)
    \leq\frac{|\alpha|}{2}.
$$
Consequently
$$
    \mathbb E d(\Gamma_1,\Gamma'_1)
    \leq
    \left(1-\frac1{2|\mathcal P_*|}\right)
    d(\Gamma,\Gamma').
$$
Path coupling extends this contraction to every pair.  Applying it to an eigenfunction through its Lipschitz seminorm shows that the polymer heat-bath chain has spectral gap at least
$1/(2|\mathcal P_*|)$.

We next compare this block chain with the restricted single-vertex heat-bath chain for the $R$-marginal.  Let $\pi_R$ denote this marginal conditioned on $\mathcal A_{m_*}$.  The ideal marginal chain chooses each $v\in R$ with probability $1/N$ (and holds with the remaining probability), heat-bath resamples $v$ under $\pi_R$, and rejects moves leaving the core.  The crude bound obtained by dropping the saving in
\eqref{eq:appendix-jpp-rooted-count} is
\begin{equation}
    M_*
    \defeq
    \max_{v\in R}
    \#\{\gamma\in\mathcal P_*:v\in\gamma\}
    \leq m_*^2 2^{m_*}.
    \label{eq:appendix-retained-root-count}
\end{equation}
Choose, once and for all, an order of the vertices of every contour in which each prefix is connected in $G^2$.  Simulate its insertion by adding the vertices in this order, and simulate its deletion by reversing the path.  Every intermediate configuration remains in $\mathcal A_{m_*}$.  If $v\notin B$, the marginal one-vertex weight ratio is
$$
    r(B,v)
    =\lambda\prod_{u\in N_G(v)}
       \frac{1+\eta^{k_u(B)+1}}
            {1+\eta^{k_u(B)}}
    \leq \lambda\leq1.
$$
Thus marginal stationary weights do not increase during insertion.  The conductance of a legal single-vertex edge
$e=(B,B\cup\{v\})$ is at least
$\pi_R(B\cup\{v\})/(2N)$, whereas the conductance of a polymer insertion is at most the stationary mass of its final configuration divided by $|\mathcal P_*|$.  Once $e$ and the inserted contour are fixed, the background polymer configuration is determined.  Each comparison path has length at most $m_*$, and at most $M_*$ contours can use $e$.

To state the comparison precisely, let $c_{\mathrm{blk}}$ and $c_R$ be the conductances of the polymer block chain and the restricted single-vertex marginal chain.  Route every block transition along the single-vertex path just described and define
$$
    \rho_*
    =\max_e\frac1{c_R(e)}
    \sum_{\{\Gamma,\Gamma'\}:\,e\in\gamma_{\Gamma,\Gamma'}}
    c_{\mathrm{blk}}(\Gamma,\Gamma')
    \abs{\gamma_{\Gamma,\Gamma'}}.
$$
Cauchy--Schwarz along each path gives
$\mathcal E_{\mathrm{blk}}(f,f)\leq\rho_*\mathcal E_R(f,f)$ for every $f$, and hence $\mathrm{Gap}_R\geq\mathrm{Gap}_{\mathrm{blk}}/\rho_*$.  The conductance and counting estimates above give
$$
    \rho_*
    \leq\frac{2Nm_*M_*}{|\mathcal P_*|}.
$$
Since the preceding coupling gives
$\mathrm{Gap}_{\mathrm{blk}}\geq1/(2|\mathcal P_*|)$, \eqref{eq:appendix-retained-root-count} now gives the restricted marginal gap
\begin{equation}
    \mathrm{Gap}_R
    \geq
    \frac1{4Nm_*M_*}
    \geq
    \frac1{4Nm_*^3 2^{m_*}}.
    \label{eq:appendix-restricted-marginal-gap}
\end{equation}
Finally partition the core by its occupied set $B\subseteq R$.  Conditional on $B$, the $L$-spins are independent: their occupation probabilities are $\eta^{k_u(B)}/(1+\eta^{k_u(B)})$ in the soft model, while in the hard-core model the unblocked spins are fair and the blocked spins are fixed to zero.  Random-scan heat bath within every nontrivial fiber therefore has gap at least $1/N$.

It remains to compare the projected heat-bath chain with the ideal marginal chain used in \eqref{eq:appendix-restricted-marginal-gap}.  For a legal addition of $v$, put
$X=\lambda\eta^{|N_G(v)\cap A|}$ under the conditional law of $A$ given $B$.  Then $0\leq X\leq \lambda\leq1$ and the marginal ratio is $r(B,v)=\mathbb E[X]$.  The projected heat-bath addition probability is
$$
    \frac1N\mathbb E\frac{X}{1+X}
    \geq\frac{r(B,v)}{N(1+\lambda)}
    \geq\frac12\frac{r(B,v)}{N(1+r(B,v))}.
$$
The same flow comparison holds for deletion by reversibility.  Writing the
projected heat-bath gap as $\mathrm{Gap}_{\mathrm{proj}}$, we therefore have
$\mathrm{Gap}_{\mathrm{proj}}\geq\mathrm{Gap}_R/2$.

We justify the bound $\mathrm{Gap}_{\mathrm{proj}}\leq1$ used below.  A random-scan heat-bath matrix is an average of conditional-expectation projections and is therefore positive semidefinite.  Restricting the chain to a set preserves this property: extending a function by zero, the restricted quadratic form is the original quadratic form plus the nonnegative holding contribution from rejected exits.  Lifting a function on the set of fibers to a function constant on each fiber then identifies the projected quadratic form with that of the original heat-bath chain.  Thus both the restricted heat-bath chain and its projection are positive semidefinite, and $\mathrm{Gap}_{\mathrm{proj}}\leq1$.

In the notation of \cite[Theorem~1]{JSTV04}, the largest probability of leaving a fiber is at most $\gamma=1$, while the smallest restriction gap is at least $1/N$.  Their decomposition theorem yields
$$
    \mathrm{Gap}_{\mathrm{full}}
    \geq
    \min\left\{
       \frac{\mathrm{Gap}_{\mathrm{proj}}}{3},
       \frac{\mathrm{Gap}_{\mathrm{proj}}/N}
            {3+\mathrm{Gap}_{\mathrm{proj}}}
    \right\}
    \geq\frac{\mathrm{Gap}_{\mathrm{proj}}}{4N}
    \geq\frac1{32N^2m_*^3 2^{m_*}}.
$$
Taking the reciprocal proves \eqref{eq:fixed-gap-restricted-pi}.
\end{proof}

\begin{proof}[Proof of \zcref{prop:fixed-gap-weak-action}]
    We give the calculation simultaneously for $\eta=0$ and for fixed
    $\eta\in(0,\eta_0)$, using $\pi_s=\pi_{\lambda_s,\eta}$ in the latter case.
    Write $\mathcal A=\mathcal A_{m_*}$ and set
    \begin{equation*}
        b_*=Ne^{-cm_*},
        \qquad
        C_*=32N^2m_*^3 2^{m_*}.
    \end{equation*}
    Let $b_s=\pi_s(\mathcal A^{\mathsf c})$ and
    $\mu_s=\pi_s(\cdot\mid\mathcal A)$.  Decompose
    \begin{equation*}
        \pi_s=(1-b_s)\mu_s+\pi_s\mathbf1_{\mathcal A^{\mathsf c}}
    \end{equation*}
    and transport the part
    $u_s=(1-b_s)\partial_s\mu_s$ inside $\mathcal A$.  Differentiating along
    \eqref{eq:fixed-gap-log-path} gives
    \begin{equation*}
        \partial_s\log\pi_s(A,B)
        =\mathsf B\tp{|B|-\mathbb E_{\pi_s}|B|},
    \end{equation*}
    so the conditional score has range at most $|R|\mathsf B$.  By
    \eqref{eq:fixed-gap-restricted-pi} and the pointwise version of
    \zcref{prop:PI-implies-action}, the integrated quadratic cost is at most
    $C_*|R|^2\mathsf B^2$.

    The remaining part of the path derivative is
    \begin{equation*}
        r_s
        =\partial_s\tp{\pi_s\mathbf1_{\mathcal A^{\mathsf c}}}
          -(\partial_sb_s)\mu_s.
    \end{equation*}
    By the same score identity, the first term has $L^1$ norm at most
    $|R|\mathsf Bb_*$.  The absolute value of $\partial_sb_s$ is bounded by the
    same quantity, so
    \begin{equation*}
        \frac12\norm{r_s}_1\leq |R|\mathsf Bb_*
    \end{equation*}
    for every $s$.  Integrating over $s\in[0,1]$ proves
    \eqref{eq:fixed-gap-weak-action-bound}.
\end{proof}

\begin{proof}[Proof of \zcref{lem:fixed-gap-initialization}]
    We again give the construction for every $\eta\in[0,\eta_0)$; setting
    $\eta=0$ proves the stated lemma.
    Let $\rho_{\rm ref}$ be the law with $B=\varnothing$ in which each
    $L$-vertex is occupied independently with probability $1/2$.  For soft
    core, let $\chi$ be uniform on the full spin space.  For hard core, sample
    a uniformly random vertex subset and return it if it is independent, and
    return the empty set otherwise.  Set
    \begin{equation*}
        \theta=\frac\zeta4,
        \qquad
        \nu_0=(1-\theta)\rho_{\rm ref}+\theta\chi.
    \end{equation*}
    Choose
    \begin{equation}
        \lambda_0
        \leq
        \min\left\{\frac{1}{2|R|},
        \frac{\zeta}{8|R|(N\log2+\log(4/\zeta))}\right\}.
        \label{eq:fixed-gap-initial-lambda}
    \end{equation}
    Every state in the relevant state space has $\chi$-mass at least $2^{-N}$,
    so $\nu_0$ has full support.  Since all interaction weights are at most one,
    the marginal formula \eqref{eq:soft-core-marginal-weight} gives
    \begin{equation}
        q_0\defeq\pi_{\lambda_0,\eta}(B\neq\varnothing)
        \leq(1+\lambda_0)^{|R|}-1\leq2|R|\lambda_0.
        \label{eq:appendix-fixed-gap-initial-tail}
    \end{equation}
    Conditional on $B=\varnothing$, the law is exactly $\rho_{\rm ref}$.  On
    that event, $\nu_0\geq(1-\theta)\rho_{\rm ref}$, while on its complement
    $\nu_0\geq\theta2^{-N}$.  Therefore
    \begin{equation*}
        D_{\mathrm{KL}}(\pi_{\lambda_0,\eta}\|\nu_0)
        \leq
        \log\frac1{1-\theta}
        +q_0\left(N\log2+\log\frac1\theta\right).
    \end{equation*}
    Equation \eqref{eq:appendix-fixed-gap-initial-tail} shows that the two terms
    sum to at most $\zeta$,
    proving \eqref{eq:fixed-gap-initial-kl}.
\end{proof}

The three preceding proofs also establish all claims in
\zcref{prop:fixed-gap-soft-core-bounds}.

\begin{proof}[Proof of
\zcref{thm:fixed-gap-sampling,thm:fixed-gap-soft-core-sampling}]
    Set $\eta=0$ for \zcref{thm:fixed-gap-sampling}, and fix
    $\eta\in(0,\eta_0)$ for \zcref{thm:fixed-gap-soft-core-sampling}.
    Set $\zeta=\tau^2/16$ and choose $\lambda_0$ according to
    \eqref{eq:fixed-gap-initial-lambda}.  Then
    the initialization calculation above gives
    \begin{equation*}
        D_{\mathrm{KL}}(\pi_{\lambda_0,\eta}\|\nu_0)\leq\tau^2/16.
    \end{equation*}
    Choose
    \begin{equation*}
        m_*=O_{\delta,\eta}\tp{\log(N\mathsf B|R|/\tau)},
        \qquad
        C_*=32N^2m_*^3 2^{m_*}
    \end{equation*}
    so that $\varepsilon_*\leq\tau/4$ in
    \eqref{eq:fixed-gap-weak-action-bound}.  Applying the one-path bound
    \eqref{eq:continuous-weak-action-bound} and the weak-action calculation
    above give
    \begin{equation*}
        D_{\mathrm{TV}}(\pi_{1,\eta},\nu_T)
        \leq\frac\tau4+
        \sqrt{\frac12\left(
          \frac{\tau^2}{16}
          +\frac{C_*|R|^2\mathsf B^2}{4T}
        \right)}.
    \end{equation*}
    Taking
    \begin{equation*}
        T=C\frac{C_*|R|^2\mathsf B^2}{\tau^2}
    \end{equation*}
    for a sufficiently large absolute constant $C$ gives TV error at most
    $\tau$.  The choices above make $\lambda_0^{-1}$, $\mathsf B$, $m_*$, and $C_*$
    polynomially bounded in $N$ and $1/\tau$, so the expected number of
    updates is polynomial.  When $\eta=0$, use
    $\pi_{\lambda,0}=\pi_\lambda$.
\end{proof}

\subsection{Action lower bound on complete bicliques}
\label{subsec:appendix-complete-biclique-action}

We prove \zcref{prop:complete-biclique-hard-core-action}.  Let
$G=K_{D,d}$ with $D/d\to\rho>1$, where $|L|=D$ and $|R|=d$, and write
$\vartheta=\log \lambda$.  Let $\mathfrak a_\vartheta$ be the pointwise action
of the curve $\vartheta\mapsto\pi_{e^\vartheta}$, computed using the heat-bath
capacities, and let
$$
    \mathcal C=\{I\in\mathcal I(G):I\cap L=\varnothing\}.
$$
Every independent set of $K_{D,d}$ lies entirely in one side, so
\begin{equation}
    Z_G(1,\lambda)=2^D+(1+\lambda)^d-1.
    \label{eq:appendix-hard-core-biclique-partition}
\end{equation}
Uniformly for $\vartheta\in[-1/d,0]$, this gives
\begin{equation*}
    Z_G(1,\lambda)=2^D\bigl(1+e^{-\Omega(d)}\bigr),
    \qquad
    \pi_\lambda(\mathcal C)=2^{d-D+O(1)}.
\end{equation*}

The exact derivative identity
$$
    \partial_\vartheta\pi_\lambda(\mathcal C)
    =\pi_\lambda(\mathcal C)
      \left(\frac{d\lambda}{1+\lambda}
      -\mathbb E_{\pi_\lambda}\abs{I\cap R}\right)
$$
and \eqref{eq:appendix-hard-core-biclique-partition} give
$$
    \mathbb E_{\pi_\lambda}\abs{I\cap R}
    =\frac{d\lambda(1+\lambda)^{d-1}}{Z_G(1,\lambda)}
    =d\,e^{-\Omega(d)},
    \qquad
    |\partial_\vartheta\pi_\lambda(\mathcal C)|
    =\Theta\tp{d2^{d-D}}.
$$
The cut can be crossed only from the empty independent set, by adding a vertex
of $L$.  Hence its heat-bath flow is
$$
    Q_\vartheta(\mathcal C,\bar{\mathcal C})
    =\frac{D}{2(D+d)Z_G(1,\lambda)}
    =\Theta\tp{2^{-D}}.
$$
By \zcref{prop:metric-derivative-flux-formulation}, for almost every
$\vartheta$, $\mathfrak a_\vartheta$ is the minimal cost
$\sum_eJ(e)^2/c(e)$ over fluxes $J$ with
$\partial_\vartheta\pi_{e^\vartheta}+\divergence_{\-{d}}J=0$.  Summing this
equation over $\mathcal C$ shows that the net flux across the cut equals
$\partial_\vartheta\pi_\lambda(\mathcal C)$ up to sign.  Since
$Q_\vartheta(\mathcal C,\bar{\mathcal C})$ is the total capacity of the cut
edges, the Cauchy--Schwarz inequality on these edges gives
$$
    \mathfrak a_\vartheta
    \geq
    \frac{|\partial_\vartheta\pi_\lambda(\mathcal C)|^2}
         {Q_\vartheta(\mathcal C,\bar{\mathcal C})}
    \geq d^2 2^{2d-D-O(1)}.
$$
Integrating over an interval of length $1/d$ gives
\begin{equation*}
    \int_{-1/d}^0\mathfrak a_\vartheta\,\dd\vartheta
    \geq d\,2^{2d-D-O(1)}.
\end{equation*}
Along \eqref{eq:fixed-gap-log-path},
$\vartheta_s=\log\lambda_0+s\mathsf B$.  Hence
$\partial_s\pi_{\lambda_s}
=\mathsf B\partial_\vartheta\pi_{e^\vartheta}$ and
$\dd s=\dd\vartheta/\mathsf B$.  Since
$\lambda_0\leq e^{-1/d}$, the full path contains the interval used above, and
therefore
\begin{equation*}
    \+A\tp{\tp{\pi_{\lambda_s}}_{s\in[0,1]}}
    \geq
    \mathsf B d\,2^{2d-D-O(1)},
\end{equation*}
which is \eqref{eq:hard-core-complete-biclique-action}.

\section{Technical Lemmas}

\begin{lemma}
    \label{lem:cosh-quadratic-bound}
    For every \(j,z\in\bb R\) and \(c>0\),
    \[
        jz-2c\tp{\cosh z-1}\leq\frac{j^2}{4c}.
    \]
\end{lemma}

\begin{proof}
    The power-series expansion gives
    \(\cosh z-1-z^2/2=\sum_{k\geq2}z^{2k}/(2k)!\geq0\).  Therefore,
    \[
        jz-2c\tp{\cosh z-1}
        \leq jz-cz^2
        =\frac{j^2}{4c}-c\tp{z-\frac{j}{2c}}^2
        \leq\frac{j^2}{4c}.
    \]
\end{proof}

The following standard moment-generating-function bound is Hoeffding's
lemma; see \cite{Hoe63}.

\begin{lemma}[Hoeffding's lemma]
    \label{lem:hoeffding}
    Let \(Y\) be a random variable such that \(a\leq Y\leq b\) almost surely.
    Then, for every \(\lambda\in\bb R\),
    \[
        \log\bb E\stp{e^{\lambda Y}}
        \leq\lambda\bb E[Y]+\frac{\lambda^2\tp{b-a}^2}{8}.
    \]
\end{lemma}

\end{document}